\documentclass[sn-nature,pdflatex]{sn-jnl}

\usepackage[latin1]{inputenc}
\usepackage{parskip}%
\usepackage{graphicx}%
\usepackage{multirow}%
\usepackage{amsmath,amssymb,amsfonts}%
\usepackage{amscd}%
\usepackage{amsthm}%
\usepackage{mathrsfs}%
\usepackage[title]{appendix}%
\usepackage{xcolor}%
\usepackage{textcomp}%
\usepackage{manyfoot}%
\usepackage{booktabs}%
\usepackage{algorithm}%
\usepackage{algorithmicx}%
	\usepackage{algpseudocode}%
\usepackage{listings}%
\usepackage{dsfont}%
\usepackage{anyfontsize}%
\usepackage{graphicx}%
\usepackage[all,poly]{xy}%

\usepackage{rotating}%
\usepackage{floatpag}
\rotfloatpagestyle{empty} 

\newtheorem{definition}{Definition}[section]%
\newtheorem*{lipschitz*}{Bi-Lipschitz regularity}
\newtheorem{assumption}{Assumption}[section]%
\newtheorem*{redshift*}{The phenomenological redshift parametrization}%
\newtheorem{lemma}{Lemma}[section]
\newtheorem{proposition}{Proposition}[section]

\newtheorem{theorem}{Theorem}[section]

\theoremstyle{definition}

\theoremstyle{remark}
\newtheorem{remark}[theorem]{Remark}
\newtheorem*{appendix*}{\textbf{The role of the Appendix}}
\numberwithin{equation}{section}

\makeindex  

\begin{document}
\title[Article Title]{How Dark is Dark Energy?
A Lightcones Comparison Approach}
\author[1,2]{\fnm{Mauro} \sur{Carfora}}\email{mauro.carfora@unipv.it}
\author[2]{\fnm{Francesca} \sur{Familiari}}\email{francesca.familiari@unipv.it}
\affil[1]{\orgdiv{Department of Physics}, \orgname{Pavia University},\country{Italy}}
\affil[2]{\orgname{National Group for Mathematical Physics (GNFM)}}
\abstract{We present a detailed analysis that may significantly impact understanding the relationship between structure formation in the late--epoch Universe and dark energy as described by the Friedmann--Lema\^itre--Robertson--Walker (FLRW) cosmological constant density 
${\widehat\Omega_\Lambda}$. Our geometrical approach provides a non--perturbative technique that allows the standard FLRW observer to evaluate a measurable, scale--dependent distance functional between her idealized FLRW past light cone and the actual physical past light cone. From the point of view of the FLRW observer, gathering data from sources at cosmological redshift $\widehat{z}$, this functional generates a geometry--structure--growth contribution ${\Omega_\Lambda(\widehat{z})}$ to   ${\widehat\Omega_\Lambda}$. This redshift--dependent contribution erodes the interpretation 
of ${\widehat\Omega_\Lambda}$ as representing constant dark energy. In particular, ${\Omega_\Lambda(\widehat{z})}$ becomes significantly large at very low $\widehat{z}$, where structures dominate the cosmological landscape. At the pivotal galaxy cluster scale, where cosmological expansion decouples from the local gravitation dynamics, we get
${\Omega_\Lambda(\widehat{z})/\widehat\Omega_\Lambda}\,=\,\mathscr{O}(1)$, showing that late--epoch structures provide an effective field contribution to the FLRW cosmological constant that is of the same order of magnitude of its assumed value. We prove that ${\Omega_\Lambda(\widehat{z})}$ is generated by a scale--dependent effective field governed by structures formation and related to the comparison between the idealized FLRW past light cone and the actual physical past light cone. These results are naturally framed in mainstream FLRW cosmology; they do not require the existence of exotic fields and provide a natural setting for analyzing the coincidence problem, leading to an interpretative shift in the current interpretation of constant dark energy.} 
\maketitle
\tableofcontents
\section{Introduction}\label{sec1}
Dark energy, in the form of a cosmological constant $\Lambda$, cold dark matter (CDM), and inflation are the key players in the $\Lambda\mathrm{CDM}$ model, the standard paradigm of modern cosmology. A rich repertoire of hypotheses, modeling assumptions, and observations\cite{Planck} indicate that this dark backstage is responsible for most of the large--scale physics in our Universe, from the early stage of structure formation to the observed late--epoch acceleration of the cosmological expansion\cite{Amendola}. We are familiar with the complex astrophysical landscape that dominates our neighborhood,  yet the Universe appears statistically\cite{SEH} isotropic and homogeneous when observed on large scales. This cosmic view encompasses space and time, and the Universe was indeed extremely uniform, except for tiny density fluctuations when the cosmic microwave background (CMB) radiation formed. Surveys of the CMB indicate that these primordial density fluctuations have a Gaussian \cite{ThomGaussian}, almost scale--invariant power spectrum originating from inflation and suggest that the inhomogeneities around us are the gravitational evolution of these inflationary seeds. We are dealing with a top--down scenario that is difficult to verify in a model--independent way\cite{Maartens}, but which is instrumental in describing cosmological dynamics in terms of a Friedmann--Lema\^itre--Robertson--Walker (FLRW) spacetime (with flat space sections) and its perturbations. The FLRW evolution of the $\Lambda\mathrm{CDM}$ model depends quite rigidly on a minimal set of observational data, the most relevant of which are the present value of the Hubble expansion parameter,  the matter density (comprising baryonic and dark matter),\; the radiation density, the primordial fluctuation amplitude related to inflation, and the scalar spectral index characterizing how fluctuations change with the scale. \; From the Friedmann equations, derived from Einstein's theory, we eventually get also the presence of a dark energy contribution in the vague terms of a cosmological constant.\;  Observations show that, at the present epoch, the estimated value of this dark energy contribution is dominant over the other components, a dominance that we see at work in the actual accelerated phase of the cosmological expansion. This narrative is the $\Lambda\mathrm{CDM}$ concordance model of cosmology\footnote{To set notation, a capsule of the $\Lambda\mathrm{CDM}$ model is presented in Appendix \ref{capsule}.}. Simple, somewhat vague, surprisingly predictive\cite{Planck}. Yet, the simplicity of this description is deceptive since we live in a perturbed universe that bears, in a rather unexpected and subtle way, the imprint of the fluctuations induced by gravitational instability. The successful role of  FLRW spacetime geometry and its perturbations
in the early--epoch $\Lambda\mathrm{CDM}$ scenario, is related to the observed transition to homogeneity\cite{Scrimgeour, Dias}  that occur for comoving radii of the order $100\mathrm{h}^{-1}\; \mathrm{Mpc}$, where $\mathrm{h}$ is the dimensionless parameter describing the relative uncertainty of the present value, $H_0$, of the Hubble parameter.
It is essential to emphasize that the characterization of this homogeneity scale relies on the statistical measure employed, typically two--point correlation function statistics that utilize an FLRW background as a reference. Model--independent statistics, which encompass all orders of the correlation function, promote the transition to homogeneity at larger scales \cite{ThomMinkowski}. Here, we do not belabor on this delicate issue, and for simplicity, we use the standard reference value $100\mathrm{h}^{-1}\; \mathrm{Mpc}$ adopted in mainstream cosmology.

In the late--epoch Universe, the role of FLRW geometry in the $\Lambda\mathrm{CDM}$ model is more delicate to handle. Gravitational clustering generates fluctuations that grow larger and larger and eventually become non--perturbative. Observations show that, in the pre--homogeneity region surrounding us, this clustering gives rise to a complex network of structures, characterized by a foam--like web of voids and galaxy filaments often extending well into the  $100\; \mathrm{Mpc}$ range\cite{Wiltshire}. At these scales,  the Einstein evolution of the fiducial  FLRW geometry uncouples from the dynamics of the matter sources and survives more as a practical computational assumption, often only assisted by  Newtonian and general--relativistic (model--dependent) numerical simulations, rather than as a  \emph{bona fide} perturbative background gravitationally determined by the actual matter distribution as  Einstein's theory requires. We have little perturbative control over spacetime geometry in such pre--homogeneity regions. In particular,  the transition from their large--scale FLRW point of view to the actual inhomogeneous and anisotropic spacetime geometry emergent at these local scales is poorly understood in a model--independent way, and the idea that around $100\mathrm{h}^{-1}\; \mathrm{Mpc}$ we have a gradual and smooth transition between these two regimes is somewhat illusive (J-P. Uzan provides a detailed analysis of this issue in the remarkable \cite{Uzan}). The existence of this issue was presciently put forward long ago by G. F. R. Ellis \cite{EllisPadova} and explored in depth by T. Buchert\cite{DarkBuchert, BuchertSyksy} (there is a vast literature on the subject, see\cite{Ellis2} for a review). Ellis suggested that non--perturbative inhomogeneities can generate backreaction effects due to the nonlinear nature of Einstein equations. These effects can contribute to the overall FLRW dynamics by modifying the cosmological parameters and even motivate a shift in perspective on the nature of dark energy\cite{DarkBuchert, Carfora0, Carfora1}. The lively debate raised by backreaction\cite{BuchertSyksy, Ellis2, Buchert, Wald1, Wald2, Wald3}, and an in--depth analysis of nonlinear FLRW perturbative effects\cite{Durrer, Fanizza, Heinesen, Obinna4} showed that backreaction could affect cosmological parameters up to the percent level, and it has a role in high--precision cosmology. But what about the impact of inhomogeneities \cite{Durrer} and backreaction on dark energy? Evidence in this sense is not conclusive. Perturbative results and averaging techniques leave the solution of the problem in an unclear state. Observations indicate that the expansion of the Universe appears to have started accelerating at the same epoch when complex, nonlinear structures emerged, a dynamical evolution that is technically challenging to control with the available backreaction and perturbative methods.
Recent observational analyses further indicate that dark energy might not be a strictly constant component. In particular, the recent DESI \cite{DESI1} and DES--SN \cite{DESI2} data have reported evidence that the dark energy equation of state parameter $w(z)$ may evolve with redshift, suggesting a phantom crossing at low redshift $z$. This renewed observational trend reopens the debate on the dynamical nature of dark energy and on the potential role of inhomogeneities in shaping its apparent evolution, highlighting the need for non--perturbative approaches able to connect structure formation and the effective cosmological constant. In this complex landscape, the primary motivation for our paper is to illustrate a new non--perturbative approach and a result that may significantly impact understanding the relationship between dark energy, as described by the FLRW cosmological constant $\widehat{\Lambda}^{(FLRW)}$, and structure formation in the late--epoch Universe.

\begin{appendix*}
The results discussed in our paper rely on the thorough characterization, in the non--smooth setting induced by strong gravitational clustering, of the properties of the observer's celestial spheres and their associated past light cones. The natural mathematical framework that is invoked in such cases is that of Lipschitz analysis and Hausdorff measures. The corresponding definitions of the area and angular diameter distances, the study of their non--perturbative fluctuations, and finally, the comparison with the idealized FLRW point of view all require a mathematical sophistication that is hardly discussed in the current cosmology literature. Thus, for the reader's convenience, the paper is accompanied by a lengthy and detailed Appendix, which collects, in a pedagogical form, definitions, and proofs of the more technical aspects of our results. The text explicitly references this material by using the header "Appendix" followed by a section or equation number. 
\end{appendix*}

\section{Motivation and a new best fit FLRW strategy} 
 Let us recall (see Appendix \ref{capsule} for notation) that with the FLRW observer, we 	can associate  the  dimensionless density parameters (we use units with $G=1=c$)
\begin{equation} 
\label{cosmoparameters}
\widehat\Omega_m(\widehat{z})\,:=\,\frac{8\pi\,\widehat\rho(\widehat{z})}{3\,H^2(\widehat{z})},\;\;\;\widehat\Omega_\Lambda^{(FLRW)}(\widehat{z})\,:=\,\frac{\widehat{\Lambda}^{(FLRW)}}{3\,H^2(\widehat{z})},\;\;\; \widehat\Omega_k(\widehat{z})\,:=\,-\,\frac{k}{a^2(\tau)\,H^2(\widehat{z})},
\end{equation}
respectively associated, at the given redshift $\widehat{z}$, with the FLRW matter density 
$\widehat\rho$, cosmological constant $\widehat{\Lambda}^{(FLRW)}$, and spatial curvature $k/a^2(\tau)$ of the comoving slicing, where $a(\tau)$ is the time--dependent FLRW scale factor, and $k$ is the spatial curvature constant. At the present epoch, observations (from the CMB data) are consistent with spatial near 
flatness, so the FLRW observer can set $\widehat\Omega_k\,\simeq\, 0$ (up to local perturbative curvature fluctuations), and reduce the curvature Friedmann equation (in the form of the Bahcall cosmic triangle\cite{Bahcall}) to the relation   
\begin{equation}
\label{triangleB}
\widehat\Omega_m\,+\,\widehat\Omega_\Lambda\,\simeq\,1\,,
\end{equation}
where the cosmological density parameters refer to the present--day values. The evidence for an accelerating Universe derived from supernova data\cite{Riess, Garnavich, Perlmutter}, and supernovae SNeIa compilations, (\emph{e.g.}, Pantheon\cite{Scolnic}) provide the best--fit values, $\widehat\Omega_m\,\simeq\,0.3$ and $\widehat\Omega_\Lambda\,\simeq\,0.7$, indicating that the cosmological constant interpretation of the dark energy is a good match to existing data from the point of view of the FLRW observer (for an assumed spatially flat $\Lambda\mathrm{CDM}$ model). It is worthwhile to stress that whereas $\widehat\Omega_m$ represents the contribution to the energy density of all matter and fields present (but not from $\Lambda^{(FLRW)}$), the interpretation of the present--day cosmological constant density $\widehat\Omega_\Lambda$ is typically, but not necessarily, associated with \emph{constant} dark energy. Attempts to establish the presence of a partial contribution due to density fluctuations generated by the presence of late--epoch structures (\emph{i.e.}, a geometry--structure--growth contribution of the type often used in cosmological model building in the presence of structure formations \cite{Huterer}) have been confined to the role of a \emph{nuisance parameter} responsible for percent level corrections. Yet, a careful geometrical analysis of the interplay between  
the reference role of the FLRW observer and the actual physical observer shows that the overall picture is more complex, leaving room for an explanation of the $\Lambda\mathrm{CDM}$  accelerating cosmic expansion as driven by an effective scalar field connected with the geometry of the physical past light cone as described by the reference FLRW observer. To enforce this non--perturbative FLRW description, we prove that the past light cones of the FLRW observer and the physical observer can be optimally compared at a given scale. The comparison is a consequence of a rigorous geometrical result\cite{CarFam} that allows the FLRW observer to evaluate a measurable scale--dependent functional distance between the celestial sphere of her idealized past light cone and the celestial sphere of the physical past light cone\footnote{This functional and its properties are discussed in full detail in Appendix \ref{TSKSM} ff.}. This functional is related to harmonic map theory, and it can be expressed in terms of the mean--square quadratic fluctuations of the physical area distance with respect to the reference FLRW area distance. It can be 	optimized with respect to the choice of the FLRW observer and the standard candles used to set the given observational scale, providing a scale--dependent \emph{best fit} between the physical and the FLRW past lightcones. 
 By exploiting the Lorentzian version of the classical Bertrand--Puiseaux formulas\cite{Berthiere} (familiar in Riemannian geometry when defining curvatures in terms of the surface areas and volumes of geodesic balls--see (\ref{Areaasymp1})), we show that the distance functional can be directly related to the spacetime scalar curvature, and thus, to the presence of local structures at a given FLRW redshift $\widehat{z}$. This relation provides a non--perturbative fractional contribution, generated by the presence of structures, that spits the  FLRW cosmological constant density $\widehat\Omega_\Lambda$. This splitting complies with the Bahcall triangle relation (\ref{triangleB}) and does not alter the present--day estimated values of the cosmological parameters (\ref{cosmoparameters}). Yet, it erodes the strict correlation between $\Lambda^{{(FLRW)}}$ and the interpretation of dark energy as a constant field since the fractional contribution to $\widehat\Omega_\Lambda$ due to structure formation becomes significantly large at very low $\widehat{z}$, where structures dominate the cosmological landscape.
We discuss what happens at the pivotal scale where cosmological expansion decouples from the local virialized dynamics of gravitational structures when we approach, for decreasing values of $\widehat{z}$, the outskirts of galaxy clusters. In this case, the fractional contribution to ${\widehat\Omega_\Lambda}$ is $\mathcal{O}(1)$, suggesting a significant role of the geometry--structure--growth term in the interpretation of dark energy as a dynamical field.

\section{A tale of two observers} 
According to the best--fit point of view advocated long ago\cite{EllisStoeger} by G. F. R. Ellis and W. R. Stoeger, the actual Universe should be modeled by a physical spacetime manifold $M$ with a Lorentzian metric $g$ that is statistically isotropic and homogeneous, \emph{viz.}\, a perturbed FLRW,  only on scales larger than the scale $100\mathrm{h}^{-1}\; \mathrm{Mpc}$ that marks the transition to homogeneity\cite{Gott, Hogg, Scrimgeour}.
More explicitly (see Definition \ref{PhenDefin} in Appendix \ref{avatar}), we  assume that $(M, g)$  is associated with the evolution of a universe that is (statistically) isotropic and homogeneous on sufficiently large scales
 $L\,>\,L_0\, \cong\, 100h^{-1}\; \mathrm{Mpc}$ and let local inhomogeneities dominate for $L\,<\,L_0$.    We also assume that in $(M, g)$, the motion of the cosmic structures characterizes a \emph{phenomenological Hubble flow} that generates a family of preferred geodesic worldlines $\mathbb{R}_{>0}\,\ni\tau\,\longmapsto\,\gamma_s(\tau)\,\in\,(M, g)$   parametrized by proper time $\tau$, and labeled by suitable comoving (Lagrangian) coordinates  $s$ adapted to the flow.
 We denote by $\dot{\gamma}_s\,:=\,\frac{d\gamma_s(\tau)}{d\tau}$,\,\,$g(\dot\gamma_s,  \dot\gamma_s)\,=\,-1$, the corresponding  $4$--velocity field.  
Since in our analysis, we fix our attention on a given observer, we drop the subscript $s$  and describe a finite portion of the observer's worldline with the timelike geodesic segment  $\tau\,\longmapsto\,\gamma(\tau)$,\,$-\delta<\tau<\delta$, \, for some $\delta>0$, \,\, where  $p\,:=\,\gamma(\tau=0)$ is the selected event corresponding to which the cosmological data are gathered. To organize and describe these data in the local rest frame of the observer $p\,:=\,\gamma(\tau=0)$, let $\left(T_p M,\,g_p,\,\{E_{(i)}\}\right)$ be the tangent space to $M$ at $p$ endowed with a $g$--orthonormal frame $\{E_{(i)}\}_{i=1,\ldots,4}$,\,\,$g_p\left(E_{(i)}, E_{(k)}\right)=\eta_{ik}$, where $\eta_{ik}$ is the Minkowski metric, and where we identify $E_{(4)}$ with the observer $4$--velocity $\dot{\gamma}(\tau)|_{\tau=0}$. It is worthwhile to stress that in the  Ellis and Stoeger paper, three possible best--fitting methods were qualitatively described, one of which involves fitting the null data, the subject of the current article. Attempts to fit the null data using observable coordinates are far from being new, and \cite{EllisPhRep} can still be considered as a foundational paper shaping the coming era of high--precision cosmology and promoting interest in the structure of our past light cone\cite{Fanizza}, \cite{VanElst}, \cite{Veneziano}, \cite{Adamek}. In particular, to whatever degree one is willing to accept the foundational role of FLRW perturbation theory in the $\Lambda\mathrm{CDM}$ model, we must face the fact that  
the analysis of the past light cone data and the best--fit point of view emphasize that only on large scales\footnote{It is important to stress that a large--scale reference
model may not even be a solution of General Relativity. An example is the
spacetime--averaged light cone that can be (almost) isotropic when averaged over
large scales, while not being dynamically characterized by the Friedmannian scale
factor \cite{VanElst}.} the physical spacetime $(M, g)$ can be (perturbatively) described in terms of a background FLRW spacetime $(M, \widehat{g})$. On smaller scales and late--epoch, the description of $(M, g)$ is more complex  since, according to the Einstein equations, the spacetime geometry must comply with the network of structures that dominate the cosmological landscape for  
$\lesssim\,100\mathrm{h}^{-1}\; \mathrm{Mpc}$. In this pre--homogeneity region, the typical observer is associated with an event $p$, the \emph{here and now} corresponding to which she gathers cosmological data. In practice, the \emph{here} can be identified with the barycenter of our small cluster of galaxies (the Local Group), and the \emph{now} is characterized by the actual temperature $T_{CMB}\,=\,2.725$ of the cosmic microwave background (CMB) as measured in the frame centered on us but stationary with respect to the CMB.   The interpretation of data in the pre--homogeneity region surrounding us is complex since the FLRW model should be used only over the largest scales. Yet, FLRW observers and the associated FLRW spacetime $(M, \widehat{g})$ have a natural role also in the presence of inhomogeneities as advocated by E. Kolb, V. Marra, and S. Matarrese in\cite{KolbMarraMatarrese}, and nicely put into work in their inspiring paper \cite{KolbMarraMatarrese2}.
Thus, together with the physical observer, it is natural to consider her Friedmannian avatar, a reference FLRW observer located at the same observational event $p$ and possibly moving, with respect to the physical observer, with a relative velocity $v(p)$. The two (instantaneous) observers may compare the geometry that astrophysical data induce on the celestial spheres associated with the physical past light cone, $\mathcal{C}^-(p, g)$, and the FLRW past light cone, $\mathcal{C}^-(p, \widehat{g})$. It is important to stress that this FLRW instantaneous observer  strictly adopts  the weak cosmological principle, a statistical rendering of the Copernican point of view  originated in the work of Jerzy Neyman\cite{Neyman} and emphasized by 
P.J.E. Peebles\cite{Peebles}.   In particular, in her reference role, she interprets the dynamics of the Universe as a realization of a homogeneous and isotropic random process on the FLRW background \cite {Peebles}, where the Binomial, Poisson, and Gaussian distributions govern fluctuations. This idealized point of view\cite{SEH} is clearly at variance with respect to the skewed distributions of sources expected to dominate in the pre--homogeneity region\cite{Adamek}, \cite{Fanizza2} and which offer a perspective that may significantly deviate from the reference FLRW picture, a further reason to introduce and keep distinct the roles of the physical observer and her FLRW avatar, and of their associated celestial spheres at a given redshift\footnote{These celestial spheres are described in detail in Appendices \ref{CelSphs} and \ref{skysezioni}.}. It is also important to note that the distance--redshift relation is complicated to control in terms of the physical redshift $z$ in the pre--homogeneity region. In this region, it is challenging to characterize a sensible notion of observed Hubble flow, and we must resort to the FLRW redshift $\widehat{z}$ and to the introduction of the peculiar velocities with respect to the underlying reference FLRW Hubble flow to parametrize the physical redshift $z$ of the sources. We emphasize that we will consider the FLRW observer in her FLRW reference frame modeling role, as is typically done in the $\Lambda\mathrm{CDM}$ model, where the weak cosmological principle is enforced at every scale to exploit FLRW perturbation theory. Thus, at very low redshift, where the FLRW model is problematic, the FLRW observer is confined to the reference role of providing a background against which we can compare the more realistic modeling described by the physical observer. These remarks on the FLRW observer's reference role in the pre--homogeneity region are critical to ensuring the nature of our result is understood.

To put the FLRW reference role described above into practice, we need, even at this expository level, some notation whose details will be streamlined in Appendices \ref{CelSphs} and \ref{skysezioni}. We begin by introducing the past null cone in the observer's local rest frame $(T_pM,\,\{E_{(i)}\})$, 
\begin{align}
\label{MinkLightCone}
C^-\left(T_pM,\,\{E_{(i)}\} \right)\,&:=\,\left\{X\,=\,\mathbb{X}^iE_{(i)}\,\not=\,0\,\in\,T_pM\,\,|\right.\\
&\left.\,g_p(X, X)\,=\,0,\,\,\mathbb{X}^4+r=0 \right\}\;,\nonumber
\end{align}
where 
$r:=(\sum_{a=1}^3(\mathbb{X}^a)^2)^{1/2}$ is the  comoving radius in  the hyperplane 
$\mathbb{X}^4\,=\,0\,\subset\,T_pM$. Let us stress that $C^-\left(T_pM,\,\{E_{(i)}\} \right)$ is not to be confused with the past lightcone $\mathcal{C}^-(p,g)$ with vertex at $p$, the set of all events $q\in (M, g)$ that can be reached along the past--pointing null geodesics issued from $p$, \emph{viz.}
\begin{equation}
\mathscr{C}^-(p,g)\,:=\,\exp_p\left[C^-\left(T_pM, g_p \right)\right]\,,
\end{equation}
where
\begin{equation}
\label{expmapdef0}
\left(T_pM, g_p \right)\,\ni\,X\,\longmapsto\,exp_p\,(X)\,:=\,\lambda_X(1)\,\in\,(M, g)
\end{equation}
is the exponential map at the observation event $p$, and where
$\lambda_X\,:\,I_W\,\longrightarrow\,(M, g)$, for some maximal interval $I_W\subseteq\mathbb{R}_{\geq0}$,  is the past--directed causal geodesic emanating from the point $p$ with initial tangent vector $\dot{\lambda}_X(0)\,=\, X\in T_pM$. We define in a  similar way the above quantities for the reference FLRW spacetime $(M, \widehat{g})$ and respectively denote by $\widehat{C}^-(T_pM,\,\{\widehat{E}_{(a)}\})$,\,\,$\widehat{r}$,\,
$\widehat{\mathscr{C}}^-(p,\widehat{g})$,\,\, and \, $\widehat{\exp}_p$\, the corresponding FLRW past null cone, comoving radius, past lightcone, and exponential map. 
Let $\widehat{\lambda}(\widehat{q})\,=\,c/\widehat{\nu}(\widehat{q})$ denote the wavelength of a signal emitted, at the event $\widehat{q}\in\,\widehat{\mathscr{C}}^-(p,\widehat{g})$,  by a source instantaneously at rest at $q$ with respect to an FLRW fundamental observer $(\widehat{q},\,\widehat{\dot{\gamma}}(\widehat{q}))$. If  $\widehat{\lambda}(p)\,=\,c/\widehat{\nu}(p)$ is the corresponding  wavelength received by the FLRW observer $(p,\,\widehat{\dot{\gamma}}(p))$ at $p$, we get the familiar expression of the FLRW cosmological redshift $\widehat{z}(\widehat{q})\,:=\,(\widehat{\lambda}(p)\,-\,\widehat{\lambda}(\widehat{q}))/\widehat{\lambda}(\widehat{q})$ of the emitting source $(\widehat{q},\,\dot{\gamma}(\widehat{q}))$.
We can measure redshifts accurately even for extremely faint galaxies, and in a regime where one has a linear Hubble (cor)relation between expansion velocity and distance, redshifts can be used as a distance indicator. More generally, we have an explicit relation between the comoving radius $\widehat{r}$ (which is not an observable quantity), the cosmological redshift $\widehat{z}$, and the (FLRW) Hubble parameter $H(\widehat{z})$, given by $\widehat{r}(\widehat{z})\,=\,\int_0^{\widehat{z}}\,\frac{d\widehat{z}\,'}{H(\widehat{z}\,')}$.
 The use of redshift by the FLRW observer as a distance indicator of sorts, implicit in the above remarks, is quite delicate for the physical observer $(p, \dot\gamma(0))$ since the knowledge of the corresponding distance--redshift relation is tantamount knowing the actual cosmological geometry in the pre--homogeneity region. In this region, the motion of cosmic structures hardly follows a linear Hubble relation, as galaxies gravitationally interact with each other. Their motion is characterized by peculiar velocities, not due to the expansion of the Universe, which can be\cite{TMDavies} of the order of $100--300\,\mathrm{Km}/\mathrm{sec}$, a significant fraction of the recession velocity over the scales that characterize the pre--homogeneity region. This fraction becomes less significant for distant galaxies, and one typically assumes that peculiar velocities can be ignored on the largest scales, where homogeneity prevails. Conversely, if 
$\widehat{z}\,(c)$ denote the reference FLRW redshift marking the decoupling of the phenomenological Hubble flow from the cosmological expansion (\emph{i.e.}, the redshift formally associated with the observer's cluster radius), then peculiar velocities becomes dominant. According to this latter remark and the FLRW best--fit strategy, it makes sense, as detailed in Appendix \ref{RefCosmRed} (see Definition \ref{RefRed}), to use the candidate FLRW spacetime as providing the cosmological expansion reference standard in the pre--homogeneity region only for $\widehat{z}\,>\,\widehat{z}\,(c)$. Explicitly, if $v_{\hat{z}}$ is the relative velocity of the physical observer $(p, \dot\gamma(0))$ with respect to the FLRW observer  $(p, \widehat{\dot\gamma}(0))$ and 
if ${z}(q):=\,\tfrac{{\lambda}(p)\,-\,{\lambda}({q})}{{\lambda}({q})}$ denote the observed redshift of a signal emitted at the event ${q}\in\,{\mathscr{C}}^-(p,{g})$ by a source in the physical spacetime $(M, g,\,\dot\gamma(\tau))$, then we can factorize, for $\widehat{z}\,>\,\widehat{z}\,(c)$,  the physical redshift in the pre--homogeneity region  according to 
\begin{equation}
\label{zetafactor0}
\left(1\,+\,{z}(q)\right)\,=\,\sqrt{\frac{1+v_{\hat{z}}}{1-v_{\hat{z}}}}\left(1\,+\,\widehat{z}(q)-\widehat{z}\,(c)\right)\left(1+\widehat{z}^{(doppl)}(q)\right)\left(1+\widehat{z}^{(grav)}(q)\right)\,,
\end{equation}
where $\widehat{z}(q)$ is the FLRW reference cosmological redshift of $q$, \; the redshift $\widehat{z}^{(doppl)}(q)$ comes from the peculiar velocity of the given source $q$ with respect to the reference FLRW Hubble flow;\; finally, the redshift $\widehat{z}^{(grav)}(q)$ comes from the presence of strong local gravitational fields affecting null geodesics propagation from $q$ to $p$. 
Note that these peculiar velocities not only comprise the  CMB dipole anisotropy induced by our solar system motion\footnote{This motion is characterized by the velocity  $\simeq 370\; km\,s^{-1}$ along the direction defined by the celestial coordinates $\mathrm{RA}\,=\,168^{\circ}$, and $\mathrm{Dec}\,=\, -7^{\circ}$.}, a contribution that can be easily factorized out in (\ref{zetafactor0}),  but also the dipole--like effect due to the peculiar motions of local galaxies and cluster of galaxies in the pre--homogeneity zone. These are the only terms that characterize ${z}(q)$ when the source is within the cosmological decoupled region. 
It is not easy to resolve these latter contributions in (\ref{zetafactor0}), even more so if these motions are generated by local isotropic over--densities or under--densities in the pre--homogeneity region or by the coherent motion of sources attracted by nearby large--scale structures.

The parametrization (\ref{zetafactor0}) is a delicate issue in any best--fit strategy. To take care of it,  we fix, for $\widehat{z}\,>\,\widehat{z}\,(c)$,  the reference FLRW redshift $\widehat{z}$  and adjust the relative velocity $v_{\hat{z}}$ (and the associated spatial orientation) of the physical $(p,\,{\dot{\gamma}}(p),\,\{{E}_{(\kappa)}\}_{\kappa=1,2,3})$ and of the FLRW observer $(p,\,\widehat{\dot{\gamma}}(p),\,\{\widehat{E}_{(\kappa)}\}_{\kappa=1,2,3})$ at $p$ in such a way to optimize, in a sense that we discuss in the following sections, the correspondence between the reference FLRW lightcone $\widehat{\mathscr{C}}^-(p,\widehat{g})$ and the physical lightcone ${\mathscr{C}}^-(p,{g})$ at the FLRW reference scale set by the chosen $\widehat{z}$.

\section{Celestial spheres and celestial cartography} 
The celestial sphere at redshift $\widehat{z}$ associated with the FLRW observer
$(p,\, \widehat{\dot\gamma}(0),\,\{\widehat{E}_{(\kappa)}\}_{\kappa=1,2,3})$ is defined by the $2$-sphere  (see Appendix \ref{CelSphs})
\begin{equation}
\label{celestialS00}
\widehat{\mathbb{CS}}_{\hat{z}}\,:=\,\left\{\widehat{X}\,=\,\widehat{\mathbb{X}}^i\widehat{E}_{(i)}\,\not=\,0\,\in\,T_pM\,\,|\,\,\widehat{\mathbb{X}}^4=0,\,\,\sum_{a=1}^3(\widehat{\mathbb{X}}^a)^2=\,\widehat{r}^2(\widehat{z}) \right\}\,,
\end{equation}
characterizing  the field of vision of the observer $(p,\, \widehat{\dot\gamma}(0),\,\{\widehat{E}_{(\kappa)}\}_{\kappa=1,2,3})$ at the given cosmological redshift $\widehat{z}$. In the sense described by R. Penrose  \cite{RindPen},  this is a representation of 
the abstract 2--sphere $\mathcal{S}^-(p)$ of past null directions parameterizing, in the reference FLRW spacetime $(M, \widehat{g})$, the past--directed null geodesics through $p$. In order to define a corresponding celestial sphere\; ${\mathbb{CS}}_{\hat{z}}$ associated with the physical observer $(p,\, {\dot\gamma}(0),\,\{{E}_{(\kappa)}\}_{\kappa=1,2,3})$ who wish to use the given  FLRW redshift $\widehat{z}$\, as a reference, let
$v_{\hat{z}}$ denote the relative $3$--velocity of the physical $(p,\,{\dot{\gamma}}(p))$ and of the FLRW observer $(p,\,\widehat{\dot{\gamma}}(p))$. This relative velocity implies that the celestial spheres $\widehat{\mathbb{CS}}_{\hat{z}}$ and ${\mathbb{CS}}_{\hat{z}}$ are related by  
 a Lorentz transformation that we can realize as a M\"obius transformation $\zeta_{(\widehat{z})}$ in the projective special linear group $\mathrm{PSL}(2,\mathbb{C})$ (the group of automorphisms of the Riemann sphere $\mathbb{S}^2\simeq\,\mathbb{C}\,\cup\,\{\infty\}$). Explicitly, we have (see Proposition \ref{PSL2mappingFLRW} in Appendix \ref{PSELLEDUE})
\begin{equation}
\label{pesse2}
\mathrm{PSL}(2, \mathbb{C})\times\widehat{\mathbb{C\,S}}_{\hat{z}}\,\ni\,\left( \zeta_{(\widehat{z})},\,\widehat{w}\right)\,\longmapsto\,\zeta_{(\widehat{z})}(\widehat{w})\,=
\,w\,=\,\sqrt{\frac{1\,+\,v_{\hat{z}}}{1\,-\,v_{\hat{z}}}}\,e^{i\,\alpha_{\hat{z}}}\,\widehat{w}\,\in\,{\mathbb{C\,S}}_{\hat{z}}\,,
\end{equation}
where $e^{i\,\alpha_{\,\widehat{z}}}$ is a rotation through an angle $\alpha$ about the $E^3$ direction, and $\widehat{w}\,:=\,(\widehat{\mathbb{X}}^1+i\,\widehat{\mathbb{X}}^2)/(1-\widehat{\mathbb{X}}^3)$ are directional coordinates of the generic point on the FLRW celestial sphere $\widehat{\mathbb{C\,S}}$.   The corresponding directional coordinates on the physical celestial sphere 
${\mathbb{C\,S}}_{\,\widehat{z}}$ are similarly defined by 
$\zeta_{(\widehat{z})}(\widehat{w})\,=\,
{w}\,:=\,({\mathbb{X}}^1+i\,{\mathbb{X}}^2)/(1-{\mathbb{X}}^3)$.
From the physical point of view, (\ref{pesse2}) corresponds to the composition 
of  the relative orientation of the spatial bases $\{{E}_{(\alpha)}\}$ with respect to the FLRW  $\{\widehat{{E}}_{(\alpha)}\}$,\;$\alpha=1,2,3$,\; followed by
a Lorentz boost with rapidity $\beta_{\hat{z}}\,:=\,\log\,\sqrt{\frac{1\,+\,v_{\hat{z}}}{1\,-\,v_{\hat{z}}}}$ along the $E^3$ axis, adjusting for the relative velocity of $(p, \dot{\gamma}(0))$ with respect to the FLRW $(p, \widehat{\dot{\gamma}}(0))$. The map (\ref{pesse2}) is fully 
determined by the images, on $\mathbb{C\,S}$, of three distinct astrophysical sources of choice (\emph{e.g.} Cepheid variable stars) on the FLRW celestial sphere $\widehat{\mathbb{C\,S}}$. This remark is simply the statement that, by adjusting her velocity and orientation, the FLRW observer can associate three specific positions $\{\widehat{w}_{(j)}\}$,\,\,$j=1,2,3$,\, on her reference $\widehat{\mathbb{C\,S}}$ with three (or less) given sources on the physical celestial sphere $\mathbb{C\,S}$. 
The aberration map (\ref{pesse2}) allows us to define, at the reference redshift $\widehat{z}$, the physical celestial sphere ${\mathbb{C\,S}}_{\,\widehat{z}}$ in the instantaneous rest space of the physical observer, $(T_pM,\,\{E_{(i)}\})$ as the image,
${\mathbb{C\,S}}_{\hat{z}}\,=\,\zeta_{(\widehat{z})}\,\left(\widehat{\mathbb{C\,S}}_{\hat{z}}\right)$, under the local aberration map $\zeta_{(\widehat{z})}$ of the FLRW celestial sphere $\widehat{\mathbb{C\,S}}_{\hat{z}}$, \emph{i.e.},
\begin{equation}
\label{celestialS0}
{\mathbb{CS}}_{\hat{z}}\,:=\,\left\{X\,=\,\mathbb{X}^iE_{(i)}\,\not=\,0\,\in\,T_pM\,\,|\,\,\mathbb{X}^4=0,\,\,\sum_{a=1}^3(\mathbb{X}^a)^2= 
\left(\frac{1\,+\,v_{\hat{z}}}{1\,-\,v_{\hat{z}}}\right)\,\widehat{r}^2(\widehat{z})
\right\}\,.
\end{equation}
Note that $\sqrt{(1\,+\,v_{\hat{z}})/(1\,-\,v_{\hat{z}})}\,\widehat{r}(\widehat{z})$ characterizes the ($v_{\hat{z}}$--Lorentz boosted) reference comoving radius ${r}(\widehat{z})$ adopted by the physical observer using the reference FLRW redshift $\widehat{z}$. The images of $\mathbb{CS}_{\hat{z}}$ and $\widehat{\mathbb{CS}}_{\hat{z}}$ under the corresponding exponential maps define the physical and FLRW cosmological sky sections\footnote{The geometry of the sky sections and their properties are discussed in Appendix \ref{skysezioni}.}   
\begin{equation}
\label{sigmapr0}
\Sigma_{\hat{z}}\,:=\,\exp_p\left[\mathbb{C\,S}_{\hat{z}} \right]\,\subset\,\mathscr{C}^-(p,g),\;\;\;\;\;
\widehat{\Sigma}_{\hat{z}}\,:=\,\widehat{\exp}_p\left[\widehat{\mathbb{C\,S}}_{\hat{z}} \right]\,\subset\,\widehat{\mathscr{C}}^-(p,\widehat{g})\,.
\end{equation}
These surfaces are orthogonal to the null geodesics generators of the past lightcones $\mathscr{C}^-(p,g)$ and 
$\widehat{\mathscr{C}}^-(p,\widehat{g})$. They carry the respective 2--dimensional metrics induced by the corresponding lightcones, and we write $(\Sigma_{\hat{z}},\,g^{(2)}_{\hat{z}})$ and $(\widehat{\Sigma}_{\hat{z}},\,\widehat{g}^{(2)}_{\hat{z}})$.
As  ${\mathbb{C\,S}}_{\hat{z}}$ and 
$\widehat{\mathbb{C\,S}}_{\hat{z}}$ vary on the associated null cones 
$C^-(T_pM, \{{E}_{(i)}\})$ and $C^-(T_pM, \{\widehat{E}_{(i)}\})$, the corresponding sky sections $(\Sigma_{\hat{z}},\,g^{(2)}_{\hat{z}})$ and $(\widehat{\Sigma}_{\hat{z}},\,\widehat{g}^{(2)}_{\hat{z}})$ foliate the physical and the FLRW past lightcones. Our best--fit strategy is to carefully exploit these past null cones and lightcone foliations to probe and compare the physical and reference FLRW spacetimes in the pre--homogeneity region. In this connection, it is important to stress that besides providing the apparent location of sources in the observer's sky at the given FLRW  redshift $\widehat{z}$, the celestial spheres ${\mathbb{C\,S}}_{\hat{z}}$ and $\widehat{\mathbb{C\,S}}_{\hat{z}}$ can provide a more accurate rendering then simple directional cartography.

\section{The Method: Comparing celestial spheres} 

Data from astrophysical sources at a given reference FLRW redshift $\widehat{z}$ 
reach the physical observer and are portrayed on her celestial sphere $\mathbb{CS}_{\hat{z}}$, by traveling along null geodesics. Physical null geodesics are not smooth since they may develop conjugate and cut locus points with the ensuing formations of past light cone caustics; see Proposition \ref{PropPhysLightCone}, Appendix \ref{LipNatSky}. The presence of multiple images of the same astrophysical source on the observer's celestial sphere, gravitational lensing\cite{Perlick}, and significant image distortions offer a dramatic demonstration of this behavior 
and also show that data gathered on $\mathbb{CS}_{\hat{z}}$ provide explicit geometric information.
These data are encoded on the physical sky section $(\Sigma_{\hat{z}},\,g^{(2)}_{\hat{z}})$, and are transferred on the physical observer's celestial sphere $\mathbb{CS}_{\hat{z}}$ by the pull--back action associated with the exponential map $\exp_p$. As stressed above, $\exp_p$ is not a smooth map; magnification,  distortions, and multiple imaging are the rule rather than the exception, and the information null geodesics convey to the physical celestial sphere $\mathbb{CS}_{\hat{z}}$ is quite corrupted by the presence of the gravitational inhomogeneities that they encounter evolving from the astrophysical sources to the observer located at $p$. Since we are not interested in the details and nature of the specific strong lensing event, one can safely make the working assumption (see \ref{Lipassumpt1}) that the physical past--directed null geodesics are bi--Lipschitz. Roughly speaking, this technical assumption sacrifices the safe harbor of smoothness in favor of more effective control of how null geodesics handle the distance and area measures\footnote{The Lipschitz regularity for the null geodesics flow is better suited to defining the tools of the trade of observational cosmology: The angular diameter distance, area distance, and luminosity distance are all naturally defined on a Lipschitz lightcone in contrast to their smooth definition that forces them diverge in the presence of caustics.}. 
In such an extended setting, we can define (see Proposition \ref{lemmaDiffF} ff.) the observable two--dimensional metric $h_{(\widehat{z})}$ on the physical celestial sphere $\mathbb{CS}_{\hat{z}}$ according to
$h_{(\widehat{z})}:=\exp_p^*\,g^{(2)}_{\hat{z}}$, where $\exp_p^*$ denotes the pull--back action associated with the (null) exponential map $\exp_p$. It is important to stress that $h_{(\widehat{z})}$  coexists with the directional angular metric  $\widetilde{h}(\mathbb{S}^2):=\,d\theta^2+\sin^2\theta\,d\phi^2$ naturally defined on $\mathbb{CS}_z$. As the FLRW redshift $\widehat{z}$ varies, the set of physical celestial spheres, decorated with both the physical and the directional metrics $\{(\mathbb{CS}_{\hat{z}},\,h_{(\widehat{z})},\,\widetilde{h}(\mathbb{S}^2))\}$,  
allows us to track down the actual spacetime geometry along the past light cone $\widehat{\mathscr{C}}^-(p,\widehat{g})$ and of the null geodesics distortion effect due to inhomogeneities. 

Astrophysical data from the FLRW sky section $(\widehat{\Sigma}_{\hat{z}},\,\widehat{g}^{(2)}_{\hat{z}})$ are similarly gathered by the reference FLRW observer on her celestial sphere 
$\widehat{\mathbb{CS}}_{\hat{z}}$. The corresponding exponential map $\widehat{\exp}_p$ has the smoothness (a diffeomorphism) induced by the FLRW null geodesics flow. The associated pull--back action induces a metric $\widehat{h}_{(\widehat{z})}\,:=\,\widehat{\exp}_p^*\,\widehat{g}^{\,(2)}_{\hat{z}}$ that gives the FLRW celestial sphere $(\widehat{\mathbb{CS}}_{\hat{z}},\,\widehat{h}_{(\widehat{z})})$ the ($\widehat{z}$-- scaled) round geometry of a sphere. Whereas the celestial sphere  $(\mathbb{CS}_z,\,h_z)$  has a rather rough landscape, data on $(\widehat{\mathbb{CS}}_{\hat{z}},\,\widehat{h}_{(\widehat{z})})$ are visualized and interpreted in terms of the smooth FLRW spacetime geometry.
At the common observation event $p$, the area elements $d\mu_{{h}_{({z})}}$ and    $d\mu_{\widehat{h}_{(\widehat{z})}}$ of  $(\mathbb{CS}_z,\,h_z)$  and $(\widehat{\mathbb{CS}}_{\hat{z}},\,\widehat{h}_{(\widehat{z})})$ provide the cross--sectional areas of the given source as seen by the physical and the FLRW observers in their local rest space. In particular, if the source subtends the observer's visual solid angle $d\mu_{\mathbb{S}^2}$ on the unit sphere, and if the observers know the intrinsic size of the given source, then the  physical area distance $D_{\hat{z}}\,:=\,{D}(\zeta_{(\widehat{z})}(y))$ and the FLRW area distance $\widehat{D}_{\hat{z}}(y)$,  between the observation event $p$ and the  astrophysical sources located at the celestial coordinates $y$ on $\widehat{\mathbb{CS}}_{\hat{z}}$ and $\zeta_{(\widehat{z})}(y)$ on $\mathbb{CS}_z$, are defined by
$d\mu_{{h}_{({z})}}\,=\,D^2_{\hat{z}}\,d\mu_{\mathbb{S}^2}\,:=\,{D}^2(\zeta_{(\widehat{z})}(y))\,d\mu_{\mathbb{S}^2}$ and 
$d\mu_{\widehat{h}_{(\widehat{z})}}\,=\,\widehat{D}^2_{\hat{z}}(y)\,d\mu_{\mathbb{S}^2}$ (see Appendix \ref{ArDistLipCone}).

Area distances are, in principle, measurable quantities whenever the source is an astrophysical object of known type. In this connection, it is worthwhile to stress that the celestial spheres area elements $d\mu_{{h}_{({z})}}$ and  $d\mu_{\widehat{h}_{(\widehat{z})}}$ are related to the intrinsic size of the source through the mapping action of the null geodesic flow reaching the observer from the source. $D_{\hat{z}}:={D}(\zeta_{(\widehat{z})})$ and  $\widehat{D}_{\hat{z}}$ are distinct from each other. In particular, ${D}(\zeta_{(\widehat{z})})$ has a memory of all the distortions that affect null geodesic propagation in the pre--homogeneity region. For this reason, the geometric landscape offered by the celestial sphere $({\mathbb{CS}}_{{z}},\,{h}_{({z})})$ is quite different from the one described by  
$(\widehat{\mathbb{CS}}_{\hat{z}},\,\widehat{h}_{(\widehat{z})})$, and their comparison requires a rather delicate approach\cite{CarFam}.
 As in cartography, where a smooth geographical globe is compared with a physical 3--dimensional scaled rendering of the actual Earth surface, one can exploit\cite{CarFam} the fact that the underlying geometries are conformally 
related,  
$\zeta_{(_{\hat{z}})}^*\,{{h}_{({z})}}\,=\,\Phi^2(\zeta_{(_{\hat{z}})})\,{\widehat{h}_{(\widehat{z})}}$, where $\zeta_{(_{\hat{z}})}^*$ denotes the action (the pull--back) of the Lorentz aberration map connecting $\widehat{\mathbb{CS}}_{\hat{z}}$ with ${\mathbb{CS}}_{{z}}$. The Poincar\'e--Koebe uniformization theorem for 2--dimensional surfaces is the basic tool that allows us to conformally represent bumpy 2--spheres with a complex pattern of raised relief and large landforms (such as the Earth's surface) over a smooth and round geographical globe. In the same spirit, the conformal factor  $\Phi(\zeta_{(_{\hat{z}})})$ does the job here and keeps track of what happens in the pre--homogeneity region (see Proposition \ref{comp1}). Not surprisingly, it can be related to the area distances according to 
$ \Phi^2(\zeta_{(_{\hat{z}})})\,=\,{D}_{\hat{z}}^2(\zeta_{(_{{z}})})/\widehat{D}_{\hat{z}}^2$, a connection that characterizes\cite{CarFam, CarFam2, Carfora3} the  functional which, at the given reference redshift $\widehat{z}$, allows us to compare the FLRW celestial sphere $\widehat{\mathbb{CS}}_{\hat{z}}$ and the physical celestial sphere $\mathbb{CS}_{z}$,\,\emph{i.e.}  
\begin{equation}
\label{Efunctn}
E_{\widehat{\mathbb{C\,S}}_{\hat{z}},\;{\mathbb{C\,S}}_{z}}[\zeta_{(\widehat{z})}]\,:=\,\int_{\widehat{\mathbb{CS}}_{\hat{z}}}\left({\Phi}(\zeta_{(_{\hat{z}})})\,-\,1 \right)^2\,d\mu_{\hat{h}_{(\widehat{z})}}\,=\,
\int_{\widehat{\mathbb{CS}}_{\hat{z}}}\,
\left[\frac{D_{\hat{z}}(\zeta_{(\widehat{z})})\,-\,\widehat{D}_{\hat{z}}}{\widehat{D}_{\hat{z}}}\right]^2
\,d\mu_{\hat{h}_{(\widehat{z})}}\,,
\end{equation}
(see Appendices \ref{CompFunctRedZ} and \ref{TSKSM}). \, 
This functional belongs to an important class of scale--dependent functionals efficiently used in image visualization\cite{HassKoehl}. 
From a geometric analysis point of view, the (bi)--Lipschitz continuous exponential map $\exp_p$ admits a Sobolev $(1,2)$--norm, requiring  that $\exp_p$ and its first differential be square integrable. This property implies that (\ref{Efunctn}) is well--defined in the assumed bi--Lipschitzian setting here adopted, and that (\ref{Efunctn}) can be minimized (see Theorem \ref{distheorem}  and \cite{CarFam}) with respect to the choice of the aberration map $\zeta_{(_{\hat{z}})}$, providing a scale--dependent  "distance functional" $E_{\widehat{\mathbb{C\,S}}_{\hat{z}},\;{\mathbb{C\,S}}_{z}}\,:=\inf_{\zeta_{(\widehat{z})}}\,E_{\widehat{\mathbb{C\,S}}_{\hat{z}},\;{\mathbb{C\,S}}_{z}}[\zeta_{(\widehat{z})}]$\;
between the FLRW celestial sphere $(\widehat{\mathbb{CS}}_{\hat{z}},\,\widehat{h}_{(\widehat{z})})$  and the physical celestial sphere $({\mathbb{CS}}_{{z}},\,{h}_{({z})})$.  In particular, it vanishes if and only if the two celestial spheres are isometric.

\section{The coupling with spacetime scalar curvature}
In terms of the sky averages of the area distances relative fluctuations,
\begin{equation} 
\delta^{(n)}_{\widehat{\mathbb{C\,S}}_{\hat{z}},\;{\mathbb{C\,S}}_{\hat{z}}}:=\frac{1}{4\pi}
\int_{\mathbb{S}^2}
[\frac{D_{\hat{z}}-\widehat{D}_{\hat{z}}}{\hat{D}_{\hat{z}}}]^n
d\mu_{\mathbb{S}^2}\,,
\end{equation}
with $n=1,2$, the functional $E_{\widehat{\mathbb{C\,S}}_{\hat{z}},\;{\mathbb{C\,S}}_{z}}$ can be equivalently rewritten in adimensional form as\footnote{Relation (\ref{EDelta2}) follows from Lemma \ref{DiArea},\,(\ref{FLRWArea1}), and (\ref{formulDiffD2}), Further properties of  $E_{\widehat{\mathbb{C\,S}}_{\hat{z}},\;{\mathbb{C\,S}}_{z}}$ are discussed in detail in Appendix \ref{CompFunctRedZ}.}      
\begin{equation}
\label{EDelta2}
\frac{E_{\widehat{\mathbb{C\,S}}_{\hat{z}},\;{\mathbb{C\,S}}_{z}}}{4\pi\,\widehat{D}^{2}_{(\widehat{z})}}\,=\,
\delta^{(2)}_{[\widehat{\mathbb{C\,S}}_{\hat{z}},\;{\mathbb{C\,S}}_{\hat{z}}]}\,
\,=\,
\frac{A({\mathbb{C\,S}}_{z})\,-\,A(\widehat{\mathbb{C\,S}}_{\hat{z}})}{4\pi\,\widehat{D}^{2}_{(\widehat{z})}}\,-2\,\delta^{(1)}_{[\widehat{\mathbb{C\,S}}_{\hat{z}},\;{\mathbb{C\,S}}_{\hat{z}}]}\,,
\end{equation} 
where $A(\widehat{\mathbb{C\,S}}_{\hat{z}})$ and $A({\mathbb{C\,S}}_{z})$ are the areas of the celestial spheres  $(\widehat{\mathbb{CS}}_{\hat{z}},\,\widehat{h}_{(\widehat{z})})$ and $({\mathbb{CS}}_{{z}},\,{h}_{({z})})$. Notice that the directional averages 
$\delta^{\;(1)}_{\widehat{\mathbb{C\,S}}_{\hat{z}},\;{\mathbb{C\,S}}_{{z}}}$  and $\delta^{\;(2)}_{\widehat{\mathbb{C\,S}}_{\hat{z}},\;{\mathbb{C\,S}}_{{z}}}$,
 control  the relative fluctuation  $(D_{\hat{z}}-\widehat{D}_{\hat{z}})/\widehat{D}_{\hat{z}}$ of the physical area distance $D_{\hat{z}}:= D_{\hat{z}}(\zeta_{(\widehat{z})})$  with respect to the reference area distance $\widehat{D}_{\hat{z}}$ associated with the FLRW fundamental observers. This fluctuation is characterized by the variance 
\begin{equation}
\label{samplevar}
\frac{1}{4\pi}
\int_{\mathbb{S}^2}
\left[\frac{D_{\hat{z}}-\widehat{D}_{\hat{z}}}{\widehat{D}_{\hat{z}}}\,-\,\delta^{(1)}_{[\widehat{\mathbb{C\,S}}_{\hat{z}},\;{\mathbb{C\,S}}_{\hat{z}}]}\right]^2
d\mu_{\mathbb{S}^2}\,=\,\delta^{(2)}_{[\widehat{\mathbb{C\,S}}_{\hat{z}},\;{\mathbb{C\,S}}_{\hat{z}}]}\,-\,\left(\delta^{(1)}_{[\widehat{\mathbb{C\,S}}_{\hat{z}},\;{\mathbb{C\,S}}_{\hat{z}}]} \right)^2\,,
\end{equation}
from which it immediately follows that $E_{\widehat{\mathbb{C\,S}}_{\hat{z}},\;{\mathbb{C\,S}}_{z}}=0$ if and only if, at the given FLRW redshift\footnote{Recall that the physical redshift $z$ is related to the reference FLRW redshift  $\widehat{z}$ by the $\mathrm{PSL}(2, \mathbb{C})$ transformation (\ref{pesse2}).} $D_{\hat{z}}=\widehat{D}_{\hat{z}}$.
According to the Lorentzian version of the Bertrand--Puiseaux formulas\cite{Berthiere}, these areas are related to the spacetime scalar curvature in the pre--homogeneity region by the asymptotic relation (see Proposition \ref{flucDisCurv}  and  Appendix \ref{ArDisCurvFact})

\begin{equation}
\label{AlmSureDeltasA0}
\frac{\delta^{\;(2)}_{\widehat{\mathbb{CS}}}(\,\widehat{z}\,)\,+\,2\,\delta^{\;(1)}_{\widehat{\mathbb{CS}}}(\,\widehat{z}\,)}{\left(1\,+\,\widehat{z}\,\right)^2\,\widehat{r}^{\,2}(\,\widehat{z}\,)}\,=\,\frac{1}{72}\,\left(\widehat{\mathrm{R}}({p})\,- \,{\mathrm{R}}(p)\,+\,\ldots\right)\,.
\end{equation}
 
\noindent
where $\,\ldots\,$ represents higher order corrections that can be expressed in terms of powers of the spacetime Riemann tensor and its covariant derivatives\footnote{For simplicity, the limit  (\ref{limitMem0}) and the other related expressions that follows are computed assuming that at $p$, the FLRW and the physical observer have the same 4--velocity; The general case,  where the two observers have a non--vanishing relative velocity $v$ 
(see (\ref{pesse2})) is worked out in full detail in the Appendix. For what concerns (\ref{limitMem0}), the general result is discussed in Theorem \ref{Theasympt1}.}.\,
For  $\widehat{z}\rightarrow 0$, we have \cite{CarFam} the limiting behavior
\begin{equation}
\label{limitMem0}
\lim_{\widehat{z}\,\longrightarrow\,0}\,
\frac{(1+q_0)^2\,\left(
\delta^{\;(2)}_{\widehat{\mathbb{C\,S}}_{\hat{z}},\;{\mathbb{C\,S}}_{{z}}}\,+\,2\,\delta^{\;(1)}_{\widehat{\mathbb{C\,S}}_{\hat{z}},\;{\mathbb{C\,S}}_{{z}}}\right)}{\left(1\,+\,\widehat{z}\right)^{\;2}\,\ln^2\left[1\,+\,(1+q_0)\widehat{z}\right]}\,
=\,\frac{\widehat{\mathrm{R}}(p)\,-\,\mathrm{R}(p)}{72\,H_0^2}\,,
\end{equation}  
where $q_0\,\simeq\,-\,0.55$ is the present value of the deceleration parameter, and
where $\mathrm{R}(p)$ and $\widehat{\mathrm{R}}(p)$ respectively are the scalar curvatures of the physical and of the reference FLRW spacetimes (evaluated at the common observation event $p$). As detailed in the Appendix (Theorem \ref{Theasympt1}), this relation holds asymptotically also in a neighborhood of $\widehat{z}=0$, say in the low--moderate redshift range $0<\widehat{z}\lesssim 10^{-2}$. For this reason, even if it would have been safe to employ the standard  Hubble formula  
$(c)\widehat{z}=H_0\,\widehat{r}$, with no role for the deceleration parameter $q_0$ (as in Appendix \ref{GareDisCurv}),\; we have used in (\ref{limitMem0}) the finer approximation of the Hubble parameter provided by $H(\widehat{z})\simeq H_0[1+(1+q_0)\widehat{z}]$.

Relation (\ref{limitMem0}) is a non--perturbative geometrical result that provides a robust correlation between the fluctuations in the area distance of astrophysical sources and scalar curvature fluctuations, directly connecting these distance indicators to Einstein's field equations. 
It is worthwhile to stress that the area (and the associated luminosity) distance fluctuation terms $\delta^{\;(1)}$ and $\delta^{\;(2)}$  often are objects of study in cosmological numerical simulations (see \emph{e.g.}, \cite{Adamek2}). They are relevant indicators in establishing the fluctuation between the physical and the reference FLRW distance--redshift relation, which is essential for establishing evidence of dark energy. Thus, one may wonder if these simulations may provide a direct estimate of the impact that the correlation between curvature and area distance fluctuations described by (\ref{limitMem0}) may have on the description of dark energy in terms of the FLRW cosmological constant $\widehat{\Lambda}^{(FLRW)}$. The smallness of $\delta^{\;(1)}$ and $\delta^{\;(2)}$, often estimated by ray tracing in relativistic N--body simulations, indicates that there is no significant bias of the distance--redshift correlation due to the presence of inhomogeneities\cite{Adamek2}. This result may suggest that (\ref{limitMem0}) and its implications are in contrast with the current framework, according to which late--epoch structures do not significantly affect the values of the cosmological parameters. To avoid this and other related misunderstandings, let us stress that what matters in (\ref{limitMem0})  is the behavior of the ratio (\ref{limitMem0}) in the $\widehat{z}\rightarrow\,0$ limit. This ratio does not alter the value of the cosmological parameters. It is a rigorous geometrical relation strictly related to the Lorentzian version of the Bertrand--Puiseaux formulas, and numerical simulations cannot question its validity. Order of magnitude estimates often assume that $\widehat{\mathrm{R}}(p)\,-\,\mathrm{R}(p)\approx 0$ as a consequence of an old conjecture by S. Weinberg\cite{Weinberg} stating that the areas $A(\mathbb{CS}_z)$ should be almost unaffected by the presence of inhomogeneities. Even if occasionally brought back to life by some authors (see \emph{e.g.}, \cite{Breton1}), this is a wrong conjecture, both from the physical point of view as shown by \cite{EllisBassett,EllisSolomons,EllisDurrer}, as well as from the geometrical point of view (whatever meaning we would like to give to \emph{almost unaffected}) since it underestimates the global consequences that  (Lorentzian) Bertrand--Puiseaux formulas can have. The lesson here is that even tiny fluctuations in curvature can have surprising effects, as illustrated by the example of a closed surface with tiny curvature bumps sparsely distributed throughout. A flatlander cosmologist with faith in a (surface version of a) cosmological principle would conclude that his observations strongly support that the observed portion of the universe is flat, with sparse and very tiny perturbative fluctuations in curvature of no statistical relevance. He does not know the global topology of his model universe because the whole surface is not accessible to his observations. Yet, according to the cosmological principle, he assumes that what he experiences is typical. Thus, from the gathered data and the cosmological principle, he can conclude that the surface universe is (perturbatively) flat. Yet, he cannot reach such a conclusion with the partial data he has. Locally, his analysis of the accessible data is correct. Still, the Gauss--Bonnet theorem demonstrates that the global average of the allegedly irrelevant, tiny curvatures (locally measured by a Bertrand--Puiseaux formula) is not a freely fluctuating quantity that can be easily dismissed since the surface's topology constrains it. The example is particular, and we do not have a Gauss--Bonnet theorem at work in our case; yet, it suggests that dismissing curvature fluctuations based on order--of--magnitude estimates is not a good strategy.

\section{Memories of late epoch inhomogeneities}
 According to the $\Lambda\mathrm{CDM}$ point of view and Neyman's weak cosmological principle \cite{Neyman, Peebles}, structure formation can be interpreted as a realization of a homogeneous and isotropic random process on the FLRW background. In such a scenario, structures are generated by the gravitational clustering of the initially small CMB density fluctuations seeded by a primordial anisotropy power spectrum of inflationary origin. This clustering is typically handled in a linear regime of adiabatic fluctuations coupled with FLRW perturbation theory. However, in the pre--homogeneity region, fluctuation growth breaks down the linear regime, and Gaussian modes interact with each other. As $\widehat{z}\rightarrow 0$, it becomes difficult to obtain a non--perturbative control over gravitational instabilities. Indeed, there is evidence that in such a regime, skewed distributions characterize the pre--homogeneity region \cite {Adamek} and \cite{Fanizza2}. In particular, the area distance fluctuations are distributed in a rather complex and unpredictable way when we approach the critical redshift  $\widehat{z}_{(c)}$, and we enter the cosmological decoupling region where the $\widehat{z}\,\rightarrow\,0$ limit is attained. Yet, the structure of the left member of the (\ref{limitMem0}) holds a strong connection with the moment generating function $\mathbb{M}_{N(m,\sigma^2)}$ associated with a normal distribution of the area distance fluctuations, with mean $m\,:=\,\delta^{\;(1)}_{\widehat{\mathbb{CS}}}(\,\widehat{z})$ and variance $\sigma^2\,:=\,\delta^{\;(2)}_{\widehat{\mathbb{CS}}}(\,\widehat{z})\,-\,(\delta^{\;(1)}_{\widehat{\mathbb{CS}}}(\,\widehat{z}))^2$. To wit, for $\xi\,\in\,\mathbb{R}$, we have\footnote{See Appendix \ref{ArDisFluc} for a detailed discussion.}
\begin{equation}
\label{MomGenFunct000}
\mathbb{M}_{N(m,\sigma^2)}\,(\xi)\,=\,\exp\,\left[\frac{\delta^{\;(2)}_{\widehat{\mathbb{CS}}}(\,\widehat{z})\,\xi^2\,+\,2\,\delta^{\;(1)}_{\widehat{\mathbb{CS}}}(\,\widehat{z})\,\xi}{2}\,-\,\frac{1}{2}\left(\delta^{\;(1)}_{\widehat{\mathbb{CS}}}(\,\widehat{z})\,\xi \right)^2\right]\,.
\end{equation}  
This relation suggests that upon stabilization (\emph{i.e.} normalizing the area distance fluctuations to have mean value $0$, and possibly variance $1$), we may recover a suitable form of Gaussianity in the $\widehat{z}\,\rightarrow\,0$ regime by exploiting central limit theorem arguments. We discuss this possibility in the worst possible scenario associated with the cosmological decoupling regime generated by strong gravitational clustering. To this end, we will begin by looking at a sequence  $\{\widehat{\mathbb{CS}}_{(\,\widehat{z}\,(i))}\,:\;i\in\mathbb{Z}^+\}$ of FLRW celestial spheres which foliate the FLRW null cone in the cosmological decoupling region surrounding $p$, \emph{i.e.}, in the FLRW redshift interval $0\,\leq\,\widehat{z}\,(i)\,\leq\,\widehat{z}\,(c)$,   
with $\widehat{\mathbb{CS}}_{(\,\widehat{z}\,(1))}=\widehat{\mathbb{CS}}_{(\,\widehat{z}\,(c))}$ and converging to $p$ as $i\rightarrow\,\infty$. We respectively denote by $\{\widehat{\Sigma}_{(\,\widehat{z}\,(i))}\,:\;i\in\mathbb{Z}^+\}$ and $\{{\Sigma}_{(\,\widehat{z}\,(i))}\,:\;i\in\mathbb{Z}^+\}$ the sequence of FLRW and physical sky sections associated with the celestial spheres $\widehat{\mathbb{CS}}_{(\,\widehat{z}\,(i))}$ according to (\ref{sigmapr0}) and the $\mathrm{PSL}(2,\mathbb{C})$ map $\zeta_{(\,\widehat{z}\,(i))}:\widehat{\mathbb{CS}}_{(\,\widehat{z}\,(i))}\rightarrow {\mathbb{CS}}_{(\,\widehat{z}\,(i))}$. To understand the relations among the corresponding
 sequence of area distance fluctuations $Y_{(\,\widehat{z}\,(i))}\,:=\,\tfrac{D_{\widehat{z}\,(i)}(\zeta_{(\,\widehat{z}\,(i))})\,-\,\widehat{D}_{\widehat{z}\,(i)}}{\widehat{D}_{\widehat{z}\,(i)}}$ sampled along distinct directions of the FLRW directional celestial sphere $\widehat{\mathbb{CS}}$, we assume that the $Y_{(\,\widehat{z}\,(i))}$ exhibit no mutual relation. We are dealing with the minimal assumption we can make because, in the real world, there is no FLRW Laplace's demon propagating the assumed Gaussianity of fluctuations from the observed portion of the CMB last scattering surface through the vagaries of the pre--homogeneity region up to the cosmological decoupling region surrounding $p$, and able to control all the moments of the random variables $\{Y_{(\,\widehat{z}\,(i))}\}$. We can only assume that the fluctuations $Y_{(\,\widehat{z}\,(i))}$ are mutually independent, square--integrable random variables which are identically distributed according to the directional probability measure $\mathbb{P}_{\widehat{\mathbb{CS}}}$\; on $\widehat{\mathbb{CS}}$\, defined (see (\ref{CSprobab}))
by $\mathbb{P}_{\widehat{\mathbb{CS}}}\,\left(d\mu_{\widehat{\mathbb{S}}^2}\right)\,:=\,
\frac{1}{4\pi}\,d\mu_{\widehat{\mathbb{S}}^2}$,\, where $d\mu_{\widehat{\mathbb{S}}^2}$ is the standard solid angle measure on the round 2--sphere. In particular, if 
\begin{equation}
{Y}^{-\,1}_{(\,\widehat{z}\,(i))}(I)\,:=\,\left\{\left.(\,\widehat\theta,\,\widehat\phi\,)\,\in\,\widehat{\mathbb{CS}}\,\,\right|\,Y_{(\,\widehat{z}\,(i))}(\widehat\theta,\,\widehat\phi)\,\in\,I\,\subset\,\mathbb{R} \right\}
\end{equation}
denote the subset of directions $(\,\widehat\theta,\,\widehat\phi\,)\,\in\,\widehat{\mathbb{CS}}$ whose corresponding area distance relative fluctuations  $Y_{(\,\widehat{z}\,(i))}(\widehat\theta,\,\widehat\phi)$ \, take values in the interval \,$I$, then the $Y_{(\,\widehat{z}\,(i))}$     are assumed to be independently distributed  under the push forward measure
\begin{equation}
\label{FlucDisDistr00}
\mathbb{P}_{\widehat{\mathbb{CS}}}\left({Y}^{-\,1}_{(\,\widehat{z}\,(i))}(I)\right)\,=\,
 \frac{1}{4\pi}\,\int_{{Y}^{-\,1}_{(\,\widehat{z}\,(i))}(I)}\, d\mu_{\widehat{\mathbb{S}}^2}\,.
\end{equation}
Under this assumption, the behavior of the sequence $\{Y_{(\,\widehat{z}\,(i))}\}$ is quite wild since in the pre--homogeneity region dominated by the vagaries of strong gravitational clustering, non--Gaussian randomness is at work, and even convergence to an average value may be questionable. Yet, the sequence of empirical means $\{\tfrac{1}{n}\,\sum_{i=1}^n\, Y_{(\,\widehat{z}\,(i))}\}$, can be controlled by a Kolmogorov's Strong Law (of large numbers) argument, and we can show that we have the $\mathbb{P}$--almost sure convergence (see Theorem \ref{AlmSureConv}) result
\begin{equation}
\label{AlmSureP0}
\mathbb{P}_{\widehat{\mathbb{CS}}}\,\left(
\lim_{n\rightarrow\,\infty}\,\frac{1}{n}\,
\sum_{i=1}^n\,Y_{(\,\widehat{z}\,(i))}\,=\, \delta^{\;(1)}_{\widehat{\mathbb{CS}}}(\,\widehat{z}\,(c))\right)\,=\,1\,
\end{equation} 
that computes the limit of the sample mean $\frac{1}{n}\,
\sum_{i=1}^n\,Y_{(\,\widehat{z}\,(i))}$ in terms of the directional average $\delta^{\;(1)}_{\widehat{\mathbb{CS}}}(\,\widehat{z}\,(c))$ of the area distance fluctuations evaluated on the cosmological decoupling celestial sphere $\widehat{\mathbb{CS}}_{\widehat{z}\,(c)}$. Moreover, if we stabilize the sequence $\{Y_{(\,\widehat{z}\,(i))}\}$ by normalizing its terms to have mean $0$ and variance $1$,\; \emph{i.e.}, if we define 
\begin{equation}
\label{RandVarX0}
X_{(\,\widehat{z}\,(i))}\,:=\,\left(
Y_{(\,\widehat{z}\,(i))}\,-\,\delta^{\;(1)}_{\widehat{\mathbb{CS}}}(\,\widehat{z}\,(i)\,)\right)\,\mathrm{Var}^{-\,1/2}_{\widehat{\mathbb{CS}}}(Y_{(\,\widehat{z}\,(i))})\,, 
\end{equation} 
then, the distribution of  corresponding sequence of  sample means $\breve{S}_m(\,\widehat\theta,\,\widehat\phi\,)\,:=\,
\frac{1}{m}\,\sum_{i}^{m}\,X_{(\,\widehat{z}\,(i))}(\,\widehat\theta,\,\widehat\phi\,)$ converges to a Gaussian distribution $N_{(0,1)}$ of $\breve{S}_m$, with mean $0$ and variance $1$, \emph{i.e.} for any interval $I\subset\,\mathbb{R}$, we have (see Theorem \ref{GaussCentLim})

\begin{equation}
\label{CentLimThm20}
\lim_{m\,\rightarrow\,\infty}\,\mathbb{E}^{\mathbb{P}}\left[\breve{S}_m,\,I \right]\,=\,
\frac{1}{\sqrt{2\pi}}\,\int_I\,\exp\left[-\,\frac{u^2}{2}\right]\,du\,.
\end{equation}
\vskip .2cm\noindent
Thus, as $m\rightarrow\infty$, the normalized area distance fluctuations $\breve{S}_m$ around the expectation $\delta^{\;(1)}_{\widehat{\mathbb{CS}}}(\,\widehat{z}\,(c))$ are normally distributed (with mean $0$ and variance $1$). Moreover, these fluctuations are independent from $\delta^{\;(1)}_{\widehat{\mathbb{CS}}}(\,\widehat{z}\,(c))$, and, in turn, $\delta^{\;(1)}_{\widehat{\mathbb{CS}}}(\,\widehat{z}\,(c))$ is independent from the $Y_{(\,\widehat{z}\,(c))}(\widehat\theta,\,\widehat\phi)$ sample variance $\mathrm{Var}_{\widehat{\mathbb{CS}}}(Y_{(\,\widehat{z}\,(c))})=\delta^{\;(2)}_{\widehat{\mathbb{CS}}}(\,\widehat{z}\,(c))\,-\,\left(\delta^{\;(1)}_{\widehat{\mathbb{CS}}}(\,\widehat{z}\,(c)) \right)^2$. In particular, $\delta^{\;(2)}_{\widehat{\mathbb{CS}}}(\,\widehat{z}\,(c))$ is indipendent from $\delta^{\;(1)}_{\widehat{\mathbb{CS}}}(\,\widehat{z}\,(c))$. This analysis can be readily extended to the curvature deviation term in (\ref{limitMem0}) (see Proposition \ref{PropA} and Appendix \ref{ArDisCurvFact}). To this end, by tracing the Einstein equations (\ref{FLRWeinstein}) and  (\ref{einsteinEFE}), respectively associated with the reference FLRW spacetime $(M, \widehat{g}, \widehat{\gamma}_s(\widehat\tau))$ and the physical spacetime $(M, {g}, {\gamma}_s(\tau))$, we get
\begin{equation}
\label{RmattCosmCon0}
\,\widehat{\mathrm{R}}(p)\,-\,\mathrm{R}(p)\,=\,4\,\left(\frac{\widehat{\Lambda}-\Lambda}{\widehat{\Lambda}}\right)\,\widehat{\Lambda}\,
+\,8\pi\,\left(\frac{\widehat{\rho}(p)\,-\,\rho(p)}{\widehat{\rho}(p)}\right) \widehat{\rho}(p)\,,
\end{equation}  
where for notational ease, we dropped the FLRW superscript from the FLRW cosmological constant term $\widehat{\Lambda}^{(FLRW)}$ and matter density $\widehat{\rho}^{\,(FLRW)}$ \,(just maintaining the carat $\widehat{}\,$), and where $\widehat{\rho}(p)$ and ${\rho}(p)$ respectively denote the FLRW and the physical matter density at $p$.  
In full generality, we have associated with the physical spacetime a cosmological constant $\Lambda^{(phys)}$, potentially distinct from $\widehat{\Lambda}^{(FLRW)}$,  and keeping track of a possible global deviation effect due to non--linearities in structures formation.
We denote by 
\begin{equation}
\Lambda^{(cont)}\,:=\,\widehat{\Lambda}^{(FLRW)}\,-\,\Lambda^{(phys)}\,,
\end{equation}
this potential contribution. In (\ref{RmattCosmCon0}), we have assumed a pure matter scheme for both the reference FLRW as well as for the physical spacetime. The characterization of  ${\rho}(p)$ in the pre--homogeneity region should take care of the fact that the physical matter density data are gathered from the past light cone (directly for the baryonic component, indirectly via the cold dark matter hypothesis). If 
$\rho_{(\,\widehat{z}\,(i))}(\widehat\theta, \widehat\phi\,)\,:=\,
\left(\exp_p\,\circ\zeta_{(\,\widehat{z}\,(i))}\right)^*{\rho}_{\,\Sigma_{\,\widehat{z}\,(i)}}(q)$ is the celestial sphere representation of ${\rho}_{\,\Sigma_{\,\widehat{z}\,(i)}}(q)$, we denote by
\begin{equation}
\rho^{(v)}\left(\widehat{\mathbb{CS}}_{\widehat{z}(i)}\right)\,=\,
\frac{1}{4\pi}\,\int_{\widehat{\mathbb{CS}}}\,
\rho^{(v)}_{\widehat{z}\,(i)}(\widehat\theta, \widehat\phi\,)\,
d\mu_{\widehat{\mathbb{S}}^2}\,,
\end{equation} 
the corresponding
 average density over $\widehat{\mathbb{CS}}_{\widehat{z}(i)}$, and by $\widehat{\rho}(\widehat{\mathbb{CS}}_{\widehat{z}(i)})$ the FLRW reference density.
For large $\widehat{z}$, the  density contrast 
\begin{equation}
\frac{\Delta_{\hat{z}}\rho}{\widehat{\rho}_{\hat{z}}}\,:=\,\frac{(\widehat{\rho}_{\hat{z}}\,-\,\rho_{\hat{z}})}{
\widehat{\rho}_{\hat{z}}},
\end{equation}
is governed by the density perturbations (with respect to the FLRW background)  seeded by inflation ($|\Delta_{\hat{z}}\rho/\widehat{\rho}_{\hat{z}}|_{{CMB}}\simeq\,10^{-5}$). As stressed in the introductory remarks, these perturbations are amplified by gravitational clustering to the point that the FLRW perturbation theory becomes completely unreliable at a late epoch when we probe the pre--homogeneity region at very low reference FLRW redshift $\widehat{z}\rightarrow\,\widehat{z}_{(c)}$. In such a non--perturbative regime, we cannot a priori assume that the relative matter density fluctuations have a Gaussian distribution. To handle the situation and recover some control, we assume that the fluctuating terms 
\begin{equation}
\Gamma_{\widehat{z}\,(i)}\left(\widehat\theta, \widehat\phi\, \right)\,:=\,  \frac{\widehat{\rho}(\widehat{\mathbb{CS}}_{\widehat{z}(i)})\,-\,\rho_{\widehat{z}\,(i)}(\widehat\theta, \widehat\phi\,)}{\widehat{\rho}(\widehat{\mathbb{CS}}_{\widehat{z}(i)})}
\end{equation} 
characterize a sequence $\{\Gamma_{\widehat{z}\,(i)}\,:\;i\,\in\,\mathbb{Z}^+\}$ of independent, identically distributed (i.i.d.) random variables with respect to $\mathbb{P}_{\widehat{\mathbb{CS}}}$.
As in the previous case, the law of large numbers and the i.i.d. assumption imply the $\mathbb{P}$--almost sure convergence result
\begin{equation}
\lim_{n\rightarrow\,\infty}\,\frac{1}{n}\,\sum_{i=1}^{n}\Gamma_{\widehat{z}\,(i)}\,=\,
\frac{\widehat{\rho}(\widehat{\mathbb{CS}}_{\widehat{z}(c)})\,-\,
\mathbb{E}^{\mathbb{P}}[\rho_{\widehat{z}\,(c)}]}{\widehat{\rho}(\widehat{\mathbb{CS}}_{\widehat{z}(c)})}\,,
\end{equation}  
where $\widehat{z}(c)$ is the FLRW reference redshift marking the cosmological decoupling sky section $\Sigma_{\widehat{z}(1)}$. We identify this expression with the matter density fluctuation at the observation point $p$, \emph{i. e.} 
\begin{equation}
\label{mattMeanValue0}
\frac{\widehat{\rho}(p)-\rho^{(v)}(p)}{\widehat{\rho}(p)}\,:=\,\frac{\widehat{\rho}(\widehat{\mathbb{CS}}_{\widehat{z}(c)})\,-\,
\mathbb{E}^{\mathbb{P}}[\rho_{\widehat{z}\,(c)}]}{\widehat{\rho}(\widehat{\mathbb{CS}}_{\widehat{z}(c)})}\,=:\,
\frac{\Delta_{\widehat{z}\,(c)}\,\rho}{\widehat{\rho}_{\widehat{z}\,(c)}}\,,
\end{equation} 
and use this characterization in the expression in 
 (\ref{RmattCosmCon0}). If we shift and normalize $\Gamma_{\widehat{z}\,(i)}$ in such a way to have mean value $0$ and variance $1$, then, in full analogy with (\ref{CentLimThm20}),  the distribution of the empirical means
$\breve{\Gamma}_n\,:=\,\tfrac{1}{n}\,\sum_{i=1}^n\,\overline{\Gamma}_{\widehat{z}\,(i)}$ 
 with respect to  the push forward measure $\mathbb{P}_{\widehat{\mathbb{CS}}}(\breve{\Gamma}^{- 1}_n(I))$\,,\, $I\,\subset\,\mathbb{R}$,\, converges for $n\,\rightarrow\,\infty$\;to\, a Gaussian distribution $N_{(0,1)}$ with mean $0$ and variance $1$. Thus, for $n\,\rightarrow\,\infty$,  the (normalized) matter density fluctuations are normally distributed around the expected mean value (\ref{mattMeanValue0}). If we note that the cosmological term in (\ref{RmattCosmCon0}), being a constant,   is statistically independent of the fluctuating term (\ref{mattMeanValue0}), then by gathering all the results above, we can factorize the relation (see (\ref{AlmSureDeltasA0}) 
  
\begin{align}
&\,\frac{72\left(\delta^{\;(2)}_{\widehat{\mathbb{CS}}}(\,\widehat{z}\,(c))\,+\,2\,\delta^{\;(1)}_{\widehat{\mathbb{CS}}}(\,\widehat{z}\,(c))\right)}{\left(1\,+\,\widehat{z}\,(c)\right)^2\,\widehat{r}^{\,2}(\,\widehat{z}\,(c))}\,=\,\left(\widehat{\mathrm{R}}({p})\,-\,{\mathrm{R}}(p)\right)\nonumber\\
\\
&= 4\,\left(\frac{\widehat{\Lambda}-\Lambda}{\widehat{\Lambda}}\right)\,\widehat{\Lambda}\,
+\,8\pi\,\left(\frac{\Delta_{\widehat{z}\,(c)}\,\rho}{\widehat{\rho}_{\widehat{z}\,(c)}}\right) \widehat{\rho}(p)\,,
\nonumber
\end{align}
in the $N_{(0,1)}$--distribution sense according to (Theorem \ref{MattCentLim})
\begin{align}
\label{OmegatermsCorr}
\delta^{\;(1)}_{\widehat{\mathbb{CS}}}(\,\widehat{z}\,(c))\,&=\,\frac{\widehat{\Omega}_m}{48}\,\left(\frac{\Delta_{\widehat{z}\,(c)}\,\rho}{\widehat{\rho}_{\widehat{z}\,(c)}}\right)\,\left(1\,+\,\widehat{z}\,(c)\right)^2\,\widehat{z}^{\,2}(c)\,\nonumber\\
\\
\delta^{\;(2)}_{\widehat{\mathbb{CS}}}(\,\widehat{z}\,(c))\,&=\,\frac{\widehat{\Omega}_{\widehat\Lambda}}{6}\,\left(\frac{\widehat{\Lambda}-\Lambda}{\widehat{\Lambda}}\right)\,\left(1\,+\,\widehat{z}\,(c)\right)^2\,\widehat{z}^{\,2}(c)\,,\nonumber
\end{align}  
where\, $\widehat\Omega_\Lambda:=\widehat{\Lambda}/{3\,H_0^2}$,\,and\,  $\widehat\Omega_m:=8\pi\widehat{\rho}/{3\,H_0^2}$ are the present values of the FLRW density cosmological parameters (see Appendix \ref{capsule}). In particular, from the 
$\delta^{\;(2)}_{\widehat{\mathbb{CS}}}(\,\widehat{z}\,(c))$ factorization in 
(\ref{OmegatermsCorr}) we explicitly get the correction terms to the FLRW cosmological constant $\widehat{\Lambda}^{(FLRW)}$ induced by the area distance fluctuations, \emph{i.e.},
\begin{equation}
\label{LambdaRenorm}
\widehat{\Lambda}^{(FLRW)}\,=\,{\Lambda}^{(phys)}\,+\,\frac{18\,H_0^2}{\left(1\,+\,\widehat{z}\,(c)\right)^2\,\widehat{z}^{\,2}(c)}\,\delta^{\;(2)}_{\widehat{\mathbb{CS}}}(\,\widehat{z}\,(c))\,,
\end{equation}
where, for clarity, we have reintroduced the full notation for the FLRW and physical cosmological constants. Actually, the correction is induced by the matter density fluctuations in an explicit and quantitative way. To see this, notice that from (\ref{OmegatermsCorr}) and the 
variance inequality (see (\ref{samplevar0}) and (\ref{samplevar0i}) )
\begin{align}
\label{samplevar0iii0}
\mathrm{Var}_{\widehat{\mathbb{CS}}}(Y_{(\,\widehat{z}\,(c))})\,&:=\,
\frac{1}{4\pi}
\int_{\widehat{\mathbb{CS}}}
\left(\tfrac{D_{\widehat{z}\,(c)}(\zeta_{(\,\widehat{z}\,(c))})\,-\,\widehat{D}_{\widehat{z}\,(c)}}{\widehat{D}_{\widehat{z}\,(c)}}\,-\,\delta^{\;(1)}_{\widehat{\mathbb{CS}}}(\,\widehat{z}\,(i))\right)^2
	d\mu_{\widehat{\mathbb{S}}^2}\nonumber\\
\\
&=\,\delta^{\;(2)}_{\widehat{\mathbb{CS}}}(\,\widehat{z}\,(c))\,-\,\left(\delta^{\;(1)}_{\widehat{\mathbb{CS}}}(\,\widehat{z}\,(c)) \right)^2\,>\,0,\nonumber
\end{align}
we directly get the lower bound 

\begin{equation}
\label{CosmCorr}
\frac{\widehat{\Lambda}^{(FLRW)}-\Lambda}{\widehat{\Lambda}^{(FLRW)}}\,\geq\,\frac{6}{(48)^2}\,\frac{\widehat{\Omega}^2_m}{\widehat\Omega_{\widehat\Lambda}}\,\left(\frac{\Delta_{\widehat{z}\,(c)}\,\rho}{\widehat{\rho}_{\widehat{z}\,(c)}}\right)^2\,\left(1\,+\,\widehat{z}\,(c)\right)^2\,\widehat{z}^{\,2}(c)
\end{equation}

\section{Structures and the FLRW cosmological constant}
The relation between the cosmological constant deviation term $\tfrac{\widehat{\Lambda}-\Lambda}{\widehat{\Lambda}}$ and the quadratic area distance fluctuation term $\delta^{\;(2)}_{\widehat{\mathbb{CS}}}(\,\widehat{z}\,(c))$ may appear somehow accidental, and to some extent without a physical rationale. However, this relation has a deep physical interpretation. To this end, let us notice that from (\ref{EDelta2}) and the expression (see (\ref{FLRWDi})) of the FLRW area distance\footnote{These relations are here written in the case of the spatially flat FLRW model. The general case is dealt with in the Appendix.}, we get
\begin{equation}
\label{rewd20}
E_{\widehat{\mathbb{C\,S}}_{\widehat{z}\,(c)},\;{\mathbb{C\,S}}_{\widehat{z}(c)}}[\overline\zeta_{(\widehat{z}\,(c))}]\,=\,
\frac{4\pi\,\widehat{r}^{\;2}_{(\widehat{z}\,(c))}}{\left(1\,+\,\widehat{z}\,(c)\right)^{\;2}}
\,\delta^{\;(2)}_{\widehat{\mathbb{CS}}}(\,\widehat{z}\,(c))\,
\end{equation}
where $\overline\zeta_{(\widehat{z}\,(c))}$ is the $\mathrm{PSL}(2, \mathbb{C})$--minimizing map (see Theorem  \ref{distheorem} and \ref{FactorEnergyD}). This minimized  $E_{\widehat{\mathbb{C\,S}}_{\widehat{z}\,(c)},\;{\mathbb{C\,S}}_{\widehat{z}(c)}}[\overline\zeta_{(\widehat{z}\,(c))}]$ defines the (functional) distance 
\begin{equation}
d_{\widehat{z}\,(c)}\left[\widehat{\mathbb{C\,S}}_{\widehat{z}\,(c)},\;{\mathbb{C\,S}}_{\widehat{z}\,(c)}\right]\,:=\,
E_{\widehat{\mathbb{C\,S}}_{\widehat{z}\,(c)},\;{\mathbb{C\,S}}_{\widehat{z}\,(c)}}[\overline\zeta_{(\widehat{z}\,(c))}]
\end{equation}
between the FLRW celestial sphere $\widehat{\mathbb{C\,S}}_{\widehat{z}\,(c)}$ and the physical celestial sphere ${\mathbb{C\,S}}_{\widehat{z}\,(c)}$. This is also the distance, at the given reference scale set by $\widehat{z}\,(c)$,  between the FLRW and the physical past lightcones $\widehat{\mathscr{C}}^-(p,\widehat{g})$ and  $\mathscr{C}^-(p,g)$. 
From (\ref{rewd20}) and (\ref{LambdaRenorm}) we easily get
  
\begin{equation}
\label{LambdaRenorm2}
\widehat{\Lambda}^{(FLRW)}\,=\,{\Lambda}^{(phys)}\,+\,\frac{9}{2\pi}\,\frac{d_{\widehat{z}\,(c)}\left[\widehat{\mathbb{C\,S}}_{\widehat{z}\,(c)},\;{\mathbb{C\,S}}_{\widehat{z}\,(c)}\right]}{\widehat{r}^{\,4}(\,\widehat{z}\,(c))}\,.
\end{equation}
\vskip 0.2cm
\noindent
If we follow the mainstream hypothesis that $\widehat{\Lambda}^{(FLRW)}$ is indeed a constant (describing Dark Energy), then  the lower bound (\ref{CosmCorr}), and the explicit  scale dependence of (\ref{LambdaRenorm2}),    indicate that we must more correctly rewrite the above relation as

\begin{equation}
\label{LambdaRenorm2sc}
\widehat{\Lambda}^{(FLRW)}\,=\,{\Lambda}^{(phys)}_{\widehat{z}\,(c)}\,+\,\frac{9}{2\pi}\,\frac{d_{\widehat{z}\,(c)}\left[\widehat{\mathbb{C\,S}}_{\widehat{z}\,(c)},\;{\mathbb{C\,S}}_{\widehat{z}\,(c)}\right]}{\widehat{r}^{\,4}(\,\widehat{z}\,(c))}\,,
\end{equation}
\noindent
which implies that the physical cosmological constant ${\Lambda}^{(phys)}_{\widehat{z}\,(c)}$ is scale dependent and that there is a contribution to $\widehat{\Lambda}^{(FLRW)}$ coming from structure formation. As long as this latter correction is negligibly small, its presence and the scale--factorization (\ref{LambdaRenorm2sc}) are not surprising since it is qualitatively in line with the backreaction effects mentioned in the introductory remarks. It turns out, however, that the relative correction $(\widehat{\Lambda}^{(FLRW)}\,-\,{\Lambda}^{(phys)}_{\widehat{z}\,(c)})/{\widehat{\Lambda}^{(FLRW)}}$ is not negligible when in the presence of the typical nonlinear structures dominating our late--epoch cosmological neighborhood, and
the explicit connection (\ref{LambdaRenorm2}), between  $\widehat{\Lambda}^{(FLRW)}\,-\,\widehat{\Lambda}^{(phys)}$ and the functional distance  between  $\widehat{\mathscr{C}}^-(p,\widehat{g})$ and  $\mathscr{C}^-(p,g)$, is a significant aspect of our analysis. To discuss this point, let us observe that we can naturally extend the correction term in (\ref{LambdaRenorm2sc}) to the (instantaneous) observers along the past--directed worldline $\widehat\tau\,\rightarrow\,\widehat{p}_{\widehat\tau}\,:=\,\widehat{\gamma}(\widehat\tau)$,\, $\widehat\tau\,<\,0$,\;of $p:=\,\widehat{\gamma}(0)$. If $\widehat{z}\,(\widehat{\tau},\,c)$ denotes the cosmological decoupling FLRW redshift associated\footnote{The redshift $\widehat{z}\,(\widehat\tau)$ is the reference FLRW redshift  measured by the FLRW observer $(\widehat{p}_{\widehat\tau},\,\widehat{\dot{\gamma}}_{\widehat\tau})$. As usual, we assume that, at a given $\widehat\tau$, there is a corresponding (instantaneous) physical observer $p_{\tau}$\,with\, $\widehat{p}_{\widehat\tau}\,=\,p_{\tau}$, and that the corresponding 4--velocities are distinct. The notational simplifications used in the paper also apply here. The general case handling the role of a nonvanishing relative velocity can be easily obtained using the techniques described in the Appendix.} with $\widehat{p}_{\widehat\tau}$, then we can characterize the cosmological constant correction term 

\begin{equation}
\label{LambdaRenorm3}
\widehat{p}_{\widehat\tau}\,\longmapsto\,\frac{9}{2\pi}\,\frac{d_{\widehat{z}\,(\tau,c)}\left[\widehat{\mathbb{C\,S}}_{\widehat{z}\,(\tau, c)},\;{\mathbb{C\,S}}_{\widehat{z}\,(\tau, c)}\right]}{\widehat{r}^{\,4}(\,\widehat{z}\,(\tau, c))}\,,
\end{equation}
\vskip 0.2cm
\noindent
as an effective field along the reference FLRW observer's worldline $\widehat\tau\,\rightarrow\,\widehat{p}_{\widehat\tau}$. If ${p}_{\widehat\tau}$ is in the homogeneity region, (\ref{LambdaRenorm3}) is perturbatively small, since at that scale, the physical spacetime $(M, g, \gamma_s(\tau))$ is, according to our assumptions, an FLRW perturbation of the reference FLRW background $(M, \widehat{g}, \widehat{\gamma}_s(\widehat\tau))$. The relative fluctuations in matter density  
\begin{equation}
\frac{\Delta_{\widehat{z}\,(\widehat\tau, c)}\,\rho}{\widehat{\rho}_{\widehat{z}\,(\widehat\tau, c)}}
\end{equation}
are very small (for a FLRW redshift $\widehat{z}\,(\widehat\tau)$ probing the CBM surface we have $|(\Delta_{\widehat{z}\,(\widehat\tau,)}\,\rho)/{\widehat{\rho}_{\widehat{z}\,(\widehat\tau)}}|\,\simeq\, 10^{\,-\,5}$) with a correspondingly small contribution,  of perturbative nature, to the area distance fluctuating 
terms  $\delta^{\;(1)}_{\widehat{\mathbb{CS}}}(\,\widehat{z}\,(\widehat{\tau}, c))$, \; and $\delta^{\;(2)}_{\widehat{\mathbb{CS}}}(\,\widehat{z}\,(\widehat{\tau}, c))$. The cosmological decoupling region marked, according to $\widehat{p}_{\widehat\tau}$,  by the reference redshift $\widehat{z}\,(\widehat{\tau}, c)$, is very small (approaching $\widehat{z}\,(\widehat{\tau}, c)\rightarrow\,0$) since the physical (instantaneous) observer ${p}_{\tau}$, associated with  $\widehat{p}_{\widehat\tau}$, is only perturbatively affected by the surrounding structure formation gravitational dynamics. 
However, as $\widehat{p}_{\widehat\tau}$ enters the pre--homogeneity region and approaches $p$,\, the situation changes in a significant way.
On can easily rework out the lower bound (\ref{CosmCorr}) at $\widehat{z}\,(\widehat{\tau}, c)$ and exploit the relation (\ref{LambdaRenorm2sc})  to write

\begin{align}
\label{ourBound}
&\frac{9}{2\pi}\,\frac{d_{\widehat{z}\,(\tau, c)}\left[\widehat{\mathbb{C\,S}}_{\widehat{z}\,(\tau, c)},\;{\mathbb{C\,S}}_{\widehat{z}\,(\tau, c)}\right]}{\,\Lambda^{(FLRW)}\,\widehat{r}^{\,4}(\,\widehat{z}\,(\tau, c))}\nonumber\\
\\
&\geq\,
\frac{6}{(48)^2}\,\frac{\widehat{\Omega}^2_m (\widehat{p}_{\widehat{\tau}})}{\widehat\Omega_{\widehat\Lambda}}\,\left(\frac{\Delta_{\widehat{z}\,(\tau, c)}\,\rho}{\widehat{\rho}_{\widehat{z}\,(\tau, c)}}\right)^2\,\left(1\,+\,\widehat{z}\,(\tau, c)\right)^2\,\widehat{z}^{\,2}(\tau, c)\,,\nonumber
\end{align}
\vskip 0.2cm\noindent
where the normalization to $\Lambda^{(FLRW)}$ is a consequence of the assumed constancy of 
the FLRW cosmological constant. The other cosmological parameters, such as  $\widehat{\Omega}^2_m (\widehat{p}_{\widehat{\tau}})$, are as the notation suggests, referred to the observer  $\widehat{p}_{\widehat\tau}$. Thus, the contribution of the effective field (\ref{LambdaRenorm3}) in the balance relation (\ref{LambdaRenorm2sc}) increases as structure formation proceeds. As we enter into the fully non--perturbative regime associated with an FLRW observer $p$ located deeply into the pre--homogeneity region, the contribution of late--epoch structures can be quite significant. To wit, for very low FLRW redshift  $\widehat{z}\,\sim\,10^{-\,4}$ we typically probe late--epoch  cosmic structures with density contrast ${\tfrac{\Delta_{\widehat{z}_c}\rho}{\widehat{\rho}_{\widehat{z}_c}}\sim 10^6}$. We are in a high nonlinear regime for which the right member of  (\ref{ourBound}) is $\mathcal{O}(1)$, as the following analysis shows.

\section{An example: Galaxy Clusters}
Galaxy clusters are quasi--equilibrium systems that provide a primary testing ground in modern high--precision cosmology\cite{Allen}. The statistics likelihood of their occurrence in a given model universe is quite sensible to the corresponding cosmological model parameters, in particular from $\widehat\Omega_\Lambda$, and observational sky surveys have promoted their use to test and constrain cosmological parameters, even at the level to discriminating among the possible alternative explanations of the nature of dark energy\cite{Allen}. Gravitational nonlinearities on scales of a few Mpc are typical of clusters, and current surveys, such as the Dark Energy Survey\cite{Abbott}, provide data and information about the nonlinear regime, whose analysis requires a delicate interplay between numerical simulations and analytical tools. Thus, it is unsurprising that clusters feature in our analysis, providing further evidence of the strong constraint that (\ref{ourBound}) puts on the interpretation of $\Omega^{(FLRW)} _\Lambda $ as the density of constant dark energy.
Clusters are structures that tend to virialize, where redshifts are due to the local relative motion of the constituent galaxies and their gravitational fields. The neighborhood of a physical observer located at a cluster barycenter $p$ is decoupled from cosmological expansion. The local redshift $z$ does not have a cosmological meaning, and the physical observer refers to the FLRW redshift $\widehat{z}$ to characterize the peculiar velocities of her neighboring galaxies with respect to the idealized FLRW Hubble flow. The situation changes as we gather our astrophysical data from the cluster's outskirts, where the infalling galaxies couple to the cosmological expansion. The radius of this region defines the cluster's \emph{physical radius}. A typical cluster has a physical radius that can be estimated to be in the range of $\sim\,(1\,-\,3)\,\mathrm{Mpc}$. In our case, for an observer located at the barycenter of the Local Group, the physical radius is of the order  $(0.95\,-\,1.05)\,\mathrm{Mpc}$ (see \emph{e.g.}, section 4 of \cite{Obinna3}). 
 The reference FLRW redshift $\widehat{z}_c$  associated with a cluster radius is formally given by the approximated formula derived from the linear Hubble law, \emph{viz.}, $\mathrm{CRad}_{(c)}\simeq\,\tfrac{\widehat{z}_c}{H_0}$, \,\, where $\mathrm{CRad}_{(c)}\,=\,r_{(c)}\,\mathrm{Mpc}$ is the radius of the cluster (in $\mathrm{Mpc}$), and where, for pedagogical purposes, we use  $H_0= 100\,h\,Km\,s^{-1}\,\mathrm{Mpc}^{-1}$ with $h=\tfrac{70}{100}$. Thus, the FLRW reference redshift $\widehat{z}_c$ associated with the cluster radius $\mathrm{Rad}_{(c)}$ can be estimated as $\widehat{z}_{c}\,\simeq \,\frac{7}{3}\,r_{(c)}
\times\,10^{-4} \simeq\,10^{-4}$, where we factorized out  the prefactor $(7/3)\, r_{(c)}$ that depends on the vagaries of the cluster considered (\cite{Obinna3} provides, as a reference, the figure $\widehat{z}_c\,\simeq\,2.4\,\times\,10^{-4}$).  This is the pivotal value of the redshift marking the \emph{inner cosmological boundary} of the pre--homogeneity region in the cluster (the Local Group) below which spacetime uncouples from the FLRW background cosmological expansion. It is worthwhile to stress that we are here dealing with an explicit version\footnote{Another significant development of the notion of Ellis' \emph{finite infinity} has been put forward by D. L. Wiltshire\cite{DavidW}.} of the notion of \emph{finite infinity}, presciently introduced by George Ellis in 1984 \cite{EllisPadova}. Explicitly, let us stress that the sphere centered on the observer, with a radius equal to the physical radius of the cluster, isolates the virialized core of the cluster from the expanding Universe. The 3--dimensional timelike cylinder generated by the world history of this sphere is the associated Ellis' finite infinity surface. This 3--dimensional cylinder is the hypersurface on which to impose the cosmological coupling of the subsystem defined by our cluster (Local Group) and the outside cosmological spacetime.   From this finite infinity point of view, the minimum $E_{\widehat{\mathbb{C\,S}}_{\hat{z}},\;{\mathbb{C\,S}}_{z}}:=\inf_{\zeta_{(\widehat{z})}}\,E_{\widehat{\mathbb{C\,S}}_{\hat{z}},\;{\mathbb{C\,S}}_{z}}[\zeta_{(\widehat{z})}]$ of the celestial spheres distance functional (\ref{Efunctn}) takes its maximum at the redshift $\widehat{z}_c$, and characterizes the functional distance between the physical and the reference FLRW past light cones.     
Inhomogeneities, in the form of the inner core density contrast,  leave their imprint on the boundary defined by the cluster's physical radius through two related mechanisms: \emph{(i)}\; A direct geometrical alteration of signal propagation that via the relation (\ref{OmegatermsCorr}) 
is described by  $\delta^{\;(1)}_{\widehat{\mathbb{C\,S}}_{\hat{z}},\;
{\mathbb{C\,S}}_{{z}}}$; \; and \emph{(ii)} A positive contribution to the FLRW cosmological constant $\widehat{\Lambda}^{(FLRW)}$ due to  $\delta^{\;(2)}_{\widehat{\mathbb{C\,S}}_{\hat{z}},\;
{\mathbb{C\,S}}_{{z}}}$, or, more expressively, to the distance functional $E_{\widehat{\mathbb{C\,S}}_{\hat{z}},\;
{\mathbb{C\,S}}_{{z}}}$. If we assume the typical cluster density contrast  $|\widehat{\rho}_{\hat{z}}(p)-\rho_{\hat{z}}(p)|/\widehat\rho_{\hat{z}}(p)\sim\,10^6$, and take the \emph{Planck} value $\widehat\Omega_m\sim 0.3$ for the average FLRW density\cite{Planck}, then as we approach the \emph{inner cosmological boundary} at $\widehat{z}\sim\ 10^{-4}$, where cluster dynamics decouples from cosmological expansion,  the relation (\ref{OmegatermsCorr}) provides 
$|\delta^{\;(1)}_{\widehat{\mathbb{C\,S}}_{\hat{z}},\;
{\mathbb{C\,S}}_{{z}}}|\sim\,10^{-5}$.\, A result that is in very good agreement with the estimates of $\delta^{\;(1)}$ obtained in this regime using FLRW perturbation theory\cite{Durrer, Fanizza, Heinesen, Obinna4}. From the inequality  
$\delta^{\;(2)}_{\widehat{\mathbb{C\,S}}_{\hat{z}},\;
{\mathbb{C\,S}}_{{z}}}\,\geq\,(\delta^{\;(1)}_{\widehat{\mathbb{C\,S}}_{\hat{z}},\;
{\mathbb{C\,S}}_{{z}}})^2$ we have 
$\delta^{\;(2)}_{\widehat{\mathbb{C\,S}}_{\hat{z}},\;
{\mathbb{C\,S}}_{{z}}}\,\geq\, 10^{-10}$, and by taking the concordance  $\Lambda\mathrm{CDM}$ model value $\widehat\Omega_\Lambda\simeq\,0.7$, the relations (\ref{OmegatermsCorr}) and (\ref{ourBound}), evaluated for $\widehat{z}=\widehat{z}_c\sim 10^{-4}$, provide 
\begin{equation}
\label{OmegatermsCorr2}
\frac{\Lambda^{(FLRW)}_{\hat{z}}\,-\,{\Lambda}^{(phys)}}{{\Lambda}^{(phys)}}\,=\,\left.
\frac{6\,\delta^{\;(2)}_{\widehat{\mathbb{C\,S}}_{\hat{z}},\;{\mathbb{C\,S}}_{{z}}}}{\widehat\Omega_\Lambda\left(1\,+\,\widehat{z}\right)^{\;2}\,\widehat{z}}\right|_{\widehat{z}_c}\,\gtrsim\,\frac{1}{12}.
\end{equation}   
 Thus, as we approach the \emph{inner cosmological boundary} in the pre--homogeneity region, the contributions (\ref{LambdaRenorm2}) and (\ref{LambdaRenorm2sc}) to the cosmological constant are of the same order of magnitude as the assumed FLRW contribution $\widehat{\Lambda}^{(FLRW)}$. 
If we assume a Copernican point of view, then the FLRW observer is strongly tempted to  identify 
$\Lambda^{(cont)}_{\hat{z}}$ with an average contribution to $\widehat{\Lambda}^{(FLRW)}$, statistically mediated over all fundamental observers, and generated by the growth of structures in the pre--homogeneity regions. This latter point of view is potentially significant for interpreting late--epoch observations in the presence of inhomogeneities.

\section{Discussion}
We have described how a rigorous analysis of the relation between the reference FLRW and the physical celestial spheres can strongly impact the interpretation of dark energy as a constant field described by the FLRW cosmological constant ${\widehat\Lambda^{(FLRW)}}$. In this analysis, we highlighted a novel technique based on the role of a distance functional between celestial spheres, a functional that originates in imaging and that can be directly related to the relative fluctuations in area distance, an observable quantity whose connection with the growth of cosmic structure is intensively investigated (see \emph{e.g.}, \cite{Huterer} for an informative review). In the late epoch, when the cosmological expansion decouples from the local virialized gravitational dynamics of cosmic structures, our distance functional keeps the memory of the non--perturbative density contrast associated with gravitational clustering. It provides a (FLRW)redshift--dependent positive contribution to the cosmological constant, which, at late epoch, is of the same order of magnitude as the FLRW cosmological constant itself. It is important to stress that regardless of its size, this contribution is a rigorous mathematical consequence of describing the physical past light cone from the reference point of view of the FLRW light cone. It is a highly nonlinear effect that cannot be easily controlled by order--of--magnitude estimates, as was often stated in the past\cite {Weinberg}. According to the scale--dependent factorization (\ref{LambdaRenorm2sc}),  the result proven here can be interpreted, from the mainstream FLRW point of view, as a parameter splitting of ${\widehat\Omega_\Lambda^{(FLRW)}}$ in two dark energy sectors  induced by 
the effective field (see (\ref{LambdaRenorm3}))
  
\begin{equation}
\label{LambdaRenorm3effF}
\frac{9}{2\pi}\,\frac{d_{\widehat{z}\,(\tau,c)}\left[\widehat{\mathbb{C\,S}}_{\widehat{z}\,(\tau, c)},\;{\mathbb{C\,S}}_{\widehat{z}\,(\tau, c)}\right]}{\widehat{r}^{\,4}(\,\widehat{z}\,(\tau, c))}\,,
\end{equation}
\vskip 0.2cm\noindent
that can be thought of as associated with the \emph{Finite Infinity Surface} surrounding each fundamental FLRW observer when evolving in the pre--homogeneity
 region. We have a physical sector ${\Lambda^{(phys)}_{\hat{z}}/3H_0^2}$ describing, for low FLRW redshift ${\widehat{z}}$, dark energy as a  term of yet unknown origin; and the cosmic structure term  ${\Lambda^{(cont)}_{\hat{z}}/3H_0^2}$ describing  the effect of nonlinear structures. For the cosmological mass parameter ${\widehat\Omega^{(FLRW)}_m}$, geometry--growth splittings of this sort are familiar\cite{Huterer}  when comparing geometry and growth of structures. But it must be stressed that whereas the usual geometry--growth splitting can be quite ambiguous, our analysis provides a strong foundation to the splitting for the dark energy parameter ${\widehat\Omega_\Lambda^{(FLRW)}}$.  
Our splitting is scale--dependent, reaching its maximum on the finite infinity surface surrounding the given FLRW observer, where the virialized structures couple with the background FLRW cosmological expansion. In these regions, typically identified with the physical radius of clusters, this effective field contributes as an effective cosmological constant of the same order of magnitude as the assumed $\Lambda^{(FLRW)}$. According to the nature of the effective field  (\ref{LambdaRenorm3}), it is clear that the explicit redshift parametrization of our ${\widehat\Omega_\Lambda^{(FLRW)}}$  splitting bears relevance to the coincidence problem,  according to which the expansion of the Universe appears to have started accelerating at the same epoch when complex nonlinear structures emerged.
As such, it is worth emphasizing that the effective field contribution emerging from the geometrical relation between the physical and reference light cones is by construction scale--dependent and, consequently, redshift--dependent. This feature arises naturally within our non--perturbative framework and does not rely on additional dynamical assumptions. Remarkably, this behavior is consistent with recent observational analyses by the DESI Collaboration \cite{DESI1}, which indicates that dark energy may exhibit a redshift evolution of its equation--of state parameter $w(z)$. In this respect, our result offers a non--exotic interpretation of such an evolution, suggesting that the apparent dynamics of dark energy could reflect the underlying scale dependence of spacetime inhomogeneities. 
It would be particularly relevant, in future work, to quantify explicitly the redshift dependence of this correction term and compare it with the DESI findings. Such an analysis is technically delicate, since it requires a consistent treatment of the distance functional at different redshifts and a refined understanding of how the geometrical distance redshift relation emerges in the inhomogeneous regime. Nevertheless, this investigation could shed light on the link between structure formation and the apparent dynamical behavior of dark energy, offering new insights into the coincidence problem.

\subsection*{Acknowledgments}
It is our pleasure to record our appreciation to Thomas Buchert, George Ellis, Giuseppe Fanizza, and Sabino Matarrese for their invaluable comments and suggestions. 

\textbf{Data availability statement}  No datasets were generated or analyzed during the current study.

\textbf{Authors contributions} The two authors contributed to the paper and extensively discussed the results.

\textbf{Competing interests} The authors declare no competing interests.

\vfill\eject

\section{A Geometrical Analysis Appendix}
\label{Appendice}
\appendix
 In this commented appendix, we provide a detailed discussion of the geometrical analysis results we have used. It comprises approximately seventy pages of relevant mathematical guidelines to the complex landscape that, by way of contrast, we have too briefly described in the main body of our text. We start with some (well-known) notational remarks.
 
\section{The standard FLRW model in a nutshell}
\label{capsule}
The weak cosmological principle and the attendant Friedman-Lema\^itre-Robertson-Walker (FLRW) spacetime play a primary role in the current interpretative standard of cosmology. Let us recall that in terms of  the FRLW coordinates $\left(\hat{r}, \hat{\theta},\hat{\varphi}\right)$, and of the proper time $\widehat{\tau}$ of the comoving fundamental FLRW observers, the FLRW spacetime metric $\widehat{g}$ reads 
\begin{align}
\widehat{g}\,&:=\,-d\hat\tau^2\,+\,a^2(\hat\tau)\,\left[d\hat{r}^2\,+\,f^2(\hat{r})\,\left(d\hat\theta^2\,+\,\sin^2\hat\theta\,d\hat\varphi^2  \right)\right]\,,\nonumber\\
\label{FLRWg0}\\
 f(\hat{r})\, &:=\,
    \begin{cases}
       & \frac{1}{\sqrt{k}}\sin\,\left(\sqrt{k}\,\hat{r}\right),\;\;\;k\,>\,0\\
       & \hat{r},\;\;\;\;\;\;\;\;\;k\,=\,0\\
       & \frac{1}{\sqrt{-\,k}}\sinh\,\left(\sqrt{-\,k}\,\hat{r}\right),\;\;k\,<\,0\,,
    \end{cases}       
\nonumber
\end{align}
where $a(\widehat\tau)$ is the time-dependent (dimensionless) scale factor, and $k$ is the spatial curvature (in units of inverse area) of the 3-dimensional spatial sections
\begin{equation}
\left(\widehat{V}_{\widehat\tau}, \widehat{g}^{(3)}_{\widehat{\tau}}\right)\,:=\,\left\{\left.q=(\widehat{r}(\widehat{\tau}), \widehat{\theta}(\widehat{\tau}), \widehat{\phi}(\widehat{\tau}))\in\,M\,\right|\,\widehat{g}^{(3)}_{\widehat{\tau}}\,=:\,\widehat{g}|_{\tau=const}    \right\}\,,
\end{equation}
$\widehat{g}$-orthogonal to the fundamental FLRW observers. These latter characterize the FLRW Hubble flow defined by the family of preferred chronological geodesics 
\begin{align}
\label{FLRWobservers0}
\widehat{\gamma}_s\,:\,\mathbb{R}_{>0}\,&\longrightarrow  \,(M, \widehat{g})\\
\widehat{\tau}\,&\longmapsto \,\widehat{\gamma}_s(\widehat{\tau})=(\widehat{r}(\widehat{\tau}), \widehat{\theta}(\widehat{\tau}), \widehat{\phi}(\widehat{\tau}))\;,\nonumber
\end{align}
with $4$-velocity $\dot{\gamma}_s\,:=\,\frac{d\gamma_s(\tau)}{d\tau}$,\,\,$g(\dot\gamma_s,  \dot\gamma_s)\,=\,-1$, and where the label $s:=(\widehat{r}, \widehat{\theta}, \widehat{\phi})|_{\tau=0}$ is a shorthand notation for the comoving (Lagrangian) coordinates adapted to the flow.

The mass-energy content of FLRW spacetimes $(M, \widehat{g}, \widehat{\gamma}_s(\widehat{\tau}))$ has the perfect fluid structure (but does not necessarily comply with the equation of state of a perfect fluid). If $\widehat{\rho}^{(FLRW)}$ and $\widehat{p}^{(FLRW)}$ respectively denote the matter density and the pressure of the fluid, then
\begin{equation}
\widehat{T}_{ab}^{(FLRW)}\,=\,\left(\widehat{\rho}^{(FLRW)}\,+\,\widehat{p}^{(FLRW)}   \right)\,\widehat{u}_a\,\widehat{u}_b\,+\,\widehat{p}^{(FLRW)}\,\widehat{g}_{ab}\,,
\end{equation}
where $\widehat{u}\,=\,\frac{d\,\widehat{\gamma}(\widehat\tau)}{d\widehat{\tau}}$ is the 4-velocity of the FLRW fundamental observers $\tau\,\longrightarrow\,\widehat{\gamma}(\widehat{\tau})$ describing the cosmic fluid flow world-lines. Note that $\widehat{\rho}^{(FLRW)}$ describes all matter contributions to the energy density (baryonic, radiation, and dark matter) but not the dark energy contribution. The FLRW Hubble parameter $H(\widehat{\tau})$ and the deceleration parameter $q(\widehat{\tau})$ describing the relative variation of the scale factor $a(\widehat{\tau})$ and its (dimensionless) second-time derivative are respectively defined by the relations
\begin{equation}
\label{Handq}
H(\widehat{\tau})\,:=\,\frac{d}{d\widehat{\tau}}\,\ln\,a(\widehat{\tau})\,,\;\;\;
q(\widehat{\tau})\,:=\,-\,\frac{\ddot{a}(\widehat{\tau})}{a(\widehat{\tau})\,H^2(\widehat{\tau})}\,,
\end{equation}
where $\ddot{a}(\widehat{\tau}):=\tfrac{d^2}{d\widehat{\tau}^2}\,{a}(\widehat{\tau})$.
Another basic quantity related to the scale factor $a(\widehat{\tau})$ is the cosmological redshift $\widehat{z}$ of an emitting source,  stationary with respect to the FLRW fundamental observers, \emph{i.e.}
\begin{equation}
\label{defredshift}
\widehat{z}(\widehat\tau)\,:=\,\frac{\widehat{\lambda}(0)-\widehat{\lambda}(\widehat{\tau})}{\widehat{\lambda}(\widehat{\tau})}\,=\,\frac{a(\widehat\tau=0)}{a(\widehat\tau)}
\end{equation}
where $\widehat{\lambda}(\widehat{\tau})$ and $\widehat{\lambda}(\widehat{\tau}=0)$ respectively are the emitted and the received wavelengths of the signal.
If we normalize (at the present time $\widehat{\tau}=0$) the FLRW metric by setting $a(\widehat\tau=0)\,=\,1$, then we can express the scale factor $a(\widehat\tau)$ in terms of $\widehat{z}$\;\, by setting
\begin{equation}
\label{aScaling}
a(\widehat{z})\,=\,\frac{1}{1\,+\,\widehat{z}(\widehat\tau)}\,.
\end{equation}
The Einstein field equations and the matter evolution equations associated with FLRW spacetimes are  
\begin{equation}
\label{FLRWeinstein}
R_{ab}(\widehat{g})\,-\,\frac{1}{2}\,R(\widehat{g})\,\widehat{g}_{ab}\,+\,\widehat{\Lambda}^{(FLRW)}\,\widehat{g}_{ab}\,=\,8\pi\,\widehat{T}_{ab}^{(FLRW)}\,,\;\;\;\;\widehat{\nabla}^a\,\widehat{T}_{ab}^{(FLRW)}\,=\,0\,,
\end{equation}
where $R_{ab}(\widehat{g})$ and $R(\widehat{g})$ are the Ricci tensor and the scalar curvature of the FLRW metric; we denoted by $\widehat{\Lambda}^{(FLRW)}$ the FLRW cosmological constant, and where $\widehat{\nabla}$ is the covariant derivative of the Levi-Civita connection associated with $\widehat{g}$. For the FLRW family of metrics (\ref{FLRWeinstein}) reduces to the Friedmann-Lema\^itre equations (see \emph{e.g.} \cite{Ellis2}) 
\begin{equation}
\label{E1}
\frac{\ddot{a}(\widehat{\tau})}{{a}(\widehat{\tau})}\,=\,-\,\frac{4\pi}{3}\left(\widehat{\rho}^{(FLRW)}(\widehat{\tau})\,+\,{3 \widehat{p}^{(FLRW)}(\widehat{\tau})}  \right)\,
+\,\frac{1}{3}\,\widehat{\Lambda}^{(FLRW)}\,.
\end{equation}
\begin{equation}
\label{E2}
\frac{d}{d\tau}\,\widehat{\rho}^{(FLRW)}(\widehat{\tau})\,+\,3\,H(\widehat{\tau})
\left(\widehat{\rho}^{(FLRW)}(\widehat{\tau})+\widehat{p}^{(FLRW)}(\widehat{z}) \right)\,=\,0\,,
\end{equation}
\begin{equation}
\label{E3} 
H^2(\widehat{\tau})\,=\,\frac{8\pi}{3}\,\widehat{\rho}^{(FLRW)}(\widehat{\tau})\,+\,\frac{\widehat{\Lambda}^{(FLRW)}}{3}\,-\,\frac{k}{a^2(\widehat{\tau})}\,.
\end{equation}
If $\dot{a}(\widehat{\tau})\not=0$, these expressions reduce to the conservation law (\ref{E2}) and the Friedmann equation (\ref{E3}). In this connection, it is useful to introduce the dimensionless density parameters respectively associated with the matter density $\widehat{\rho}^{(FLRW)}$, the radiation density $\widehat{\rho}_r^{(FLRW)}$,  the cosmological constant $\widehat{\Lambda}^{(FLRW)}$, and the curvature $k$ of the spatial sections, 
\begin{align}
&\widehat{\Omega}^{(FLRW)}_m(\widehat{\tau})\,:=\,
\frac{\widehat{\rho}^{(FLRW)}}{\widehat{\rho}_{crit}(\widehat{\tau})},\;\;\;\widehat{\Omega}^{(FLRW)}_r(\widehat{\tau})\,:=\,
\frac{\widehat{\rho}_r^{(FLRW)}}{\widehat{\rho}_{crit}(\widehat{\tau})}
\,,\nonumber\\
\\
&\widehat{\Omega}^{(FLRW)}_\Lambda(\widehat{\tau})\,:=\,
\frac{\widehat{\Lambda}^{(FLRW)}}{8\pi\,G\,\widehat{\rho}_{crit}(\widehat{\tau})},\;\;\;
\widehat{\Omega}^{(FLRW)}_k(\widehat{\tau})\,:=\,
-\,\frac{3\,k}{8\pi\,G\,a^2(\widehat\tau)\widehat{\rho}_{crit}(\widehat{\tau})}\,\,,\nonumber
\end{align}
where
\begin{equation}
\widehat{\rho}_{crit}(\widehat{\tau})\,:=\,\frac{3\,H^2(\widehat{\tau})}{8\pi\,G}
\end{equation}
is the critical matter density of the FLRW model. (Even if we typically assume $c=1=G$, for clarity, we have written the gravitational constant $G$ in the above densities.) The present-day values (denoted by a subscript "$\,0$") of the various contributions of these constituents  are estimated as \cite{Planck0, Planck} 
\begin{equation}
\label{cosmoparameters1}
\widehat{\Omega}^{(FLRW)}_{k\,0}\le0.01\,, \,\,\,\,\,\,\,\, \widehat{\Omega}^{(FLRW)}_{r\,0}=9.4\times10^{-5}\,, \,\,\,\,\,\,\, \widehat{\Omega}^{(FLRW)}_{m\,0}=0.32\,.
\end{equation}
Notice that $\widehat{\Omega}^{(FLRW)}_{k\,0}$ is usually considered negligible (assuming flat spatial section and that the local fluctuations in curvature are irrelevant); the matter contribution splits into the ordinary matter term $(\sim 5\%)$ indicated with the subscript "$b\,0$"(baryons) and the cold dark matter term $(\sim 27\%)$ indicated with the subscript "$CDM\,0$": 
\begin{equation}
 \widehat{\Omega}^{(FLRW)}_{b\,0}=0.05\,,\,\,\,\,\,\, \widehat{\Omega}^{(FLRW)}_{CDM\,0}=0.27\,.
\end{equation}
In terms of the present-day densities, the Friedmann equation (\ref{E3}) for $\tau\,=\,0$ provides the Bahcall \emph{Cosmic Triangle} relation \cite{Bahcall}  
\begin{equation}
\widehat{\Omega}^{(FLRW)}_{k\,0}+\widehat{\Omega}^{(FLRW)}_{r\,0}+\widehat{\Omega}^{(FLRW)}_{m\,0}+\widehat{\Omega}^{(FLRW)}_{\Lambda\,0}=1\,,
\end{equation}
conveniently representing the dynamical status of the Universe. Evaluated for the spatially flat case, $\widehat{\Omega}^{(FLRW)}_{k\,0}\,\approx\,0$, it provides the present-day contribution to dark energy,
\begin{equation}
\label{cosmoparameters2}
\widehat{\Omega}^{(FLRW)}_{\Lambda\,0}=0.68\,.
\end{equation} 
As long as the pressure $\widehat{p}^{(FLRW)}$ is assumed to be small as compared with the matter density $\widehat{\rho}^{(FLRW)}$, we can exploit (\ref{aScaling}) and the scaling properties of the matter and radiation densities \cite{Ellis2}, to write the Friedmann equation (\ref{E3}) in terms of the (normalized) Hubble rate of expansion ratio $H(\widehat{z})/H_0$, according to 
\begin{equation}
\label{AccaZeta}
\left(\frac{H(\widehat{z})}{H_0}\right)^2\, =\,\left(\widehat{\Omega}^{(FLRW)}_{b\,0}\,+\,\widehat{\Omega}^{(FLRW)}_{CDM\,0}\right)\,(1\,+\,\widehat{z})^3\,+\, 
\widehat{\Omega}^{(FLRW)}_{r\,0}\,(1\,+\,\widehat{z})^4\,+\,\widehat{\Omega}^{(FLRW)}_{\Lambda\,0}\,.
\end{equation}
In the FLRW modeling, the spacetime exhibits a constant matter density distribution $\widehat{\rho}(\widehat{\tau})$ on spatial hyper-surfaces at constant proper time $\widehat{\tau}$, and we can confidently deal, via FLRW perturbation theory, with small amplitude fluctuations $({\rho}(\widehat{\tau})-\widehat{\rho}(\widehat{\tau}))/\widehat{\rho}(\widehat{\tau})$ in the observed physical matter density field ${\rho}(\widehat{\tau})$. As small as they may be, these fluctuations play a fundamental role in the early  Universe. Acting as seeds of gravitational clustering, they give rise to the complex network of structures we observe today. This straightforward geometric scenario together with the assumption of a \emph{constant Dark Energy}, in the form of a cosmological constant ${\widehat{\Lambda}^{(FLRW)}}$, \emph{Cold Dark Matter} (CDM), and inflation, is the stage of the standard paradigm of modern cosmology, the $\Lambda\mathrm{CDM}$ model. The resulting picture is indeed quite successful in delivering a relatively accurate physical and geometrical representation of the Universe in the present era\footnote{Characterized by the actual temperature of the cosmic microwave background $T_{CMB}\,=\,2.725\,K$ as measured  in the frame centered on us
but stationary with respect to the CMB.} and  over spatial scales ranging from\footnote{
The actual averaging scale marking the statistical onset of isotropy and homogeneity is still much debated \cite{ThomMinkowski}. For the sake of the argument presented in this paper, we adopt the rather conservative estimate of the scales over which an average
isotropic expansion is seen to emerge, namely $70-120\,h^{-1}Mpc$, and ideally extending to a few times this scale \cite{Wiltshire}.} $\approx  \, 100\,h^{-1}\; \mathrm{Mpc}$ to the visual horizon of our past lightcone  \cite{Gott}, \cite{Hogg}, \cite{Scrimgeour}, where $h$ is the dimensionless parameter describing the relative uncertainty of the actual value of the present-epoch  Hubble-Lema\^itre constant $H_0$.

\section{The physical spacetime and its FLRW avatar}
\label{avatar}
The key ingredient of our analysis is a list of properties of the actual physical spacetime, which, for the reader's convenience, we introduce as a phenomenological definition of sorts.
\vskip 0.5cm \noindent
\begin{definition}(\emph{The physical cosmological model})
\label{PhenDefin}
We model the actual Universe with a cosmological spacetime
$(M, g, \gamma_s(\tau))$ complying with the following assumptions
\begin{itemize}
\item{$M$ is a four-dimensional smooth manifold, diffeomorphic to $V^{(3)}\,\times\,\mathbb{R}$ where $V^{(3)}\simeq\, \mathbb{R}^3$\, or $\mathbb{S}^3$ (the three-sphere) is the three-manifold topologically modeling the cosmological space sections.}
\item{$M$ is endowed with a smooth Lorentzian metric $g$. When considering the geometry of the associated lightcone, the regularity of $g$  can be lowered by requiring only the continuity and H\"older continuity of the components of the metric and their first derivatives.}
\item{The Lorentzian manifold $(M, g)$ is required to be statistically isotropic and homogeneous in Neyman's sense, \emph{viz.}\, $(M, g)$ is diffeomorphic to a perturbed FLRW metric $\widehat{g}$,  only on scales larger than the scale $L_0\,\gtrsim\,100\mathrm{h}^{-1}\; \mathrm{Mpc}$ that marks the transition to homogeneity.}
\item{ The metric $g$ is a solution of the Einstein field equations and of the associated  matter evolution equations 
\begin{equation}
\label{einsteinEFE}
R_{ab}({g})\,-\,\frac{1}{2}\,R({g})\,{g}_{ab}\,+\,{\Lambda}^{(phys)}\,{g}_{ab}\,=\,8\pi\,{T}_{ab}\,,\;\;\;\;{\nabla}^a\,{T}_{ab}\,=\,0\,,
\end{equation}
where $R_{ab}({g})$ and $R({g})$ are the Ricci tensor and the scalar curvature of the metric $g$; we denoted by ${\Lambda}^{(phys)}$ the physical cosmological constant (assumed to be potentially distinct from the corresponding FLTW value ${\Lambda}^{(FLRW)}$), and where ${\nabla}$ is the covariant derivative of the Levi-Civita connection associated with ${g}$.}
\item{The energy-momentum tensor $T_{ab}$ describes matter and radiation that are distributed in a statistically isotropic and homogeneous way only on scales larger than $L_0$; their local physical properties and dynamics are determined by the appropriate theoretical and phenomenological models provided by astrophysical observations.}
\item{We associate to the spacetime manifold $(M, g)$ 
the \emph{phenomenological fundamental observers} defined by a family of preferred chronological geodesics 
\begin{align}
\label{observers0}
\gamma_s\,:\,\mathbb{R}_{>0}\,&\longrightarrow  \,(M, g)\\
\tau\,&\longmapsto \,\gamma_s(\tau)\;,\nonumber
\end{align}
parametrized by proper time $\tau$ and labeled by suitable comoving (Lagrangian) coordinates  $s$ adapted to the flow. 
 We denote by $\dot{\gamma}_s\,:=\,\frac{d\gamma_s(\tau)}{d\tau}$,\,\,$g(\dot\gamma_s,  \dot\gamma_s)\,=\,-1$, the corresponding  $4$-velocity field.}
\item{The fundamental observers (\ref{observers0}) describe the average motion of matter in the Universe and characterize a phenomenological Hubble flow which is (statistically) isotropic and homogeneous only on sufficiently large scales
 $L\,\gtrsim\,L_0$.}
\end{itemize}
\end{definition}
The patchwork of assumptions collected in Definition \ref{PhenDefin}  is admittedly vague, and an explicit characterization of $(M, g, \gamma_s(\tau))$ as a phenomenological background solution of the Einstein equations \cite{KolbMarraMatarrese} is a task that remains to be settled. Yet, the critical point for present consideration is whether a less ambitious task, still providing an explanatory framework in the pre-homogeneity region $L\,<\,L_0$, can be technically addressed. The currently most profitable strategy is to confine the analysis to our past lightcone, where we have direct control over the physical and geometric description of the model and its relation to observations. To make sense of Definition \ref{PhenDefin} in this restricted setting, we must focus on our role as a given fundamental observer. The information we receive travels along our past lightcone, and our perspective is limited to the vantage point associated with a single spacetime event. We are unable to move away from this "here and now," a consequence of the fact that the observed region of the Universe is approximately on the order of $10^{10}$ light-years, whereas the diameters of our galaxy and the Local Group are much smaller. Thus, even with long-term astronomical observation programs, the observative time scales are tiny, even at the galactic and cluster scales. Yet, for technical reasons, we must associate with the  {phenomenological observer} a small proper time scale around the chosen data-gathering event. We collect these remarks in the following characterization of the phenomenological (instantaneous) observer.

\begin{definition}(\emph{The phenomenological observer})
\label{DefPhysObs}
Let $\mathbb{R}_{>0}\,\ni\,\tau\,\longmapsto\,\gamma(\tau)\,\in\,(M, g)$ denote the worldline of an observer in the family of phenomenological fundamental observers defined by (\ref{observers0}). To select a (small) proper time observational scale around the data gathering event $p\,\in\,(M, g)$, let fix a proper-time interval $-\delta<\tau<\delta$,\;\,  $\delta>0$, and let
\begin{equation}
\mathbb{R}_{>0}\,\ni\,\tau\,\longmapsto\,\gamma(\tau)\,\in\,(M, g),\;\;\;-\delta<\tau<\delta,\;\;\; \gamma(\tau=0)\,=:\,p
\end{equation}
be the corresponding geodesic segment contained in $\tau\,\longmapsto\,\gamma(\tau)$ and centered around the spacetime point $p\,:=\,\gamma(\tau=0)$ reached with a 4-velocity 
$\dot{\gamma}(p)\,:=\,\dot{\gamma}(\tau=0)$ by the chosen observer. The pair 
\begin{equation}
\label{physObserver1}
\left(p,\;\dot{\gamma}(p)\right)
\end{equation}
characterizes the (instantaneous) phenomenological observer who gathers cosmological data at the selected event $p$. Received data are organized and described in the phenomenological observer's \emph{Local Rest Frame} (LRF)  defined by the spacetime tangent space $T_pM$ endowed with a $g$-orthonormal frame $\{E_{(i)}\}_{i=1,\ldots,4}$, associated with the $4$-velocity  $\dot{\gamma}(p)$, \emph{i.e.}
\begin{equation}
\label{PhysObserv}
\mathrm{LRF}_{p}\,:=\,\left(T_p M,\,\left.\{E_{(i)}\}_{i=1,\ldots,4}\right|\,g_p\left(E_{(i)}, E_{(k)}\right)=\eta_{ik}, \,E_{(4)}\,:=\,\dot{\gamma}(p)\right)
\end{equation}
where $\eta_{ik}$ is the Minkowski metric in $T_pM$. 
\end{definition}
\vskip 0.5cm\noindent
For simplicity, we often refer to the phenomenological observer as the \emph{physical observer}.

As emphasized by E. Kolb, V. Marra, and S. Matarrese in\cite{KolbMarraMatarrese, KolbMarraMatarrese2},  
FLRW observers and the associated FLRW spacetime have, from a best-fit point of view, a natural role also in the presence of nonperturbative inhomogeneities.
 We implement this role by introducing on the same differential manifold $M$, arena of the physical $(M, g, \gamma(\tau))$, a family of reference FLRW cosmological models that, following the notation established in Section \ref{capsule}  (see (\ref{FLRWg0}) and (\ref{FLRWobservers0})), we denote by $(M, \widehat{g}, \widehat{\gamma}(\widehat{\tau}))$. This family provides the candidate FLRW spacetime potentially best fitting the physical $(M, g, \gamma(\tau))$. We characterize the corresponding reference FLRW observer as follows.
 
\begin{definition}(\emph{The reference FLRW observer})
Let us consider the family (\ref{FLRWobservers0}) of fundamental observers, with $4$-velocity $\dot{\gamma}_s\,:=\,\frac{d\gamma_s(\tau)}{d\tau}$,\,\,$g(\dot\gamma_s,  \dot\gamma_s)\,=\,-1$, defining the Hubble flow of the candidate reference FLRW model $(M, \widehat{g}, \widehat{\gamma}(\widehat{\tau}))$,  
\begin{align}
\label{FLRWobservers1}
\widehat{\gamma}_s\,:\,\mathbb{R}_{>0}\,&\longrightarrow  \,(M, \widehat{g})\\
\widehat{\tau}\,&\longmapsto \,\widehat{\gamma}_s(\widehat{\tau})=(\widehat{r}(\widehat{\tau}), \widehat{\theta}(\widehat{\tau}), \widehat{\phi}(\widehat{\tau}))\;.\nonumber
\end{align}
 Let   $\mathbb{R}_{>0}\,\ni\,\widehat\tau\,\longmapsto\,\widehat{\gamma}(\widehat\tau)\,\in\,(M, \widehat{g})$,\; with $\widehat{\gamma}(\widehat\tau=0)\,\equiv\,p$, denote the worldline of the observer in the family (\ref{FLRWobservers1}) passing through the same observational event $p$ chosen by the phenomenological observer (\ref{physObserver1}).\;For a given proper-time interval $-\widehat\delta\,<\,\widehat\tau\,<\,\widehat\delta$,\;\,with\,  $\widehat\delta>0$, let
\begin{equation}
\mathbb{R}_{>0}\,\ni\,\widehat\tau\,\longmapsto\,\widehat{\gamma}(\widehat\tau)\,\in\,(M, \widehat{g}),\;\;\;-\widehat{\delta}\,<\,\widehat\tau\,<\,\widehat\delta,\;\;\; \gamma(\tau=0)\,=:\,p
\end{equation}
be the geodesic segment contained in $\tau\,\longmapsto\,\gamma(\tau)$ and centered around the spacetime point $p$. If we denote by $\widehat{\dot{\gamma}}(p)\,:=\,\widehat{\dot{\gamma}}(\widehat\tau=0)$ the corresponding FLRW 4-velocity, then  
\begin{equation}
\label{FLRWObs}
\left(p,\;\widehat{\dot{\gamma}}(p)\right)
\end{equation}
characterizes the (instantaneous) FLRW observer who gathers cosmological data at the selected event $p$. The associated  FLRW observer's \emph{Local Rest Frame} is defined by 
\begin{equation}
\label{FLRWObserv}
\widehat{\mathrm{LRF}}_{p}\,:=\,\left(T_p M,\,\left.\{\widehat{E}_{(i)}\}_{i=1,\ldots,4}\right|\,\widehat{g}_p\left(\,\widehat{E}_{(i)},\, \widehat{E}_{(k)}\right)\,=\,\eta_{ik}, \,\widehat{E}_{(4)}\,:=\,\widehat{\dot{\gamma}}(p)\right)\,,
\end{equation}
where $\{\,\widehat{E}_{(i)}\}_{i=1,\ldots,4}$ denotes a $\widehat{g}_p$-orthonormal frame
 associated with the $4$-velocity  $\widehat{\dot{\gamma}}(p)$.
\end{definition}

The reference FLRW observer $(p, \widehat{\dot\gamma}(p))$ is chosen in such a way as to eliminate the dipole anisotropy induced by the motion with respect to the CMB. \, If $(p, \widehat{\dot\gamma}(p))$ is identified with our solar system, this motion is characterized by the velocity  $\simeq 370\; km\,s^{-1}$ along the direction defined by the celestial coordinates $\mathrm{RA}\,=\,168^{\circ}$, and $\mathrm{Dec}\,=\, -7^{\circ}$. However, for our cosmological best-fit purposes it is more appropriate to identify $(p, \widehat{\dot\gamma}(p))$ with the center of mass of our Local Group of galaxies whose motion relative to the CMB is characterized by a velocity $v_{LG}=620\,\pm\,151\,\mathrm{Km}\,s^{-1}$ \cite{Planck}. This CMB-normalization also controls the dipole anisotropy of the physical spacetime $(M, g, \gamma_s(\tau))$ when we probe the homogeneity scale; however, this control does not extend to the pre-homogeneity region, where 
the phenomenological Hubble flow is characterized by the background motion of nearby clusters of galaxies and differs quite significantly from the reference FLRW Hubble flow. To take care of this important degree of freedom,  we assume that the physical observer $(p, {\dot\gamma}(p))$ moves, relative to the CMB-normalized FLRW observer $(p, \widehat{\dot\gamma}(p))$, with a relative velocity $v(p)$, possibly depending on the observational scale. Explicitly, at the observational event $p$, in the FLRW local rest frame (\ref{FLRWObserv}) we can write (with $c=1$)
\begin{equation}
\label{3velocity}
{\dot\gamma}(p)\,=\,\frac{\sum_{a=1}^3\,v^a(p)\widehat{E}_{(a)}+\widehat{\dot{\gamma}}(p)}{\sqrt{1\,-\,|v(p)|^2}}\,,
\end{equation}
where $v(p)$ is the relative 3-velocity of the physical observer $(p, {\dot\gamma}(p))$ with respect to the reference FRW observer $(p, \widehat{\dot\gamma}(p))$ introduced above.

\begin{remark}
It must be stressed that the FLRW instantaneous observer so introduced strictly adopts Neyman's weak cosmological principle \cite{Neyman, Peebles}.   In particular, in her reference role, she interprets the dynamics of the Universe as a realization of a homogeneous and isotropic random process on the FLRW background \cite{Peebles, SEH}. Deviations from this controlled statistical behavior are associated with large, nonperturbative, peculiar velocities and gravitational vagaries with respect to the FLRW Hubble flow. These fluctuations can generate skewed statistics in the pre-homogeneity region\cite{Adamek2}, \cite{Fanizza2} that may signal significant deviation from an FLRW perturbative picture and provide a further reason to introduce and keep distinct the roles of the physical observer and her FLRW avatar. These remarks on the FLRW observer's reference role in the pre-homogeneity region are critical to ensuring the nature of our result is understood. \;\;\;\;\;$\square$
\end{remark}

\section{The directional Celestial Spheres}
\label{CelSphs} 
To put the FLRW reference role described above at work, we begin by introducing the past null cone in the physical observer's local rest frame $(T_pM,\,\{E_{(i)}\})$, 
\begin{equation}
\label{MinkLightCone0}
C^-\left(T_pM,\,\{E_{(i)}\} \right)\,:=\,\left\{\left.X\,=\,\mathbb{X}^iE_{(i)}\,\not=\,0\,\in\,T_pM\,\right|\,\,g_p(X, X)\,=\,0,\,\,\mathbb{X}^4+r=0 \right\}\;,
\end{equation}
where 
$r:=(\sum_{a=1}^3(\mathbb{X}^a)^2)^{1/2}$ is the comoving radius parametrizing the family of $2$-spheres 
\begin{equation}
\label{celestialR}
 \mathbb{S}^2_r(T_pM)\,:=\,\{X\in C^-\left(T_pM\right)\,|\,\, \mathbb{X}^4\,=\,-\,r,\,\,\,\sum_{a=1}^3(\mathbb{X}^a)^2=r^2,\,\,r\in\,\mathbb{R}_{> 0} \}\,,
 \end{equation}
which foliate $C^-\left(T_pM\right)/\{p\}$.  We think of $\mathbb{S}^2_r(T_pM)$  as providing the apparent representation of the sky, at a given value of the comoving radius $r$, in the local rest space of the (instantaneous) observer $(p,\dot\gamma(p))$. In particular, the 2-sphere   $\left.\mathbb{S}^2_r(T_pM)\right|_{r=1}$ or, equivalently, its projection on the hyperplane 
$\mathbb{X}^4\,=\,0$ in $T_pM$, 
\begin{equation}
\label{celestialS}
\mathbb{S}^2\left(T_pM\right)\,:=\,\left\{X\,=\,\mathbb{X}^iE_{(i)}\,\not=\,0\,\in\,T_pM\,\,|\,\,\mathbb{X}^4=0,\,\,\sum_{a=1}^3(\mathbb{X}^a)^2=1 \right\}\;,
\end{equation}
can be used to label the (spatial) past directions of sight constituting the field of vision of the physical observer $(p,\, \dot\gamma(p))$.
In the sense described by R. Penrose  \cite{RindPen},  this is a representation of 
the abstract 2-sphere of past null directions 
\begin{equation}
\mathcal{S}^-(p)\simeq \mathbb{S}^2\simeq\,\mathbb{C}\,\cup\,\{\infty\}\,,
\end{equation}
describing, in the physical spacetime $(M, {g}, \gamma_s(\tau))$, the past-directed null geodesics reaching the physical observer $(p,\, \dot\gamma(p))$. Explicitly, let  
\begin{equation}
{n}(\theta, \phi)\,:=\,\sum_{a=1}^3\,{n}^a(\theta, \phi)\,{E}_{(a)}\,,\,\,\,\,
0\leq\theta\leq\pi,\,\,0\leq\phi<2\pi\,,\nonumber
\end{equation}  
denote the spatial direction in $T_pM$ associated with the celestial angular coordinates $(\theta, \phi)\,\in\,\mathbb{S}^2\left(T_pM\right)$, (by abusing notation, we write ${n}(\theta, \phi)\,\in\,\mathbb{S}^2\left(T_pM\right)$). Any such spatial direction characterizes a corresponding past-directed null vector $\ell(\theta, \phi)\,\in\,C^-\left(T_pM, \{E_{(i)}\} \right)$,  
\begin{equation}
\ell(\theta, \phi)\,:=\,\left(n(\theta,\,\phi),\,-\,\dot{\gamma}(p)\right)\,,
\end{equation} 
 normalized according to 
 $g_p\left(\ell(\theta, \phi),\dot{\gamma}(p)\right)\,=\,1$.
The corresponding past-directed null rays 
\begin{equation}
\mathbb{R}_{\geq 0}\,\ni\,r\,\longmapsto\,r\,\ell({n}(\theta, \phi))\,,\,\,\,\,\,\,\,\,(\theta, \phi)\,\in\,\mathbb{S}^2\left(T_pM\right)\,,
\end{equation}

are the generators of the null-cone $C^-\left(T_pM, \{E_{(i)}\} \right)$. According to the geometrical optics approximation, we can assume that optical (and gravitational wave) data propagate along null geodesics, and    a photon reaching  $(p,\, \dot\gamma(p))$ from the past-directed null direction $\ell(\theta, \phi)$, is characterized  by  the (future-pointing) wave vector 
\begin{equation}
\label{wavevect0}
k_{(p)}(\theta, \phi)\,:=\,-\,{\nu}\,\ell(\theta, \phi)\,\in\, T_pM\,,
\end{equation}
where $\nu\,=\,-\,g_p\left(k_{(p)},\, \dot\gamma(p)\right)$ is the photon frequency as measured  by the physical observer $(p,\, \dot\gamma(p))$. 
We explicitly characterize the parametrization role of (\ref{celestialS}) by introducing the following definition.

\begin{definition}(\emph{The physical celestial sphere})
The spherical surface $\mathbb{S}^2\left(T_pM\right)$ endowed with the round metric
\begin{equation}
\label{roundmet_0}
{h}_{\mathbb{S}^2}(\theta, \phi)\,=\,d\theta^2\,+\,\sin^2\theta\,d\phi^2\,,
\end{equation}
and the associated solid angle (area) measure
$d\mu_{\mathbb{S}^2}(\theta, \phi)\,=\,\sin\theta\,d\theta d\phi$,
defines the \emph{physical celestial sphere of sky-directions} 
\begin{equation}
\label{celstSphere}
\mathbb{C\,S}(p)\,:=\,\left(\mathbb{S}^2\left(T_pM\right),\, {h}_{\mathbb{S}^2}   \right)\,,
\end{equation}
providing, in the instantaneous rest space $\left(T_pM, \{E_{(i)}\} \right)$, the geometrical representation of the set of all directions towards which the physical observer $(p,\, \dot\gamma(p))$ can look and collect (optical and gravitational wave) data from astrophysical sources.

\end{definition}
Another basic ingredient in the
 cartography of the astrophysical sources on the celestial sphere (\ref{celstSphere}) is provided by the past-directed \emph{null exponential map} based at the observation 
event $p$.  To characterize it, let $W_p\,\subseteq \,T_pM$ denote the maximal domain of existence of the set of causal geodesics emanating from the point $p$, 
\begin{align}
\lambda_X\,:\,I_X\,&\longrightarrow\,(M, g)\\
\eta\,&\longmapsto\,\lambda_X(\eta),\;\;\;\; \lambda(0)\,=\,p,\;\;\;\dot{\lambda}(0)\,=\,X\in\,W_p,\;\;g_p(X, X)\,\leq\,0,  \nonumber
\end{align} 
where $I_X\subseteq\mathbb{R}_{\geq0}$ denotes the maximal interval of existence associated with the chosen $X\in\,W_p$. The past-directed \emph{null exponential map} based at the observation event $p$ is defined by restricting $W_p\,\subseteq \,T_pM$ to the set of past-directed null vectors according to 
\begin{align}
&\exp_p\,:\,W_p\,\cap C^-\left(T_pM, \{E_{(i)}\} \right)\,\longrightarrow \;\;\;\;M\nonumber\\
&X\;\;\;\;\;\;\longmapsto \;\;\;\;exp_p\,(X)\,:=\,\lambda_X(1)\,.
\label{expmapdef_0}
\end{align}
Let us recall that the past lightcone $\mathscr{C}^-(p,g)\,\in\,(M, g)$ with the vertex at $p$,\, is the set of all events $q\in (M, g)$ that can be reached from $p$ along the past-pointing null geodesics issued from $p$. In our cosmological setting, we have the following more specific characterization of $\mathscr{C}^-(p,g)$.

\begin{proposition}
\label{PropPhysLightCone}
As $n(\theta, \phi)$ varies over the celestial sphere $\mathbb{C\,S}(p)$, the corresponding
set of null rays in $C^-\left(T_pM, \{E_{(i)}\} \right)$  
\begin{equation}
r\ell(n(\theta, \phi))\,\in\,C^-\left(T_pM, \{E_{(i)}\} \right),\;\;\;r\in\,I_{\ell},
\,\,\,\,(\theta, \phi)\in\mathbb{C\,S}(p)\,,
\end{equation}
parametrizes the null geodesic generators 
\begin{equation}
\label{radgeodesics_0}
r\ell(n(\theta, \phi))\,\longmapsto\,\exp_p(r \ell(n(\theta, \phi)))\,,
\end{equation}
of the physical past lightcone 
\begin{equation}
\label{pastcone2_0}
\mathscr{C}^-(p,g)\,:=\,\exp_p\left[W_p\cap C^-\left(T_pM, \{E_{(i)}\} \right)\right]\,.
\end{equation}
The geometrical structure of $\mathscr{C}^-(p,g)$ is characterized by its  \emph{terminal points}, the last\textendash{}points on the 
null geodesic generators that lie on the boundary $\partial\mathrm{I}^-(p,g)$ of the chronological past of $p$.  Any such terminal point 
\begin{equation}
q\left(r_*,\,\theta, \phi\right):=\exp_p\left(r_*\,\ell(n(\theta, \phi))\right)
\end{equation}
 is said to be:\, \emph{i)}\; a \emph{conjugate terminal point} if the exponential map $\exp_p$ is singular at $(r_*,\,\theta,\phi)$;\;\emph{ii)}\; a \emph{cut locus terminal point} 
if the exponential map $\exp_p$ is non--singular at $(r_*,\,\theta,\phi)$ and there exists another null geodesic, issued from $p$, passing through $q(r_*, \theta, \phi)$, (see also \cite{Beem}, \cite{Perlick}). We denote \cite{Klainer} by $\mathcal{T}^-(p)$, the set of all terminal points associated with the past null geodesic flow issuing from $p$. In the presence of cut points,  $\mathscr{C}^-(p,g)$ fails to be an embedded submanifold of $(M, g)$. 
\end{proposition}

\begin{proof}
The characterization of  $\mathscr{C}^-(p,g)$ is a direct consequence of the definition of the celestial sphere $\mathbb{C\,S}(p)$ and terminal points \cite{Klainer},  and of the properties of the (null) exponential map \cite{Beem}, \cite{Perlick}. \;\;\;\;\;$\square$
\end{proof}

The exponential map representation of the past lightcone $\mathscr{C}^-(p,g)$ afforded by Proposition \ref{PropPhysLightCone}  provides a natural setup for a description of observational data gathered from $\mathscr{C}^-(p,g)$. It emphasizes the basic role of past-directed null geodesics and provides the framework for interpreting the physical data in the observer's local rest frame at $p$. In particular, as we shall prove in the next sections, it allows us to represent on the celestial sphere  
$\mathbb{C\,S}(p)$  the actual geometry of the observed sky at a given length scale. This role of the null geodesics and the associated exponential map $\exp_p$ is well-known and quite effective in a neighborhood of $p$, where we can introduce normal coordinates associated with $\exp_p$. Yet, it is delicate to handle in regions where $\exp_p$ is not a diffeomorphism of $W_p\,\cap\, C^-\left(T_pM, \{E_{(i)}\} \right)$  onto its image. 
Our strategy is to start with the standard description \cite{EllisPhRep}, \cite{Ellis2} of observational data coming from $\mathscr{C}^-(p,g)$, associated with the usual assumption that the exponential map is a diffeomorphism\footnote{From an observational point of view, this geometrical setup is appropriate also in the \emph{weak lensing regime} describing the alteration, due to the effect of gravity, of the apparent shape and brightness of astrophysical sources.} 
in a sufficiently small neighborhood of $p$. Then, we move from this smooth description to the more general, low regularity (Lipschitz) case typical of the actual past lighcone $\mathscr{C}^-(p,g)$ . In this connection, it is worthwhile to stress that the standard normal coordinates description is strictly associated with the assumption that the metric of $(M, g)$ is sufficiently regular, with components $g_{ij}(x^\ell)$ which are at least twice continuously differentiable, \emph{i.e.}\, $g_{ij}(x^\ell)\,\in\, C^k(\mathbb{R}^4, \mathbb{R})$, \, for $k\geq 2$. Under this hypothesis, there is  a star-shaped neighborhood $N_0(g)$ of  $0$ in $W_p\subseteq T_pM$ and a corresponding geodesically convex neighborhood of $p$, $U_p 	\subseteq\,(M, g)$, restricted to which $\exp_p\,:\,N_0\,\subseteq \,T_pM\,\longrightarrow\;U_p\,\subseteq\,M$ is a diffeomorphism. In such $U_p$ we can  introduce the geodesic normal coordinates $(x^i)$  of the observed event $q\,\in\,\mathscr{C}^-(p,g)\cap\,U_p$ according to 
\begin{align}
\label{geodnormal_0}
x^i\,:=\,\mathbb{X}^i\,\circ \,\exp_p^{-1}\,:\,M\cap\,U_p\,&\longrightarrow \,N_0(g)\subseteq \mathbb{R}^4\\
q\,\,\,\,\,\,\,&\longmapsto\,x^i(q)\,:=\, \mathbb{X}^i\left(\exp_p^{-1}(q)\right)\nonumber
\end{align}
where $\mathbb{X}^i\left(\exp_p^{-1}(q)\right)$ are the components, in the $g$-orthonormal frame $\{E_{(i)}\}$ of the vector $\exp_p^{-1}(q)\,\in\,W_p\subseteq T_pM$. Thus, in $\mathscr{C}^-(p,g)\cap\,U_p$  we can write,  
\begin{align}
\label{exp1_0}
\exp_p\
\,:\,C^-\left(T_pM,\,\{E_{(i)} \}\right)\cap\,N_0(g)\,&\longrightarrow \,\mathscr{C}^-(p,g)\,\cap\,U_p\\
r \ell(n(\theta, \phi))\,=\,r \left(n^a(\theta, \phi) E_{(a)}\,-\,E_{(4)}\right)\,\,\,\,\,\,\,\,\,\,&\longmapsto \,\exp_p(r \ell(n))\,=\,q\nonumber\,,
\end{align}
with $x^i(q)\,=\,(x^a(q),\,x^4(q))$,\, $a\,=\,1,2,3$,\, explicitly provided by
\begin{equation}
x^i(q)\,:=\,\mathbb{X}^i\left(\exp_p^{-1}(q)\right)\,=\,\left(r\,n^a(\theta,\phi),\,-\,r\right)\,.
\end{equation}

If we move from the rather complex picture associated with the physical cosmological spacetime $(M,\,{g},\,\gamma(\tau))$   to the reference FLRW scenario provided by $(M,\,\widehat{g},\,\widehat\gamma(\widehat\tau))$, the definitions of the corresponding celestial sphere and exponential map are quite standard, and we can easily write down their characterization. We start by introducing the FLRW observer's past null cone according to
\begin{equation}
\label{MinkLightConeFLRW0}
\widehat{C}^-\left(T_pM,\,\{\widehat{E}_{(i)}\} \right)\,:=\,\left\{\left.\widehat{X}\,=\,\widehat{\mathbb{X}}^i\widehat{E}_{(i)}\,\not=\,0\,\in\,T_pM\,\right|\,\,\widehat{g}_p(\widehat{X},\widehat{X})\,=\,0,\,\,\widehat{\mathbb{X}}^4+\widehat{r}=0 \right\}\;,
\end{equation}
where $\widehat{r}$ is the FLRW comoving radius in the hyperplane $\widehat{\mathbb{X}}^4\,=\,0$. In analogy with (\ref{expmapdef_0}), we introduce characterize the the past-directed \emph{FLRW null exponential map} based at the observation 
event $(p, \widehat{\dot\gamma}(p))$. Let $\widehat{W}_p\,\subseteq \,(T_pM, \,\{\widehat{E}_{(i)}\})$ denote the maximal domain of existence of the set of FLRW causal geodesics emanating from the point $p$,  
\begin{align}
\widehat\lambda_{\widehat{X}}\,:\,\widehat{I}_{\widehat{X}}\,&\longrightarrow \,(M, \widehat{g})\\
\widehat\eta\,&\longmapsto \,\widehat\lambda_{\widehat{X}}(\widehat\eta),\;\;\;\; \widehat\lambda(0)\,=\,p,\;\;\;\widehat{\dot{\lambda}}(0)\,=\,\widehat{X}\in\,\widehat{W}_p,\;\;\widehat{g}_p({\widehat{X}}, {\widehat{X}})\,\leq\,0,  \nonumber
\end{align} 
where $\widehat{I}_{\widehat{X}}\subseteq\mathbb{R}_{\geq0}$ denotes the maximal interval of existence associated with the chosen ${\widehat{X}}\in\,\widehat{W}_p$. The past-directed \emph{FLRW null exponential map} based at the observation event $p$ is defined by restricting $\widehat{W}_p\,\subseteq \,T_pM$ to the set of past-directed FLRW null vectors according to 
\begin{align}
\label{expmapdef0FLRW}
\widehat{\exp}_p\,:\,\widehat{W}_p\,\cap C^-\left(T_pM, \{\widehat{E}_{(i)}\} \right)\,&\longrightarrow\;\;\;\;(M, \widehat{g})\\
\widehat{X}\;\;\;\;\;\;&\longmapsto\;\;\;\;\widehat{exp}_p\,({\widehat{X}})\,:=\,\widehat{\lambda}_{\widehat{X}}(1)\,,\nonumber
\end{align}
where $\widehat\eta\,\longmapsto\,\widehat\lambda_{\widehat{X}}(\widehat\eta)$ is the FLRW past-directed causal geodesic emanating from the point $p$ with initial tangent vector $\widehat{\dot{\lambda}}_{\widehat{X}}(0)\,=\,\widehat{X}\in (T_pM, \{\widehat{E}_{(i)}\})$. 
With these notational remarks, we characterize the celestial sphere describing the field of vision of the FLRW observer $(p,\, \widehat{\dot\gamma}(0),\,\{\widehat{E}_{(\kappa)}\})$, and the associated FLRW past lightcone $\widehat{\mathscr{C}}^{\,-}(p,\widehat{g})$ according to the following definition.

\begin{definition}(\emph{The directional FLRW celestial sphere and the associated past lightcone}).
Let us consider the $2$-sphere in $T_pM$ defined by\footnote{To emphasize that the round 2-sphere $\mathbb{S}^2(T_pM)$ is parametrized by the FLRW observer in terms of the celestial coordinates $(\widehat\theta, \widehat\phi)$ we denote it as $\widehat{\mathbb{S}}^2(T_pM)$ to avoid notational confusion with the round 2-sphere  $\mathbb{S}^2(T_pM)$ parametrized in terms of $(\theta, \phi)$ by the physical observer.} 
\begin{equation}
\label{celestialS001}
\widehat{\mathbb{S}}^2(T_pM)\,:=\,\left\{\widehat{X}\,=\,\widehat{\mathbb{X}}^i\widehat{E}_{(i)}\,\not=\,0\,\in\,T_pM\,\,|\,\,\widehat{\mathbb{X}}^4=0,\,\,\sum_{\kappa=1}^3(\widehat{\mathbb{X}}^{\kappa})^2=\,1 \right\}\,,
\end{equation}
providing the celestial coordinates, in $(T_pM,\, \widehat{\dot\gamma}(p),\,\{\widehat{E}_{(\kappa)}\})$,  of the past-directed null geodesics issuing from the reference FLRW observer $(p,\, \widehat{\dot\gamma}(p))$.
As in the case of the physical celestial sphere, $\widehat{\mathbb{S}}^2(T_pM)$ is naturally endowed with the directional round metric
\begin{equation}
\label{FLRWroundmet0}
{h}_{\widehat{\mathbb{S}}^2}(\widehat\theta, \widehat\phi)\,:=\,d\widehat\theta^2\,+\,\sin^2\widehat\theta\,d\widehat\phi^2,\;\;\;
0<\widehat{\theta}\leq\pi,\,\,0\leq\widehat{\phi}<2\pi\,,
\end{equation} 
and the associated area 2-form 
$d\mu_{\widehat{\mathbb{S}}^2}(\widehat\theta, \widehat\phi)\,=\,\sqrt{\det({h}_{\widehat{\mathbb{S}}^2})}\,d\widehat\theta d\widehat\phi\,=\,\sin\widehat\theta\,d\widehat\theta d\widehat\phi$,
that parametrize, in the instantaneous rest space $\left(T_pM, \{\widehat{E}_{(i)}\} \right)$, the solid angle measures that the FLRW observer associates with the astrophysical sources. These notational remarks characterize the directional FLRW celestial sphere $\widehat{\mathbb{C\,S}}(p)$ and the associated past lightcone $\widehat{\mathscr{C}}^{\,-}(p,\widehat{g})$  according to
\begin{equation}
\label{FLRWcelstSphere}
\widehat{\mathbb{C\,S}}(p)\,:=\,\left(\widehat{\mathbb{S}}^2\left(T_pM\right),\,
{h}_{\widehat{\mathbb{S}}^2}\right)\,,
\end{equation}
and 
\begin{equation}
\label{pastcone2_0FLRW}
\widehat{\mathscr{C}}^{\,-}(p,\widehat{g})\,:=\,\widehat{\exp}_p\left[\widehat{C}^-\left(T_pM,\,\{\widehat{E}_{(i)}\} \right)\right]\,\subset\,(M, \widehat{g}).
\end{equation}
\end{definition}
\vskip 0.5cm\noindent
We denote by 
\begin{equation}
\widehat{\mathscr{C}}^-(p,\widehat{g};\,\widehat{r}_0)\,:=\,\left\{\left. q\,\in\,M\,\right|\,\,q\,=\,\widehat{\exp}_p(\widehat{r} \widehat{\ell}(\widehat{n}(\widehat\theta, \widehat\phi))),\, \, 0\leq \widehat{r} < \widehat{r}_0,\,\, (\widehat\theta, \widehat\phi)\in\widehat{\mathbb{C\,S}}(p)  \right\}\,,
\end{equation}
the portion of $\widehat{\mathscr{C}}^-(p, \hat{g})$ accessible to observations for a given value $\widehat{r}_0$ of the FLRW comoving radius $\widehat{r}$. \, In the FLRW case (\ref{expmapdef0FLRW}) is, for any finite interval $\widehat{I}_{\widehat{W}}$,  a smooth diffeomorphism\footnote{As already stressed, the situation is quite more complex for (\ref{expmapdef_0}), since the physical exponential map $\exp_p$ is of low regularity (Lipschitz) along the past null cone ${\mathscr{C}}^-(p,{g})$.}, and 
in convex neighborhoods, $\hat{U}_{p}\subset\,(M, \hat{g})$,  of $p$ we can introduce normal coordinates by
\begin{equation}
\label{geodnormalY_0}
\widehat{x}^i\,:=\,\widehat{\mathbb{X}}^i\,\circ \,\widehat{\exp}_p^{-1}\,:\,(M, \widehat{g})\cap\,\widehat{U}_{p}\,\longrightarrow\,\mathbb{R}\,,
\end{equation}
where  $\widehat{\mathbb{X}}^i$ are the components of  the vector $\widehat{\mathbb{X}}\in\,{T}_pM$ with respect to the $\widehat{g}$-orthonormal frame $\{\widehat{E}_{(i)}\}_{i=1,\ldots,4}$ with $\widehat{E}_{(4)}\,:=\,\widehat{\dot{\gamma}}(p	)$.

\section{The reference cosmological redshift}
 \label{RefCosmRed} 
In comparing the physical and reference celestial spheres ${\mathbb{C\,S}}(p)$ and $\widehat{\mathbb{C\,S}}(p)$, a fundamental role is played by locating astrophysical sources in terms of their cosmological redshifts. This role is (apparently) simple for the  FLRW observer due to the relation (\ref{defredshift}) connecting the wavelength of the signal of an emitting source, stationary with an FLRW fundamental observer, with the scale factor $a(\widehat{\tau})$. This connection associates the FLRW cosmological redshift with a distance and scale indicator. Conversely, the redshift of a source comoving with the phenomenological Hubble flow can hardly be used as a reliable distance indicator. In particular, in the pre-homogeneity region near our observation point, gravitational clustering generates galactic structures, typically clusters, that tend to virialize. Hence,  redshifts are not of cosmological origin; they are mainly due to the local relative motion of the constituent galaxies and their gravitational fields. The clustering neighborhood, $V_{(c)}(p)$, of the physical observer located at the barycenter $p$ of such structures is decoupled from cosmological expansion, and the local redshift $z$ does not have a cosmological meaning. In the situation at hand, the physical observer refers to the FLRW redshift $\widehat{z}$ only to characterize the peculiar velocities of her neighboring galaxies with respect to the idealized FLRW Hubble flow. The situation changes as we gather our astrophysical data from the cluster's outskirts, where the infalling galaxies couple to the cosmological expansion. The radius of this region defines the cluster's \emph{physical radius}, and we need to keep track of this length scale when handling the cosmological comparison between the physical and the reference FLRW past lightcones. To circumvent these difficulties, we must discuss the relationship between the FLRW and the physical redshift. We start with a detailed analysis of the structure of the redshift of an astrophysical source as measured by the FLRW observer $(p,\,\widehat{\dot{\gamma}}(p))$.

Let $\widehat{\lambda}(\widehat{q})\,=\,c/\widehat{\nu}(\widehat{q})$ denote the wavelength of a signal emitted, at the event $\widehat{q}\in\,\widehat{\mathscr{C}}^-(p,\widehat{g})$,  by a source locally at rest with respect to the instantaneous FLRW fundamental observer $(\widehat{q},\,\widehat{\dot{\gamma}}(\widehat{q}))$. If  $\widehat{\lambda}(p)\,=\,c/\widehat{\nu}(p)$ is the corresponding  wavelength received by the FLRW observer $(p,\,\widehat{\dot{\gamma}}(p))$ at $p$, we get the expression (\ref{defredshift}) of the FLRW \emph{cosmological redshift} of the emitting source $(\widehat{q},\,\dot{\gamma}(\widehat{q}))$, \emph {i.e.}  
\begin{equation}
\widehat{z}(\widehat{q})\,:=\,\frac{a(\widehat\tau=0)}{a(\widehat\tau_{\widehat{q}})}\,-\,1\,=
\frac{\widehat{\lambda}(p)\,-\,\widehat{\lambda}(\widehat{q})}{\widehat{\lambda}(\widehat{q})}\,,
\end{equation} 
where $a(\widehat\tau)$ is the scale factor of the FLRW spacetime metric $\widehat{g}$ (see (\ref{FLRWg0})). We can measure redshifts accurately even for extremely faint galaxies, and in a regime where one has a linear Hubble (cor)relation between the cosmological expansion velocity and distance, the associated cosmological FLRW redshift $\widehat{z}$  can be used as a distance indicator. We have an explicit relation between the comoving radius $\widehat{r}$ (which is not an observable quantity), the cosmological redshift $\widehat{z}$, and the (FLRW) Hubble parameter $H(\widehat{z})$ (see (\ref{AccaZeta})), given by
\begin{equation}
\label{Dandri0} 
\widehat{r}(\,\widehat{z}\,)\,=\,\int_0^{\;\widehat{z}}\,\frac{d\widehat{z}\,'}{H(\widehat{z}\,')}\,.
\end{equation}
Conversely, the use of redshift by the physical observer $(p, \dot\gamma(0))$  as a distance indicator of sorts is quite delicate since the knowledge of the corresponding distance-redshift relation is tantamount to knowing the actual cosmological geometry in the pre-homogeneity region. In this region, the motion of cosmic structures characterizes the phenomenological Hubble flow, which deviates significantly from a linear Hubble relation due to the gravitational interactions between galaxies. Their motion is described by peculiar velocities, not due to the expansion of the Universe, which can be\cite{TMDavies} of the order of $100-300\,\mathrm{Km}/\mathrm{sec}$, a significant fraction of the reference FLRW recession velocity over the scales that characterize the pre-homogeneity region. In line with our modeling assumption on the nature of the physical spacetime $(M, g, \gamma(\tau))$, this fraction becomes less significant for distant galaxies, and one typically assumes that peculiar velocities can be ignored on the largest scales where homogeneity takes over. According to this latter remark  and the FLRW best-fit strategy, it makes sense to use the candidate FLRW spacetime as providing the cosmological expansion reference standard also in the pre-homogeneity region.

 To characterize explicitly the reference role of the FLRW redshift $\widehat{z}$, let us consider the FLRW description of a signal emitted, at an event ${q}\in\,{\mathscr{C}}^-(p,{g})\cap\,\widehat{\mathscr{C}}^{\,-}(p,\widehat{g})$, by a source associated with the \emph{phenomenological Hubble flow} $\tau\longrightarrow\,\gamma_s(\tau)$. On the physical celestial sphere $\mathbb{CS}(p)$ the source emission is described as coming from a region $B(q)\,\subset\,\mathbb{CS}(p)$ centered around the spatial unit vector $n(q)\,:=\,n(\theta(q), \phi(q))\in\,\mathbb{CS}(p)$ pointing to the source. This region subtends a solid angle measure
\begin{equation}
\label{apparentArea}
\mu\left(B(q)\right)\,:=\,\int_{B(q)}\,d\mu_{\mathbb{S}^2}(\theta,\phi)\,,
\end{equation}
describing the apparent area extension
of the source as seen by the physical observer $(p,\,{\dot{\gamma}}(p))$ at $p$. Moreover,  if $n(a)\in\,\mathbb{CS}(p)$, with $a\,=\,1, 2$, are two points\footnote{By a slight abuse of notation we indicate the angular coordinates $(\theta, \phi)$ of a point on $\mathbb{S}^2$ with the corresponding versor $n(\theta, \phi)$.} on the physical celestial sphere  $\mathbb{CS}(p)$ subtending the diameter of the region $B(q)\,\subset\,\mathbb{CS}(p)$, then the physical observer associates with the source its apparent angular diameter
\begin{equation}
\label{apparentDiamater}
\mathrm{diam}_{\mathbb{CS}}\left(B(q)\right)\,=\,d_{\mathbb{S}^2}\left[n(1),\,n(2)\right]\,,
\end{equation}  
where $d_{\mathbb{S}^2}$ denotes the distance function on the round unit sphere $(\mathbb{S}^2, {h}_{\mathbb{S}^2})$.

With an obvious adaptation of the above framework, the FLRW 
observer $(p,\,\widehat{\dot{\gamma}}(p))$ describes the source emission on the  FLRW celestial sphere $\widehat{\mathbb{CS}}(p)$ as coming from a region $\widehat{B}(q)\,\subset\,\widehat{\mathbb{CS}}(p)$ centered around the spatial unit vector 
$\widehat{n}(q):=\,\widehat{n}(\widehat\theta(q), \widehat\phi(q))$. In analogy with (\ref{apparentArea}) and (\ref{apparentDiamater}), we can associate with the source its solid angle area measure 
\begin{equation}
\label{FLRWapparentArea}
\widehat{\mu}\left(\widehat{B}(q)\right)\,:=\,\int_{\widehat{B}(q)}\,d\mu_{\widehat{\mathbb{S}}^2}(\widehat\theta,\widehat\phi)\,,
\end{equation}  
and its apparent angular diameter 
\begin{equation}
\label{FLRWapparentDiamater}
\widehat{\mathrm{diam}}_{\widehat{\mathbb{CS}}}\left(\widehat{B}(q)\right)\,=\,\widehat{d}_{\widehat{\mathbb{S}}^2}\left[\widehat{n}(1),\,\widehat{n}(2)\right]\,,
\end{equation} 
where $\widehat{d}_{\widehat{\mathbb{S}}^2}$ denotes the distance function on the round unit sphere $(\widehat{\mathbb{S}}^2, \widetilde{h}_{\widehat{\mathbb{S}}^2})$.

The physical and the FLRW observers are in relative motion with respect to each other; consequently, the above geometrical data on the appearance of the source on the celestial spheres ${\mathbb{CS}}$ and $\widehat{\mathbb{CS}}$ are distinct. To compare them requires a delicate nonperturbative analysis, which is the central theme of this work. The first step in this comparison analysis is to assume that the FLRW observer can identify:\; \emph{(i)}\; The spacetime region $\subset\,(M, \widehat{g},\, \widehat{\dot\gamma}(\widehat\tau))$ in which the emission event $q$ occurs; \; \emph{(ii)}\; The corresponding peculiar velocities of the phenomenological Hubble flow with respect to the reference FLRW Hubble flow;\; \emph{(iii)}\; The presence of strong local gravitational fields affecting the null geodesics propagation from the emission event $q$ to the observation event $p$;\; \emph{(iv)}\; The FLRW redshift $\widehat{z}_{(c)}$ marking the boundary of the region surrounding $p$ where the phenomenological Hubble flow decouples from the reference FLRW cosmological expansion. 
Under these assumptions, the FLRW observer can associate with the astrophysical source moving along the phenomenological Hubble flow and (instantaneously) located at $q$, the following redshifts:\; 
\begin{itemize}
\item {\emph{(I)}\; The FLRW cosmological redshift $\widehat{z}(q)$ associated with the location $q$ in the reference FLRW model. This term is present as long as the source is coupled to the cosmological expansion, namely if  $\widehat{z}(q)\,\geq\,\widehat{z}_{(c)}$;}
\item  {\emph{(II)}\; The Doppler redshift $\widehat{z}^{(doppl)}(q)$ associated to the peculiar velocity of the source with respect to the reference FLRW Hubble flow;}
\item {\emph{(III)}\; The gravitational redshift $\widehat{z}^{(grav)}(q)$ associated with the null geodesics distortions with respect to the reference FLRW null geodesic connecting $q$ to $p$ according to the FLRW observer $(p,\,\widehat{\dot{\gamma}}(p))$. In selecting the FLRW observer, we have factorized out the  CMB dipole anisotropy induced by our Local Group motion; it follows that  $\widehat{z}^{(doppl)}(q)$ is associated only with the peculiar motions of galaxies (and clusters) generated by the local isotropic over-densities or under-densities in the pre-homogeneity zone. According to the FLRW observer $(p,\,\widehat{\dot{\gamma}}(p))$, the phenomenological Hubble flow terms $\widehat{z}^{(doppl)}(q)$ and $\widehat{z}^{(grav)}(q)$ are always present in the pre-homogeneity region. They are dominant in the decoupling region $\widehat{z}(q)\,<\,\widehat{z}_{(c)}$, where they eventually account for the total redshift the FLRW observer attributes to the source $q$. This redshift is no longer of cosmological origin.}
\end{itemize}

 Under the assumptions $(i), (ii), (iii)$ described above\footnote{Even for the FLRW reference observer it is notoriously difficult \cite{Ellis2} to resolve the  contributions $\widehat{z}^{(doppl)}(q)$ and $\widehat{z}^{(grav)}(q)$   from the  total redshift $\widehat{z}^{(tot)}$ of the source, and the factorization (\ref{FLRWzetafactor0}) is possible only under the assumptions $(i), (ii), (iii)$.}, the actually measured \emph{total redshift}, $\widehat{z}^{(tot)}$ of the source can be factorized, for $\widehat{z}\,\geq\,\widehat{z}_{(c)}$, according to 

 \begin{align}
 \label{FLRWzetafactor0}
 \left(1\,+\,\widehat{z}^{\,(tot)}(q)\right)&:=
    \begin{cases}
       & (1+\widehat{z}(q)-\widehat{z}_{(c)})(1+\widehat{z}^{\,(doppl)}(q))(1+\widehat{z}^{\,(grav)}(q)),\;\;\;\widehat{z}\,\geq\,\widehat{z}_{(c)}\\
       & \\
       & (1+\widehat{z}^{\,(doppl)}(q))(1+\widehat{z}^{\,(grav)}(q)),,\;\;\widehat{z}\,<\,\widehat{z}_{(c)}\,.
    \end{cases}       
\end{align}
 
\noindent
This factorization of the FLRW redshift implies a corresponding factorization in the standard relation (\ref{Dandri0})  that connects, in terms of the (FLRW) Hubble parameter $H(\widehat{z})$ (see (\ref{AccaZeta})),  the comoving FLRW radius $\widehat{r}$ associated with a source $q$ to the corresponding FLRW redshift $\widehat{z}$. Let us assume that
the observation point $p$ is located at the barycenter of a cluster of galaxies, then the comoving radius associated with a source located at the FLRW redshift $\widehat{z}(q)$ factorizes according to

\begin{equation}
\label{factorErre}
\widehat{r}\,\left(\widehat{z}(q)\right)\,=\,\widehat{r}\,\left(\widehat{z}_{(c)}\right)\,+\,{\int_{\;\widehat{z}_{(c)}}^{\widehat{z}}\,\frac{d\widehat{z}\,'}{H(\widehat{z}\,')}}\,,
\end{equation}
where $\widehat{r}\,\left(\widehat{z}_{(c)}\right)$ is is identified with the \emph{cluster radius}. This radius can be formally evaluated according to the standard FLRW relation (\ref{Dandri0}), \emph{i.e.}  
\begin{equation}
\label{FLRWcritRad}
\widehat{r}\,\left(\widehat{z}_{(c)}\right)\,:=\,{\int_0^{\;\widehat{z}_{(c)}}\,\frac{d\widehat{z}\,'}{H(\widehat{z}\,')}}\,,
\end{equation}
an expression that has only an indicative reference value, as it does not take into account the fact that in the decoupling region, the cosmological redshift-distance relation does not hold. Yet, since we use $\widehat{z}$ only in an FLRW reference role, in what follows, we adopt (\ref{FLRWcritRad}), and factorize $\widehat{r}\,\left(\widehat{z}(q)\right)$ according to
\begin{equation}
\label{factorErreBis}
\widehat{r}\,\left(\widehat{z}(q)\right)\,=\,{\int_0^{\;\widehat{z}_{(c)}}\,\frac{d\widehat{z}\,'}{H(\widehat{z}\,')}}\,+\,{\int_{\;\widehat{z}_{(c)}}^{\widehat{z}}\,\frac{d\widehat{z}\,'}{H(\widehat{z}\,')}}\,.
\end{equation}
 The radius $\widehat{r}\,\left(\widehat{z}_{(c)}\right)$ is typically derived from physical models exploiting the fact that a cluster is a gravitationally virialized system. The typical physical radius can be estimated in the range $\sim\,(1\,-\,3)\,\mathrm{Mpc}$. By exploiting (\ref{factorErreBis}) and the linear Hubble law $\widehat{z}\,\simeq\,H_0\,\widehat{r}\,(\widehat{z})$ in the above formal reference sense, we can associate to $\widehat{r}\,\left(\widehat{z}_{(c)}\right)$  a FLRW reference redshift $10^{-\,4}\,\lesssim\,\widehat{z}_{(c)}\,\lesssim\,10^{-\,3}$.

The vagaries of the peculiar velocities 
in the pre-homogeneity region imply that there is no privileged choice of the local rest frame   $(T_pM, \{E_{(a)}\},\,\dot{\gamma}(p))$ of the physical observer
with respect to the FLRW local rest frame $(T_pM, \{\widehat{E}_{(a)}\},\,\widehat{\dot{\gamma}}(p))$ (with $a=1,2,3$).  From the point of the reference FLRW observer who is looking at sources in the pre-homogeneity region, the natural choice for the physical observer is to \emph{optimize the relative velocity and orientation of $(T_pM, \{E_{(a)}\}, \dot{\gamma}(p))$ in such a way that the FLRW celestial sphere at a given redshift $\widehat{z}$ is as close as possible to the physical celestial sphere}, in a sense that we make precise later on by introducing a natural optimization functional. In other words, in the pre-homogeneity region, we need to optimize the choice of $(T_pM, \{E_{(a)}\},\dot{\gamma}(p))$ with respect to  
$(T_pM, \{\widehat{E}_{(a)}\},\widehat{\dot{\gamma}}(p))$, an optimization that will be $\widehat{z}$-dependent. To this end, let $v_{\,\widehat{z}}(p)$ denote the relative 3-velocity of the physical observer $(p, \dot\gamma(p))$ with respect to the FLRW observer  $(p, \widehat{\dot\gamma}(p))$, defined by the relation 
\begin{equation}
\label{3velocity2}
{\dot\gamma}_{\;\widehat{z}}(p)\,=\,\frac{\sum_{a=1}^3\,v^a_{\,\widehat{z}}(p)\,\widehat{E}_{(a)}+\widehat{\dot{\gamma}}(p)}{\sqrt{1\,-\,|v_{\;\widehat{z}}|^2}}\,,
\end{equation}
(see (\ref{3velocity})) where we have explicitly introduced the FLRW redshift dependence in the 4-velocity vector ${\dot\gamma}(p)$ labeling. According to (\ref{FLRWzetafactor0}), and the Doppler rescaling induced by the observers' relative velocity $v_{\;\widehat{z}}$, we can characterize the physical redshift $z$ in terms of the total redshift $\widehat{z}^{\;(tot)}(q)$. To this end, let us consider, at an event $q\,\in\,(M, g, \gamma_s(\tau))$,
 an astrophysical source $(q,\,\dot\gamma(q))$ stationary with respect to the phenomenological Hubble flow $\tau\,\longrightarrow\,\gamma_s(\tau)$. Let us assume that physical signals from $(q,\,\dot\gamma(q))$, travelling along the null geodesics generators of the physical lightcone ${\mathscr{C}}^-(p,{g})$, are gathered at the event $p$ by the physical observer $(p, \dot\gamma(p))$ as well as by the FLRW observer $(p, \widehat{\dot\gamma}(p))$. This latter, enforcing the FLRW point of view, analyzes the data with respect to the reference FLRW Hubble flow by attributing to the source $(q,\,\dot\gamma(q))$ a cosmological redshift $\widehat{z}(q)\,\geq\,\widehat{z}_{(c)}$, a Doppler redshift $\widehat{z}^{\,(doppl)}(q)$ due to the peculiar velocity of $(q,\,\dot\gamma(q))$ with respect to the FLRW Hubble flow at $q$, and possibly a local gravitational redshift $\widehat{z}^{\,(grav)}(q)$ due to the presence of strong gravitational fields in the region where $(q,\,\dot\gamma(q))$ evolves. Under the hypotheses 
\emph{(i),\,(ii),\,(iii)}  described above, the FLRW observer collects this redshift according to the formula (\ref{FLRWzetafactor0}) providing the total phenomenological redshift\; $\widehat{z}^{\;(tot)}(q)$ of the source $(q,\,\dot\gamma(q))$ as measured by the FLRW observer $(p, \widehat{\dot\gamma}(p))$. Since $(q,\,\dot\gamma(q))$ is stationary with respect to the phenomenological Hubble flow, its phenomenological redshift  measured at $p$ by the physical observer $(p,{\dot\gamma}(p))$ is provided by
\begin{equation}
\left(1\,+\,{z}(q)\right)\,=\,\sqrt{\frac{1+v_{\,\widehat{z}}}{1-v_{\,\widehat{z}}}}(1\,+\,\widehat{z}^{\;(tot)}(q))\,,
\end{equation}   
where $v_{\,\widehat{z}}$ denotes the relative 3-velocity of the physical observer $(p,{\dot\gamma}(p))$ with respect to the FLRW observer $(p, \widehat{\dot\gamma}(p))$. It follows that we can profitably factorize the physical redshift $z$  in terms of the reference FLRW cosmological redshift $\widehat{z}$. We have

\begin{definition} {\emph{\textbf{The physical redshift parametrization}}}
\label{RefRed}

Let $(q,{\dot\gamma}(q))$ denote an astrophysical source that according to the reference FLRW observer $(p, \widehat{\dot\gamma}(p))$ is located at cosmological redshift $\widehat{z}(q)$, and is affected by a Doppler and gravitational redshift $\widehat{z}^{\,(doppl)}(q))$ and $\widehat{z}^{\,(grav)}(q))$. If $v_{\,\widehat{z}}$ denotes the relative 3-velocity of the physical observer $(p,{\dot\gamma}(p))$ with respect to the FLRW observer $(p, \widehat{\dot\gamma}(p))$, then we define  the physical cosmological redshift ${z}(q)$ of the source  $(q,{\dot\gamma}(q))$ according to     
\begin{equation}
\label{zetafactor0A}
\left(1\,+\,{z}(q)\right)\,=\,\sqrt{\frac{1+v_{\,\widehat{z}}}{1-v_{\,\widehat{z}}}}\,(1\,+\,\widehat{z}^{\;(tot)}(q))\,,
\end{equation}
where $(1\,+\,\widehat{z}^{\;(tot)}(q))$ is given by (\ref{FLRWzetafactor0}).
\end{definition}
\vskip 0.5cm\noindent
\begin{remark}
The parametrization (\ref{zetafactor0A}) is a delicate issue in any best-fit strategy since the relative velocity $v_{\;\widehat{z}}$ is $\widehat{z}$ dependent. As stressed above, to take care of this dependence,  we fix the reference FLRW redshift $\widehat{z}$  and adjust the relative velocity $v_{\;\widehat{z}}$ (and the associated spatial orientation) of the physical observer $(p,\,{\dot{\gamma}}_{\,\widehat{z}}(p))$ in such a way to optimize, in a sense that we discuss in the next sections, the correspondence between the reference FLRW lightcone $\widehat{\mathscr{C}}^-(p,\widehat{g})$ and the physical lightcone ${\mathscr{C}}^-(p,{g})$ at the FLRW reference scale set by the chosen $\widehat{z}$. \;\;\;\;\;$\square$
\end{remark}
\section{The geometry of celestial cartography}
\label{PSELLEDUE}
As stressed in the previous section, the celestial sphere $\mathbb{C\,S}(p)$ of the physical observer $(p, \dot\gamma(p))$, and the celestial sphere $\widehat{\mathbb{C\,S}}(p)$ of the FLRW ideal observer $(p, \widehat{\dot\gamma}(p))$ cannot be directly identified as they stand. The velocity fields $\dot{\gamma}(p)$ and $\widehat{\dot{\gamma}}(p)$ are distinct and to compensate for the induced  aberration, $\mathbb{C\,S}(p)$ and $\widehat{\mathbb{C\,S}}(p)$ can be identified only up to Lorentz transformations. In the standard FLRW view, this reduces to the familiar kinematical boost taking care of the dipole component in the CMB spectrum due to our peculiar motion with respect to the standard of rest provided by the CMB.
In the pre-homogeneity region, \emph{i.e.} when we consider scales $\lesssim  \, 100h^{-1}\; \mathrm{Mpc}$, the situation is relatively more complex. As discussed above,  even if we factor out the effect of coherent bulk flows due to the non-linear local gravitational dynamics, the variance in the peculiar velocity of the sources,  with respect to the average rate of expansion, can be significant \cite{Wiltshire}. This variance implies that a global kinematical boost cannot handle the complex pattern of local inhomogeneities and the peculiar motion of the sources in this region. To address this issue, we need to provide a scale-dependent characterization of the Lorentz transformations connecting the past null directions on the physical celestial sphere $\mathbb{C\,S}(p)$ with the past null directions on the reference FLRW celestial sphere $\widehat{\mathbb{C\,S}}(p)$.

To describe the scale-dependent  Lorentz transformations connecting the null directions of  $\widehat{\mathbb{C\,S}}(p)$ and $\mathbb{C\,S}(p)$,  we use the well-known correspondence between the restricted Lorentz group and the six-dimensional projective special linear 
group, describing the automorphisms of the Riemann sphere  $\mathbb{S}^2\,\simeq\,\mathbb{C}\,\cup\,\{\infty\}$,  \emph{i.e.} 
\begin{equation}
\mathrm{PSL}(2, \mathbb{C})\,=\,\left\{\left. A\,\in\,\mathrm{Mat}(2, \mathbb{C})\,\right|\;\det(A)\,=\,1\right\}/\left\{\,\pm\,I_2  \right\}\,,
\end{equation}
the group of unimodular $2\times \,2$ complex matrices quotiented by its center ($\pm$\, the $2\times 2$ identity matrix $I_2$). More expressively, $\mathrm{PSL}(2, \mathbb{C})$ can be viewed as the group of the conformal transformations of the celestial spheres that correspond to the restricted Lorentz transformations connecting $\mathbb{C\,S}(p)$ to $\widehat{\mathbb{C\,S}}(p)$.   
Let us recall that the elements of  $\mathrm{PSL}(2, \mathbb{C})$ can be identified with the M\"obius transformations of the Riemann sphere $\mathbb{S}^2\,\simeq\,\mathbb{C}\,\cup\,\{\infty\}$, \emph{i.e.} the fractional linear transformations of the form\footnote{To avoid a notational conflict with the redshift parameter ${z}$, we have labeled the complex coordinate in $\mathbb{C}\,\cup\,\{\infty\}$ with  $w$ rather than with the standard  $z$.} 
\begin{align}
\label{zetaPSL_0}
\zeta\,:\,\mathbb{C}\,\cup\,\{\infty\}\,&\longrightarrow\,\mathbb{C}\,\cup\,\{\infty\}\\
w\,&\longmapsto\,\zeta (w)\,:=\,\frac{aw+b}{cw+d}\,,\,\,\,\,\,a, b, c, d\,\in\,\mathbb{C}\,,\,\,\,ad\,-\,bc\,\not=\,0\,,\nonumber
\end{align}
 Strictly speaking, the map (\ref{zetaPSL_0}) is in the projective general linear group $\mathrm{PGL}(2, \mathbb{C})$, but this is the same as $\mathrm{PSL}(2, \mathbb{C})$ since every nonzero element of $\mathbb{C}$ is the square of a nonzero element. Thus,    without loss in generality, we can normalize the transformation $\zeta$ 
by assuming that 
\begin{equation}
ad\,-\,bc\,=\,1\,.
\end{equation}
Let  $\widehat{X}\,=\,\widehat{n}(\widehat\theta, \widehat\phi)$ denote a point on the FLRW celestial sphere $\widehat{\mathbb{C\,S}}(p)$, and let $\widehat{w}$
denote its stereographic projection\footnote{From the north pole $\theta\,=\,0\in\,\widehat{\mathbb{C\,S}}(p)$.} on the Riemann sphere $\mathbb{C}\cup\,\{\infty\}$, \emph{i.e.},  
\begin{align}
\label{zetaypsilon}
&{\Pi}_{\mathbb{S}^2}\,:\,\widehat{\mathbb{C\,S}}(p)\,
\longrightarrow\,\mathbb{C}\cup\{\infty \}\\
&\widehat{\mathbb{X}}^\alpha\,\longmapsto\,
{\Pi}_{\mathbb{S}^2}(\widehat{\mathbb{X}}^\alpha)\,=\,\widehat{w}\,:=\,\frac{\widehat{\mathbb{X}}^1+i\,\widehat{\mathbb{X}}^2}{1-\widehat{\mathbb{X}}^3}\,=\,\frac{\cos\widehat\phi\sin\widehat\theta\,+\,i\,\sin\widehat\phi\sin\widehat\theta}{1\,-\,\cos\widehat\theta}\,,\nonumber\
\end{align}
with\, $0<\theta\leq\pi,\,\,0\leq\phi<2\pi$. It is worthwhile to stress once more that the celestial spheres $\widehat{\mathbb{C\,S}}(p)$ and ${\mathbb{C\,S}}(p)$ play the role of a mapping frame. A celestial globe where astrophysical positions are registered, and where the Lorentz transformations $\widehat{\mathbb{C\,S}}(p)\,\longrightarrow\,{\mathbb{C\,S}}(p)$ must be interpreted actively as affecting only the recorded astrophysical data. In other words, the Lorentz transformations affect the null directions 
in $\widehat{\mathbb{C\,S}}(p)$, mapping them in the corresponding directions in ${\mathbb{C\,S}}(p)$.  To set notation, we provide a few illustrative examples \cite{RindPen} of the 
$\mathrm{PSL}(2, \mathbb{C})$ transformations associated with the redshift dependent Lorentz group action between the celestial spheres $\widehat{\mathbb{C\,S}}(p)$ and ${\mathbb{C\,S}}(p)$. If we sample the pre-homogeneity region, the relative $3$-velocity $v_{\;\widehat{z}}(p)$ of the physical observer $(p,\,\dot\gamma(p))$ with respect to the  FLRW ideal observer $(p,\,\widehat{\dot\gamma}(p))$ is redshift-dependent and given by (\ref{3velocity2}).  We stress that the physical observer adopts the given  FLRW $\widehat{z}$\, as the reference cosmological redshift. If the map between $\widehat{\mathbb{C\,S}}(p)$ and ${\mathbb{C\,S}}(p)$ is a pure Lorentz boost in a common direction, say $E^3\,=\,\widehat{E}^3$,  then the associated $\mathrm{PSL}(2, \mathbb{C})$ transformation is provided by   
\begin{align}
\label{PSLfactor1}
\mathrm{PSL}(2, \mathbb{C})\times \widehat{\mathbb{C\,S}}(p)\,&\longrightarrow \,{\mathbb{C\,S}}(p)\\
\left(\zeta^{\,boost}_{(\widehat{z})},\,\widehat{w}\right)\,&\longmapsto\,
\zeta^{\,boost}_{(\widehat{z})}(\widehat{w})\,=\,w\,:=\,\sqrt{\frac{1\,+\,v_{\,\widehat{z}}(p)}{1\,-\,v_{\,\widehat{z}}(p)}}\,\widehat{w}\,,\nonumber
\end{align}
where $\sqrt{\frac{1\,+\,v_{\,\widehat{z}}}{1\,-\,v_{\,\widehat{z}}}}$ is the relativistic Doppler factor, and $w$ is the point in the Riemann sphere corresponding, under stereographic projection, to the direction $n(\theta, \phi)\in {\mathbb{C\,S}}(p)$. \; 
Conversely, if $\widehat{\mathbb{C\,S}}(p)$ and ${\mathbb{C\,S}}(p)$ differ by a pure rotation through an angle $\alpha_{\;\widehat{z}}$ about a common direction (say $E^3$), then the associated $\mathrm{PSL}(2, \mathbb{C})$ transformation is given by 
\begin{align}
\label{PSLfactor2}
\mathrm{PSL}(2, \mathbb{C})\times \widehat{\mathbb{C\,S}}(p)\,&\longrightarrow \,{\mathbb{C\,S}}(p)\\
\left(\zeta^{\,rot}_{(\widehat{z})},\,\widehat{w}\right)\,&\longmapsto \,
\zeta^{\,rot}_{(\widehat{z})}(\widehat{w})\,=\,w\,=\,e^{i\,\alpha_{\;\widehat{z}}}\,\widehat{w}\,.\nonumber
\end{align}

The $\mathrm{PSL}(2, \mathbb{C})$ map  connecting the celestial spheres $\widehat{\mathbb{C\,S}}(p)$ and ${\mathbb{C\,S}}(p)$ is obtained by composing 
(\ref{PSLfactor1}) and (\ref{PSLfactor2}) and  is fully 
determined by the images, on the physical $\mathbb{C\,S}$, of three distinct astrophysical sources of choice (\emph{e.g.} Cepheid variable stars) selected, at the given redshift $\widehat{z}$,  on the FLRW celestial sphere $\widehat{\mathbb{C\,S}}(p)$. Explicitly, we have the following result.

\begin{proposition} (\emph{The $\mathrm{PSL}(2, \mathbb{C})$ map})
\label{PSL2mappingFLRW}
Let $\{\widehat{w}_{(j)}\}$,\;$j=1,2,3$,\; denote the three null directions on $\widehat{\mathbb{C\,S}}(p)$, associated to the three selected astrophysical sources 
$\{\widehat{q}_{(j)}\}$ that according to the FLRW observer $(p, \widehat{\dot{\gamma}}(p))$ are located at  FLRW redshift $\widehat{z}$. Given such reference data, the M\"obius transformation $\zeta_{(\widehat{z})}\,\in\,\mathrm{PSL}(2, \mathbb{C})$ connecting $\widehat{\mathbb{C\,S}}(p)$ and ${\mathbb{C\,S}}(p)$ is determined by the images $\{\zeta_{(\widehat{z})}(\widehat{w}_{(j)})\}\,\in\,{\mathbb{C\,S}}(p)$. The choice of these images characterizes 
a rotation through an angle $\alpha_{\;\widehat{z}}$  followed by a boost with rapidity $\beta_{\;\widehat{z}}\,:=\,\log\,\sqrt{\frac{1\,+\,v_{\;\widehat{z}}}{1\,-\,v_{\;\widehat{z}}}}$, where $v_{\;\widehat{z}}(p)$ is the relative velocity  of the physical observer $(p, \dot{\gamma}(p))$ with respect to the reference FLRW observer $(p, \widehat{\dot{\gamma}}(p))$. The resulting fractional linear transformation mapping the celestial spheres $\widehat{\mathbb{C\,S}}(p)$ and ${\mathbb{C\,S}}(p)$ is provided by 
\begin{align}
\label{psl2caction_0}
\mathrm{PSL}(2, \mathbb{C})\times \widehat{\mathbb{C\,S}}(p)\,&\longrightarrow \,\mathbb{C\,S}(p)\\
\left( \zeta_{(\widehat{z})},\,\widehat{w}\right)\,&\longmapsto \,\zeta_{(\widehat{z})}(\widehat{w})\,=
\,w\,=\,\sqrt{\frac{1\,+\,v_{\;\widehat{z}}(p)}{1\,-\,v_{\;\widehat{z}}(p)}}\,e^{i\,\alpha_{\;\widehat{z}}}\,\widehat{w}\,.
\nonumber
\end{align}
\end{proposition}
\begin{proof}
From a geometrical point of view, the characterization of (\ref{psl2caction_0}) is a direct consequence of the definition of the M\"obius transformation (\ref{zetaPSL_0}). We can write
\begin{equation}
\zeta_{(\widehat{z})} (w)\,:=\,\frac{aw+b}{cw+d}\,=\,\left(\frac{b}{c}\right)\,\frac{\left(\frac{a}{b}\,w\,+\,1 \right)}{\left(w\,+\,\frac{d}{c}\right)}\,\,\,\,\,a, b, c, d\,\in\,\mathbb{C}\,,
\end{equation}
which implies that $\zeta_{(\widehat{z})} (w)$ is fixed by the three complex equations
\begin{equation}
\frac{b}{c}\,=\,\alpha_1,\;\;\;\frac{a}{b}\,=\,\alpha_2,\;\;\;\frac{d}{c}\,=\,\alpha_3\,,
\end{equation}
where the complex numbers $\alpha_1,\,\alpha_2,\,\alpha_3$ describe three, non-coincident, null directions on the celestial sphere $\widehat{\mathbb{C\,S}}(p)$. From a physical point of view, this geometrical result is simply the statement that, by adjusting the relative  
velocity and orientation of the physical observer, we can associate the three specified positions on $\widehat{\mathbb{C\,S}}(p)$ with three given sources on the physical celestial sphere ${\mathbb{C\,S}}(p)$. It is important to stress that  under a general M\"obius transformation, two or all the three distinct points $\{\widehat{w}_{(j)}\}\in\,\widehat{\mathbb{C\,S}}(p)$ may collapse into a single image point $\{\zeta_{(\widehat{z})}(\widehat{w}_{(j)})\}\,\rightarrow\,w_0\,\in\,{\mathbb{C\,S}}(p)$. We need to control these possible degeneracies since distinct sources for the FLRW observer $(p, \widehat{\dot\gamma}(p))$  must go into corresponding distinct sources for the physical observer $(p, \dot\gamma(p))$. Identifications are indeed possible under the effect of gravitational lensing, caused (via the exponential map $\exp_p$) by the behavior of the physical null geodesics, and not\footnote{Unless one is willing to consider the admittedly artificial configuration (at least in our specific cosmological setting) of a relative velocity, between the two observers  $(p, \widehat{\dot\gamma}(p))$  and  $(p, \dot\gamma(p))$, that approaches the speed of light.} by the action of $\zeta_{(\widehat{z})}\,\in\,\mathrm{PSL}(2, \mathbb{C})$. \;\;\;\;\;$\square$
\end{proof}
Besides the directional aberration and the velocity-induced boost, the $\mathrm{PSL}(2,\mathbb{C})$ action relates  also the physical comoving radius $r(\widehat{z})$ in $(T_pM,\,\{{E}_{(a)}\},\, \dot{\gamma}(p))$ to the FLRW comoving radius $\widehat{r}\,(\,\widehat{z}\,)$ (see (\ref{Dandri0}) and (\ref{factorErre})) according to 
\begin{align}
\label{radialconnect0}
r(\,\widehat{z}\,)\,&=\,\sqrt{\frac{1\,+\,v_{\;\widehat{z}}(p)}{1\,-\,v_{\;\widehat{z}}(p)}}\,\,\widehat{r}(\,\widehat{z}\,)\\
&=\,\sqrt{\frac{1\,+\,v_{\;\widehat{z}}(p)}{1\,-\,v_{\;\widehat{z}}(p)}}\,
\left(\widehat{r}\,\left(\widehat{z}_{(c)}\right)\,+\,{\int_{\;\widehat{z}_{(c)}}^{\widehat{z}}\,\frac{d\widehat{z}\,'}{H(\widehat{z}\,')}}\right)\,.\nonumber
\end{align}
From a physics point of view, this is a special relativistic effect associated with the radial parametrization of the null geodesics reaching the observers, 
$(p,\,\dot\gamma(p))$\,  and\, $(p,\,\widehat{\dot\gamma}(p))$,\, from a given source. Explicitly,  we can consider, in the local inertial frame $(\widehat{T}_p M,\,\{\widehat{E}_{(a)}\}, \widehat{\dot{\gamma}}(p))$, the image of two photons emitted, at distinct times, by a source ${q}$. As seen in $({T}_p M,\,\{\widehat{E}_{(a)}\}, \widehat{\dot{\gamma}}(p))$, the (images of the) two photons travel from the past towards the origin (\emph{i.e.}, $p\,\equiv\,\widehat{r}=0$) along the radial axis by keeping a constant radial separation, say $\widehat{r}_0$, between them. On the other hand, in the local inertial frame $({T}_p M,\,\{{E}_{(a)}\}, \dot\gamma(p))$  associated with the physical observer $(p,\,\dot\gamma(p))$, and moving with a (signed) radial velocity $v_{\;\widehat{z}}(p)$,\, the separation between the (images of the) two photons is provided by the well-known special relativistic transformation rule
$\widehat{r}_0  \,=\,\widehat{r}_0\,\sqrt{\frac{1\,-\,v_{\;\widehat{z}}(p)}{1\,+\,v_{\;\widehat{z}}(p)}}$, describing optical aberration and from which (\ref{radialconnect0}) follows by synchronizing the emission time of the second photon with the arrival, at $p$, of the first photon.

Proposition \ref{PSL2mappingFLRW}, the reference role of the FLRW redshift $\widehat{z}$, and the map (\ref{radialconnect0}) suggest that along with the directional celestial spheres $\widehat{\mathbb{CS}}(p)$ and ${\mathbb{CS}}(p)$ we can profitably introduce their versions at the given cosmological redshift $\widehat{z}$.

\begin{definition} 
The celestial spheres associated with the field of vision at redshift $\widehat{z}$ of the FLRW observer $(p,\, \widehat{\dot\gamma}(p))$ and of the physical observer $(p,\, {\dot\gamma}(p))$ are respectively defined by 
\begin{equation}
\label{celestialS002}
\widehat{\mathbb{CS}}_{\;\widehat{z}}(p)\,:=\,\left\{\widehat{X}\,=\,\widehat{\mathbb{X}}^i\widehat{E}_{(i)}\,\not=\,0\,\in\,T_pM\,\,|\,\,\widehat{\mathbb{X}}^4=0,\,\,\sum_{a=1}^3(\widehat{\mathbb{X}}^{a})^2=\,\widehat{r}^{\;2}(\,\widehat{z}) \right\}\,
\end{equation}
and
\begin{equation}
\label{celestialS01}
{\mathbb{CS}}_{\;\widehat{z}}\,:=\,\left\{X\,=\,\mathbb{X}^iE_{(i)}\,\not=\,0\,\in\,T_pM\,\,|\,\,\mathbb{X}^{\,4}=0,\,\,\sum_{a=1}^3(\mathbb{X}^a)^2= 
\left(\frac{1\,+\,v_{\,\widehat{z}}}{1\,-\,v_{\,\widehat{z}}}\right)\,\widehat{r}^{\,2}(\,\widehat{z}\,)
\right\}\,,
\end{equation}
where (see (\ref{Dandri0}) and (\ref{factorErre})) 
\begin{equation}
\label{Dandri0bis} 
\widehat{r}(\,\widehat{z}\,)\,=\,\widehat{r}\,\left(\widehat{z}_{(c)}\right)\,+\,{\int_{\;\widehat{z}_{(c)}}^{\widehat{z}}\,\frac{d\widehat{z}\,'}{H(\widehat{z}\,')}}\,.
\end{equation} 
is the FLRW comoving radius expressed in terms of the (FLRW) Hubble parameter $H(\widehat{z})$ (see (\ref{AccaZeta})) and of the presence of a cosmological decoupling region of comoving radius $\widehat{r}\,\left(\widehat{z}_{(c)}\right)$.

The physical celestial sphere ${\mathbb{CS}}_{\;\widehat{z}}$ is naturally endowed with the round directional  metric and measure defined by 
\begin{align}
\label{roundmet0erre}
{h}_{\mathbb{S}^2(\,\widehat{z}\,)}\,&:=\,{r}^{\,2}(\,\widehat{z}\,)\,\left(d\theta^2\,+\,\sin^2\theta\,d\phi^2\right)\,,\nonumber\\
\\
d\mu_{\mathbb{S}^2(\,\widehat{z}\,)}\,&=\,{r}^2(\,\widehat{z}\,)\,\sin\theta\,d\theta d\phi\,\nonumber\,,
\end{align}
with $d_{{\mathbb{S}}^2(\,\widehat{z}\,)}$ denoting the ${h}_{\mathbb{S}^2(\,\widehat{z}\,)}$-distance function. Similarly, the FLRW celestial 
sphere $\widehat{\mathbb{CS}}_{\;\widehat{z}}$  is endowed with the radial counterpart of the round metric (\ref{FLRWroundmet0}) and measure, 
\begin{align}
\label{FLRWroundmet0erre}
{h}_{\widehat{\mathbb{S}}^2(\,\widehat{z}\,)}\,&:=\,\widehat{r}^{\;2}(\,\widehat{z}\,)\,\left(d\widehat\theta^{\;2}\,+\,\sin^2\widehat\theta\,d\widehat\phi^{\;2}\right)\,=\,
\left(\frac{1\,-\,v_{\,\widehat{z}}\,(p)}{1\,+\,v_{\,\widehat{z}}\,(p)\,}\right)\,
{r}^{\,2}(\,\widehat{z}\,)\,\left(d\widehat\theta^2\,+\,\sin^2\widehat\theta\,d\widehat\phi^2\right)\,,\nonumber\\
\\
d\mu_{\widehat{\mathbb{S}}^2(\,\widehat{z}\,)}\,&=
\,\widehat{r}^{\;2}(\,\widehat{z}\,)\,\sin\widehat\theta\,d\widehat\theta d\widehat\phi\,
=\,
\left(\frac{1\,-\,v_{\,\widehat{z}}\,(p)}{1\,+\,v_{\,\widehat{z}}\,(p)}\right)\,
{r}^{\,2}(\,\widehat{z}\,)\,\sin\widehat\theta\,d\widehat\theta d\widehat\phi\,,\nonumber
\end{align}
where $(\theta,\,\phi)\,=\,e^{i\alpha_{\,\widehat{z}}}\,(\widehat\theta,\,\widehat\phi)$ is the rotation described in Proposition \ref{PSL2mappingFLRW}. We denote by $d_{\widehat{\mathbb{S}}^2(\,\widehat{z}\,)}$ the distance associated with the round metric ${h}_{\widehat{\mathbb{S}}^2(\,\widehat{z}\,)}$.
\end{definition}
\vskip 0.5cm\noindent
These characterizations imply the following result.
\vskip 0.3cm\noindent
\begin{lemma}
 The $\mathrm{PSL}(2,\,\mathbb{C})$ map
\begin{equation}
\mathrm{PSL}(2, \mathbb{C})\,\ni\,\zeta_{(\widehat{z})}\,:\widehat{\mathbb{C\,S}}(p)\,\longrightarrow\,\mathbb{C\,S}(p)\,,
\end{equation}
defined by (\ref{psl2caction_0}) can be extended to the physical  $({\mathbb{C\,S}}_{\,\widehat{z}}(p),\,{h}_{\mathbb{S}^2(\,\widehat{z}\,)})$ and the FLRW celestial sphere $(\widehat{\mathbb{C\,S}}_{\;\widehat{z}},\,{h}_{\widehat{\mathbb{S}}^2(\,\widehat{z}\,)})$ according to
\begin{equation}
\label{extensionZeta}
\left(\widehat{\mathbb{C\,S}}_{\;\widehat{z}},\,{h}_{\widehat{\mathbb{S}}^2(\,\widehat{z}\,)}
\,=\,\zeta_{(\widehat{z})}^*({h}_{{\mathbb{S}}^2(\,\widehat{z}\,)})
\right)\,\longrightarrow\,
\left({\mathbb{C\,S}}_{\;\widehat{z}}\,=\,\zeta_{(\widehat{z})}\left(\widehat{\mathbb{C\,S}}_{\;\widehat{z}}\right),\,{h}_{{\mathbb{S}}^2(\,\widehat{z}\,)}\right)\,.\;\;\;\square
\end{equation} 
\end{lemma} 
\vskip 0.3cm\noindent
In other words, the \emph{directional geometry} of  $(\widehat{\mathbb{C\,S}}_{\;\widehat{z}},\,{h}_{\widehat{\mathbb{S}}^2(\,\widehat{z}\,)})$ and $({\mathbb{C\,S}}_{\;\widehat{z}}\,,\,{h}_{{\mathbb{S}}^2(\,\widehat{z}\,)})$ are mapped into each other under the action of the $\mathrm{PSL}(2, \mathbb{C})$ map  $\zeta_{(\widehat{z})}$ and its pull back. This is not surprising since it is a direct consequence of the above definitions. However, the celestial spheres carry more sophisticated information than the directional geometry used to characterize the sky coordinates of astrophysical sources. This further geometric structure is often ignored in the usual approaches to cosmology\footnote{Some of the relevant aspects of this geometric structure are implicitly dealt with on the past lightcone without formalizing, as we do, the role of the null exponential map.}, yet, it can be made explicit  by transferring, by means of the null exponential map, geometric data from the physical and FLRW past lightcones $\mathscr{C}^-(p,g)$ and 
$\widehat{\mathscr{C}}^-(p,\widehat{g}\,)$ to ${\mathbb{C\,S}}_{\;\widehat{z}}$ and $\widehat{\mathbb{C\,S}}_{\;\widehat{z}}$.

\section{The geometry of the sky sections}
\label{skysezioni}
As the FLRW reference redshift $\widehat{z}$ varies, the images of the celestial spheres 
${\mathbb{C\,S}}_{\,\widehat{z}}(p)$ and $\widehat{\mathbb{C\,S}}_{\;\widehat{z}}$ under the exponential maps $\exp_p$ and $\widehat{\exp}_p$ define the sky sections of the physical and FLRW past lightcones $\mathscr{C}^-(p,g)$ and 
$\widehat{\mathscr{C}}^-(p,\widehat{g}\,)$. We discuss in detail the geometry of these sky sections and its subtle connection with the actual description of the astrophysical sources on the physical and FLRW celestial spheres.
 
Before addressing the problems related to the Lipschitz nature of the physical past lightcone, 
let us assume, for the moment, that the exponential map $\exp_p$ is a local diffeomorphism from a star-shaped neighborhood $N_0(g)$ of  $0\,\in\,W_p\subseteq T_pM$ into a corresponding geodesically convex neighborhood of $p$, $U_p 	\subseteq\,(M, g)$ (see (\ref{geodnormal_0})). Then, according to (\ref{pastcone2_0}) and to the Gauss lemma applied to $\exp_p\,:\,C^-\left(T_pM,\, \{E_{(i)} \}\right)\cap\,N_0(g)\,\longrightarrow\,\mathscr{C}^-(p,g)\,\cap\,U_p$, we can assume that the past lightcone region
$\mathscr{C}^-(p, g)\,\cap\,U_p\setminus \{p\}$ is smoothly foliated by the $r$-dependent family of   2--dimensional surfaces $\Sigma_{\;\widehat{z}}$, the \emph{cosmological sky sections at (reference) redshift} $\widehat{z}$, defined by
\begin{equation}
\label{sigmapr_0}
\Sigma_{\;\widehat{z}}\,:=\,\exp_p\left[\mathbb{C\,S}_{\;\widehat{z}} \right]\,=\, \left\{\left.\exp_p\left(r(\widehat{z})\,\ell({n}(\theta, \phi))\right)\,\right|\,\, n(\theta, \phi) \,\in\,\mathbb{C\,S}(p)\right\}\,.
\end{equation}
These surfaces are $g$-orthogonal to all null geodesics originating at $p$, \emph{i.e.} 
\begin{equation}
\label{GaussLemma_0}
\left.g\left(T\exp_p(r(\widehat{z})\ell),\,T\exp_p(\underline{u})\right)\right|_{\exp_p(r(\widehat{z})\ell)}\,=\,0\,,
\end{equation}
where $T\exp_p(...)$ denotes the tangent mapping associated to $\exp_p$, and  $\underline{u}$ is the generic vector tangent to $\mathbb{CS}_{\;\widehat{z}}$. As long as $\exp_p$ is a diffeomorphism, each surface $\Sigma_{\;\widehat{z}}\in\,\mathscr{C}^-(p, g)\,\cap\,U_p\setminus \{p\}$ is topologically a 2-sphere endowed with the $\widehat{z}$-dependent two-dimensional Riemannian metric
\begin{equation}
\label{Sigmametric_0}
g^{(2)}_{\;\widehat{z}}\,:=\,\iota_{\;\widehat{z}}^*\,\left.g\right|_{\mathscr{C}^-(p, g)}
\end{equation}
induced by the inclusion $\iota_{\;\widehat{z}}:\Sigma_{\;\widehat{z}}\,\hookrightarrow\,\mathscr{C}^-(p, g)$ of $\Sigma_{\;\widehat{z}}$  into $\mathscr{C}^-(p, g)\,\cap\,U_p\setminus \{p\}$.

By proceeding similarly, we can define the FLRW cosmological sky sections 
\begin{align}
\label{sigmapr_0FLRW}
\widehat{\Sigma}_{\;\widehat{z}}\,&:=\,\widehat{\exp}_p\left[\widehat{\mathbb{C\,S}}_{\;\widehat{z}} \right]\,\subset\,\widehat{\mathscr{C}}^-(p,\widehat{g})\\
&= 
\left\{\left.\widehat{\exp}_p\left(\,\widehat{\,r}(\,\widehat{z}\,)\,\widehat{\ell}(\,\widehat{n}(\,\widehat\theta, \widehat\phi))\right)\,\right|\,\, (\widehat\theta, \widehat\phi) \,\in\,\widehat{\mathbb{C\,S}}(p)\right\}\,.\nonumber
\end{align}
as smooth  2-spheres, orthogonal to the null geodesics generators of the FLRW past lightcone
$\widehat{\mathscr{C}}^-(p,\widehat{g})$. Contrary to what happens in the case of the physical sky sections $\Sigma_{\;\widehat{z}}$, the FLRW sky sections are always well defined (since $\widehat{\exp}_p$ is a diffeomorphism\footnote{As long as we do not extend the FLRW modeling to the Big Bang singularity, clearly not the case in the analysis discussed here.}), and  each $\widehat{\Sigma}_{\;\widehat{z}}$ carries the 2-dimensional round metric
\begin{equation}
\label{inducemetric0FLRW}
\widehat{g}^{(2)}_{\;\widehat{z}}\,:=\,\left.\widehat{\iota}^*_{\;\widehat{z}}\,\widehat{g}\right|_{\widehat{\mathscr{C}}^-(p,\widehat{g})}\,,
\end{equation}
induced by the inclusion  
$\widehat{\Sigma}_{\;\widehat{z}}\,\hookrightarrow\,\widehat{\mathscr{C}}^-(p,\widehat{g})$. When stressing the metric structures of the sky sections  $\Sigma_{\;\widehat{z}}$ and $\widehat{\Sigma}_{\;\widehat{z}}$, we respectively write
\begin{equation}
\left(\Sigma_{\;\widehat{z}},\,g^{(2)}_{\;\widehat{z}}\right)\;\;\;\mathrm{and} \;\;\;\left(\widehat{\Sigma}_{\;\widehat{z}},\,\widehat{g}^{(2)}_{\;\widehat{z}}\right)\,.
\end{equation}
Actually, under the assumptions we have made on the nature of the physical spacetime $(M, g)$, we can be quite more specific on the nature of the two-dimensional geometry of the surface $\left(\Sigma_{\;\widehat{z}},\,g^{(2)}_{\;\widehat{z}}\right)$. We have the following result.

\begin{proposition} (The physical metric on the celestial sphere $\mathbb{C\,S}_{\hat{z}}$)
\label{lemmaDiffF}

Let $\left(\Sigma_{\widehat{z}_1},\,g^{(2)}_{\widehat{z}_1}\right)$ and $\left(\Sigma_{\widehat{z}_2}\,,\,g^{(2)}_{\widehat{z}_2}\right)$,\,denote two sky sections of ${\mathscr{C}}^-(p,{g})$ respectively associated with the reference redshifts $\widehat{z}_1$ and $\widehat{z}_2$, with $\widehat{z}_2\,>\,\widehat{z}_1$, and let 
\begin{equation}
r(\widehat{z}_\varepsilon)\,\ell(n)\,:=\,
r((1-\varepsilon)\widehat{z}_1\,+\,\varepsilon\widehat{z}_2)\,\ell(n),\;\;\;0\,\leq\,\varepsilon\,\leq\,1,\;\;\;n(\theta, \phi) \,\in\,\mathbb{C\,S}(p),
\end{equation}
 be the null ray segments interpolating between the celestial spheres $\mathbb{C\,S}_{\widehat{z}_1}$ and $\mathbb{C\,S}_{\widehat{z}_2}$.
Then, the  map
\begin{equation}
\label{Finterpolation}
F_{\widehat{z}_1,\,\widehat{z}_2}(\epsilon)\,:=\,\exp_p\left(r(\widehat{z}_\varepsilon)\,\ell(n)\right),\;\;\;\;0\,\leq\,\varepsilon\,\leq\,1\,
\end{equation}
which interpolates between the sky sections $\Sigma_{\widehat{z}_1}=F_{\widehat{z}_1,\,\widehat{z}_2}(0)$ and $\Sigma_{\widehat{z}_2}=F_{\widehat{z}_1,\,\widehat{z}_2}(1)$,
is a diffeomorphism, and
\begin{equation}
\label{Fpullback}
g^{(2)}_{\widehat{z}_1}\,=\, F^*_{\widehat{z}_1\,\widehat{z}_2}\,g^{(2)}_{\widehat{z}_2}\,,
\end{equation}
where $F^*_{\widehat{z}_1,\,\widehat{z}_2}$ denotes the pull-back action associated with the map (\ref{Finterpolation}). This action can be naturally extended to the celestial sphere $\mathbb{CS}_{\,\widehat{z}\,}(p)$ by using the null exponential map
\begin{equation}
\exp_p\,:\,\mathbb{CS}_{\,\widehat{z}\,}(p)\,\longrightarrow\,\left(\Sigma_{\,\widehat{z}\,},\,g^{(2)}_{\,\widehat{z}\,}\right)\,,
\end{equation}
 to characterize on the family of redshift-dependent celestial spheres $\{\mathbb{CS}_{\,\widehat{z}\,}(p)\}\,\subset\,C^{\,-}(T_pM,\,\{E_{(i)}\})$ the corresponding family of pullback metrics
\begin{equation}
\label{metrich}
\widehat{z}\,\longrightarrow\,h_{\,\widehat{z}\,}\,:=\,\exp_p^*\,g^{(2)}_{\,\widehat{z}\,}\circ\exp_p\,,
\end{equation}
transferring to the physical observer's null cone $C^{\,-}(T_pM,\,\{E_{(i)}\})$ the redshift dependent geometry of the past lightcone sections $\{\Sigma_{\,\widehat{z}}\,\}$.
It follows that each celestial sphere $\mathbb{CS}_{\,\widehat{z}\,}(p)$ is naturally endowed with two distinct metrics,
\begin{equation}
\left(\mathbb{CS}_{\,\widehat{z}\,}(p),\,{h}_{\mathbb{S}^2(\,\widehat{z}\,)},\,h_{\,\widehat{z}\,}\right)\,.
\end{equation}
the directional round metric $\widetilde{h}_{\mathbb{S}^2(\,\widehat{z}\,)}$, and the physical metric $h_{\,\widehat{z}\,}$ transferred, by the null-geodesic flow, from the sky section $\Sigma_{\,\widehat{z}\,}\,=\,\exp_p(\mathbb{CS}_{\,\widehat{z}\,}(p))$. Whereas the former gives the celestial coordinates $n(\theta, \phi)$ and comoving radius $r(\,\widehat{z}\,)$ of the astrophysical sources $q\,=\,\exp_p(r(\,\widehat{z}\,)\ell(n(\theta, \phi)))$, the latter provides more specific geometrical and physical information on the sky section $\Sigma_{\,\widehat{z}\,}$. 
\end{proposition} 
\begin{proof}
These results are a direct consequence of the definition (\ref{sigmapr_0}) of the 
sky section $\Sigma_{\;\widehat{z}}$ and of the assumed diffeomorphic nature of the exponential map $\exp_p$. The relation (\ref{Finterpolation}) simply reflects the fact that, as $r(\widehat{z})$ varies,  the geodesic null rays
\begin{equation}
\exp_p\left(r(\widehat{z})\,\ell(n)\right),\;\;\;\;n(\theta, \phi) \,\in\,\mathbb{C\,S}(p)
\end{equation}
are the generators of the past lightcone ${\mathscr{C}}^-(p,{g})$ transferring geometric information among the sky sections $\{\Sigma_{\;\widehat{z}}\}$.   
Later on, we prove that Proposition \ref{lemmaDiffF} can be suitably extended to the less regular case of a bi-Lipschitzian exponential map $\exp_p$. \;\;\;\;\;$\square$
\end{proof}

Proposition \ref{lemmaDiffF} extends naturally to the sky sections $\{\widehat{\Sigma}_{\;\widehat{z}}\}$ of the reference FLRW past lightcone $\widehat{\mathscr{C}}^-(p,\widehat{g})$. In analogy with (\ref{Finterpolation}), we denote by
\begin{equation}
\label{FinterpolationFLRW}
\widehat{F}_{\widehat{z}_1,\,\widehat{z}_2}(\epsilon)\,:=\,\widehat{\exp}_p\left(\,\widehat{r}\,(\widehat{z}_\varepsilon)\,\widehat{\ell}\,(\,\widehat{n}\,)\right),\;\;\;\;0\,\leq\,\varepsilon\,\leq\,1\,
\end{equation}
the diffeomorphism interpolating between the sky sections $\widehat{\Sigma}_{\widehat{z}_1}\,=\,\widehat{F}_{\widehat{z}_1,\,\widehat{z}_2}(0)$ and $\widehat{\Sigma}_{\widehat{z}_2}\,=\,\widehat{F}_{\widehat{z}_1,\,\widehat{z}_2}(1)$.
In particular, the FLRW version of the map (\ref{Fpullback}) reduces to a conformal diffeomorphism, \emph{i.e.}  
\begin{equation}
\label{FpullbackFLRW}
\widehat{g}^{\;(2)}_{\widehat{z}_1}\,=\, \widehat{F}^{\,*}_{\widehat{z}_1\,\widehat{z}_2}\,g^{(2)}_{\widehat{z}_2}\,=\,\Theta^2_{\widehat{z}_1\,\widehat{z}_2}\,g^{(2)}_{\widehat{z}_2}\,,
\end{equation}
with
\begin{equation}
\Theta_{\widehat{z}_1\,\widehat{z}_2}\,:=\,
\frac{f\left(\hat{r}(\widehat{z}_1)\right)\,\left(1\,+\,\widehat{z}_2\right)}{f\left(\hat{r}(\widehat{z}_2)\right)\,\left(1\,+\,\widehat{z}_1\right)}\,.
\end{equation}
\vskip 0.4cm\noindent
Given the geometrical nature of the FLRW coordinates (see (\ref{FLRWg0}) and (\ref{aScaling})), we have an explicit characterization of the FLRW counterpart $\widehat{h}_{\,\widehat{z}\,}$ of the family of pullback metrics (\ref{metrich}),\, on the reference FLRW family of celestial spheres 
$\{\widehat{\mathbb{CS}}_{\,\widehat{z}\,}(p)\}$.\, Using the the standard FLRW coordinates $(\,\widehat{r},\,\widehat{\theta},\,\widehat{\phi},\,\widehat{\tau}\,)$, the metric induced by the inclusion of the sky section $\widehat{\Sigma}_{\,\widehat{r}\,}$, at comoving radius $\widehat{r}$, into $\widehat{\mathscr{C}}^-(p, \hat{g})$, is given by
 \begin{equation}
 \label{etametric01}
 \widehat{g}^{\,(2)}_{\widehat{r}}\,:=\,
a^2(\widehat{\tau}(\widehat{r}))\,f^2\left(\widehat{r}\right)\left(d\widehat{\theta}^2\,+\,\sin^2\widehat\theta d\widehat{\phi}^2\right)\;,
\end{equation} 
where $a(\widehat{\tau}(\widehat{r}))$ is the  FLRW expansion factor $a(\widehat{\tau})$  evaluated in correspondence with the given value of the radial coordinate  $\widehat{r}\in\widehat{T}_pM$. To compute the pullback of (\ref{etametric01}) on the celestial sphere $\widehat{\mathbb{C\,S}}_{\,\widehat{zeta}\,}(p)$,
 let $y^i_q=(\widehat{r}_q, \widehat{\theta}_q, \widehat{\phi}_q, \widehat{\tau}_q)$ the normal coordinates of the event $q\in\,\widehat{\Sigma}_{\,\widehat{r}\,}\,\subset\,\widehat{\mathscr{C}}^-(p, \hat{g})$ associated with the observation of a given astrophysical source whose worldline intercepts  $\widehat{\Sigma}_{\,\widehat{r}\,}$\,  at  $q$. The equation for the radial, past-directed, null geodesic connecting $q$ to the observation event $p$ is  \cite{EllisElst}
\begin{equation}
d\widehat{r}\,=\,-\,\frac{d\widehat\tau}{a(\widehat\tau)}\,,\,\,\,\,\widehat\tau(p)\,=\,0\,=\,
\widehat{r}(p)\,,
\end{equation}
which integrates into the expression that provides the (matter-comoving) radial coordinate distance between $p$ and $q$
\begin{equation}
\label{radialerreq}
\widehat{r}_q\,=\,\int_0^{\widehat{\tau}_q}\,\frac{d\widehat\tau}{a(\widehat\tau)}\,.
\end{equation}
Thus, the pull back of the metric (\ref{etametric01}), evaluated at $\widehat{\exp_p}^{-1}(q)$, can be written in terms of $\widehat\tau_q$ as 
\begin{equation}
\label{FLRWhq}
\widehat{h}_q\,:=\,\widehat{h}(\widehat{r}_q, \widehat\theta_q, \widehat\phi_q)\,=\, 
a^2(\widehat{\tau}_q)\,f^2\left(\widehat{r}_q\right)\left(d\widehat{\theta}_q^2\,+\,
\sin^2\widehat\theta_q d\widehat{\phi}_q^2\right)\;, 
\end{equation}
If we introduce the  dimensionless FLRW cosmological redshift corresponding to the event $q$,
\begin{equation}
\label{zetaFLRW}
\widehat{z}\,(q)\,:=\,\widehat{z}\left(\,\widehat{\tau}_q\right)\,=\,\frac{a_0}{a(\widehat{\tau}_q)}\,-\,1\,,
\end{equation}
where $a_0\,:=\,a(\widehat{\tau}=0)$, and normalize $a_0$ by setting $a_0=1$, then we can rewrite (\ref{FLRWhq}) as
\begin{equation}
\label{metrichatzeta}
\widehat{h}_{\widehat{z}_q}\,=\,
\frac{f^2\left(\widehat{r}(\,\widehat{z}_q)\right)}{(1\,+\,\widehat{z}_q)^2}\,\left(d\widehat{\theta}_q^2\,+\,
\sin^2\widehat\theta_q d\widehat{\phi}_q^2\right)\;. 
\end{equation}  
Since $q$ is the generic point on $\widehat{\Sigma}_{\,\widehat{z}}$, the family of pullback metrics 
\begin{equation}
\label{FLRWhpullBack}
\widehat{z}\,\longrightarrow\,
\widehat{h}_{\,\widehat{z}\,}\,=\,
\frac{f^2\left(\hat{r}(\widehat{z})\right)}{(1\,+\,\widehat{z})^2}\,\left(d\hat\theta^2\,+\,\sin^2\hat\theta\,d\hat\varphi^2 \right)\,,
\end{equation} 
characterizes on the reference FLRW family of celestial spheres 
$\{\widehat{\mathbb{CS}}_{\,\widehat{z}\,}(p)\}$,  the FLRW counterpart of (\ref{metrich}).

\subsection{The Lipschitz nature of the physical sky sections} 
\label{LipNatSky}
The presence of a pre-homogeneity region around the observer at $p$ makes the above smoothness assumptions characterizing the definition of the physical sky sections $\Sigma_{\;\widehat{z}}$ and of the associated celestial spheres ${\mathbb{C\,S}}_{\;\widehat{z}}$ quite unrealistic, even for large $\widehat{z}$ where the foliation $\{(\Sigma_{\;\widehat{z}},\,g^{(2)}_{\;\widehat{z}})\}$ probes the homogeneity region of $\mathscr{C}^-(p,g)$. The point is that null geodesics coming from large ${\;\widehat{z}}$ sources must travel through the low redshift pre-homogeneity region to reach us at the event $p$, and the vagaries of the local distribution of astrophysical sources imply that the physical past lightcone is not smooth. In particular,  $\mathscr{C}^-(p,g)$ may fail to be the boundary $\partial\,\mathrm{I}^-(p,g)$ of the chronological past $\mathrm{I}^-(p,g)$ of $p$, (the set of all events $q\in M$ that can be connected to $p$ by a past-directed timelike curve), because past-directed null geodesics generators of $\mathscr{C}^-(p,g)$,\;  $\lambda\,:\,[0, \delta)\,\longrightarrow\,(M, g)$, 
with $\lambda(0)\,=\,p$,  may leave $\partial\mathrm{I}^-(p,g)$ and, under the action of the local spacetime curvature, plunge into the interior
 $\mathrm{I}^-(p,g)$. In such a situation, the mapping
\begin{equation}
\label{sectionexp}
\left.\exp_p\right|_{\mathscr{C}^-(p,g)}\,:\,\mathbb{C\,S}_{\;\widehat{z}}\,\longrightarrow\,
\Sigma_{\;\widehat{z}}\,:=\,\exp_p\,[\mathbb{C\,S}_{\;\widehat{z}}]
\end{equation} 
is no longer one-to-one, and the cosmological sky section $\Sigma_{\;\widehat{z}}$ fails to be a smooth surface. From the physical point of view, this is the geometrical setting associated with the generation of multiple images of astrophysical sources in the observer celestial sphere  $\mathbb{C\,S}_{\;\widehat{z}}$. The restriction of the exponential map $\exp_p$ to the past lightcone $\mathscr{C}^{-}(p, g)$ is quite difficult to handle in such a low-regularity framework, and 
as the reference redshift $\widehat{z}$ varies,  the development of caustics in $\mathscr{C}^-(p,g)$ generates cusps and crossings in the surfaces $\{\Sigma_{\;\widehat{z}}\}$, to the effect that they are no longer homeomorphic to 2-spheres.  

The mathematical framework for handling such a low-regularity past lightcone scenario is to assume that the sky sections $\Sigma_{\;\widehat{z}}$ are Lipschitz surfaces. Under such assumption, their metric geometry retains a controllable behavior with respect to the Lebesgue and Hausdorff measures even when  $\Sigma_{\;\widehat{z}}$ is everything but a smooth surface. The Lipschitz structure of the generic sky section $\Sigma_{\;\widehat{z}}$ is characterized by a maximal atlas $\mathcal{A}\,=\,\left\{(U_\alpha, \varphi_\alpha)\right\}$ such that all transition maps between the coordinate charts $(U_\alpha, \varphi_\alpha)$ of $\Sigma_{\;\widehat{z}}$,
\begin{equation}
\varphi_{\alpha\beta}\,:=\,\varphi_\beta\circ\varphi_\alpha^{-1}\,:\,\varphi_\alpha\left(U_\alpha\cap\,U_\beta\right)\,
\longrightarrow\,\varphi_\beta\left(U_\alpha\cap\,U_\beta\right),\,
\end{equation}
are locally\footnote{Since $\Sigma_{\,\widehat{z}}$ is compact, the transition maps between the coordinate charts $(U_\alpha, \varphi_\alpha)$ can be taken to be bi-Lipschitz and not just locally bi-Lipschitz.} bi-Lipschitz maps between domains of the Euclidean space $(\mathbb{R}^2, \delta)$, \emph{i.e.}, there exist constants $c_{\alpha\beta}\,\geq\,1$ such that
\begin{equation}
c_{\alpha\beta}^{\,-,1}\,\left|x\,-y\right|\,\leq
\,\left|\varphi_{\alpha\beta}(x)\,-\,\varphi_{\alpha\beta}(y)\right|\,\leq\,c_{\alpha\beta}\,\left|x\,-y\right|,\,\;\;\forall\,x,\,y\,\in\,\varphi(U_\alpha)\cap\varphi(U_\beta)\,.
\end{equation}
Rademacher's theorem \cite{Gariepy}, \cite{Rosenberg} implies that the transition maps
$\varphi_{\alpha\beta}$ on $\Sigma_{\;\widehat{z}}$ have differentials $d\varphi_{\alpha\beta}$ that are defined almost everywhere with respect to the Lebesgue measure, and which are locally bounded and measurable on their domains. They determine \cite{EellsFuglede} a bounded measurable version of the tangent bundle and associated tensor bundles\footnote{In full generality, we can associate with a Lipschitz manifold a version of the tangent bundle which is only a topological fiber bundle and not a vector bundle.} These assumptions allow us to characterize the metric $g^{(2)}_{\;\widehat{z}}$, induced by the inclusion $\Sigma_{\;\widehat{z}}\,\hookrightarrow\,\mathscr{C}^-(p,g)$ (see (\ref{Sigmametric_0})), as a  positive definite symmetric 2-tensor defined almost everywhere\footnote{In presence of cut points on $\mathscr{C}^-(p,g)$,   the inclusion map 
$\iota_r:\Sigma_{\;\widehat{z}}\,\hookrightarrow\,\mathscr{C}^-(p, g)$ of the sky section $\Sigma_{\;\widehat{z}}$  into $\mathscr{C}^-(p, g)$ is Lipschitz, thus Rademacher's theorem allows us to define the pullback metric $g|_{\Sigma_{\;\widehat{z}}}\,:=\,\iota_r^*\,\left.g\right|_{\mathscr{C}^-(p, g)}$  almost-everywhere.} 
on $\Sigma_{\;\widehat{z}}$. For technical reasons \cite{EellsFuglede} we also assume that there are constants $c_\beta\,\geq\,1$, associated with the Lipschitz atlas $(U_\beta, \varphi_\beta)$ of $\Sigma_{\;\widehat{z}}$,  such that   
\begin{equation}
\label{RiemLipSurf}
c_\beta^{\,-\,2}\,\left|d\varphi_\beta(q)\,w\right|^2_{\mathbb{R}^2}\,\leq\,    
g^{(2)}_{\;\widehat{z}}(q)(w, w)\,\leq\,
c_\beta^2\,\left|d\varphi_\beta(q)\,w\right|^2_{\mathbb{R}^2}\,,\;\;\;q\in\,U_\beta\,,
\end{equation} 
holds, almost everywhere, for all tangent vectors $w\in\,T_q\Sigma_{\;\widehat{z}}$. Under these assumptions $(\Sigma_{\;\widehat{z}},\,g^{(2)}_{\;\widehat{z}})$ defines a \emph{Riemannian Lipschitz surface} \cite{EellsFuglede}. 
The characterization of the distance function associated with this extended version of the sky section $(\Sigma_{\;\widehat{z}},\,g^{(2)}_{\;\widehat{z}})$ is slightly delicate since 
the metric $g^{(2)}_{\;\widehat{z}}$ is only defined  almost everywhere on $\Sigma_{\;\widehat{z}}$. To take care of this problem, let us denote by $N$ the generic set of Lebesgue measure $0$ in $(\Sigma_{\;\widehat{z}},\,g^{(2)}_{\;\widehat{z}})$. Given any two points $q_1,\,q_2\,\in\,\Sigma_{\;\widehat{z}}$, consider the set
$\mathrm{Lip}^{N}([a,b]\rightarrow \Sigma_{\;\widehat{z}})$ of all Lipschitz paths
\begin{equation}
\rho_{(1,2)}\,:\, [a,b]\,\ni\,t\,\mapsto\,\rho_{(1,2)}(t)\,\in\,\Sigma_{\;\widehat{z}}\,,\;\;\;
\rho_{(1,2)}(a)=q_1,\;\;\rho_{(1,2)}(b)=q_2\,,
\end{equation} 
such that the pre-images, $\rho_{(1,2)}^{\,-1}\,(N)$, of the null set $N$ have Lebesgue measure $0$ in the interval $[a,b]$. To any such path, we can associate the length functional
\begin{equation}
L^N\left(\rho_{(1,2)}\right)\,:=\,\int_{[a,b]}\,\sqrt{g^{(2)}_{\;\widehat{z}}\left(\frac{d \rho_{(1,2)}(t)}{dt},\,\frac{d \rho_{(1,2)}(t)}{dt} \right)}\,dt\,,
\end{equation}
that can be minimized, for a given zero-measure set $N\,\subset\,\Sigma_{\;\widehat{z}}$, to provide the corresponding distance
\begin{equation}
d^N\left(q_{(1)},\,q_{(2)}\right)\,:=\, \inf\,\left\{L^N\left(\rho_{(1,2)}\right),\;\,
\rho_{(1,2)}\,\in\,\mathrm{Lip}^{N}([a,b]\rightarrow \Sigma_{\;\widehat{z}} )\right\}\,.
\end{equation}
Finally, if we let the zero-measure set $N$ vary over $\Sigma_{\;\widehat{z}}$ we define the distance function associated with $(\Sigma_{\;\widehat{z}},\,g^{(2)}_{\;\widehat{z}})$ according to\footnote{Details of the characterization of the distance function on Riemann Lipschitz manifolds are discussed in \cite{EellsFuglede, DeCecco}.}  
\begin{equation}
\label{Lipdistance}
d_{\Sigma_{\;\widehat{z}}}\left(q_{(1)},\,q_{(2)}\right)\,:=\,\sup_{N}\,d^N\left(q_{(1)},\,q_{(2)}\right)\,.
\end{equation}
These remarks imply that $(\Sigma_{\;\widehat{z}},\,g^{(2)}_{\;\widehat{z}})$ has a Lebesgue measure, locally  defined, in any Lipschitz coordinate chart $(U_\alpha, \varphi_\alpha=(\theta,\,\phi))$, by
\begin{equation}
\label{RiemLipMeas}
d\mu_{g^{(2)}_{\,\widehat{z}}}\:=\,\sqrt{\det\,g^{(2)}_{\,\widehat{z}}}\,d\theta\,d\phi
\end{equation}
and characterizing the canonical full-measure  class of 
$(\Sigma_{\;\widehat{z}},\,g^{(2)}_{\;\widehat{z}})$. We denote by
\begin{equation}
\label{RiemLipArea}
A\left(\Sigma_{\;\widehat{z}}\right)\,:=\,\int_{\Sigma_{\;\widehat{z}}}\,d\mu_{g^{(2)}_{\,\widehat{z}}},
\end{equation}
the corresponding area of $(\Sigma_{\;\widehat{z}},\,g^{(2)}_{\;\widehat{z}})$.

Let us consider the physical celestial sphere $\left(\mathbb{C\,S}_{\;\widehat{z}},\,d_{\mathbb{S}^2(\,\widehat{z}\,)} \right)$ endowed with the distance function $d_{\mathbb{S}^2(\,\widehat{z}\,)}$ associated with the round metric (\ref{roundmet0erre}), and the corresponding Riemann-Lipschitz sky section $(\Sigma_{\;\widehat{z}},\,d_{\Sigma_{\;\widehat{z}}})$ endowed with the distance function (\ref{Lipdistance}). To preserve the Riemann-Lipschitz nature of the sky sections $(\Sigma_{\;\widehat{z}},\,d_{\Sigma_{\;\widehat{z}}})$ as the reference FLRW redshift $\widehat{z}$ varies, we need to assume that the (null) exponential map
\begin{equation}   
\label{bilipexp0}
\left.\exp_p\right|_{\;\widehat{z}}\,:\,\left(\mathbb{C\,S}_{\;\widehat{z}},\,d_{\mathbb{S}^2(\,\widehat{z}\,)} \right)\,\subset\,T_pM\,\longrightarrow\,(\Sigma_{\;\widehat{z}},\,d_{\Sigma_{\;\widehat{z}}})\,
\subset\,\mathscr{C}^-(p,g)\,,
\end{equation} 
is a bi-Lipschitz map. Namely, we make the following assumption.

\begin{lipschitz*}
\label{Lipassumpt1} 
For any given reference redshift $\widehat{z}$\;, the exponential map (\ref{bilipexp0})   
associated with the past-directed null geodesics flow between the celestial sphere $(\mathbb{C\,S}_{\;\widehat{z}},\,d_{\widetilde{h}_{\,\widehat{z}}})$ and the  sky section $(\Sigma_{\;\widehat{z}},\,d_{\Sigma_{\;\widehat{z}}})$,
is a bi-Lipschitz homeomorphism. Namely, there is a redshift-dependent constant $c_{\;\widehat{z}}\,\geq\,1$ such that $\left.\exp_p\right|_{\;\widehat{z}}$ together with its inverse $\left.\exp_p^{\,-1}\right|_{\;\widehat{z}}$  is a bijection satisfying the uniform Lipschitz condition 
\begin{align}
\label{lipexp2bis}
c_{\,\widehat{z}}^{\,-\,1}\;
d_{\mathbb{S}^2(\,\widehat{z}\,)}(r(\,\widehat{z}\,)\,\ell(n_1),\,r(\,\widehat{z}\,)\,\ell(n_2))
\,&\leq \,
d_{\Sigma_{\;\widehat{z}}}\left(\exp_p(r(\,\widehat{z}\,)\,\ell(n_1)), \exp_p(r(\,\widehat{z}\,)\,\ell(n_2))\right)\nonumber\\
&\leq \,c_{\,\widehat{z}}\;
d_{\mathbb{S}^2(\,\widehat{z}\,)}(r(\,\widehat{z}\,)\,\ell(n_1),\,r(\,\widehat{z}\,)\,\ell(n_2))\,,
\end{align} 
for all \,$r(\,\widehat{z}\,)\,\ell(n_1),\,r(\,\widehat{z}\,)\,\ell(n_1)\,\in\,\mathbb{C\,S}_{\;\widehat{z}}$\;. The smallest constant $c_{\;\widehat{z}}$ such that (\ref{lipexp2bis}) holds for all 
$r(\,\widehat{z}\,)\,\ell(n_1),\,r(\,\widehat{z}\,)\,\ell(n_1)\,\in\,\mathbb{C\,S}_{\;\widehat{z}}$
characterizes the Lipschitz constant $\mathrm{Lip}_{\,\widehat{z}\,}(\exp_p)$ of the map $\exp_p\,:\,
\mathbb{C\,S}_{\;\widehat{z}}\,\longrightarrow\,\Sigma_{\;\widehat{z}}$,\; \emph{i.e.},
\begin{equation}
\label{LipExpZ}
\mathrm{Lip}_{\,\widehat{z}\,}\left(\exp_p  \right)\,:=\,\sup\left\{\frac{d_{\Sigma_{\;\widehat{z}}}\left(\exp_p(r(\,\widehat{z}\,)\,\ell(n_1)), \exp_p(r(\,\widehat{z}\,)\,\ell(n_2))\right)}{d_{\mathbb{S}^2(\,\widehat{z}\,)}(r(\,\widehat{z}\,)\,\ell(n_1),\,r(\,\widehat{z}\,)\,\ell(n_2))}
 \right\}\,,
\end{equation}
where the $\sup$ is over all couples $r(\,\widehat{z}\,)\,\ell(n_1),\,r(\,\widehat{z}\,)\,\ell(n_1)\,\in\,\mathbb{C\,S}_{\;\widehat{z}}$\; with $n_1\,\not=\,n_2$. \;\;\;\; $\square$
\end{lipschitz*}
\vskip 0.5cm
A result geometrically proved by  M . Kunzinger, R. Steinbauer, M. Stojkovic \cite{Kuzinger}, (based on work by B.-L. Chen and P. LeFloch \cite{LeFloch}),  and by E. Minguzzi \cite{Minguzzi}, implies that in the bi-Lipschitz setting the exponential map retains an appropriate form of regularity in the sense that locally, for each point $p\in\,M$, there exist open star-shaped neighborhoods, $N_0(p)$ of $0\in\,T_pM$ and $U_p\,\subset\,(M, g)$, such that $\exp_p\,:\,N_0(p)\,\longrightarrow\,U_p$ is a bi-Lipschitz homeomorphism \cite{Kuzinger}. In particular, each point $p\in (M, g)$ possesses a basis of totally normal neighborhoods. It is worthwhile to stress that geodesic normal coordinates can still be defined. Still, the transition from the current smooth coordinate systems\footnote{Recall that $M$ is a smooth manifold and that the low Lipschitz $C^{1,\,1}$ regularity is caused by the metric $g$, and not by the differentiable structure of $M$.} used around $p\in M$ to the normal coordinates associated with $\exp_p$ is only continuous. It follows that we can assume that $\Sigma_{\;\widehat{z}}$ is topologically a 2-sphere, and we mimic the effect of the many lensing events that may affect $\Sigma_{\;\widehat{z}}$ by assuming that it has the irregularities of a metric surface with the fractal geometry of a 2-sphere with the locally-finite Hausdorff 2-measure associated with $d_{\Sigma_{\;\widehat{z}}}$ described in Section \ref{HausdorffMeasDist} below. 

As the formulas above show, keeping track of notation when working with the maps $\exp_p\,:\,
\mathbb{C\,S}_{\;\widehat{z}}\,\longrightarrow\,\Sigma_{\;\widehat{z}}$,\; gives rise to quite unwieldy expressions, thus, whenever possible, we introduce the following shorthand notation.

\begin{remark}(\emph{A notational remark})
\label{RemExpNot}
Recall that if $x$ is the generic point 
on the celestial sphere $\mathbb{C\,S}_{\;\widehat{z}}\,\in\,(T_pM,\,\{E_{(i)}\})$ pointed by the spatial vector $r(\,\widehat{z}\,)\,n(x)$ \,\emph{i.e.}
\begin{equation}
\mathbb{C\,S}_{\;\widehat{z}}\,\ni\,x\,:=\,
r(\,\widehat{z}\,)\,n^a(x) E_{(a)}\,,
\end{equation}
where $n(x)\,\in\,\mathbb{C\,S}$,\, then the corresponding null vector on  $C^{\,-}(T_pM,\,\{E_{(i)}\})$ is given by
\begin{equation}
\label{nullX}
r(\,\widehat{z}\,) \ell(x)\,=\,r(\,\widehat{z}\,) \left(n^a(x) E_{(a)}\,-\,E_{(4)}\right)\,. 
\end{equation} 
Since notation wants to travel light, the correspondence $x\,\leftrightarrow\,\ell(x)$ suggests denoting with the same symbol $"x"$ both the points on the celestial sphere $\mathbb{C\,S}_{\;\widehat{z}}(p)$ as well as the associated null vectors on the past null cone $C^{\,-}(T_pM,\,\{E_{(i)}\})$. We can indeed unambiguosly associate to the given point $x\,\in\,\mathbb{C\,S}_{\;\widehat{z}}$, the vector  $r(\,\widehat{z}\,)\,n(x)\,\in\,\mathbb{C\,S}_{\;\widehat{z}}(p)$, characterizing the direction of sight $n(x)\,\in\,\mathbb{S}^2\,\simeq\,\mathbb{C\,S}(p)$, and the comoving radius $r(\,\widehat{z}\,)$ associated with  the null vector (\ref{nullX}). Thus, if there is no danger of confusion in what follows, we can safely use the shorthand notation
\begin{equation}
\exp_p\,(x)\,:=\,\exp_p\left(r(\,\widehat{z}\,) \ell(x)\right)\,,
\end{equation}
to denote the exponential map associated with the null geodesic issued from $p$ along the past-directed null vector $r(\,\widehat{z}\,)\ell(x)\,\in\,C^{\,-}(T_pM,\,\{E_{(i)}\})$\;. \;\;\;\;\;$\square$
\end{remark}
\vskip 0.5cm\noindent
This simplified notation proves particularly useful in discussing the geometric relations between the observer's celestial sphere and the past light cone sections. We start by comparing the metric geometries of $\mathbb{C\,S}_{\;\widehat{z}}\,\in\,T_pM$ and of the associated sky section $\Sigma_{\;\widehat{z}}$. To this end, let us pullback on the physical celestial sphere $\mathbb{C\,S}_{\;\widehat{z}}\,\in\,T_pM$ the distance function $d_{\Sigma_{\;\widehat{z}}}$,  and define on $\mathbb{C\,S}_{\;\widehat{z}}\,\in\,T_pM$ the \emph{physical distance function} $d_{\mathbb{C\,S}_{\;\widehat{z}}}$ according to 
\begin{align}
\label{pullbackDist}
d_{\mathbb{C\,S}_{\;\widehat{z}}}(x,\, y)\,&:=\,
\left(\exp_p^*d_{\Sigma_{\;\widehat{z}}}\right)(x,\,y)\\
&=\,
d_{\Sigma_{\;\widehat{z}}}\left(\exp_p(x), \exp_p(y)\right)\,,\nonumber
\end{align}
for all $x,\,y\,\in\,\mathbb{C\,S}_{\hat{z}}$ such 
that $\exp_p(x),\,\exp_p(y)\,\in\,(\Sigma_{\;\widehat{z}},\,d_{\Sigma_{\;\widehat{z}}})$.
The bi-Lipschitz condition (\ref{lipexp2bis}) and (\ref{LipExpZ}) can be rewritten (in the adopted shorthand notation) as
\begin{equation}
\label{BiLippe}
\mathrm{Lip}_{\,\widehat{z}\,}^{\,-\,1}\,\left(\exp_p  \right)\,d_{\mathbb{S}^2(\,\widehat{z}\,)}(x,y)\,\leq\,
d_{\mathbb{C\,S}_{\;\widehat{z}}}(x,\, y)\,\leq\,\mathrm{Lip}_{\,\widehat{z}\,}\left(\exp_p  \right)\,d_{\mathbb{S}^2(\,\widehat{z}\,)}(x,y)\,,
\end{equation}
\vskip 0.5cm \noindent
showing that
 on the physical celestial sphere $\mathbb{C\,S}_{\;\widehat{z}}$ two distances coexist:\;
$d_{\mathbb{S}^2(\,\widehat{z}\,)}$\; and 
$d_{\mathbb{C\,S}_{\;\widehat{z}}}$\,. The former is the distance associated with the natural round metric  (\ref{roundmet0erre}) that encodes the \emph{measured angular separation} between the two astrophysical sources $x$ and $y$ as they appear on the celestial sphere $\mathbb{C\,S}_{\;\widehat{z}}$.\;  The physical distance $d_{\mathbb{C\,S}_{\;\widehat{z}}}$ describes the metric structure of 
$(\Sigma_{\;\widehat{z}},\,d_{\Sigma_{\;\widehat{z}}})$ as conveyed on $\mathbb{C\,S}_{\;\widehat{z}}$ by the null geodesic flow via the null exponential 
map pull-back action $\exp_p^*d_{\Sigma_{\;\widehat{z}}}$. The distance function $d_{\mathbb{C\,S}_{\;\widehat{z}}}$  provides the actual   separation of the sources $q_1\,=\,\exp_p(x)$ and $q_2\,=\,\exp_p(y)$ on the sky section $\Sigma_{\;\widehat{z}}$.
Thus, with the celestial sphere of the physical observer  $\mathbb{C\,S}_{\;\widehat{z}}$ we can naturally associate two distinct metric structures
\begin{equation}   
\left(\mathbb{C\,S}_{\;\widehat{z}},\,d_{\mathbb{S}^2(\,\widehat{z}\,)};\,d_{\mathbb{C\,S}_{\;\widehat{z}}}  \right)\,.
\end{equation} 
\begin{remark}
\label{remBilip}
Strong gravitational lensing phenomena provide evidence that two distinct images $x,\,y\,\in\,\mathbb{C\,S}_{\;\widehat{z}}$,\, with $d_{\mathbb{S}^2(\,\widehat{z}\,)}(x,y)\not=0$, may correspond to the apparent images of a single source located at the event $q\in\,\Sigma_{\;\widehat{z}}$ with $\exp_p(x)=\exp_p(y)\,:=\,q$, and  $d_{\Sigma_{\;\widehat{z}}}(q_1, q_2)\,=\,0$. Since 
\begin{equation}
d_{\Sigma_{\;\widehat{z}}}\left(\exp_p(x),\,\exp_p(y)\right)\,\leq\,\mathrm{Lip}_{\,\widehat{z}\,}\left(\exp_p  \right)\,d_{\mathbb{S}^2(\,\widehat{z}\,)}(x,y)
\end{equation}
strong lensing, taking the form of multiple images of a single physical source,\,\emph{i.e.},\, $d_{\Sigma_{\;\widehat{z}}}\left(\exp_p(x),\,\exp_p(y)\right)\,=\,0$ with $d_{\mathbb{S}^2(\,\widehat{z}\,)}(x,y)\,\not=\,0$, is compatible with a Lipschitz characterization of $\exp_{\widehat{z}\,}$. However, we cannot have reasonable control over the very complex topological structure of the sky section $\Sigma_z$ induced by a cascade of (strong) lensing events of this type. Moreover, the corresponding caustics and singularities at the terminal points on $\Sigma_{\hat{z}}$ provide a level of detail irrelevant to the present cosmological analysis. Such a complex behavior is not allowed if $\exp_{\widehat{z}\,}$ is restricted by a bi-Lipchitz condition, providing another compelling reason for adopting (\ref{BiLippe}). Explicitly, the bound  
\begin{equation}
\label{biLipSep}
\mathrm{Lip}^{\,-\,1}_{\,\widehat{z}\,}\left(\exp_p  \right)\,d_{\mathbb{S}^2(\,\widehat{z}\,)}(x,y)\,\leq\,d_{\Sigma_{\;\widehat{z}}}\left(\exp_p(x),\,\exp_p(y)\right)\,,
\end{equation}
implies the injectivity of $\exp_{\widehat{z}\,}$. In particular, any two distinct source images $x,\,y\,\in\,\mathbb{C\,S}_{\;\widehat{z}}$ with $d_{\mathbb{S}^2(\,\widehat{z}\,)}(x,y)\,\not=\,0$ necessarily correspond to two distinct physical sources $\exp_p(x)$ and $\exp_p(y)$ on the sky section $\Sigma_{\widehat{z}\,}$ separated by a distance  that, according to (\ref{biLipSep}), is controlled by the Lipschitz distortion of the null exponential map,
\begin{equation}
d_{\Sigma_{\;\widehat{z}}}\left(\exp_p(x),\,\exp_p(y)\right)\,\geq\,
\mathrm{Lip}^{\,-\,1}_{\,\widehat{z}\,}\left(\exp_p  \right)\,d_{\mathbb{S}^2(\,\widehat{z}\,)}(x,y)\,>\,0\,.
\end{equation}
This bi-Lipschitz scenario can be relaxed to a Lipschitz scenario under suitable circumstances, which will be described later. \;\;\;\;\;$\square$
\end{remark}

\section{Area distance on a biLipschitz past lightcone}
\label{ArDistLipCone}
The bi-Lipschitz framework described in the previous sections allows us to do analysis on the sky sections $(\Sigma_{\;\widehat{z}},\,d_{\Sigma_{\;\widehat{z}}})$, in a way not too dissimilar to the smooth case. Yet, the complicated nature of the metric geometry of the sky sections $\Sigma_{\;\widehat{z}}$ is better explored if we introduce the associated $s$-dimensional Hausdorff measure $\mathscr{H}^s\left(\Sigma_{\,\widehat{z}}  \right)$, \; $0\leq\,s\,\leq\,2$,\; and its connection with the standard Riemannian measure of $(\Sigma_{\;\widehat{z}},\,g^{(2)}_{\;\widehat{z}})$ as a Riemannian Lipschitz surface  (see (\ref{RiemLipSurf})). It is relatively easy to show that in the bi-Lipschitz setting, the two measures,\, the Hausdorff $\mathscr{H}^2\left(\Sigma_{\,\widehat{z}} \right)$ and the Riemannian area $A(\Sigma_{\,\widehat{z}})$, coincide as measures of full support. Yet, $\mathscr{H}^s\left(\Sigma_{\,\widehat{z}}\right)$ better captures the actual observable aspects of the metric geometry of sky sections, providing a direct connection with the operative way the angular diameter distance and the area distance of astrophysical sources are introduced in astrophysics. 

\subsection{The angular diameter distance}
The definition of the angular diameter distance does not directly involve the Hausdorff measure of the region over which the visible portion of the source is extended. Yet, the Hausdorff measure provides an operational characterization of the angular and area distances of the source. Their relevance in cosmology cannot be underestimated since both are observable quantities as long as we consider sources of known intrinsic size.

We start by defining, in a rigorous geometrical way, the angular diameter distance of an astrophysical extended source described by a region $B^{(phys)}_{\;\widehat{z}}(q)\,\subset\,\Sigma_{\;\widehat{z}}$ from which signals are gathered by the physical observer $(p,\,\dot\gamma(p))$. In the bi-Lipschitz case, the image\footnote{If 
$\exp_p$ is only Lipschitz, we may have multiple inverse images of $B^{(phys)}_{\;\widehat{z}}(q)$.} of this region on the physical observer's celestial sphere is provided by a corresponding region $B_{\;\widehat{z}}(q)\,\subset\,\mathbb{CS}_{\;\widehat{z}}(p)$ such that
\begin{equation}
B^{(phys)}_{\;\widehat{z}}(q)\,=\,\exp_p(B_{\;\widehat{z}}(q))\,.
\end{equation}
As the notation suggests, we assume that $B^{(phys)}_{\;\widehat{z}}(q)$ is centered around a point $q\in\,\Sigma_{\;\widehat{z}}$ that, for technical reasons, we  characterize as the \emph{geometrical barycenter} of the region considered. Explicitly, let $q_1, \ldots\,q_k$ denote points in $B^{(phys)}_{\;\widehat{z}}(q)$. We define $q$ as the minimizer of the function
\begin{equation}
\label{CenterMass}
B^{(phys)}_{\;\widehat{z}}(q)\,\ni\, q\,\longmapsto\,\sigma(q)\,:=\,\frac{1}{2}\,
\sum_{j=1}^k\,d^2_{\Sigma_{\,\widehat{z}}}\left(q,\,q_j\right)\,.
\end{equation} 
Since $\Sigma_{\;\widehat{z}}$ is a Riemann-Lipschitz surface, such a minimizer exists and is unique. We denote by $r(\,\widehat{z}\,)\,n(x)\,\in\,\mathbb{C\,S}_{\;\widehat{z}}$  its celestial sphere coordinates , \emph{i.e.},
\begin{equation}
\Sigma_{\;\widehat{z}}\,\ni\,q\,=\,exp_p\,\left(x(q)\right)\,=\,exp_p\,\left(r(\,\widehat{z}\,)\ell(n(x))\right)\,,
\end{equation}
where $r(\,\widehat{z}\,)\ell(n(x))=r(\,\widehat{z}\,)(n(x)-E_{(4)})\,\in\,C^{\,-}(T_pM,\,\{E_{(i)}\})$. Together with these technical and notational remarks,
it is also useful to introduce, at the source location $q\,\in\,\Sigma_{\;\widehat{z}}$,  the associated directional celestial sphere $\mathbb{CS}(q)$ in the local rest frame $(T_qM,\,\{E{(q)}_{(i)},\,E(q)_{(4)}:=\,\dot\gamma(q)\}$. The actual visual size of the source on $\Sigma_{\;\widehat{z}}$ spans the region 
\begin{align}
B^{(phys)}_{\;\widehat{z}}(q)\,&=\,\exp_p\left(B_{\;\widehat{z}}(q)\right)\,\subset\,\Sigma_{\;\widehat{z}}\\
&:=\,\left\{\left.q'\in\,\Sigma_{\;\widehat{z}}\;\right|\;q'\,=\,\exp_p\,\left(x'\right),\;\;\;x'\,\in\, B_{\;\widehat{z}}(q)\right\}\,,\nonumber
\end{align}
where, adopting the shorthand notation intoduced in Remark \ref{RemExpNot}, we have associated to the celestial sphere points $x'\,\in\,B_{\;\widehat{z}}(q)\,\subset\,\mathbb{C\,S}_{\;\widehat{z}}$ the corresponding past-directed null vector $r(\,\widehat{z}\,)\ell(x')\,\in\,C^{\,-}(T_pM,\,\{E_{(i)}\})$ and set
\begin{equation}
\exp_p\,\left(x'\right)\,:=\,\exp_p\left(r(\,\widehat{z}\,)\ell(x')\right)\,.  
\end{equation}
The diameter of the region $B^{(phys)}_{\;\widehat{z}}(q)$ is given by 
\begin{equation}
\mathrm{diam}_{\Sigma_{\;\widehat{z}}\,}\left(B^{(phys)}_{\;\widehat{z}}(q)\right)\,:=\,\sup\,\left\{\left.d_{\Sigma_{\;\widehat{z}}}\left(q_1,\,q_2\right)\;\right|\; q_1,\,q_2\,\in\,B^{(phys)}_{\;\widehat{z}}(q)  \right\}\,.
\end{equation} 
 
Since $B^{(phys)}_{\;\widehat{z}}(q)\,=\,\exp_p(B_{\;\widehat{z}}(q))$, we can exploit the pull-back of the function $\mathrm{diam}_{\Sigma_{\;\widehat{z}}\,}$  to associate the \emph{physical diameter} of $B^{(phys)}_{\;\widehat{z}}(q)$ with
the image of the source on 
the observer's celestial sphere $\mathbb{CS}_{\;\widehat{z}}(p)$.

\begin{remark}
We cannot directly measure   $\mathrm{diam}_{\Sigma_{\;\widehat{z}}\,}(B^{(phys)}_{\;\widehat{z}}(q))$. Yet, reliable estimates are possible as long as we have physical information on the nature of the sources. Typically, this information (still gathered via the signals reaching us along null-geodesics) is provided by the total rate of emission of radiant energy, the \emph{luminosity} $\mathds{L}$ of the source considered. Further properties of the signals (\emph{e.g.}, a periodicity) can allow us to classify the source in a type whose physics is largely understood. In such a scenario (familiar in astrophysics for the Cepheid variable stars and the $\mathrm{SN1a}$ supernovae), the intrinsic size of the source can be reliably inferred.  \;\;\; $\square$    
\end{remark}
\vskip 0.5cm \noindent 
The knowledge of the intrinsic size of the source $B^{(phys)}_{\;\widehat{z}}(q)$ allows us to estimate the diameter  $\mathrm{diam}_{\Sigma_{\;\widehat{z}}\,}(B^{{(phys)}}_{\;\widehat{z}}(q))$ and attribute, by exploiting the pull back action under $\exp_p$, its intrinsic size to the region $B_{\;\widehat{z}}(q)\,\subset\,\mathbb{CS}_{\;\widehat{z}}(p)$.  
 Taking into account the definition  of the physical distance function $d_{\mathbb{C\,S}_{\;\widehat{z}}}$ \; (see (\ref{pullbackDist})), we get 
\begin{align}
\label{DiametersPlay}
\mathrm{diam}_{\mathbb{C\,S}_{\;\widehat{z}}}\left(B_{\;\widehat{z}}(q)\right)\,&:=\,
\sup\,\left\{\left.d_{\mathbb{C\,S}_{\;\widehat{z}}}
\left(x_1,\,x_2\right)\;\right|\; x_1,\,x_2\,\in\,B_{\;\widehat{z}}(q)  \right\}\\
&=\,\sup\,\left\{\left.\left(\exp_p^*d_{\Sigma_{\;\widehat{z}}}\right)\left(x_1,\,x_2\right)\;\right|\; x_1,\,x_2\,\in\,B_{\;\widehat{z}}(q)  \right\}\nonumber\\
&=\,\sup\,\left\{\left.d_{\Sigma_{\;\widehat{z}}}\left(\exp_p(x_1),\,\exp_p(x_2)\right)\;\right|\; x_1,\,x_2\,\in\,B_{\;\widehat{z}}(q)  \right\}\nonumber\\
&=\,\sup\,\left\{\left. d_{\Sigma_{\;\widehat{z}}}\left(q_1,\,q_2\right)\;\right|\; \left.q_a:=\exp_p(x_a)\right|_{a=1,2},\,x_a\,\in\,B_{\;\widehat{z}}(q)  \right\}\nonumber\\
&=\,\sup\,\left\{\left. d_{\Sigma_{\;\widehat{z}}} \left(q_1,\,q_2\right)\;\right|\; q_1, q_2\,\in\,B^{(phys)}_{\;\widehat{z}}(q)  \right\}\nonumber\\
&=:\,\mathrm{diam}_{\Sigma_{\;\widehat{z}}\,}\left(B^{(phys)}_{\;\widehat{z}}(q)\right)\nonumber\,.
\end{align} 
Note that on the celestial sphere $(\mathbb{C\,S}(p),\,d_{\mathbb{S}^2})$, the diameter $\mathrm{diam}\left(B_{\;\widehat{z}}(q)\right)$ has a visual angular span given by
\begin{equation}
\label{AngSpan}
\mathrm{diam}_{\mathbb{S}^2}(B_{\;\widehat{z}\,}(q))\,:=\,\sphericalangle_{\;n(x_1)}^{\;n(x_2)}\,,
\end{equation}
where $\sphericalangle_{n(x_1)}^{n(x_2)}$ is the angle, subtended on $\mathbb{C\,S}(p)\,\simeq\,\mathbb{S}^2$ by the  vectors $n(x_1)$\; and $n(x_2)$ poynting to the apparent sky directions of $q_1$ and $q_2$\,.

The knowledge of  $\mathrm{diam}_{\Sigma_{\;\widehat{z}}\,}(B^{(phys)}_{\;\widehat{z}}(q))$
allows us to introduce the following
 
\begin{definition} (\emph{The angular diameter distance})
\label{AngDiamDist}
If  $B^{(phys)}_{\;\widehat{z}}(q)\in\,\Sigma_{\;\widehat{z}}$ is the visible region of an extended source of known intrinsic diameter $\mathrm{diam}_{\Sigma_{\;\widehat{z}}\,}(B^{(phys)}_{\;\widehat{z}}(q))$, then its angular diameter distance from the observer $p$ is defined by
\begin{equation}
\label{AngDiamDist1}
D^{(ang)}_{\;\widehat{z}}(q)\,:=\,\frac{\mathrm{diam}_{\Sigma_{\;\widehat{z}}\,}(B^{(phys)}_{\;\widehat{z}}(q))}{\mathrm{diam}_{\mathbb{S}^2}(B_{\;\widehat{z}}(q))}\,
=\,\frac{\mathrm{diam}_{\mathbb{C\,S}_{\;\widehat{z}}}\left(B_{\;\widehat{z}}(q)\right)}{\mathrm{diam}_{\mathbb{S}^2}(B_{\;\widehat{z}}(q))}
\end{equation}
where $\mathrm{diam}_{\mathbb{S}^2}(B_{\;\widehat{z}}(q))$ is the angular span (\ref{AngSpan}) of the source on the observer's directional celestial sphere $\mathbb{S}^2\,\simeq\,\mathbb{CS}(p)$\,.
\end{definition}
\vskip 0.5cm\noindent 
In the assumed bi-Lipschitz setting, the angular diameter distance can be bounded in terms of the isometric distortion of the null exponential map between the metric spaces $(\mathbb{C\,S}_{\;\widehat{z}},\,d_{S^2(\,\widehat{z}\,)})$ and $(\Sigma_{\;\widehat{z}},\,d_{\Sigma_{\;\widehat{z}}})$,
\begin{equation}
\exp_p\,:\,
\left(\mathbb{C\,S}_{\;\widehat{z}},\,d_{S^2(\,\widehat{z}\,)}   \right)\,\longrightarrow\,\left(\Sigma_{\;\widehat{z}},\, d_{\Sigma_{\;\widehat{z}}}  \right).
\end{equation}
We have

\begin{lemma}
\label{EstAngDist}
If $\mathrm{Lip}_{\,\widehat{z}}\left(\exp_p\right)\,\geq\,1$ denotes the Lipschitz constant of the exponential map  $\exp_p\,:\,\left(\mathbb{C\,S}_{\;\widehat{z}},\,d_{S^2(\,\widehat{z}\,)}   \right)\,\longrightarrow\,\left(\Sigma_{\;\widehat{z}},\, d_{\Sigma_{\;\widehat{z}}}  \right)$ (see (\ref{LipExpZ}) and 
(\ref{BiLippe})) then
\begin{equation}
\label{IsoDistLip}
\mathrm{Lip}_{\,\widehat{z}}^{-1}\left(\exp_p\right)\,\sqrt{\frac{1\,+\,v_{\;\widehat{z}}(p)}{1\,-\,v_{\;\widehat{z}}(p)}}\,\,\widehat{r}(\,\widehat{z}\,)\,\leq \,D^{(ang)}_{\;\widehat{z}}(q)\,
\leq\, \mathrm{Lip}_{\,\widehat{z}}\left(\exp_p\right)\,\sqrt{\frac{1\,+\,v_{\;\widehat{z}}(p)}{1\,-\,v_{\;\widehat{z}}(p)}}\,\,\widehat{r}(\,\widehat{z}\,)\,.
\end{equation}
where $\widehat{r}(\;\widehat{z}\,)$ is the FLRW comoving radius in $(T_pM,\,\{\widehat{E}_{(k)}\})$ associated to the source $B_{\;\widehat{z}}(q)$, and $v_{\;\widehat{z}}(p)$ is the relative velocity, at the given redshift ${\widehat{z}}$, of the physical observer with respect to the reference FLRW observer. 
\end{lemma}
\begin{proof}
For notational ease, let us use the shorthand notation  
\begin{equation}
\mathrm{Lip}_{\,\widehat{z}\,}\,:=\,\mathrm{Lip}_{\,\widehat{z}\,}\left(\exp_p  \right)\,.
\end{equation}
According to (\ref{BiLippe}) we have 
\begin{equation}
\label{BiLippeDiam}
\mathrm{Lip}_{\,\widehat{z}\,}^{\,-\,1}\;\mathrm{diam}_{\mathbb{S}^2(\,\widehat{z}\,)}\left(B_{\;\widehat{z}}(q)\right)\,\leq\,
\mathrm{diam}_{\mathbb{C\,S}_{\;\widehat{z}}}\left(B_{\;\widehat{z}}(q)\right)\,\leq\,\mathrm{Lip}_{\,\widehat{z}\,}\;\mathrm{diam}_{\mathbb{S}^2(\,\widehat{z}\,)}\left(B_{\;\widehat{z}}(q)\right)\,,
\end{equation}
which, from the definition (\ref{AngDiamDist}) of angular diameter distance, provides 
\begin{equation}
\label{BiLippe2}
\frac{1}{\mathrm{Lip}_{\,\widehat{z}\,}}\,
\frac{\mathrm{diam}_{\mathbb{S}^2(\,\widehat{z}\,)}\left(B_{\;\widehat{z}}(q)\right)}{\mathrm{diam}_{\mathbb{S}^2}\left(B_{\;\widehat{z}}(q)\right)}\,\leq\, 
D^{(ang)}_{\;\widehat{z}}(q)\,
\leq \,\mathrm{Lip}_{\,\widehat{z}\,}\,\frac{\mathrm{diam}_{\mathbb{S}^2(\,\widehat{z}\,)}\left(B_{\;\widehat{z}}(q)\right)}{\mathrm{diam}_{\mathbb{S}^2}\left(B_{\;\widehat{z}}(q)\right)}\,.
\end{equation}
The stated result follows by exploiting the relation 
\begin{equation}
\frac{\mathrm{diam}_{\mathbb{S}^2(\,\widehat{z}\,)}\left(B_{\;\widehat{z}}(q)\right)}{\mathrm{diam}_{\mathbb{S}^2}\left(B_{\;\widehat{z}}(q)\right)}\,=\,r(\,\widehat{z}\,)\,
=\,\sqrt{\frac{1\,+\,v_{\;\widehat{z}}(p)}{1\,-\,v_{\;\widehat{z}}(p)}}\,\,\widehat{r}(\,\widehat{z}\,)\,,
\end{equation}
where $r(\,\widehat{z})$ is the physical comoving radius and where we have exploited (\ref{radialconnect0}) connecting $r(\,\widehat{z})$ to the reference FLRW comoving radius 
$\widehat{r}(\,\widehat{z}\,)$.  \;\;\;\;\;$\square$
\end{proof}
Lemma \ref{EstAngDist}  clearly shows  how the \emph{isometric distortion} of the exponential map $\exp_p\,:\,
(\mathbb{C\,S}_{\;\widehat{z}},\,d_{S^2(\,\widehat{z}\,)}  )\,\longrightarrow\,(\Sigma_{\;\widehat{z}},\, d_{\Sigma_{\;\widehat{z}}})$, defined by the corresponding Lipschitz constant $\mathrm{Lip}_{\,\widehat{z}}\left(\exp_p\right)\,\geq\,1$,  affects the angular diameter distance $D^{(ang)}_{\;\widehat{z}}(q)$. The uniform bound  (\ref{IsoDistLip}) also indicates that we can optimize this distortion by tuning, at the given reference FLRW redshift $\widehat{z}\;$, the relative velocity $v_{\;\widehat{z}}(p)$ of the physical observer with respect the FLRW observer.

From a cosmological point of view, angular diameter distance is often characterized as a distance defined in terms of the astrophysical source's physical and angular size by considering a bundle of rays diverging from the observer to the emitting source, subtending a solid angle $d\omega_O$ at the observer and spanning a cross-sectional area $d\sigma_0$ at the emitter. Under the assumption that the source is an astrophysical object of known intrinsic size, we can define the notion of area distance $D_O$ (and the associated angular diameter distance) between the observer and source by mimicking an \emph{inverse square law} by  setting
\begin{equation}
\label{facilDist}
d\sigma_0\,=\,D^2_O\,d\omega_O\,,
\end{equation} 
The quantity $D_O$ is, in principle, measurable as long as physics allows us to estimate the  "intrinsic size" $d\omega_O$, and we can measure the apparent solid angle $d\sigma_0$. This characterization, although simple, is quite effective, as already stressed, since it can be related to energy flux measurements. Our definition above is a rewriting of  (\ref{facilDist})  where we took care of defining the relevant quantities involved in (\ref{facilDist}) on the geometrical spaces where they make mathematical sense\footnote{The usual characterizations of the angular diameter distance and area distance are directly framed on the lightcone, confusing the domain of definition of the null exponential map (the past null cone in $T_pM$) with its image (the past lightcone). From a physical point of view, this is a good approximation near $p$. Yet, it is difficult to handle as we move away from $p$.}, namely the observer celestial sphere $\mathbb{C\,S}_{\;\widehat{z}}$ and the corresponding lightcone section $\Sigma_{\hat{z}}\;$. 

\subsection{The Hausdorff measure and the area distance}
\label{HausdorffMeasDist}
In the previous section, we introduced the angular diameter distance by taking care 
of the geometry involved. The definition of the associated area distance sketched above is more delicate. As already stressed, there is a subtle interplay between the usual Lebesgue area measure of a region of the sky section $\Sigma_{\,\widehat{z}}\,$, and the corresponding Hausdorff measure. The former is easier to manipulate from a formal calculus point of view; the latter better captures the actual, observable aspects of the metric geometry of sky sections. Since in the bi-Lipschitz setting, both induce the same full two-dimensional measure on $\Sigma_{\,\widehat{z}}\,$, we may be tempted to dismiss the rather complex Hausdorff scenario however the Hausdorff measure plays a role in understanding the non-smooth geometry of $(\Sigma_{\;\widehat{z}},\,d_{\Sigma_{\;\widehat{z}}})$, and better connects with actual observations.    

Let us start by defining $s$-dimensional Hausdorff measure $\mathscr{H}^s\left(\Sigma_{\,\widehat{z}}  \right)$ of the sky section $(\Sigma_{\;\widehat{z}},\,d_{\Sigma_{\;\widehat{z}}})$. For $s\,\in\,\mathbb{R}_{\geq\,0}$,\,
the Hausdorff measure $\mathscr{H}^s\left(\Sigma_{\,\widehat{z}}  \right)$  is computed in terms of coverings of small diameter that probe the detail of the metric geometry of $(\Sigma_{\;\widehat{z}},\,d_{\Sigma_{\;\widehat{z}}})$. Explicitly, let $\{U_\alpha({\,\widehat{z}}\,)\}_{\alpha=1}^\infty$ be a collection of sets $U_\alpha({\,\widehat{z}}\,)\, \subset\,\Sigma_{\;\widehat{z}}$ such that $\Sigma_{\;\widehat{z}}\,\subset\,\cup_{\alpha=1}^\infty\,U_\alpha({\,\widehat{z}}\,)$ and $\mathrm{diam}_{\Sigma_{\;\widehat{z}}}\left(U_\alpha({\,\widehat{z}}\,)\right)\,\leq\,\delta$,\,$\forall\,\alpha$,\; with\, $\delta\,>\,0$, and where 
\begin{equation}
\label{diamU}
\mathrm{diam}_{\Sigma_{\;\widehat{z}}}\left(U_\alpha({\,\widehat{z}}\,)\right)\,:=\,\sup\,\left\{\left.d_{\Sigma_{\;\widehat{z}}}\left(q_1,\,q_2\right)\;\right|\; q_1,\,q_2\,\in\,U_\alpha({\,\widehat{z}}\,)  \right\}
\end{equation} 
denotes the diameter of $U_\alpha({\,\widehat{z}}\,)$ in the metric space $(\Sigma_{\;\widehat{z}},\,d_{\Sigma_{\;\widehat{z}}})$. With these notational remarks along the way, we have

\begin{definition}
\label{DefHausdMeas}
For $s\,\in\,\mathbb{R}_{\geq\,0}$,\,the $s$-dimensional Hausdorff measure of the sky section $(\Sigma_{\;\widehat{z}},\,d_{\Sigma_{\;\widehat{z}}})$ is provided  by 
\begin{equation}
\label{hausdEsse}
\mathscr{H}^s\left(\Sigma_{\,\widehat{z}}  \right)\,:=\,\lim_{\delta\rightarrow 0}\,
\mathscr{H}_\delta^s\left(\Sigma_{\,\widehat{z}}  \right)\,=\,\sup_{\delta>0}\,\mathscr{H}_\delta^s\left(\Sigma_{\,\widehat{z}}  \right),\,
\end{equation}
where
\begin{equation}
\label{Hdelta}
\mathscr{H}_\delta^s\left(\Sigma_{\,\widehat{z}}  \right)\,:=\,\inf\,
\left\{\sum_{\alpha}^\infty\,\frac{\pi^{s/2}}{\Gamma\left(\tfrac{s}{2}\,+\,1 \right)}\,\left(\frac{\mathrm{diam}_{\Sigma_{\;\widehat{z}}}\left(U_\alpha ({\,\widehat{z}}\,)\right)}{2}\right)^s \right\}\,,
\end{equation}
\vskip 0.3cm\noindent
$\Gamma(t)\,:=\,\int_0^\infty\,e^{\,-x}\,x^{\,t-1}\,dx$ denotes  the gamma function, 
and where the infimum is taken over all coverings $\{U_\alpha({\,\widehat{z}}\,)\}_{\alpha=1}^\infty$ such that $\mathrm{diam}_{\;\widehat{z}}\,U_\alpha({\,\widehat{z}}\,)\,\leq\,\delta$.
\end{definition}
\vskip 0.5cm\noindent
For a detailed discussion of the standard measure-theoretic properties of 
the Hausdorff measure see \cite{Gariepy}. Notice that if $U_\alpha({\,\widehat{z}}\,)$ is isometric to a Euclidean metric disk of radius $r$, then, for $s\,=\,2$, the term  $\tfrac{\pi^{s/2}}{\Gamma(\tfrac{s}{2}\,+\,1)}$ appearing in (\ref{hausdEsse}) normalizes the Hausdorff measure of $U_\alpha({\,\widehat{z}}\,)$  to the Lebesgue measure $\pi\,{r}^2$ of the disk\footnote{For $s=2\;\Rightarrow\,\tfrac{\pi^{s/2}}{\Gamma(\tfrac{s}{2}\,+\,1)}\,=\,\pi$. } of radius $r\;$.

As long as $\exp_p$ is a bi-Lipschitz map, we have the following result.

\begin{proposition}
\label{HausMeasSigma}
 Let\; $\mathrm{Lip}_{\;\widehat{z}\,}\,\geq\,1$ denote the Lipschitz constant $\mathrm{Lip}_{\;\widehat{z}\,}(\exp_p)$ describing the isometric distortion of the bi-Lipschitz exponential map $\exp_p\,:\,\left(\mathbb{C\,S}_{\;\widehat{z}},\,d_{S^2(\,\widehat{z}\,)}   \right)\,\longrightarrow\,\left(\Sigma_{\;\widehat{z}},\, d_{\Sigma_{\;\widehat{z}}}  \right)$, then the two-dimensional Hausdorff measure $\mathscr{H}^2\,\left(\Sigma_{\,\widehat{z}\,}\right)$ is the Riemann-Lipschitz area (\ref{RiemLipArea}) of $\Sigma_{\,\widehat{z}\,}$,
\begin{equation}
\mathscr{H}^2\,\left(\Sigma_{\,\widehat{z}\,}\right)\,=\,A\left(\Sigma_{\;\widehat{z}}\right)\,:=\,\int_{\Sigma_{\;\widehat{z}}}\,d\mu_{g^{(2)}_{\,\widehat{z}}}\,.
\end{equation} 
 Moreover, we have the uniform bound
\begin{equation}
\label{haussEsse2prima}
\frac{4\pi}{\mathrm{Lip}^2_{\,\widehat{z}\,}}\,\left(\frac{1\,+\,v_{\,\widehat{z}}\,(p)}{1\,-\,v_{\,\widehat{z}}\,(p)}\right)\,
\widehat{r}^{\,2}(\,\widehat{z}\,)\,\leq\,A\left(\Sigma_{\;\widehat{z}}\right)\,\leq\,
4\pi\,\mathrm{Lip}^2_{\,\widehat{z}\,}\,\left(\frac{1\,+\,v_{\,\widehat{z}}\,(p)}{1\,-\,v_{\,\widehat{z}}\,(p)}\right)\,
\widehat{r}^{\,2}(\,\widehat{z}\,)\,,
\end{equation}
where $\widehat{r}(\,\widehat{z}\,)$ is the reference FLRW comoving radius, and $v_{\,\widehat{z}}$ is the relative velocity of the physical observer with respect to the FLRW observer.
\end{proposition}
\begin{proof}
 These results easily follow from the standard behavior of the Hausdorff measure under the action of a Lipschitz map \cite{Gariepy}. Let us start with the case in which 
$\exp_p\,:\,\left(\mathbb{C\,S}_{\;\widehat{z}},\,d_{S^2(\,\widehat{z}\,)}   \right)\,\longrightarrow\,\left(\Sigma_{\;\widehat{z}},\, d_{\Sigma_{\;\widehat{z}}}  \right)$ is only  Lipschitz. 
 Let $\{V_\alpha\}_{\alpha=1}^\infty$ be a collection of sets $V_\alpha\, \subset\,\mathbb{CS}_{\;\widehat{z}}(p)$ such that $\mathbb{CS}_{\;\widehat{z}}(p)\,\subset\,\cup_{\alpha=1}^\infty\,V_\alpha$ and $\mathrm{diam}_{\mathbb{S}^2(\,\widehat{z})}\left(V_\alpha\right)\,\leq\,\delta$,\,$\forall\,\alpha$,\; with\, $\delta\,>\,0$. We have $\exp_p(\mathbb{CS}_{\;\widehat{z}})\,\subset\,\cup_\alpha\,\exp_p(V_\alpha)$, with 
\begin{equation}
\mathrm{diam}_{\Sigma_{\;\widehat{z}}}\,\left(\exp_p\left(V_\alpha \right)\right)\,\leq\,\mathrm{Lip}_{\,\widehat{z}\,}\,\mathrm{diam}_{\mathbb{S}^2(\,\widehat{z})}\left(V_\alpha \right)\, \leq\, \mathrm{Lip}_{\,\widehat{z}\,}\;\delta.
\end{equation}
If we compute the $s$-dimensional Haurdorff measure at resolution $\mathrm{Lip}_{\,\widehat{z}\,}\,\delta$, we get 
\begin{align}
\mathscr{H}^s_{\mathrm{Lip}_{\,\widehat{z}\,}\delta}\,\left(\exp_p\left(\mathbb{CS}_{\,\widehat{z}\,}\right) \right)\,&\leq\,
\sum_{\alpha}^\infty\,\frac{\pi^{s/2}}{\Gamma\left(\tfrac{s}{2}\,+\,1 \right)}\,\left(\frac{\mathrm{diam}_{\Sigma_{\;\widehat{z}}}\left(\exp_p\left(V_\alpha\right)\right)}{2}\right)^s \\
&\leq\,\left(\mathrm{Lip}_{\,\widehat{z}\,}\right)^s\,
\sum_{\alpha}^\infty\,\frac{\pi^{s/2}}{\Gamma\left(\tfrac{s}{2}\,+\,1 \right)}\,\left(\frac{\mathrm{diam}_{\mathbb{S}^2(\,\widehat{z})}\left(V_\alpha\right)}{2}\right)^s\nonumber\,.
\end{align}
Up to the $(\mathrm{Lip}_{\,\widehat{z}\,})^s$ prefactor,  the $\inf$ of the above sum over the coverings $\{V_\alpha\}_{\alpha=1}^\infty$ defines the $s$-dimensional Hausdorff measure (at resolution $\delta$) \; $\mathscr{H}^s_{\delta}\,(\mathbb{CS}_{\,\widehat{z}\,})$. Thus, we have  
\begin{equation}
\mathscr{H}^s_{\mathrm{Lip}_{\,\widehat{z}\,}\delta}\,\left(\exp_p\left(\mathbb{CS}_{\,\widehat{z}\,}\right) \right)\,\leq\,\left(\mathrm{Lip}_{\,\widehat{z}\,}\right)^s\,   \mathscr{H}^s_{\delta}\,(\mathbb{CS}_{\,\widehat{z}\,})\,,
\end{equation}
and  
\begin{align}
\label{HausIneq}
\lim_{\delta\rightarrow 0}\,\mathscr{H}^s_{\mathrm{Lip}_{\,\widehat{z}\,}\delta}\,\left(\exp_p\left(\mathbb{CS}_{\,\widehat{z}\,}\right) \right)\,&\leq\,\lim_{\delta\rightarrow 0}\,\left(\mathrm{Lip}_{\,\widehat{z}\,}\right)^s\,   \mathscr{H}^s_{\delta}\,(\mathbb{CS}_{\,\widehat{z}\,})\nonumber\\
\Rightarrow\,\mathscr{H}^s\,\left(\exp_p\left(\mathbb{CS}_{\,\widehat{z}\,}\right) \right)\,&\leq\,\left(\mathrm{Lip}_{\,\widehat{z}\,}\right)^s\,   \mathscr{H}^s\,(\mathbb{CS}_{\,\widehat{z}\,})
\end{align}
Since $\mathbb{CS}_{\,\widehat{z}\,}$ is the smooth round two-sphere, $\mathscr{H}^s\,(\mathbb{CS}_{\,\widehat{z}\,})$ can be identified with the standard area measure of $\mathbb{CS}_{\,\widehat{z}\,}$ for $s\,=\,2$. The bound (\ref{HausIneq}) implies also that
\begin{equation}
\mathscr{H}^s\,\left(\exp_p\left(\mathbb{CS}_{\,\widehat{z}\,}\right) \right)\,=\, 
\mathscr{H}^s\,\left(\Sigma_{\,\widehat{z}\,}\right)\,=\,0,\;\;\;\;\mathrm{for}\;\;s\,>\,2\,,
\end{equation}  
and that the Hausdorff dimension of $\Sigma_{\,\widehat{z}\,}$, defined by
\begin{equation}
\mathrm{dim}_{\mathscr{H}}\left(\Sigma_{\,\widehat{z}\,}\right)\,:=\,\inf\,
\left\{\left. s\,\in\,\mathbb{R}_{\geq 0}\;\right|\;\mathscr{H}^s\,\left(\Sigma_{\,\widehat{z}\,}\right)\,=\,0 \right\}\,,
\end{equation}
is $\mathrm{dim}_{\mathscr{H}}\left(\Sigma_{\,\widehat{z}\,}\right)\,=\,2$\,.

If the null exponential map 
$\exp_p\,:\,\left(\mathbb{C\,S}_{\;\widehat{z}},\,d_{S^2(\,\widehat{z}\,)}   \right)\,\longrightarrow\,\left(\Sigma_{\;\widehat{z}},\, d_{\Sigma_{\;\widehat{z}}}  \right)$ is by-Lipschitz, then it is injective (see Remark \ref{remBilip}) and it admits an inverse 
\begin{align}
\exp_p^{\,-1}\,:\,\left(\Sigma_{\;\widehat{z}},\, d_{\Sigma_{\;\widehat{z}}}\right)\,
&\longrightarrow\,\left(\mathbb{C\,S}_{\;\widehat{z}},\,d_{S^2(\,\widehat{z}\,)}\right)\\
q\,&\longmapsto x\,:=\,\exp_p^{\,- 1}(q)\,=\,x^i\,E_{(i)}\,,\nonumber
\end{align}
with $\sum_{a=1}^3(x^a)^2=\,r^2(\,\widehat{z}\,)\,=\,(x^4)^2$. The bi-Lipschitz condition 
\begin{equation}
\label{BiLippeAncora}
\frac{1}{\mathrm{Lip}_{\,\widehat{z}\,}}\,d_{\mathbb{S}^2(\,\widehat{z}\,)}(x,y)\,\leq\,
d_{\Sigma_{\hat{z}}\,}(\exp_p(x),\, \exp_p(y))\,\leq\,\mathrm{Lip}_{\,\widehat{z}\,}\,d_{\mathbb{S}^2(\,\widehat{z}\,)}(x,y)\,,
\end{equation}
implies
\begin{equation}
\frac{1}{\mathrm{Lip}_{\,\widehat{z}\,}}\,
d_{\Sigma_{\hat{z}}\,}(q_1,\, q_2)\,\leq\,d_{\mathbb{S}^2(\,\widehat{z}\,)}\left(\exp_p^{\,- 1}(q_1),\,\exp_p^{\,- 1}(q_2)\right)\,
\leq\,\mathrm{Lip}_{\,\widehat{z}\,}\,\,d_{\Sigma_{\hat{z}}\,}(q_1,\, q_2)\,,
\end{equation}
where, as usual, we set $\mathrm{Lip}_{\,\widehat{z}\,}\,:=\,\mathrm{Lip}_{\,\widehat{z}\,}\left(\exp_p  \right)$. 
It follows that also $\exp_p^{\,- 1}$ is bi-Lipschitz with the same metric distortion of $exp_p$,\,\emph{i.e.}
\begin{equation}
\mathrm{Lip}_{\,\widehat{z}\,}\left(\exp^{\,- 1}_p  \right)\,=\, 
\mathrm{Lip}_{\,\widehat{z}\,}\,\left(\exp_p  \right)\,,
\end{equation} 
and we can extend to $\exp_p^{\,- 1}$ the argument leading to (\ref{HausIneq}). With the obvious shift in notation, we get the lower bound
\begin{equation}
\mathscr{H}^s\,\left(\Sigma_{\,\widehat{z}\,}\right)\, \geq\,
\frac{1}{\mathrm{Lip}^s_{\,\widehat{z}\,}}\, \mathscr{H}^s\,(\mathbb{CS}_{\,\widehat{z}\,})\,, 
\end{equation}
which, together with (\ref{HausIneq}) and $\mathrm{dim}_{\mathscr{H}}\left(\Sigma_{\,\widehat{z}\,}\right)\,=\,2$\,,  provides
\begin{equation}
\frac{1}{\mathrm{Lip}^2_{\,\widehat{z}\,}}\, \mathscr{H}^2\,(\mathbb{CS}_{\,\widehat{z}\,})\,\leq\,\mathscr{H}^2\,\left(\Sigma_{\,\widehat{z}\,}\right)\,\leq\,
\mathrm{Lip}^2_{\,\widehat{z}\,}\, \mathscr{H}^2\,(\mathbb{CS}_{\,\widehat{z}\,})\,.
\end{equation}
Since $\mathscr{H}^2\,(\mathbb{CS}_{\,\widehat{z}\,})$ is simply the standard area measure of the round 2-sphere $(\mathbb{CS}_{\,\widehat{z}\,},\,\widetilde{h}_{\mathbb{S}^2(\,\widehat{z}\,)})$  (see (\ref{roundmet0erre})), and $\mathscr{H}^2\,\left(\Sigma_{\,\widehat{z}\,}\right)$ is the Riemann-Lipschitz area (\ref{RiemLipArea}) of $\Sigma_{\,\widehat{z}\,}$,
\begin{equation}
\mathscr{H}^2\,\left(\Sigma_{\,\widehat{z}\,}\right)\,=\,A\left(\Sigma_{\;\widehat{z}}\right)\,:=\,\int_{\Sigma_{\;\widehat{z}}}\,d\mu_{g^{(2)}_{\,\widehat{z}}}\,,
\end{equation} 
we get 
\begin{equation}
\label{HausAreaBound2}
\frac{4\pi}{\mathrm{Lip}^2_{\,\widehat{z}\,}}\,\left(\frac{1\,+\,v_{\,\widehat{z}}\,(p)}{1\,-\,v_{\,\widehat{z}}\,(p)}\right)\,
\widehat{r}^{\,2}(\,\widehat{z}\,)\,\leq\,A\left(\Sigma_{\;\widehat{z}}\right)\,\leq\,
4\pi\,\mathrm{Lip}^2_{\,\widehat{z}\,}\,\left(\frac{1\,+\,v_{\,\widehat{z}}\,(p)}{1\,-\,v_{\,\widehat{z}}\,(p)}\right)\,
\widehat{r}^{\,2}(\,\widehat{z}\,)\,,
\end{equation}
which implies the stated result. \;\;\;\; $\square$      
\end{proof}  

\begin{remark}
As already stressed, \,the Hausdorff measure captures quite efficiently the geometry of the measurement process in astrophysical observations. Signals reaching us from an extended source $B_{\;\widehat{z}}(q)$ are often organized in observational bins representing the spatial regions from which we receive the relevant energetic information. Depending on the nature of the observation, these spatial regions can be associated with a finite covering $\{V_\alpha \}$ of $B_{\;\widehat{z}}(q)$ with a spatial resolution $\delta\,>\,0$, defined by $\mathrm{diam}(V_\alpha)\,<\,\delta$\,. The corresponding  Hausdorff measure $\mathscr{H}_\delta^s(B_{\;\widehat{z}}(q))$ provides a good approximation to the actual Hausdorff measure and, when possible, to the area measure $A(B_{\;\widehat{z}}(q))$.     
\;\;\;$\square$
\end{remark}
\vskip 0.5cm\noindent
We can extend the above analysis to the visible region $B^{(phys)}_{\;\widehat{z}}(q)\,:=\,\exp_p\left(B_{\;\widehat{z}}(q)\right)\,\subset\,\Sigma_{\;\widehat{z}}$ characterizing the extended astrophysical source $q\in\,\Sigma_{\;\widehat{z}}$ introduced in Definition \ref{AngDiamDist}. The following result is a direct consequence of Proposition \ref{HausMeasSigma}.

\begin{lemma} (\emph{The Hausdorff measure of an extended source})
The $2$-dimensional Hausdorff measure of the region $B^{(phys)}_{\;\widehat{z}}(q)\,:=\,\exp_p\left(B_{\;\widehat{z}}(q)\right)\,\subset\,\Sigma_{\;\widehat{z}}$ characterizing the portion of the extended astrophysical source $q\in\,\Sigma_{\;\widehat{z}}$ visible by $(p, \dot\gamma(p))$, is given by its Riemann-Lipschitz area 
\begin{equation}
\mathscr{H}^2\,\left(B^{(phys)}_{\;\widehat{z}}(q)\right)\,=\,A\left(B^{(phys)}_{\;\widehat{z}}(q)\right)\,:=\,\int_{B^{(phys)}_{\;\widehat{z}}(q)}\,d\mu_{g^{(2)}_{\,\widehat{z}}}\,,
\end{equation} 
and we have the uniform bound 
\begin{align}
\frac{\omega(B_{\;\widehat{z}}(q))}{\mathrm{Lip}^2_{\,\widehat{z}\,}}\,\left(\frac{1\,+\,v_{\,\widehat{z}}\,(p)}{1\,-\,v_{\,\widehat{z}}\,(p)}\right)\,
\widehat{r}^{\,2}(\,\widehat{z}\,)\,&\leq\,A\left(B^{(phys)}_{\;\widehat{z}}(q)\right)\nonumber\\
&\leq\,
\omega(B_{\;\widehat{z}}(q))\,\mathrm{Lip}^2_{\,\widehat{z}\,}\,\left(\frac{1\,+\,v_{\,\widehat{z}}\,(p)}{1\,-\,v_{\,\widehat{z}}\,(p)}\right)\,
\widehat{r}^{\,2}(\,\widehat{z}\,)\,,
\label{haussEsse2}
\end{align}
where $\omega(B_{\;\widehat{z}}(q))$ is the solid angle subtended by the region $B_{\;\widehat{z}}(q)$ on the celestial sphere $\mathbb{CS}_{\,\widehat{z}\,}(p)$,  and we adopted the notation of Proposition \ref{HausMeasSigma}.\;\;\;$\square$
\end{lemma}
\vskip 0.5cm\noindent
Since $B^{(phys)}_{\;\widehat{z}}(q)\,=\,\exp_p(B_{\;\widehat{z}}(q))$, and $\exp_p$ is a bi-Lipschitz map, we have 
\begin{align}
\label{chainPullBack}
A\left(B^{(phys)}_{\;\widehat{z}}(q)\right)\,&:=\,\int_{B^{(phys)}_{\;\widehat{z}}(q)}\,d\mu_{g^{(2)}_{\,\widehat{z}}}\,=\,
\int_{\exp_p(B_{\;\widehat{z}}(q))}\,d\mu_{g^{(2)}_{\,\widehat{z}}}\\
&=\,\int_{B_{\;\widehat{z}}(q)}\,\exp_p^*\left(d\mu_{g^{(2)}_{\,\widehat{z}}}\right)\nonumber\\
&=\,
\int_{B_{\;\widehat{z}}(q)}\,\frac{\exp_p^*\left(d\mu_{g^{(2)}_{\,\widehat{z}}}\right)}{d\mu_{\mathbb{S}^2}}\,d\mu_{\mathbb{S}^2}\,,
\nonumber
\end{align} 
where $d\mu_{\mathbb{S}^2}\,:=\,\sin\theta\,d\theta\,d\phi$ is the solid angle measure on $\mathbb{CS}_{\;\widehat{z}}(p)$, and we introduced the pull back  measure on the celestial sphere $\mathbb{CS}_{\,\widehat{z}\,}(p)$ 
\begin{equation}
d\mu_{h_{\,\widehat{z}}}\,:=\,\exp_p^*\left(d\mu_{g^{(2)}_{\,\widehat{z}}}\right)\,. 
\end{equation}
According to (\ref{chainPullBack}), $d\mu_{h_{\,\widehat{z}}}$ is such that 
\begin{equation}
\int_{B_{\;\widehat{z}}(q)}\, d\mu_{h_{\,\widehat{z}}}\,=\,A\left(B^{(phys)}_{\;\widehat{z}}(q)\right),\,
\end{equation}
namely, $d\mu_{h_{\,\widehat{z}}}$ transfers to the observer's celestial sphere $\mathbb{CS}_{\;\widehat{z}}(p)$  the information on the physical area of the source emission region $B^{(phys)}_{\;\widehat{z}}(q)\,=\,\exp_p(B_{\;\widehat{z}}(q))$ visible by $(p, \dot\gamma(p))$. 
Let us consider the ratio 
\begin{align}
\label{TheRatioAD}
&\,\frac{A(B^{(phys)}_{\;\widehat{z}}(q))}{\omega(B_{\,\widehat{z}}(q))}\,=\,
\frac{A\left(\exp_p(B_{\;\widehat{z}}(q))\right)}{\omega(B_{\,\widehat{z}}(q))}\\
&=\,\frac{1}{\omega(B_{\;\widehat{z}}(q))}\,\int_{B_{\;\widehat{z}}(q)}\,\frac{d\mu_{h_{\,\widehat{z}}}}{d\mu_{\mathbb{S}^2}}\,d\mu_{\mathbb{S}^2}\nonumber\\
&=\, \frac{1}{r^2(\widehat{z\,})\omega(B_{\,\widehat{z}}(q))}\,
\int_{B_{\,\widehat{z}}(q)}\,\frac{d\mu_{h_{\,\widehat{z}}}}{d\mu_{\mathbb{S}^2}}\,
 d\mu_{\mathbb{S}^2(\widehat{z}\,)}\nonumber\,,
\end{align}
\vskip 0.4cm\noindent
where we have parametrized $B^{(phys)}_{\;\widehat{z}}(q)\,=\,
\exp_p(B_{\;\widehat{z}}(q))$  in terms of the celestial sphere image $B_{\;\widehat{z}}(q))$, and  introduced the comoving radial factor $r^2(\widehat{z\,})$ to rewrite the solid angle measure $d\mu_{\mathbb{S}^2}$ in terms of the round measure 
$d\mu_{\mathbb{S}^2(\widehat{z}\,)}=\,r^2(\widehat{z\,})\,d\mu_{\mathbb{S}^2}$ on $\mathbb{CS}_{\;\widehat{z}}(p)$. Finally, we have denoted by $\omega(B_{\;\widehat{z}}(q))$the solid angle spanned by  $B_{\;\widehat{z}}(q)$ on the directional celestial sphere $\mathbb{CS}(p)$. 
We have the following result.

\begin{proposition} (The Area Distance)
\label{PropAreaDistance}

Let $\mathrm{diam}_{\mathbb{S}^2(\widehat{z})}(B_{\,\widehat{z}}(q))\,=\,\delta\,>\,0 $, and let us denote by $x_q\,:=\,\exp_p^{\,- 1}(q)$ the normal coordinates on $(\mathbb{CS}_{\,\widehat{z}\,}(p),\, h_{\mathbb{S}^2(\,\widehat{z}\,)})$ associated with the bi-Lipschitz null exponential map 
$\exp_p\,:\,\mathbb{CS}_{\,\widehat{z}\,}(p)\,\longrightarrow\, \Sigma_{\,\widehat{z}\,}$.
The density (a Radon-Nikodym derivative)
\begin{equation}
\frac{d\mu_{h_{\,\widehat{z}}}(y)}{d\mu_{\mathbb{S}^2}(y)},\;\; y\,\in\,B_{\;\widehat{z}}(q)
\end{equation}
is locally summable over  $B_{\;\widehat{z}}(q)\,\cap\,(\mathbb{CS}_{\,\widehat{z}\,}(p),\, h_{\mathbb{S}^2(\,\widehat{z}\,)})$. It follows that the limit 
\begin{align}
\label{BesLimit}
&\,\lim_{\delta\,\rightarrow\,0}\,
\frac{A\left(\exp_p(B_{\;\widehat{z}}(q))\right)}{\omega(B_{\,\widehat{z}}(q))}\,=\,
\lim_{\delta\,\rightarrow\,0}\,\frac{1}{\omega(B_{\,\widehat{z}}(q))}\,
\int_{B_{\,\widehat{z}}(q)}\,\frac{d\mu_{h_{\,\widehat{z}}}}{d\mu_{\mathbb{S}^2}}(y)\,d\mu_{\mathbb{S}^2}(y)\\
&=\,\lim_{\delta\,\rightarrow\,0}\,\frac{1}{r^2(\widehat{z\,})\omega(B_{\,\widehat{z}}(q))}\,
\int_{B_{\,\widehat{z}}(q)}\,\frac{d\mu_{h_{\,\widehat{z}}}}{d\mu_{\mathbb{S}^2}}(y)\,
 d\mu_{\mathbb{S}^2(\widehat{z}\,)}(y)\nonumber\\
&=\,
\left(\frac{d\mu_{h_{\,\widehat{z}}}}{d\mu_{\mathbb{S}^2}}\right)(\exp_p^{\,- 1}(q))\,=:\,D^2_{\,\widehat{z}}(x_q)\nonumber\,,
\end{align}
evaluated by the observer $(p,\,\dot\gamma(p))$,  
exists almost everywhere on $B_{\;\widehat{z}}(q)\,\cap\,\mathbb{CS}_{\,\widehat{z}\,}(p)$. The scalar density $D^2_{\,\widehat{z}}(x_q)$, so defined, is the squared area distance between the physical observer $(p, \dot\gamma(p))$ and the source $q\,\in\,\Sigma_{\;\widehat{z}}$. In terms of the metric distortion $\mathrm{Lip}_{\,\widehat{z}\,}$, the reference FLRW comoving radius $\widehat{r}^{\,2}(\,\widehat{z}\,)$, and the relative velocity $v_{\,\widehat{z}}$ of the physical observer $(p, \dot\gamma(p))$ with respect to the FLRW observer $(p, \widehat{\dot\gamma}(p))$, the squared area distance $D^2_{\,\widehat{z}}(x_q)$ is uniformly bounded according to
\begin{equation}
\label{AreaDistBounds}
\frac{1}{\mathrm{Lip}^2_{\,\widehat{z}\,}}\,\left(\frac{1\,+\,v_{\,\widehat{z}}\,(p)}{1\,-\,v_{\,\widehat{z}}\,(p)}\right)\,
\widehat{r}^{\,2}(\,\widehat{z}\,)\,\leq\,D^2_{\,\widehat{z}}(x_q)
\leq\,
\mathrm{Lip}^2_{\,\widehat{z}\,}\,\left(\frac{1\,+\,v_{\,\widehat{z}}\,(p)}{1\,-\,v_{\,\widehat{z}}\,(p)}\right)\,
\widehat{r}^{\,2}(\,\widehat{z}\,)\,.
\end{equation}
\end{proposition}

\begin{proof}
The limit in (\ref{BesLimit})  
\begin{equation}
\lim_{\delta\,\rightarrow\,0}\,\frac{1}{r^2(\widehat{z\,})\omega(B_{\,\widehat{z}}(q))}\,
\int_{B_{\,\widehat{z}}(q)}\,\frac{d\mu_{h_{\,\widehat{z}}}}{d\mu_{\mathbb{S}^2}}(y)\,
 d\mu_{\mathbb{S}^2(\widehat{z}\,)}(y)\,=\,
\left(\frac{d\mu_{h_{\,\widehat{z}}}}{d\mu_{\mathbb{S}^2}}\right)(\exp_p^{\,- 1}(q))\,,
\end{equation}
exists by the Lebesgue-Besicovitch theorem (see \emph{e.g.} \cite{Gariepy}),
and the uniform bound (\ref{AreaDistBounds}) is a direct consequence of (\ref{haussEsse2}).
\end{proof}
\vskip 0.3cm\noindent
The area element $d\mu_{h_{\,\widehat{z}}}$ and the solid angle measure $d\mu_{\mathbb{S}^2}$ possess several significant properties that are worth recalling. To wit,  let 
\begin{equation}
r(\,\widehat{z}\,)\,\longmapsto\,r(\,\widehat{z}\,)\,\ell_q\,\in\,C^{-}\left(T_pM,\,\{E_{(i)}\} \right)\,,\;\;\; 0\,\leq\,r(\,\widehat{z}\,)\,\leq\,r(\,\widehat{z}_{\,q})\,,
\end{equation}
denote the null ray in $T_pM$ pointing to  $x_q\,\in\,\mathbb{CZ}_{\,\widehat{z}}$, with $\exp_p(x_q)\,=\,q\,\in\,\Sigma_{\,\widehat{z}}$. If $y$ is a point source varying along the null geodesic
\begin{equation}
\label{yexpmap}
r\,\longmapsto\,\exp_p(r\,\ell_q),\;\;\;0\,\leq\,r\leq\,r(q)\,,
\end{equation}
then, according to the Gauss Lemma (see (\ref{sigmapr_0}),\, (\ref{GaussLemma_0})) the area $2$-form $\mu_{g^{(2)}_{\widehat{z}(y)}}$, associated with\, $d\mu_{g^{(2)}_{\widehat{z}(y)}}$,  is $g$-orthogonal\, at $y$ to the null geodesic tangent vector\footnote{Obtained by pushing forward with $\exp_p$ the vector $r(y)\ell_q\,\in\,C^{-}(T_pM,\,\{E_{(i)}\})$} $(\exp_p)_*\,(r(y)\ell_q)$.   Moreover, 
if $m_{(1)},\,m_{(2)}$ is an orthonormal basis in $T_y\,\Sigma_{\widehat{z}(y)}$, then the evaluation of the cross-sectional area $\int_{\sigma_y}\,\mu_{g^{(2)}_{\widehat{z}(y)}}(m_{(1)},\,m_{(2)})$ in a small neighborhood  $\sigma_y\,\subset\,\Sigma_{\widehat{z}(y)}$ of $y$ is, for $\mathrm{diam}(\sigma_y)\,\rightarrow\,0$, independent\footnote{This is a basic property of null geodesic congruences proved and discussed by several authors \cite{Sachs}, \cite{Pirani}, \cite{EllisVarenna}} of the four-velocity $\dot\gamma(y)$ of the source at $y$.

 Conversely, the solid angle measure $d\mu_{\mathbb{S}^2}$ has a more complex behavior related to its behavior under a Lorentz transformation. Explicitly, let us denote by $d\mu_{\mathbb{S}^2_{(y)}}$ the solid angle measure on the celestial sphere
 $\mathbb{CS}(y)\simeq\,\mathbb{S}^2$ associated with the generic point $y$ defined above.
By exploiting the exponential map along the null geodesic (\ref{yexpmap}) we can pull back $d\mu_{\mathbb{S}^2_{(y)}}$ to the celestial sphere $\mathbb{CS}(p)$ associated with the physical observer $(p, \dot\gamma(p))$,
\begin{equation}
\exp_p^*\left(d\mu_{\mathbb{S}^2_{(y)}}\right)\,,
\end{equation}
\vskip 0.3cm\noindent
and compare it with the prexisting solid angle measure $d\mu_{\mathbb{S}^2_{(p)}}$. The two measures are related by the beaming effect characteristic of the Lorentz group action on solid angles, \emph{i.e.}
\begin{equation}
\label{beaming}
\exp_p^*\left(d\mu_{\mathbb{S}^2_{(y)}}\right)\,=\,\frac{\left(k^a_{(p)}\,\dot\gamma_a(p)\right)^2}{\left(\exp_p^*(k^a_{(y)}\,\dot\gamma_a(y))\right)^2}\,d\mu_{\mathbb{S}^2_{(p)}}\,=\,\left(\frac{\lambda(y)}{\lambda(p)}\right)^2\,d\mu_{\mathbb{S}^2_{(p)}}\,,
\end{equation} 
\vskip 0.3cm\noindent 
where $k$ denotes the (future-pointing) wave vector associated with the signal emitted by $(y,\,\dot\gamma(y))$ with wavelength $\lambda(y)$ and received by $(p, \dot\gamma(p))$ with a wavelength $\lambda(p)$ (see (\ref{wavevect0})). The relation (\ref{beaming}) plays a role in characterizing another basic distance parameter, the \emph{galaxy area distance} $D_{G}$ (for a thorough presentation, see \cite{EllisVarenna}, \cite{Ellis2}). Roughly speaking, $D_G$ is characterized by the ratio between the cross-sectional area of the source, $d\sigma_G$, \emph{as seen by the observer} $(p, \dot\gamma(p))$ and the solid angle $d\omega_G$ describing the emission at the source, \emph{i.e.}
\begin{equation}
\label{GalArDist}
D^2_{G}\,:=\,\frac{d\sigma_G}{d\omega_G}\,.
\end{equation} 
\vskip 0.3cm\noindent
Since we have no information on the solid angle $d\omega_G$ in which the radiation was emitted,  the galaxy area distance is not directly observable. Yet, it is related to the area distance and the luminosity distance. These relations are well-known properties on which we do not belabor. Yet, for the sake of the reader, and due to the relevant role of the area distance in our paper, we provide a brief proof of the relation between $D^2$ and $D^2_G$. The standard proof, based on the analysis of the geodesic deviation equation, is quite complex (it is described in full detail in \cite{EllisVarenna}). We provide a simpler and shorter proof here by exploiting the pullback action of the exponential map on the relevant angular and area measures.

\begin{lemma} (The Etherington duality \cite{Etherington})

If \begin{equation}
\label{zetafactor1}
\left(1\,+\,{z}(q)\right)\,=\,\sqrt{\frac{1+v_{\,\widehat{z}}}{1-v_{\,\widehat{z}}}}(1\,+\,\widehat{z}^{(tot)}(q))\,
\end{equation}
is the actual redshit of the source $q$ expressed in terms of the FLRW reference redshift $\widehat{z}$ and of the Doppler $\widehat{z}^{\,(doppl)}(q)$ and of the gravitational redshift $\widehat{z}^{\,(grav)}(q)$\; (see (\ref{zetafactor0A})), then the galaxy area distance $D_{G}$ and the area distance $D_{\,\widehat{z}}$ are related according to
\begin{equation}
D^2_G\,=\,\left(1\,+\,z(q)\right)^2\,D^2_{\,\widehat{z}}\,.
\end{equation}
\end{lemma} 
\begin{proof}
We can directly work on the central null geodesic connecting the observer $(p,\,\dot\gamma(p))$ with the source at $(q,\,\dot\gamma(q))$.  To the observer we associate the solid angle measure $d\mu_{\mathbb{S}^2}$ defined on the celestial sphere $\mathbb{CS}(p)$, and the physical area measure $d\mu_{h_{\,\widehat{z}}}$ defined on $\mathbb{CS}_{\,\widehat{z}}(p)$ by the pull back $\exp_p^*d\mu_{g^{(2)}_{\,\widehat{z}}}$ of the area measure $d\mu_{g^{(2)}_{\,\widehat{z}}}$ defined at $(q,\,\dot\gamma(q))$. Finally, if  $\mathbb{CS}(q)$ denotes the directional celestial sphere associated with $(q,\,\dot\gamma(q))$, then we can associate with the source the solid angle measure $d\mu_{\mathbb{S}^2(q)}$. In such a framework, the role of the cross-sectional area of the source, $d\sigma_G$, as seen by the observer $(p, \dot\gamma(p))$ is provided by the pullback measure $d\mu_{h_{\,\widehat{z}}}$ at $(p,\,\dot\gamma(p))$. The emission solid angle $d\mu_{\mathbb{S}^2(q)}$ at  $(q,\,\dot\gamma(q))$ corresponds to the $d\omega_G$ featuring in the ratio $(\ref{GalArDist})$. Since we need to compare all these geometrical data at $(p,\,\dot\gamma(p))$, we can pull back the source solid angle measure $d\mu_{\mathbb{S}^2(q)}$ to the observer's celestial sphere $\mathbb{CS}(p)$,\emph{i.e.}
\begin{equation}
\exp_p^*\left(d\mu_{\mathbb{S}^2(q)}\right)\,.
\end{equation}
\vskip 0.3cm\noindent
This is the representation on $\mathbb{CS}(p)$  of the source emission solid angle at $(q,\,\dot\gamma(q))$.\, We can compare this pullback solid angle with the solid angle measure $d\mu_{\mathbb{S}^2}$ defined on $\mathbb{CS}(p)$. According to
(\ref{beaming}) we have
\begin{equation}
\label{beaming2}
\exp_p^*\left(d\mu_{\mathbb{S}^2(q)}\right)\,=\,\frac{\left(k^a_{(p)}\,\dot\gamma_a(p)\right)^2}{\left(\exp_p^*(k^a_{(q)}\,\dot\gamma_a(q))\right)^2}\,d\mu_{\mathbb{S}^2_{(p)}}\,=\,\left(\frac{\lambda(q)}{\lambda(p)}\right)^2\,d\mu_{\mathbb{S}^2}\,.
\end{equation} 
The galaxy area distance is defined as the relation
\begin{equation}
d\mu_{h_{\,\widehat{z}}}\,=\,{D^2_G}\,{\exp_p^*\left(d\mu_{\mathbb{S}^2(q)}\right)}\,,
\end{equation}
whereas from (\ref{BesLimit}), defining the area distance $D_{\,\widehat{z}}$, we get 
\begin{equation}
d\mu_{h_{\,\widehat{z}}}\,=\,{D_{\,\widehat{z}}^2}\,{d\mu_{\mathbb{S}^2}}\,.
\end{equation}
If we take into account (\ref{beaming2}), 
the direct comparison between these two expressions  provides
\begin{align}
D^2_G\,&=\,\left(\frac{\lambda(p)}{\lambda(q)}\right)^2\,D_{\,\widehat{z}}^2\\
&=\,\left(1\,+\,z(q)\right)^2\,D_{\,\widehat{z}}^2\,,\nonumber
\end{align}
as stated.
\end{proof}
\vskip 0.3cm\noindent

As already stressed, the area distance $D_{\,\widehat{z}}(x_q)$ is, at least in principle, a measurable quantity for sources of known intrinsic size\footnote{To what extent $D_{\,\widehat{z}}(x_q)$ is measurable is critically discussed by G. F. R. Ellis in the Editor's note to \cite{Etherington}.}. Geometrically, it is a density, and one has the useful

\begin{lemma}
\label{DiArea}
\begin{equation}
\label{PhysArea1}
\int_{\mathbb{CS}_{\,\widehat{z}\,}(p)}\,D^2_{\,\widehat{z}}(x)\,d\mu_{\mathbb{S}^2}(x)\,=\,A\left(\Sigma_{\;\widehat{z}}\right)\,=\,A\left(\mathbb{CS}_{\,\widehat{z}\,},\,h_{\,\widehat{z}} \right)\,,
\end{equation}
\end{lemma}
\begin{proof}
This result follows easily from the same pullback chain that leads to (\ref{chainPullBack}). Explicitly, we have  
\begin{align}
\label{PhysArea2}
&\int_{\mathbb{CS}_{\,\widehat{z}\,}(p)}\,D^2_{\,\widehat{z}}\,d\mu_{\mathbb{S}^2}\,=\,
\int_{\mathbb{CS}_{\,\widehat{z}\,}(p)}\,\frac{d\mu_{h_{\,\widehat{z}}}}{d\mu_{\mathbb{S}^2}}\,d\mu_{\mathbb{S}^2}\\
&=\,
\int_{\mathbb{CS}_{\,\widehat{z}\,}(p)}\,d\mu_{h_{\,\widehat{z}}}\,=\, A\left(\mathbb{CS}_{\,\widehat{z}\,},\,h_{\,\widehat{z}} \right)\,=\,
\int_{\mathbb{CS}_{\,\widehat{z}\,}(p)}\,\exp_p^*\left(d\mu_{g^{(2)}_{\,\widehat{z}}}  \right)\nonumber\\
&=\,\int_{\exp_p(\mathbb{CS}_{\,\widehat{z}\,)}(p)}\,d\mu_{g^{(2)}_{\,\widehat{z}}}\,=\,
\int_{\Sigma_{\,\widehat{z}\,}(p)}\,d\mu_{g^{(2)}_{\,\widehat{z}}}\nonumber\\
&=\,A\left(\Sigma_{\;\widehat{z}}\right)\,.\;\;\;\;\square\nonumber 
\end{align}
\end{proof}

\subsection{The FLRW case}
The above analysis can be readily extended to the reference FLRW past lightcone  $\widehat{\mathscr{C}}^-(p, \widehat{g}\,)$, the corresponding celestial spheres $\widehat{C\,\mathbb{S}}_{\,\widehat{z}\,}(p)$ and sky sections $\widehat{\Sigma}_{\,\widehat{z}\,}$. We do not repeat all the FLRW versions of the technical results proven in the previous sections. Suffices to say that given a geometrical quantity in the physical spacetime $(M, g, \gamma(\tau))$, we denote the corresponding quantity in the FLRW spacetime with a circumflex accent. No ambiguities can arise since, as we have seen, the definition of the past light cone 
$\widehat{\mathscr{C}}^-(p, \widehat{g})$, of the corresponding celestial spheres $\widehat{\mathbb{CS}}_{\hat{z}}(\,p)$ and sky sections $\widehat{\Sigma}_{\, \widehat{z}}$, and of the exponential map $\widehat{\exp}_p$ are formally the same in $(M, g, \gamma(\tau))$ and $(M, \widehat{g}, \widehat{\gamma}(\widehat{\tau}))$,  with the obvious replacement of the relevant quantities.

As in Proposition \ref{PropAreaDistance}  the area element associated with the metric $\widehat{h}_{\hat{z}}$ on $\widehat{\mathbb{CS}}_{\hat{z}}(\,p)$ (see (\ref{FLRWhpullBack})), 
\begin{equation}
\label{Farea}
d\mu_{\widehat{h}_{\hat{z}}}\,=\,\frac{f^2\left(\widehat{r}(\,\widehat{z})\right)}{(1\,+\,\widehat{z})^2}\,\,d\mu_{\widehat{\mathbb{S}}^2}\,
\end{equation}
characterizes  the FLRW \emph{observer area distance} between the FLRW observer $(p,\widehat{\dot\gamma}\,(p))$ and a source $q$ at redshift $\widehat{z}$ according to  
\begin{equation}
\label{FLRWDi}
\widehat{D}^2_{\hat{z}}\,:=\,\frac{d\mu_{\widehat{h}_{\hat{z}}}}{d\mu_{\widehat{\mathbb{S}}^2}}\,=\, \frac{f^2\left(\widehat{r}(\,\widehat{z})\right)}{\left(1\,+\,\widehat{z}\right)^2}\,,
\end{equation}
where $d\mu_{\widehat{\mathbb{S}}^2}$ is the solid angle measure on the FLRW observer's directional celestial sphere $\widehat{\mathbb{CS}}(\,p)$. 

One can show that relations similar to (\ref{PhysArea1}) and 
(\ref{PhysArea2}) also hold among the areas $A(\widehat{\Sigma}_{\hat{z}})$,\,
$A(\widehat{\mathbb{C\,S}}_{\hat{z}})$, and the area distance      $\widehat{D}_{\hat{z}}$, \emph{i.e.} 
\begin{equation}
\label{FLRWArea1}
\int_{\widehat{\mathbb{CS}}_{\,\widehat{z}\,}(p)}\,\widehat{D}^2_{\,\widehat{z}}(x)\,d\mu_{\widehat{\mathbb{S}}^2}(x)\,=\,\widehat{A}\left(\,\widehat{\Sigma}_{\;\widehat{z}}\right)\,=\,\widehat{A}\left(\,\widehat{\mathbb{CS}}_{\,\widehat{z}\,},\,\widehat{h}_{\,\widehat{z}} \right)\,.
\end{equation}
Explicitly, 
\begin{align}
\label{FLRWArea2}
&\int_{\widehat{\mathbb{CS}}_{\,\widehat{z}\,}(p)}\,\widehat{D}^2_{\,\widehat{z}}\,d\mu_{\widehat{\mathbb{S}}^2}\,=\,
\int_{\widehat{\mathbb{CS}}_{\,\widehat{z}\,}(p)}\,\frac{d\mu_{\widehat{h}_{\,\widehat{z}}}}{d\mu_{\widehat{\mathbb{S}}^2}}\,d\mu_{\widehat{\mathbb{S}}^2}\\
&=\,
\int_{\widehat{\mathbb{CS}}_{\,\widehat{z}\,}(p)}\,d\mu_{\widehat{h}_{\,\widehat{z}}}\,=\, \widehat{A}\left(\widehat{\mathbb{CS}}_{\,\widehat{z}\,},\,\widehat{h}_{\,\widehat{z}} \right)\,=\,
\int_{\widehat{\mathbb{CS}}_{\,\widehat{z}\,}(p)}\,\widehat{\exp}_p^*\left(d\mu_{\widehat{g}^{(2)}_{\,\widehat{z}}}  \right)\nonumber\\
&=\,\int_{\widehat{\exp}_p(\widehat{\mathbb{CS}}_{\,\widehat{z}\,)}(p)}\,d\mu_{\widehat{g}^{(2)}_{\,\widehat{z}}}\,=\,
\int_{\widehat{\Sigma}_{\,\widehat{z}\,}(p)}\,d\mu_{\widehat{g}^{(2)}_{\,\widehat{z}}}\,=\,\widehat{A}\left(\widehat{\Sigma}_{\;\widehat{z}}\right)\,.\;\;\;\;\square\nonumber 
\end{align}

\subsection{Area distance and spacetime scalar curvature}
\label{GareDisCurv}
The connection between area distance and spacetime scalar curvature $\mathrm{R}(g)$ is often dismissed on the basis that $\mathrm{R}(g)$ enters only indirectly in the Raychaudhuri equation governing null geodesics evolution \cite{Ellis2}. As a consequence, the non-perturbative role of $\mathrm{R}(g)$ in governing the fluctuations of area distance is underestimated, even if it can be quite significant in high-precision cosmology. These fluctuations are particularly relevant when we approach and cross the boundary of the spatial region ${V}_{(c)}(p)$ where the gravitational dynamics of the sources uncouples from cosmological expansion. As discussed in Section \ref{RefCosmRed}, this region is typically identified with the cluster of galaxies whose barycenter characterizes the location of the physical observer $p$. We have the following characterization.

\begin{definition}
Let
\begin{equation}
\mathbb{R}_{>0}\,\ni\,\tau\,\longmapsto\,\gamma(\tau)\,\in\,(M, g),\;\;\;-\delta<\tau<\delta,\;\;\; \gamma(\tau=0)\,=:\,p
\end{equation}
be the timelike geodesic segment associated with the physical observer $(p, \dot\gamma(0))$ (see Definition \ref{DefPhysObs}). We assume that $(p, \dot\gamma(0))$ is located at the barycenter of a galactic cluster, and we denote by $\mathbb{W}(\gamma(\tau))$,\; $-\delta<\tau<\delta$,\, the (portion of) timelike world-tube generated by the bundle of timelike geodesics describing the evolution of the cluster's galaxies around the observer world history $\gamma(\tau)$. The intersection
\begin{equation}
\partial\,{V}_{(c)}(p)\,:=\,\partial\,
\mathbb{W}(\gamma(\tau))\,\cap\,{\mathcal{C}}^-(p, {g}) 
\end{equation}   
between the timelike boundary $\partial\mathbb{W}(\gamma(\tau))$ of the cluster world-tube and the observer's past lightcone ${\mathcal{C}}^-(p, {g})$ defines the boundary of the cluster as seen by $(p, \dot\gamma(0))$. We assume that $\partial\,{V}_{(c)}(p)$ is a 2-dimensional surface (topologically $\mathbb{S}^2$) bounding a 
$3$-dimensional disk ${V}_{(c)}(p)$ describing, with respect to $(p, \dot\gamma(0))$,  the spatial region occupied by the cluster\footnote{Since the cluster galaxies are in a small neighborhood of $p$, we can assume that ${V}_{(c)}(p)$ is contained in the maximal spatial domain bounded by $\partial\,{V}_{(c)}(p)$.}. To ${V}_{(c)}(p)$ we associate:
\begin{itemize}
\item{The comoving radius $r(\widehat{z}_{(c)})$ of ${V}_{(c)}(p)$, defining the cluster's \emph{physical radius},\, marking the transition between the inner cluster region, whose gravitational dynamics is uncoupled from the cosmological expansion, and the cluster outskirt where infalling galaxies are marginally affected by the outside cosmological expansion;}
\item{The reference FLRW observer $\widehat\tau\,\longmapsto\,\widehat{\gamma}(\widehat\tau)$,\;
$-\widehat\delta<\widehat\tau<\widehat\delta$,\; with\, $\widehat\gamma(0)\,=:\,p$  can characterize similarly her description of the cluster  by introducing the cluster region  
$\widehat{V}_{(c)}(p)$, its boundary $\partial\,\widehat{V}_{(c)}(p)$, and  
the corresponding FLRW cluster physical radius $\widehat{r}(\widehat{z}_{(c)})$.}
\end{itemize}
\end{definition}
Notice that, following our notation, we have used, for the physical observers' characterization of the cluster region, the redshift value that the FLRW observer associates with $\widehat{r}(\widehat{z}_{(c)})$. Rather than a matter of notation, this choice is mandatory since the uncoupling from cosmological expansion in the cluster region makes the use of the physical redshift $z$ useless\footnote{As discussed in Section \ref{RefCosmRed}, the physical redshift $z$ is not conducive to establish a reliable redshift-distance relation when peculiar velocities and gravitational clustering induce significant peculiar velocities. }. 
  
To let the interplay between the spacetime scalar curvature and area distance emerge from backstage, we  exploit the Lorentzian version of the classical Bertrand-Puiseux 
formulas \cite{SchoenYau},\,  \cite{Petersen}  that are associated with the geometrical interpretation of the various curvatures on a Riemannian manifold in terms of length, area, and volume measures of small geodesic circles, spheres, and balls. In the Lorentzian case, the corresponding formulas are quite more delicate to derive since rather than geodesic balls, we have to deal with small causal diamonds, \cite{Berthiere}, \cite{GibbSol1}, \cite{GibbSol2}, \cite{Myrheim}. To write down the relevant expression in our case and take into account the role of the cosmological uncoupling at the physical boundary $\partial\,\widehat{V}_{(c)}(p)$ of the cluster,
we foliate the physical past 
lightcone region  $\mathbb{W}(\gamma(\tau))\,\cap\,{\mathcal{C}}^-(p, {g})$ 
 with a  sequence of sky sections
 \begin{equation}
 \label{SigmaFol}
 \{\Sigma_{\widehat{z}(i)}\}_{i=1}^\infty\,,\;\;\;\mathrm{with}\,\,\Sigma_{\widehat{z}(1)}\,:=\,
 \partial\,\widehat{V}_{(c)}(p),\,\,\mathrm{and}\,\,\
\lim_{i\rightarrow\,\infty}\Sigma_{\widehat{z}(i)}\,\rightarrow\, p\,, 
 \end{equation}
starting from the physical boundary $\partial\,\widehat{V}_{(c)}(p)$ of the cluster and 
converging, as $i\,\rightarrow\,\infty$,\, to the observation event $(p, \dot\gamma(p))$. Since we are sufficiently near $p$, we can assume, without loss of generality, that there is a unique corresponding sequence $\{{V}_{(i)}(p)\}$ of maximal 3-dimensional regions \cite{GibbSol1}, embedded in the causal past of $p$,\; ${\mathrm{J}}^-(p, {g})$,\;  such that the boundary of ${V}_{(i)}(p)$ is the corresponding surface  $\Sigma_{\widehat{z}\,(i)}$. Each ${V}_{(i)}(p)$  intersects the world line segment ${\gamma}(\tau)$,\; $-\,\epsilon\,\leq\,\tau\,\leq\,0$, of the observer $p$ at the point ${p}_{(i)}=\,{\gamma}({\tau}_{(i)}\,=\,-\,r(\,\widehat{z}\,(i)))$, with
$\widehat{p}_{(1)}\,=:\,p_{(c)}$\; and\,
 $\widehat{p}_{(i)}\,\longrightarrow\,p$\, as\, $i\,\rightarrow\, \infty$.
The Lorentzian Bertrand-Puiseux 
formula for the sequence of sky sections $\{\Sigma_{\widehat{z}\,(i)}\}$  involves at leading order in ${r}({\widehat{z}\,(i)})$ the spacetime scalar curvature \cite{GibbSol1}, and we have the asymptotic expression 
\begin{equation} 
\label{Areaasymp1}
A\left(\Sigma_{\widehat{z}\,(i)}\right)\,=\,{4\pi}\,{r}^{2}({\widehat{z}\,(i)})\,\left(1\,-\,\frac{1}{72}\,{r}^{2}({\widehat{z}\,(i)})\,\mathrm{R}(p)\,+\,\ldots \right)\;,
\end{equation}
where $\mathrm{R}(p)$ is the spacetime scalar curvature of $(M, g)$ evaluated at the observational event $(p, \dot\gamma(p))$. The $\ldots$\ represents higher order corrections that can be expressed in terms of powers of the (spacetime) Riemann tensor and its (covariant) derivatives. The asymptotics  (\ref{Areaasymp1}) is particularly significant in correspondance of the initial surface $\Sigma_{\widehat{z}(1)}$ of the sequence $\{\Sigma_{\widehat{z}(i)}\}$, marking the boundary $\partial\,{V}_{(c)}(p)$ of the region ${V}_{(c)}(p)$ where the gravitational dynamics of the sources uncouples from cosmological expansion, and where the cluster's \emph{physical radius},\, 
$r(\widehat{z}_{(c)})$,  provides a natural length scale to compare with the corresponding length scale set by the local scalar curvature $\mathrm{R}(p_{(c)})$. In such a case, the asymptotic expression (\ref{Areaasymp1}) takes the expressive form  
\begin{equation} 
\label{Areaasymp1crit}
\frac{A\left(\Sigma_{\widehat{z}_{(c)}}\right)}{{4\pi}\,{r}^{2}(\widehat{z}_{(c)})}\,=\,1\,-\,\frac{1}{72}\,{r}^{2}(\widehat{z}_{(c)})\,\mathrm{R}(p)\,+\,\ldots\;.
\end{equation}
It must be stressed that the asymptotics (\ref{Areaasymp1}) (a consequence of the normal coordinate expansion of the metric) is not a perturbative result, as often stated\footnote{This misunderstanding is because the normal coordinate asymptotic expansion in terms of the curvature tensor and its derivatives is often used in the perturbative analysis of geometrical field theories}. There is no better way of stressing this than noticing that, 
as ${r}(\widehat{z}\,(i))\,\longrightarrow\,0$, \emph{i.e.} for $p_{(i)}\,\longrightarrow\,p$,\; we get the Lorentzian Bertrand-Puiseaux formula providing the geometrical definition of the spacetime scalar curvature at the observation point $(p, \dot\gamma(p))$, 
\begin{equation}
\label{Areaasymp13}
\lim_{{r}({\widehat{z}\,(i)})\,\longrightarrow\,0^+}\,\frac{18}{\pi}\,\frac{\left[4\pi\,{r}^2({\widehat{z}\,(i)})\,-\,
A(\Sigma_{\widehat{z}\,(i)})\right]}{{r}^4({\widehat{z}\,(i)})}\,
=\,{\mathrm{R}}(p)\,.
\end{equation}
\vskip 0.3cm\noindent
According to (\ref{PhysArea1}), we have 
\begin{equation}
\label{PhysAreazeta(i)}
\int_{\mathbb{CS}_{\,\widehat{z}\,(i)\,}(p)}\,D^2_{\,\widehat{z}\,(i)}(x)\,d\mu_{\mathbb{S}^2}(x)\,=\,A\left(\Sigma_{\;\widehat{z}\,(i)}\right)\,=\,A\left(\mathbb{CS}_{\,\widehat{z}\,(i)},\,h_{\,\widehat{z}\,(i)} \right)\,,
\end{equation}
from which we get 
\begin{equation}
\label{PhysAreazeta(ii)}
A\left(\Sigma_{\;\widehat{z}\,(i)}\right)\,=\,4\pi\,\left\langle\,D^2_{\,\widehat{z}\,(i)}\,  \right\rangle_{{\mathbb{CS}}}\,,
\end{equation}
where
\begin{equation}
\label{PhysAreazeta(iii)}
\,\left\langle\,D^2_{\,\widehat{z}\,(i)}\,  \right\rangle_{{\mathbb{CS}}}\,:=\,\frac{1}{4\pi}\,\int_{\mathbb{CS}(p)}\,D^2_{\,\widehat{z}\,(i)}(x)\,d\mu_{\mathbb{S}^2}(x)
\end{equation}
is the average value of the area distance with respect to the solid angle measure on the directional celestial sphere $\mathbb{CS}(p)$. Thus, we can rewrite (\ref{Areaasymp13}) according to 
\begin{equation}
\label{AreaDistasymp}
\lim_{{r}({\widehat{z}\,(i)})\,\longrightarrow\,0^+}\,\frac{18}{\pi}\,\frac{\left[4\pi\,\left({r}^2({\widehat{z}\,(i)})\,-\,
\left\langle\,D^2_{\,\widehat{z}\,(i)}\right\rangle_{{\mathbb{CS}}}\right)\right]}{{r}^4({\widehat{z}\,(i)})}\,
=\,{\mathrm{R}}(p)\,.
\end{equation}
In the reference FLRW case, we have similar asymptotics, and we can write, with the obvious meaning of the symbols adopted
 \begin{equation} 
\label{FLRWAreaasymp1}
\widehat{A}\left(\widehat{\Sigma}_{\widehat{z}\,(i)}\right)\,=\,{4\pi}\,\widehat{r}^{\,\;2}({\widehat{z}\,(i)})\,\left(1\,-\,\frac{1}{72}\,\widehat{r}^{\,\;2}({\widehat{z}\,(i)})\,\widehat{\mathrm{R}}(\widehat{p})\,+\,\ldots \right)\;,
\end{equation}
and
\begin{equation} 
\label{FLRWAreaasymp1crit}
\frac{\widehat{A}\left(\widehat{\Sigma}_{\widehat{z}_{(c)}}\right)}{{4\pi}\,\widehat{r}^{\,\;2}({\widehat{z}_{(c)}})}\,=\,1\,-\,\frac{1}{72}\,\widehat{r}^{\,\;2}({\widehat{z}_{(c)}})\,\widehat{\mathrm{R}}({p})\,+\,\ldots\;,
\end{equation}
\vskip 0.3cm\noindent
where $\widehat{\mathrm{R}}(p)$ is the FLRW spacetime scalar curvature of $(M, \widehat{g})$ evaluated at the event\footnote{Generally, the sequence of points $\{\widehat{p}_{(i)}\}$ is  distinct from the sequence $\{{p}_{(i)}\}$; yet, both are assumed to converge to $p$.} $({p}, \widehat{\dot\gamma}(p))$.  As $\widehat{r}(\widehat{z}\,(i))\,\longrightarrow\,0$, \; we get 
\begin{equation}
\label{FLRWAreaasymp13}
\lim_{\widehat{r}({\widehat{z}\,(i)})\,\longrightarrow\,0^+}\,\frac{18}{\pi}\,\frac{\left[4\pi\,\widehat{r}^2({\widehat{z}\,(i)})\,-\,
\widehat{A}\left(\widehat{\Sigma}_{\widehat{z}\,(i)}\right)\right]}{\widehat{r}^4({\widehat{z}\,(i)})}\,
=\,\widehat{\mathrm{R}}(p)\,.
\end{equation}
\vskip 0.3cm\noindent
We have also the FLRW version of the relations (\ref{PhysAreazeta(i)}),\,(\ref{PhysAreazeta(ii)}),\,and (\ref{PhysAreazeta(iii)}), \emph{i.e.}
\begin{equation}
\label{FLRWPhysAreazeta(i)}
\int_{\mathbb{\widehat{CS}}_{\,\widehat{z}\,(i)\,}(p)}\,\widehat{D}^2_{\,\widehat{z}\,(i)}(x)\,d\mu_{\widehat{\mathbb{S}}^2}(x)\,=\,\widehat{A}\left(\widehat{\Sigma}_{\;\widehat{z}\,(i)}\right)\,=\,\widehat{A}\left(\widehat{\mathbb{CS}}_{\,\widehat{z}\,(i)},\,\widehat{h}_{\,\widehat{z}\,(i)} \right)\,,
\end{equation} 
\begin{equation}
\label{FLRWPhysAreazeta(ii)}
A\left(\widehat{\Sigma}_{\;\widehat{z}\,(i)}\right)\,=\,4\pi\,\left\langle\,\widehat{D}^2_{\,\widehat{z}\,(i)}\,  \right\rangle_{\widehat{\mathbb{CS}}}\,,
\end{equation}
and
\begin{equation}
\label{FLRWPhysAreazeta(iii)}
\,\left\langle\,\widehat{D}^2_{\,\widehat{z}\,(i)}\,  \right\rangle_{\widehat{\mathbb{CS}}}\,:=\,\frac{1}{4\pi}\,\int_{\widehat{\mathbb{CS}}(p)}\,\widehat{D}^2_{\,\widehat{z}\,(i)}(x)\,d\mu_{\widehat{\mathbb{S}}^2}(x)\,,
\end{equation}
\vskip 0.2cm\noindent
from which we retrieve the FLRW Bertrand-Puiseaux formula\footnote{In the FLRW case, the limit (\ref{FLRWAreaDistasymp}) can be easily connected to the limiting behavior, as $\widehat{z}\rightarrow\,0$,\, of the Hubble parameter $H(\widehat{z})$ and of the deceleration parameter $q$ (see (\ref{Handq})). However, the explicit connection with $\widehat{\mathrm{R}}(p)$ is more useful for our purposes.  }  
\vskip 0.2cm\noindent 
\begin{equation}
\label{FLRWAreaDistasymp}
\lim_{\widehat{r}({\widehat{z}\,(i)})\,\longrightarrow\,0^+}\,\frac{18}{\pi}\,\frac{\left[4\pi\,\left(\widehat{r}^2({\widehat{z}\,(i)})\,-\,
\left\langle\,\widehat{D}^2_{\,\widehat{z}\,(i)}\right\rangle_{\widehat{\mathbb{CS}}}\right)\right]}{\widehat{r}^4({\widehat{z}\,(i)})}\,
=\,\widehat{\mathrm{R}}(p)\,.
\end{equation}
\vskip 0.3cm\noindent
If we exploit the action of the $\mathrm{PSL}(2,\,\mathbb{C})$ map (\ref{extensionZeta}), then we can profitably compare the asymptotics (\ref{Areaasymp13}) and (\ref{FLRWAreaasymp13}). To this end,  let $\{\widehat{w}_{(j)}\}$,\;$j=1,2,3$,\; denote the three null directions on $\widehat{\mathbb{C\,S}}_{\;\widehat{z}}$, associated to the three selected astrophysical sources 
$\{\widehat{q}_{(j)}\}$ that the FLRW observer $(p, \widehat{\dot{\gamma}}(p))$ locates at  FLRW redshift $\widehat{z}$. According to Proposition \ref{PSL2mappingFLRW} the map $\zeta_{(\widehat{z})}\,\in\,\mathrm{PSL}(2, \mathbb{C})$ connecting $\widehat{\mathbb{C\,S}}_{\;\widehat{z}}$ and ${\mathbb{C\,S}}_{\;\widehat{z}}$ is determined by the choice of the corresponding three images $\{\zeta_{(\widehat{z})}(\widehat{w}_{(j)})\}\,\in\,{\mathbb{C\,S}}(p)$. This 
choice  characterizes 
a rotation $e^{i\alpha_{\,\widehat{z}}}$ through an angle $\alpha_{\;\widehat{z}}$ around a common axis, (the observers can conveniently choose $E_{(3)}\,=\,\widehat{E}_{(3)}$),  followed by a boost with rapidity $\beta_{\;\widehat{z}}\,:=\,\log\,\sqrt{\frac{1\,+\,v_{\;\widehat{z}}}{1\,-\,v_{\;\widehat{z}}}}$, where $v_{\;\widehat{z}}(p)$ is the relative velocity of the physical observer $(p, \dot{\gamma}(p))$ with respect to the reference FLRW observer $(p, \widehat{\dot{\gamma}}(p))$. Given such a setup, let $(\widehat{r}(\,\widehat{z}\,)\widehat{\ell}(\,\widehat{n}\,( \widehat\theta,\widehat\phi)))$ the past-directed null vector in $T_pM$ pointing to the coordinates on $\widehat{\mathbb{C\,S}}_{\;\widehat{z}}$  of an astrophysical source $q$ at reference redshift $\widehat{z}$, and let 
\begin{equation}
\label{thePSLmap}
\widehat{\mathbb{C\,S}}_{\;\widehat{z}}\,\ni\,
(\widehat{r}(\,\widehat{z}\,)\widehat{\ell}(\,\widehat{n}\,( \widehat\theta,\widehat\phi)))\,\longmapsto\,\zeta_{(\widehat{z})}\left(\widehat{r}(\,\widehat{z}\,)\widehat{\ell}(\,\widehat{n}\,( \widehat\theta,\widehat\phi))\right)\,=\,({r}(\,\widehat{z}\,){\ell}(\,{n}\,(\theta,\phi)))
\end{equation}
the corresponding coordinates on ${\mathbb{C\,S}}_{\;\widehat{z}}$. We have (see (\ref{radialconnect0}))
\begin{equation}
\label{radialconnect0Map}
r(\,\widehat{z}\,)\,=\,\sqrt{\frac{1\,+\,v_{\;\widehat{z}}(p)}{1\,-\,v_{\;\widehat{z}}(p)}}\,\,\widehat{r}(\,\widehat{z}\,)\,,\;\;\;n(\theta, \phi)\,=\,e^{i\alpha_{\,\widehat{z}}}\,\,\widehat{n}\,(\widehat\theta, \widehat\phi)\,.
\end{equation}
\vskip 0.3cm\noindent
To proceed with the comparison of (\ref{Areaasymp13}) and (\ref{FLRWAreaasymp13}), let us start by noticing that the pullback of $D^2_{\,\widehat{z}}\,d\mu_{\mathbb{S}^2}$, under the  $\mathrm{PSL}(2,\,\mathbb{C})$ map $\zeta_{(\widehat{z})}$,\, is given by
\begin{equation}
\zeta_{(\widehat{z})}^*\left(D^2_{\,\widehat{z}}\,d\mu_{\mathbb{S}^2}   \right)\,=\,\left(D^2_{\,\widehat{z}}\,\circ\,\zeta_{(\widehat{z})}\right)\,\zeta_{(\widehat{z})}^*d\mu_{\mathbb{S}^2}\,
=\,
D^2_{\,\widehat{z}}\,d\mu_{\widehat{\mathbb{S}}^2}\,\circ\,\zeta_{(\widehat{z})}\,,
\end{equation}
where $d\mu_{\widehat{\mathbb{S}}^2}$ is the solid angle measure on $\widehat{\mathbb{C\,S}}$. Note that,
the pull back $D^2_{\,\widehat{z}}\,d\mu_{\widehat{\mathbb{S}}^2}\,\circ\,\zeta_{(\widehat{z})}$  contains the aberration factor $\sqrt{\frac{1\,+\,v_{\;\widehat{z}}(p)}{1\,-\,v_{\;\widehat{z}}(p)}}$ induced by (\ref{radialconnect0Map}). 
To simplify computations, it is useful to introduce a shorthand notation that explicitly accounts for this aberration factor. Thus, we define
\begin{equation}
\label{Dnotation}
D^2_{\,\widehat{z}}\,\left(\zeta_{(\widehat{z})},\,v_{\;\widehat{z}}\right)\,d\mu_{\widehat{\mathbb{S}}^2}:=\,\frac{1\,-\,v_{\;\widehat{z}}}{1\,+\,v_{\;\widehat{z}}}\; D^2_{\,\widehat{z}}\,\,d\mu_{\widehat{\mathbb{S}}^2}\,\circ\,\zeta_{(\widehat{z})}\,.
\end{equation} 
\vskip 0.3cm\noindent
By taking into account these notational remarks, the following chain of relations allows us to pull back the average physical area distance from the physical to the reference FLRW celestial sphere $\widehat{\mathbb{CS}}_{(\widehat{z})}$. By integrating over its directional counterpart $\widehat{\mathbb{CS}}$ we get, 
\begin{align}
\,\left\langle\,D^2_{\,\widehat{z}}\,  \right\rangle_{{\mathbb{CS}}}\,&:=\,\frac{1}{4\pi}\,\int_{\mathbb{CS}}\,D^2_{\,\widehat{z}}\,d\mu_{\mathbb{S}^2}\,=\,
\frac{1}{4\pi}\,\int_{\zeta_{(\widehat{z})}(\widehat{\mathbb{CS}})}\,D^2_{\,\widehat{z}}\,d\mu_{\mathbb{S}^2}\nonumber\\
\\
&=\,
\frac{1}{4\pi}\,\int_{\widehat{\mathbb{CS}}}\,\zeta_{(\widehat{z})}^*\left(D^2_{\,\widehat{z}}\,d\mu_{\mathbb{S}^2}\right)\,=\, 
\frac{1}{4\pi}\,\int_{\widehat{\mathbb{CS}}}\,\left(
D^2_{\,\widehat{z}}\,d\mu_{\widehat{\mathbb{S}}^2}\,\circ\,\zeta_{(\widehat{z})}\right)\nonumber\\
\nonumber\\
&=:\,\frac{1\,+\,v_{\;\widehat{z}}}{1\,-\,v_{\;\widehat{z}}}\,
\left\langle\,D^2_{\,\widehat{z}}\,(\zeta_{(\widehat{z})},\,v_{\;\widehat{z}})  \right\rangle_{\widehat{\mathbb{CS}}}\,.
\nonumber
\end{align}
By introducing this expression in (\ref{AreaDistasymp}) and taking into account the relation (\ref{radialconnect0Map}) connecting ${r}({\widehat{z}\,(i)})$ to $\widehat{r}({\widehat{z}\,(i)})$\,, an easy computation provides

\begin{equation}
\label{AreaDistasympFLRW}
\lim_{\widehat{z}\,\longrightarrow\,0^+}\,\frac{18}{\pi}\,\frac{4\pi\,\left(\widehat{r}^2({\widehat{z}})\,-\,
\left\langle\,D^2_{\,\widehat{z}}\,(\zeta_{(\widehat{z})},\,v_{\;\widehat{z}})\right\rangle_{\widehat{\mathbb{CS}}}\right)}{\widehat{r}^4({\widehat{z}})}\,
=\,\frac{1\,+\,v(p)}{1\,-\,v(p)}\,{\mathrm{R}}(p)\,,
\end{equation}
\vskip 0.3cm\noindent
where, $v(p)\,:=\,\lim_{\widehat{z}\,\longrightarrow\,0}v_{\hat{z}}$ is the limiting relative velocity of the physical observer with respect to the FLRW observer\footnote{In the case we assume that this relative velocity is $\widehat{z}$-dependent, a possibility that can be considered when optimizing the appropriate choice of the observer when gathering data at a specific redshift.}. Notice that the factor $\frac{1\,+\,v_{\;\widehat{z}}}{1\,-\,v_{\;\widehat{z}}}$ is bounded away from zero\footnote{The aberration factor is a small perturbation of $1$ since the typical peculiar velocities $v_{\;\widehat{z}}$ in the pre-homogeneity region are of the order of a few hundreds  $\mathrm{Km}\;\,s^{- 1}$. Thus, for $c=1$,\, $v_{\;\widehat{z}}\,\approx\,10^{\,- 3}$.}, 
which allows us to replace  the original ${r}({\widehat{z}})$ and $\widehat{r}({\widehat{z}}$\;  $\,\longrightarrow\,0^+$ limits with the more expressive 
${\widehat{z}}\,\longrightarrow\,0^+$ limit. The comparison of (\ref{AreaDistasympFLRW}) with its FLRW counterpart (\ref{FLRWAreaasymp13}) is immediate, and we get
 
\begin{proposition}
\label{flucDisCurv}
For low reference FLRW redshift $\widehat{z}$, the fluctuation between the directional average on $\widehat{\mathbb{CS}}$ of the physical area distance $D^2_{\,\widehat{z}\,(i)}\,(\zeta_{(\widehat{z})},\,v_{\;\widehat{z}})$ and of its FLRW counterpart 
$\widehat{D}^2_{\,\widehat{z}}$ is related to the corresponding spacetime scalar curvature fluctuation by the asymptotics 
\begin{equation}
\label{FluctAreaDistasympt2}
\frac{\left\langle\,D^2_{\,\widehat{z}}\,(\zeta_{(\widehat{z})},\,v_{\;\widehat{z}})\,-\,
\widehat{D}^2_{\,\widehat{z}}\right\rangle_{\widehat{\mathbb{CS}}}}
{\widehat{r}^2(\widehat{z})}\,
=\,\frac{\widehat{r}^2(\widehat{z})}{72}\,\left(\widehat{\mathrm{R}}(p)\,- \frac{1\,+\,v(p)}{1\,-\,v(p)}\,{\mathrm{R}}(p)\,+\,\ldots\right)\,,
\end{equation}
and by the corresponding $\widehat{z}\,\longrightarrow\,0^+$ limit
\begin{equation}
\label{FluctAreaDistasympFLRW}
\lim_{\widehat{z}\,\longrightarrow\,0^+}\,
\frac{\left\langle\,D^2_{\,\widehat{z}}\,(\zeta_{(\widehat{z})},\,v_{\;\widehat{z}})\,-\,
\widehat{D}^2_{\,\widehat{z}}\right\rangle_{\widehat{\mathbb{CS}}}}
{\widehat{r}^4(\widehat{z})}\,
=\,\frac{1}{72}\,\left(\widehat{\mathrm{R}}(p)\,- \frac{1\,+\,v(p)}{1\,-\,v(p)}\,{\mathrm{R}}(p)\right)\,. \;\;\;\square
\end{equation}
\end{proposition}

Since at a given $\widehat{z}$, the squared FLRW area distance  (\ref{FLRWDi}) is constant and given by 
\begin{equation}
\label{FLRWDiRefer}
\widehat{D}^2_{\hat{z}}\,:=\,\frac{d\mu_{\widehat{h}_{\hat{z}}}}{d\mu_{\widehat{\mathbb{S}}^2}}\,=\, \frac{f^2\left(\widehat{r}(\,\widehat{z})\right)}{\left(1\,+\,\widehat{z}\right)^2}\,,
\end{equation}
we can equivalently rewrite the left member of (\ref{FluctAreaDistasympFLRW}) as the relative fluctuation of $D^2_{\,\widehat{z}}\,(\zeta_{(\widehat{z})},\,v_{\;\widehat{z}})$ with respect to its FLRW counterpart, \emph{i.e.}
 \begin{equation}
\label{FluctAreaDist2}
\lim_{\widehat{z}\,\longrightarrow\,0^+}\,\frac{f^2\left(\widehat{r}(\,\widehat{z})\right)}{\left(1\,+\,\widehat{z}\right)^2\,\widehat{r}^4(\,\widehat{z})}
\left\langle\,\frac{D^2_{\,\widehat{z}}\,(\zeta_{(\widehat{z})},\,v_{\;\widehat{z}})\,-\,
\widehat{D}^2_{\,\widehat{z}}}{\widehat{D}^2_{\,\widehat{z}}}\right\rangle_{\widehat{\mathbb{CS}}}
\,
=\,\frac{1}{72}\,\left(\widehat{\mathrm{R}}(p)\,- \frac{1\,+\,v(p)}{1\,-\,v(p)}\,{\mathrm{R}}(p)\right)\,.
\end{equation}
This latter expression can be made more explicit from the observational point of view if we introduce the celestial sphere average of the mean fluctuation and the mean square of  the physical area distance $D_{\,\widehat{z}}\,(\zeta_{(\widehat{z})},\,v_{\;\widehat{z}})$ with respect to the reference FLRW area distance $\widehat{D}^2_{\,\widehat{z}}$, 
\begin{equation}
\label{meanFunc0}
\delta^{\;(1)}_{\widehat{\mathbb{CS}}}(\,\widehat{z}\,)\,:=\,
\frac{1}{4\pi}\,\int_{\widehat{\mathbb{CS}}}\,
\frac{\left(D_{\hat{z}}(\zeta_{(\,\widehat{z})}(y),\,v_{\;\widehat{z}})\,-\,\widehat{D}_{\hat{z}}(y)\right)}{\widehat{D}_{\hat{z}}(y)}\,
d\mu_{\widehat{\mathbb{S}}^2}(y)\,,
\end{equation} 
and
\begin{equation}
\label{meansquareFunc0}
\delta^{\;(2)}_{\widehat{\mathbb{CS}}}(\,\widehat{z}\,)\,:=\,
\frac{1}{4\pi}\,\int_{\widehat{\mathbb{CS}}}\,
\frac{\left(D_{\hat{z}}(\zeta_{(\widehat{z})}(y),\,v_{\;\widehat{z}})\,-\,\widehat{D}_{\hat{z}}(y)\right)^2}{\widehat{D}^2_{\hat{z}}(y)}\,
d\mu_{\widehat{\mathbb{S}}^2}(y)\,.
\end{equation}
The variance associated with the relative fluctuation $(D_{\hat{z}}(\zeta_{(\,\widehat{z})},\,v_{\;\widehat{z}})\,-\,\widehat{D}_{\hat{z}})/(\widehat{D}_{\hat{z}})$ is given by
\begin{equation}
\label{samplevar0}
\frac{1}{4\pi}
\int_{\widehat{\mathbb{CS}}}
\left(\frac{D_{\hat{z}}(\zeta_{(\widehat{z})},\,v_{\;\widehat{z}})-\widehat{D}_{\hat{z}}}{\widehat{D}_{\hat{z}}}\,-\,\delta^{\;(1)}_{\widehat{\mathbb{CS}}}(\,\widehat{z}\,)\right)^2
d\mu_{\widehat{\mathbb{S}}^2}\,=\,\delta^{\;(2)}_{\widehat{\mathbb{CS}}}(\,\widehat{z}\,)\,-\,\left(\delta^{\;(1)}_{\widehat{\mathbb{CS}}}(\,\widehat{z}\,) \right)^2\,.
\end{equation}
\vskip 0.3cm\noindent
Moreover, as is easily verified, we have
 
\begin{equation}
\label{formulDiffD2}
\delta^{\;(2)}_{\widehat{\mathbb{CS}}}(\,\widehat{z}\,)\,+\,2\,\delta^{\;(1)}_{\widehat{\mathbb{CS}}}(\,\widehat{z}\,)\,=\,
\left\langle\,\frac{D^2_{\,\widehat{z}}\,(\zeta_{(\widehat{z})},\,v_{\;\widehat{z}})\,-\,
\widehat{D}^2_{\,\widehat{z}}}{\widehat{D}^2_{\,\widehat{z}}}\right\rangle_{\widehat{\mathbb{CS}}}
\,,
\end{equation}
thus, (\ref{FluctAreaDist2}) has the equivalent representation
\begin{equation}
\label{FluctAreaDist3}
\lim_{\widehat{z}\,\longrightarrow\,0^+}\,\frac{f^2\left(\widehat{r}(\,\widehat{z})\right)}{\left(1\,+\,\widehat{z}\right)^2\,\widehat{r}^4(\,\widehat{z})} \left(\delta^{\;(2)}_{\widehat{\mathbb{CS}}}(\,\widehat{z}\,)\,+\,2\,\delta^{\;(1)}_{\widehat{\mathbb{CS}}}(\,\widehat{z}\,)  \right)\,
=\,\frac{1}{72}\,\left(\widehat{\mathrm{R}}(p)\,- \frac{1\,+\,v(p)}{1\,-\,v(p)}\,{\mathrm{R}}(p)\right)\,.
\end{equation}
\vskip 0.3cm\noindent
If we specialize our comparison between the physical and the FLRW reference spacetimes to the case where the FLRW has flat spatial sections, then $f^2(\widehat{r}(\,\widehat{z}))\,=\,\widehat{r}^2(\,\widehat{z})$. Moreover, since the relation (\ref{FluctAreaDist2}) involves the $\widehat{z}\,\longrightarrow\,0^+$ limit, we can evaluate the comoving radius $\widehat{r}(\,\widehat{z})$ by employing the standard  Hubble formula $\widehat{z}=H_0\,\widehat{r}(\,\widehat{z})$. In such a framework, the previous asymptotic analysis and  Proposition \ref{flucDisCurv} imply the following result.

\begin{theorem}  
\label{Theasympt1}
Let $\{\Sigma_{\widehat{z}(i)}\}$   be the sequence of sky sections foliating the physical past lightcone region  $\mathbb{W}(\gamma(\tau))\,\cap\,{\mathcal{C}}^-(p, {g})$, describing the transition to the cosmological uncoupling in the cluster region,  and converging to the observation event $(p, \dot\gamma(p))$. Similarly, we denote by  $\{\widehat{\Sigma}_{\widehat{z}(i)}\}$ the corresponding sequence of sky sections foliating the FLRW past lightcone region $\widehat{\mathbb{W}}(\widehat\gamma(\widehat\tau))\,\cap\,{\widehat{\mathcal{C}}}^-(p, \widehat{g})$, and converging to the FLRW observer $(p, \widehat{\dot\gamma}(p))$. We can associate with the sequence  $\{\Sigma_{\widehat{z}(i)}\}$ a corresponding sequence of points ${p}_{(i)}=\,{\gamma}({\tau}_{(i)}\,=\,-\,r(\,\widehat{z}\,(i)))$ on the physical observer's world line segment ${\gamma}(\tau)$,\; $-\,\epsilon\,\leq\,\tau\,\leq\,0$,\, such that $({p}_{(i)},\,\dot\gamma(\tau_{(i)}))\,\longrightarrow\,(p, \dot\gamma(p))$\, as\, $i\,\rightarrow\, \infty$. We denote by $\widehat{p}_{(i)}=\,\widehat{\gamma}(\widehat{\tau}_{(i)}\,=\,-\,\widehat{r}(\,\widehat{z}\,(i)))$,\, with $(\widehat{p}_{(i)},\,\widehat{\dot\gamma}(\widehat\tau_{(i)}))\,\longrightarrow\,(p, \widehat{\dot\gamma}(p))$\, as\, $i\,\rightarrow\, \infty$,\,    the corresponding sequence for the FLRW observer's world line segment $\widehat{\gamma}(\widehat\tau)$,\; $-\,\epsilon\,\leq\,\widehat\tau\,\leq\,0$. Note that in general $\widehat{p}_{(i)}\,\not=\,{p}_{(i)}$. If $\{\widehat{\mathbb{CS}}_{\,\widehat{z}\,(i)}\}$ is the sequence of FLRW celestial spheres associated with $\{\widehat{\Sigma}_{\widehat{z}(i)}\}$, then we have the asymptotic expansion 

\begin{equation}
\label{AsyFluctAreaDist4}
\frac{\delta^{\;(2)}_{\widehat{\mathbb{CS}}}(\,\widehat{z}\,(i))\,+\,2\,\delta^{\;(1)}_{\widehat{\mathbb{CS}}}(\,\widehat{z}\,(i))}{\widehat{z}^{\,2}\,(i)\,\left(1\,+\,\widehat{z}\,(i)\right)^2}\,
=\,\frac{1}{72\,H_0^2}\,\left(\widehat{\mathrm{R}}(\widehat{p})\,- \tfrac{1\,+\,v_{\widehat{z}\,(i)}(p_{(i)})}{1\,-\,v_{\widehat{z}\,(i)}(p_{(i)})}\,{\mathrm{R}}(p)\,+\,\ldots\,\right)\,,
\end{equation}
where $H_0$ is the present value of the Hubble parameter, and where $\ldots$\; represents higher order corrections that can be expressed in terms of powers of the (spacetime) Riemann tensor and its (covariant) derivatives. As $\widehat{z}\,\longrightarrow\,0^+$, this asymptotics provides 
\begin{equation}
\label{FluctAreaDist4}
\lim_{\widehat{z}\,\longrightarrow\,0^+}\,\frac{\delta^{\;(2)}_{\widehat{\mathbb{CS}}}(\,\widehat{z}\,)\,+\,2\,\delta^{\;(1)}_{\widehat{\mathbb{CS}}}(\,\widehat{z})}{\widehat{z}^{\,2}\,\left(1\,+\,\widehat{z}\right)^2}\,
=\,\frac{1}{72\,H_0^2}\,\left(\widehat{\mathrm{R}}(p)\,- \frac{1\,+\,v(p)}{1\,-\,v(p)}\,{\mathrm{R}}(p)\right)\,.
\end{equation}
\end{theorem}
\begin{proof}
By taking into account the assumed spatial flatness of the reference FLRW model and the low redshift Hubble relation, (\ref{FluctAreaDist4}) follows by a straightforward rewriting of (\ref{FluctAreaDist2}). The asymptotics (\ref{AsyFluctAreaDist4}) follows from (\ref{Areaasymp1}) and (\ref{FLRWAreaasymp1}) by replacing the areas $A(\Sigma_{\widehat{z}\,(i)})$ and $\widehat{A}(\widehat{\Sigma}_{\widehat{z}\,(i)})$ with the corresponding celestial sphere averages of the respective (squared) area distances (see (\ref{PhysAreazeta(ii)}),\, (\ref{FLRWPhysAreazeta(ii)}), and the normalization (\ref{Dnotation})).
\end{proof}
\begin{remark}
As we approach the critical FLRW redshift $\widehat{z}_{(c)}$ signaling the transition to the cosmological uncoupling in the cluster region, the asymptotics (\ref{AsyFluctAreaDist4}) can be rewritten in the more expressive form 

\begin{equation}
\label{AsyFluctAreaDist4CRIT}
\frac{\delta^{\;(2)}_{\widehat{\mathbb{CS}}}(\,\widehat{z}_{(c)})\,+\,2\,\delta^{\;(1)}_{\widehat{\mathbb{CS}}}(\,\widehat{z}_{(c)})}{\left(1\,+\,\widehat{z}_{(c)}\right)^2}\,
=\,\frac{\widehat{z}^{\,2}_{(c)}}{72\,H_0^2}\,\left(\widehat{\mathrm{R}}({p})\,- \tfrac{1\,+\,v_{\widehat{z}_{(c)}}(p_{(c)})}{1\,-\,v_{\widehat{z}_{(c)}}(p_{(c)})}\,{\mathrm{R}}(p)\,+\,\ldots\,\right)\,.\;\;\;\;\;\;\square
\end{equation}
\end{remark}

\section{Area distance fluctuations on the past lightcone}
\label{ArDisFluc}
The averages $\delta^{\;(1)}_{\widehat{\mathbb{CS}}}(\,\widehat{z}\,(i))$ and $\delta^{\;(2)}_{\widehat{\mathbb{CS}}}(\,\widehat{z}\,(i))$ are the mean and mean square relative fluctuations of the physical and FLRW area distance. This observation suggests that the asymptotics (\ref{AsyFluctAreaDist4}) and the limit (\ref{FluctAreaDist4}) are deeply connected with the random fluctuations of area distances. In the situation at hand, and for the large redshifts $\widehat{z}$ that probe the homogeneity region, this possibility is naturally related to Neyman's weak cosmological principle, and it fits well within the standard $\Lambda\mathrm{CDM}$ model that assumes a Gaussian power spectrum of CBM density fluctuations. However, this assumption of Gaussianity, which drives the FLRW structure formation paradigm, fails in the pre-homogeneity region due to strong gravitational clustering. In particular, the area distance fluctuations are distributed in a rather complex and unpredictable way when we approach the critical redshift  $\widehat{z}_{(c)}$, and we enter the cosmological uncoupling region where the $\widehat{z}\,\rightarrow\,0$ limit is attained. Yet, the structure of the left member of the (\ref{FluctAreaDist4}) holds a strong connection with the moment generating function $\mathbb{M}_{N(m,\sigma^2)}$ associated with a normal distribution  of the area distance fluctuations, with mean $m\,:=\,\delta^{\;(1)}_{\widehat{\mathbb{CS}}}(\,\widehat{z})$ and variance $\sigma^2\,:=\,\delta^{\;(2)}_{\widehat{\mathbb{CS}}}(\,\widehat{z})\,-\,(\delta^{\;(1)}_{\widehat{\mathbb{CS}}}(\,\widehat{z}))^2$. To wit, for $\xi\,\in\,\mathbb{R}$, we have 
\begin{align}
\label{MomGenFunct}
\xi\,\longmapsto\,\mathbb{M}_{N(m,\sigma^2)}\,(\xi)\,&:=\,\exp\,\left[\xi\,m\,+\,\frac{\sigma^2\,\xi^2}{2}\right]\nonumber\\
\\
&=\,\exp\,\left[\delta^{\;(1)}_{\widehat{\mathbb{CS}}}(\,\widehat{z})\,\xi\,+\,\frac{\delta^{\;(2)}_{\widehat{\mathbb{CS}}}(\,\widehat{z})\,\xi^2}{2}\,-\,\frac{1}{2}\left(\delta^{\;(1)}_{\widehat{\mathbb{CS}}}(\,\widehat{z})\,\xi \right)^2\right]\,\nonumber
\end{align} 
a relation suggesting that upon stabilization (\emph{i.e.} normalizing the area distance fluctuations to have mean value $0$, and possibly variance $1$), we may recover Gaussianity by exploiting central limit theorem arguments. 

Let us start   by considering the sequence of relative fluctuations
\begin{equation}
\label{random1}
\left\{\left.Y_{(\,\widehat{z}\,(i))}\,:=\,\frac{D_{\widehat{z}\,(i)}(\zeta_{(\,\widehat{z}\,(i))},\,v_{\;\widehat{z}\,(i)})\,-\,\widehat{D}_{\widehat{z}\,(i)}}{\widehat{D}_{\widehat{z}\,(i)}}\;\right|\;i\,\geq\,1\;\right\}\,,
\end{equation}   
associated with the corresponding sequence of sky sections $\{\Sigma_{\widehat{z}(i)}\}$ \,and\, $\{\widehat{\Sigma}_{\widehat{z}(i)}\}$ (and associated celestial spheres). From the point of view of the physical past lightcone ${\mathcal{C}}^-(p, {g})$, when moving from the initial sky section $\Sigma_{\widehat{z}(1)}$ to $\Sigma_{\widehat{z}(k)}$, with $k\,>\,1$ ,\; the interpolating region $\Xi(\widehat{z}(1),\,\widehat{z}(k))\, \subset\,{\mathcal{C}}^-(p, {g})$ is subjected to the vagaries of the gravitational dynamics of sources crossing $\Xi(\widehat{z}(1)\,,\,\widehat{z}(k))\,$ and interacting with its null geodesic generators, \emph{i.e.}, altering the map $\exp_p|_{\Xi(\widehat{z}(1),\,\widehat{z}(k))\,}$. This process, in the cluster region surrounding the observational event $p$,  is random in the 
real-world context: \emph{unpredictable data interact  with $\Xi(\widehat{z}(1),\,\widehat{z}(k))\,$, data that are not implied by the information on 
the initial sky section $\Sigma_{\widehat{z}(1)}$, marginally coupled with cosmological expansion}. The state of $\Sigma_{\widehat{z}(1)}$ (hence also the distribution of the corresponding area distance $D_{\widehat{z}\,(1)}$ of its points as evaluated  from $p$) does not determine the configuration of $\Sigma_{\widehat{z}(k)}$ for $k\,>\,1$.   
To put the situation at hand in a colorful way\footnote{This remark owes much to discussions with George Ellis.}, \emph{the Laplace's demon argument ignores the real world and, in particular, our physical past lightcone context in the pre-homogeneity region}. A strong, non-Gaussian randomness of the terms $Y\,(\,\widehat{z}\,(i))$ of the sequence (\ref{random1})  is at work here. Yet, from the point of view of the directional FLRW celestial sphere $\widehat{\mathbb{CS}}$, it is reasonable to assume that the area distance fluctuations are identically distributed random variables on $\widehat{\mathbb{CS}}$. These remarks motivate the following

\begin{assumption}
In the low-redshift $\widehat{z}$ regime characterizing the cluster region,
the sequence of relative area distance fluctuations  (\ref{random1}) is interpreted as a sequence of mutually independent, square integrable random variables 
\begin{align}
\widehat{\mathbb{CS}}\,&\longrightarrow\,\mathbb{R}\\
(\widehat\theta,\,\widehat\phi)\,&\longmapsto\,Y_{(\,\widehat{z}\,(i))}(\widehat\theta,\,\widehat\phi)\,,\nonumber
\end{align}
which are identically distributed according to the probability measure $\mathbb{P}_{\widehat{\mathbb{CS}}}$\; on $\widehat{\mathbb{CS}}$\, defined 
by
\begin{equation}
\label{CSprobab}
\mathbb{P}_{\widehat{\mathbb{CS}}}\,\left(d\mu_{\widehat{\mathbb{S}}^2}\right)\,:=\,
\frac{1}{4\pi}\,d\mu_{\widehat{\mathbb{S}}^2}\,.
\end{equation}
 Explicitly, for  $I\,	\subset\,\mathbb{R}$, let
\begin{equation}
{Y}^{-\,1}_{(\,\widehat{z}\,(i))}(I)\,:=\,\left\{\left.(\,\widehat\theta,\,\widehat\phi\,)\,\in\,\widehat{\mathbb{CS}}\,\,\right|\,Y_{(\,\widehat{z}\,(i))}(\widehat\theta,\,\widehat\phi)\,\in\,I \right\}
\end{equation}
denote the subset of directions $(\,\widehat\theta,\,\widehat\phi\,)\,\in\,\widehat{\mathbb{CS}}$ whose corresponding area distance relative fluctuations  $Y_{(\,\widehat{z}\,(i))}(\widehat\theta,\,\widehat\phi)$ \,take values in the interval \,$I$. The $Y_{(\,\widehat{z}\,(i))}$     are assumed to be independently distributed  under the push-forward measure
\begin{equation}
\label{FlucDisDistr}
\mathbb{P}_{\widehat{\mathbb{CS}}}\left({Y}^{-\,1}_{(\,\widehat{z}\,(i))}(I)\right)\,=\,
 \frac{1}{4\pi}\,\int_{{Y}^{-\,1}_{(\,\widehat{z}\,(i))}(I)}\, d\mu_{\widehat{\mathbb{S}}^2}\,.
\end{equation}
\end{assumption}   

Notice that, according to (\ref{meanFunc0}) and (\ref{meansquareFunc0}), the mean and the mean square fluctuations associated with (\ref{random1}) are respectively given by 
\begin{equation}
\label{meanFunc0i}
\delta^{\;(1)}_{\widehat{\mathbb{CS}}}(\,\widehat{z}\,(i))\,:=\,
\frac{1}{4\pi}\,\int_{\widehat{\mathbb{CS}}}\,
Y_{(\,\widehat{z}\,(i))}\,
d\mu_{\widehat{\mathbb{S}}^2}\,,
\end{equation} 
and
\begin{equation}
\label{meansquareFunc0i}
\delta^{\;(2)}_{\widehat{\mathbb{CS}}}(\,\widehat{z}\,(i))\,:=\,
\frac{1}{4\pi}\,\int_{\widehat{\mathbb{CS}}}\,
Y^2_{(\,\widehat{z}\,(i))}\,
d\mu_{\widehat{\mathbb{S}}^2}\,.
\end{equation} 
\vskip 0.3cm\noindent
The corresponding variance is provided by (\ref{samplevar0})\,, \emph{i.e.}  
\begin{align}
\label{samplevar0i}
\mathrm{Var}_{\widehat{\mathbb{CS}}}(Y_{(\,\widehat{z}\,(i))})\,&:=\,
\frac{1}{4\pi}
\int_{\widehat{\mathbb{CS}}}
\left(Y_{(\,\widehat{z}\,(i))}\,-\,\delta^{\;(1)}_{\widehat{\mathbb{CS}}}(\,\widehat{z}\,(i))\right)^2
	d\mu_{\widehat{\mathbb{S}}^2}\nonumber\\
\\
&=\,\delta^{\;(2)}_{\widehat{\mathbb{CS}}}(\,\widehat{z}\,(i))\,-\,\left(\delta^{\;(1)}_{\widehat{\mathbb{CS}}}(\,\widehat{z}\,(i)) \right)^2\,.\nonumber
\end{align}
\vskip 0.3cm\noindent
As for what concerns the behavior of the sequence $\{Y_{(\,\widehat{z}\,(i))}\}$ we have a first result that is an immediate consequence of Kolmogorov's Strong Law (of Large Numbers).

\begin{theorem}
\label{AlmSureConv}
Let us assume that the area distance relative fluctuations $Y_{(\,\widehat{z}\,(1))}(\widehat\theta,\,\widehat\phi)$ are integrable and have a finite mean value $\delta^{\;(1)}_{\widehat{\mathbb{CS}}}(\,\widehat{z}\,(1))\,=\,\delta^{\;(1)}_{\widehat{\mathbb{CS}}}(\,\widehat{z}\,(c))$ over the the cosmological decoupling surface $\Sigma_{\widehat{z}(1)}\,:=\,\partial\,V_{(c)}(p)$.\, 
For $(\widehat\theta,\,\widehat\phi)\,\in\,\widehat{\mathbb{CS}}$, let 
\begin{equation}
\label{empiricalQ}
\widehat{\mathbb{CS}}\,\ni\,(\widehat\theta,\,\widehat\phi)\,\longrightarrow\,
Q_n (\widehat\theta,\,\widehat\phi)\,:=\,\frac{1}{n}\,\sum_{i=1}^n\,Y_{(\,\widehat{z}\,(i))}(\widehat\theta,\,\widehat\phi)\,\in\,\mathbb{R}
\end{equation}
denote the empirical mean associated with the sequence $\{Y_{(\,\widehat{z}\,(i))}(\widehat\theta,\,\widehat\phi)\;|\;i\,\geq\,1\}$, where 
\begin{equation}
Y_{(\,\widehat{z}\,(1))}(\widehat\theta,\,\widehat\phi)\,:=\, \left.\frac{D_{\widehat{z}\,(c)}(\zeta_{(\,\widehat{z}\,(c))},\,v_{\;\widehat{z}\,(c)})\,-\,\widehat{D}_{\widehat{z}\,(c)}}{\widehat{D}_{\widehat{z}\,(c)}}\right|_{(\widehat\theta,\,\widehat\phi)}
\end{equation}
is the area distance fluctuation, in the direction $(\widehat\theta,\,\widehat\phi)\in\,\widehat{\mathbb{CS}}$,\,  evaluated at the cluster physical radius surface $\partial\,V_{(c)}(p)\,=:\,\Sigma_{\widehat{z}\,(1)}$. Then we have the $\mathbb{P}$-almost sure convergence result
\begin{equation}
\label{AlmSureP}
\mathbb{P}_{\widehat{\mathbb{CS}}}\,\left(\lim_{n\rightarrow\,\infty}\,\frac{1}{n}\,
\sum_{i=1}^n\,Y_{(\,\widehat{z}\,(i))}\,=\, \delta^{\;(1)}_{\widehat{\mathbb{CS}}}(\,\widehat{z}\,(c))  \right)\,=\,1\,,
\end{equation}
where 
\begin{equation}
\label{meanFunc0Y}
\delta^{\;(1)}_{\widehat{\mathbb{CS}}}(\,\widehat{z}\,(c))\,:=\,
\frac{1}{4\pi}\,\int_{\widehat{\mathbb{CS}}}\,
Y_{(\,\widehat{z}\,(1))}\,
d\mu_{\widehat{\mathbb{S}}^2}\,
\end{equation}
\vskip 0.3cm\noindent 
is the mean area distance fluctuation evaluated at  $\partial\,V_{(c)}(p)\,=:\,\Sigma_{\widehat{z}\,(1)}$. 
\end{theorem}
\begin{proof}
By hypothesis, the random variables $Y_{(\,\widehat{z}\,(i))}$ are mutually $\mathbb{P}_{\widehat{\mathbb{CS}}}$-independent, integrable and identically distributed.  Without loss in generality, we can shift them according to (see (\ref{meanFunc0i}))
\begin{equation}
\overline{Y}_{(\,\widehat{z}\,(i))}\,:=\, Y_{(\,\widehat{z}\,(i))}\,-\,\delta^{\;(1)}_{\widehat{\mathbb{CS}}}(\,\widehat{z}\,(i))\,,
\end{equation}
so that they have zero mean. In a sound-bite version, we can exploit Komogorov's Strong Law \cite{Stroock} to conclude that, as $n\rightarrow\infty$,\;the empirical mean $Q_n$ defined by (\ref{empiricalQ}) converges $\mathbb{P}$-almost surely to the average $\delta^{\;(1)}_{\widehat{\mathbb{CS}}}(\,\widehat{z}\,(c))$\, of\, 
$Y_{(\,\widehat{z}\,(1))}$,\,      according\footnote{Actually, the convergence is also in $L^1(\mathbb{P}_{\widehat{\mathbb{CS}}},\,\mathbb{R})$;\;see \cite{Stroock}.} to (\ref{AlmSureP}). Yet, to understand the emergence of 
the role of the cosmological decoupling surface $\partial\,V_{(c)}(p)$ in controlling the behavior of the sequence of random variables $\{\overline{Y}_{(\,\widehat{z}\,(i))}\}$, we need to provide a few details of the proof of the theorem. To this end, let us consider the interval $I_i\,:=\,(i,\infty)$,\;$i\,\in\,\mathbb{Z}^{+}$, and  associate with the generic  $\overline{Y}_{(\,\widehat{z}\,(i))}$ the subset\footnote{In Probability theory, the $\{A_i\}$ characterize  \emph{events} in the probability space $(\widehat{\mathbb{CS}},\,\sigma,\,\mathbb{P}_{\widehat{\mathbb{CS}}})$, where $\sigma$ is the Borel $\sigma$-algebra over the directional celestial sphere $\widehat{\mathbb{CS}}$. In this section, when we use the word "event," we refer to it in its probabilistic sense. } of directions $(\,\widehat\theta,\,\widehat\phi\,)\,\in\,\widehat{\mathbb{CS}}$ whose corresponding area distance  fluctuations  $\overline{Y}_{(\,\widehat{z}\,(i))}(\widehat\theta,\,\widehat\phi)$ \,are such that $|\overline{Y}_{(\,\widehat{z}\,(i))}(\widehat\theta,\,\widehat\phi)|\,\in\, I_i$, \emph{i.e.}, $|\overline{Y}_{(\,\widehat{z}\,(i))}(\widehat\theta,\,\widehat\phi)|\,>\,i$. Namely, the subset 
of $\widehat{\mathbb{CS}}$ defined  by the inverse image
\begin{equation}
\label{InvImage}
A_i\,(\widehat\theta, \widehat\phi)\,:=\,\left|\overline{Y}_{(\,\widehat{z}\,(i))}\right|^{-\,1}\left(I_i \right)\,=\,
\left\{\left.\left(\widehat\theta, \widehat\phi \right)\,\in\,
\widehat{\mathbb{CS}}\;\,\right|\;      |\overline{Y}_{(\,\widehat{z}\,(i))}\,(\widehat\theta, \widehat\phi)|\,>\,i\right\}\,.
\end{equation}
 These events are distributed according the push  forward measure $\mathbb{P}_{\widehat{\mathbb{CS}}}(|\overline{Y}_{(\,\widehat{z}\,(i))}|^{-\,1}(I_i))$ (see (\ref{FlucDisDistr})), and we can write 
\begin{equation}
\label{FlucDisDistrA}
\mathbb{P}_{\widehat{\mathbb{CS}}}\left(|\overline{Y}_{(\,\widehat{z}\,(i))}|\,>\,i  \right)\,:=\,
\mathbb{P}_{\widehat{\mathbb{CS}}}\left(A_i\right)\,=\,
 \frac{1}{4\pi}\,\int_{|\overline{Y}_{(\,\widehat{z}\,(i))}|^{-\,1}(I_i)}\, d\mu_{\widehat{\mathbb{S}}^2}\,.
\end{equation}
Let us remark that the event of having larger and larger fluctuations 
$\overline{Y}_{(\,\widehat{z}\,(i))}$,\; \emph{i.e.}, to find a (measurable) set of directions $(\widehat\theta, \widehat\phi )$ such that $|\overline{Y}_{(\,\widehat{z}\,(i))(\widehat\theta, \widehat\phi )}|\,>\,i$ for arbitrarily large $i$, is controlled by distribution of the $\limsup$\; $\cap_{k=1}^\infty\,\cup_{i\geq\,k}\,A_i$\, of the sets $A_i\,(\widehat\theta, \widehat\phi)$\,. This $\limsup$  can be characterized as  
\begin{equation}
\limsup_{i\rightarrow\infty}\,A_i\,(\widehat\theta, \widehat\phi)\,:=\,
\left\{\left.\left(\widehat\theta, \widehat\phi \right)\,\in\,
\widehat{\mathbb{CS}}\;\,\right|\;(\widehat\theta, \widehat\phi)\,\in\,
A_i\,(\widehat\theta, \widehat\phi)\,\;\;\;\mathrm{for\;infinitely\;many}\;i  \right\}\,.
\end{equation}
Since the events $\{A_i\}$ are mutually independent, we have that the probability of the event $\limsup_{i\rightarrow\infty}\,A_i$ is either $0$ or $1$. In particular, the Borel-Cantelli Lemma \cite{Stroock}, implies that
\begin{equation}
\mathbb{P}_{\widehat{\mathbb{CS}}}\left( \limsup_{i\rightarrow\infty}\,A_i  \right)\,=\,0\,
\end{equation} 
if 
\begin{equation}
\label{PSeries}
\sum_{i=1}^\infty\,\mathbb{P}_{\widehat{\mathbb{CS}}}\,\left(|\overline{Y}_{(\,\widehat{z}\,(i))}|\,>\,i  \right)\,<\,\infty\,.
\end{equation}
In other words, the convergence of the above series implies that for any sequence of mutually $\mathbb{P}_{\widehat{\mathbb{CS}}}$-independent events $A_i\,(\widehat\theta, \widehat\phi)$,\;
$\mathbb{P}_{\widehat{\mathbb{CS}}}$-almost every direction $(\widehat\theta, \widehat\phi)\,\in\,\widehat{\mathbb{CS}}$ is in at most finitely many $A_i\,(\widehat\theta, \widehat\phi)$,\;\emph{i.e.}, sampling the $Y_{(\,\widehat{z}\,(i))}(\widehat\theta, \widehat\phi)$ over all possible the celestial sphere directions, we may have $|\overline{Y}_{(\,\widehat{z}\,(i))}(\widehat\theta, \widehat\phi)|\,>\,i$ only in at most finitely many cases. To prove convergence and the theorem, we  start by bounding (\ref{PSeries}) according to\footnote{We adapt to our particular case the elegant proof of Kolmogorov's strong law provided by \cite{Stroock}.}  
\begin{equation}
\sum_{i=1}^\infty\,\mathbb{P}_{\widehat{\mathbb{CS}}}\,\left(|\overline{Y}_{(\,\widehat{z}\,(i))}|\,>\,i  \right)\,\leq\,
\sum_{i=1}^\infty\,\int_{i-1}^i\,\mathbb{P}_{\widehat{\mathbb{CS}}}\,\left(|\overline{Y}_{(\,\widehat{z}\,(i))}|\,>\,t  \right)\,dt\,.
\end{equation} 
Since the $Y_{(\,\widehat{z}\,(i))}$ are identically distributed, the probability $\mathbb{P}_{\widehat{\mathbb{CS}}}\left(|\overline{Y}_{(\,\widehat{z}\,(i))}|\,>\,t \right)$ that 
the $i-th$\, relative fluctuation $Y_{(\,\widehat{z}\,(i))}$ is such that $|\overline{Y}_{(\,\widehat{z}\,(i))}|$ is larger that $t$, \,with\, $i-1\leq\,t\,\leq\,i$, is the same for all $i$. In particular, we have  
\begin{equation}
\mathbb{P}_{\widehat{\mathbb{CS}}}\left(|\overline{Y}_{(\,\widehat{z}\,(i))}|\,>\,t \right)\,=\,\mathbb{P}_{\widehat{\mathbb{CS}}}\left(|\overline{Y}_{(\,\widehat{z}\,(1))}|\,>\,t   \right)\,,
\end{equation}
and we can write 

\begin{align}
\sum_{i=1}^\infty\,\mathbb{P}_{\widehat{\mathbb{CS}}}\,\left(|\overline{Y}_{(\,\widehat{z}\,(i))}|\,>\,i  \right)\,&\leq\,
\sum_{i=1}^\infty\,\int_{i-1}^i\,\mathbb{P}_{\widehat{\mathbb{CS}}}\,\left(|\overline{Y}_{(\,\widehat{z}\,(1))}|\,>\,t  \right)\,dt\\
&=\,\,\int_{0}^\infty\,\mathbb{P}_{\widehat{\mathbb{CS}}}\,\left(|\overline{Y}_{(\,\widehat{z}\,(1))}|\,>\,t  \right)\,dt\nonumber\,.
\end{align} 
From the relation\footnote{The relation (\ref{intRel}) is a particular case of a more general relation that holds on any $\sigma$-finite measure space and for any non-negative measurable function-see \emph{e.g.} \cite{Stroock}.}
\begin{equation}
\label{intRel}
\int_{\widehat{\mathbb{CS}}}\,\left|\overline{Y}_{(\,\widehat{z}\,(1))}\left(\widehat\theta, \widehat\phi \right)\right|^\alpha\,\mathbb{P}_{\widehat{\mathbb{CS}}}\,\left(d\mu_{\widehat{\mathbb{S}}^2}\right)\,=\,\alpha\,\int_0^\infty\,t^{\alpha-1}\,\mathbb{P}_{\widehat{\mathbb{CS}}}\,\left(|\overline{Y}_{(\,\widehat{z}\,(1))}|\,>\,t  \right)\,dt
\end{equation}
which holds for any $\alpha\,\in\,(0, \infty)$, we get (for $\alpha=1$)
\begin{equation}
\int_{0}^\infty\,\mathbb{P}_{\widehat{\mathbb{CS}}}\,\left(|\overline{Y}_{(\,\widehat{z}\,(1))}|\,>\,t  \right)\,dt\,=\,
\mathbb{E}^{\mathbb{P}}\left[\left|\overline{Y}_{(\,\widehat{z}\,(1))}(\widehat\theta, \widehat\phi)\right|\right]\,,
\end{equation}
where $\mathbb{E}^{\mathbb{P}}$ denotes the expectation value with respect to $\mathbb{P}_{\widehat{\mathbb{CS}}}(d\mu_{\widehat{\mathbb{S}}^2})\,=\,d\mu_{\widehat{\mathbb{S}}^2}/4\pi$. Thus,
\begin{equation}
\label{partialfinal}
\sum_{i=1}^\infty\,\mathbb{P}_{\widehat{\mathbb{CS}}}\,\left(|\overline{Y}_{(\,\widehat{z}\,(i))}|\,>\,i  \right)\,\leq\,\mathbb{E}^{\mathbb{P}}\left[\left|\overline{Y}_{(\,\widehat{z}\,(1))}(\widehat\theta, \widehat\phi)\right|\right]\,<\,\infty.
\end{equation}
To exploit this result for controlling  the $\mathbb{P}$-almost sure convergence (\ref{AlmSureP}), let
\begin{equation}
\mathbb{E}^{\mathbb{P}}\left[\overline{Y}_{(\,\widehat{z}\,(i))},\,A_i\right]\,:=\,
\int_{A_i}\, \overline{Y}_{(\,\widehat{z}\,(i))}(\widehat\theta, \widehat\phi)\,
\mathbb{P}_{\widehat{\mathbb{CS}}}\left(d\mu_{\widehat{\mathbb{S}}^2}\right)\,,
\end{equation}
denote the expectation value of $\overline{Y}_{(\,\widehat{z}\,(i))}$ over the region $A_i\,\subset\,\widehat{\mathbb{CS}}$ (see (\ref{InvImage})), and let us consider the empirical mean
\begin{equation}
\label{EYsequence}
\frac{1}{n}\,\sum_{i=1}^n\,\mathbb{E}^{\mathbb{P}}\left[\overline{Y}_{(\,\widehat{z}\,(i))},\,A_i\right]\,.
\end{equation}
According to (\ref{partialfinal}), the sequence $\{\mathbb{E}^{\mathbb{P}}\left[\overline{Y}_{(\,\widehat{z}\,(i))},\,A_i\right]\}$ converges to $0$, and consequently also  (\ref{EYsequence}) converges to $0$ as $n\,\rightarrow\,\infty$ for almost every direction $(\widehat\theta, \widehat\phi)\,\in\,\widehat{\mathbb{CS}}$. Finally, to prove the $\mathbb{P}$-almost sure convergence (\ref{AlmSureP}), it remains to control the convergence of $\sum_{k=1}^\infty\,\mathbb{E}^{\mathbb{P}}\left[\overline{Y}^2_{(\,\widehat{z}\,(k))},\,A_k\right]/k^2$. This last step easily follows \cite{Stroock} from the bound
\begin{equation}
\sum_{k=1}^\infty\,\frac{\mathbb{E}^{\mathbb{P}}\left[\overline{Y}^2_{(\,\widehat{z}\,(k))},\,A_k\right]}{k^2}\,\leq\,S_0\,\mathbb{E}^{\mathbb{P}}\left[\left|\overline{Y}_{(\,\widehat{z}\,(1))}(\widehat\theta, \widehat\phi)\right|\right]\,<\,\infty,
\end{equation}  
where $S_0:=\sup_{\ell\in\mathbb{Z}^+}\,\left(\ell\,\sum_{n=\ell}^\infty\,\frac{1}{n^2} \right)$. 
\end{proof}

Notwithstanding the control afforded by Theorem \ref{AlmSureConv},
the distribution of the random variables $\{Y_{(\,\widehat{z}\,(i))}\}$ can be quite complex. There is evidence of the presence of a non-trivial skewness \cite{Fanizza2}, indicating an asymmetry of the distribution of area distance fluctuations about its mean value $\delta^{\;(1)}_{\widehat{\mathbb{CS}}}(\,\widehat{z}\,(i))$. As a marker of this asymmetric behavior, one can also use the median\footnote{The interplay between skewness, mean, and the median is complex and not particularly useful since there is no direct connection between skewness and the (complex) relation between mean and median featuring in L\'evy's analysis on the convergence behavior of a collection of random variables. The discussion of this point about the area distance fluctuations is beyond our intention here.} of $Y_{(\,\widehat{z}\,(i))}$, defined \cite{Stroock} as any real number (or interval of real numbers) $a_{(i)}\in\mathbb{R}$ such that

\begin{equation}
a_{(i)}\,:=\,\min\,\left[\mathbb{P}_{\widehat{\mathbb{CS}}}\,\left(Y_{(\,\widehat{z}\,(i))}\right)\,\leq\,a_{(i)},\;
\mathbb{P}_{\widehat{\mathbb{CS}}}\,\left(Y_{(\,\widehat{z}\,(i))}\right)\,\geq\,a_{(i)} \right]\,\geq\,\frac{1}{2}\,.
\end{equation} 
Note that the assumed square summability of the $\{Y_{(\,\widehat{z}\,(i))}\}$ implies some control on the deviation of the median(s) $a_{(i)}$ from the mean value $\delta^{\;(1)}_{\widehat{\mathbb{CS}}}(\,\widehat{z}\,(i))$. We have
\begin{equation}
\left|a_{(i)}\,-\,\delta^{\;(1)}_{\widehat{\mathbb{CS}}}(\,\widehat{z}\,(i))  \right|\, \leq\,\sqrt{2\,\mathrm{Var}_{\widehat{\mathbb{CS}}}(Y_{(\,\widehat{z}\,(i))})}\,.
\end{equation}
\vskip 0.3cm\noindent 
The complexity of the distribution of the random variables $\{Y_{(\,\widehat{z}\,(i))}\}$ and the potential presence of skewness is not surprising, as already stressed, there is no FLRW Laplace's demon propagating the assumed Gaussianity of fluctuations from the observed portion of the CMB last scattering surface to the pre-homogeneity region. The necessity to control all the moments of $\{Y_{(\,\widehat{z}\,(i))}\}$, suggested by the similarity between the moment generating function (\ref{MomGenFunct}) and (\ref{FluctAreaDist4}), 
indicates that to discuss the limiting behavior of the sequence $\{Y_{(\,\widehat{z}\,(i))}\}$, we need to stabilize it with respect to the mean and variance of each term. Thus, following a standard procedure, we associate with (\ref{random1}) the sequence $\{X_{(\,\widehat{z}\,(i))}\;|\,i\,\geq\,1 \}$ of mutually independent, square integrable random variables whose generic term is defined by  
\begin{equation}
\label{RandVarX}
X_{(\,\widehat{z}\,(i))}\,:=\,\left(
Y_{(\,\widehat{z}\,(i))}\,-\,\delta^{\;(1)}_{\widehat{\mathbb{CS}}}(\,\widehat{z}\,(i)\,)\right)\,\mathrm{Var}^{-\,1/2}_{\widehat{\mathbb{CS}}}(Y_{(\,\widehat{z}\,(i))})\,, 
\end{equation} 
with mean $0$  and variance $1$,
\begin{equation}
\left\langle\,X_{(\,\widehat{z}\,(i))}\, \right\rangle_{\widehat{\mathbb{CS}}}\,=\,0\,,\;\;\;
\mathrm{Var}_{\widehat{\mathbb{CS}}}(X_{(\,\widehat{z}\,(i))})\,=\,1\,,\;\;\;\forall\, i\,\geq\,1\,.
\end{equation}
\vskip 0.3cm\noindent
Given the null direction $\widehat{r}(\widehat{z})\widehat{\ell}(\widehat{n})\,\in\,\widehat{C}^{\,-}(T_pM,\,\{\widehat{E}_{(i)}\})$  associated with the sky coordinates $\widehat{n}(\widehat\theta,\,\widehat\phi)\,\in\,\widehat{\mathbb{CS}}$,  we    
introduce on the celestial sphere $\widehat{\mathbb{CS}}$  the random variables defined by the partial sum
\begin{equation}
\label{partsum}
\widehat{\mathbb{CS}}\,\ni\,
\widehat{n}(\widehat\theta,\,\widehat\phi)\,\rightarrow\,
\breve{S}_m(\,\widehat\theta,\,\widehat\phi\,)\,:=\,
\frac{1}{m}\,\sum_{i}^{m}\,X_{(\,\widehat{z}\,(i))}(\,\widehat\theta,\,\widehat\phi\,)\,\in\,\mathbb{R}\,,
\end{equation}
which describes the empirical averages, along the given direction $\widehat{n}(\widehat\theta,\,\widehat\phi)$, of the $m$ random variables $X_{(\,\widehat{z}\,(i))}(\,\widehat\theta,\,\widehat\phi\,)$ evaluated at the points $\exp_p{\widehat{r}(\widehat{z}_{(i)})\widehat{\ell}(\widehat{n}(\widehat\theta,\,\widehat\phi))} \in\, \Sigma_{\widehat{z}(i)}$, for $i\,=\,1,\ldots\,m$. Thus, as we move through  the sequence of sky sections $\{\Sigma_{\widehat{z}(i)}\}_{i=1}^m$,\,the random variable $\breve{S}_m$\, measures the empirical average of the fluctuations of the random variable $Y_{(\,\widehat{z}\,(i))}$ with respect to its mean $\delta^{\;(1)}_{\widehat{\mathbb{CS}}}(\,\widehat{z}\,(i)\,)$\,.\,   
As in the case of the $Y_{(\,\widehat{z}\,(i))}$ (and of the associated  $X_{(\,\widehat{z}\,(i))}$), if we consider the subset of directions $(\,\widehat\theta,\,\widehat\phi\,)\,\in\,\widehat{\mathbb{CS}}$ whose corresponding empirical averages  $\breve{S}_m(\,\widehat\theta,\,\widehat\phi\,)$ \,take values in \,$I\,\subset\,\mathbb{R}$,  then, the $\breve{S}_m$ are indipendently distributed over the interval $I$ with respect to the push forward measure
\begin{equation}
\label{celestDistr}
 \frac{1}{4\pi}\,d\mu_{\widehat{\mathbb{S}}^2}\left(\breve{S}_m^{-\,1}(I)\right)\,.
\end{equation}
 In general, the empirical averages $\breve{S}_m$ do not converge as $m\rightarrow\infty$; however, according to the central limit theorem \cite{Stroock}, we have the following result regarding convergence in distribution.
 
\begin{theorem}
\label{GaussCentLim}
The distribution of $\breve{S}_m$ with respect to the push forward measure 
(\ref{celestDistr}) converges to a Gaussian distribution $N_{(0,1)}$ of $\breve{S}_m$, with mean $0$ and variance $1$, \emph{i.e.}
\begin{equation}
\label{CentLimThm1}
\lim_{m\,\rightarrow\,\infty}\,\frac{1}{4\pi}\,d\mu_{\widehat{\mathbb{S}}^2}\left(\breve{S}_m^{-\,1}(I)\right)\,=\,d\mu(N_{(0,1)})\,,
\end{equation}
explicitly
\begin{equation}
\label{CentLimThm2}
\lim_{m\,\rightarrow\,\infty}\,\mathbb{E}^{\mathbb{P}}\left[\breve{S}_m,\,I \right]\,=\,
\left\langle\,\breve{S}_m\,\in\,I\, \right\rangle_{\widehat{\mathbb{CS}}}\,=\,\frac{1}{\sqrt{2\pi}}\,\int_I\,\exp\left[-\,\frac{u^2}{2}\right]\,du\,.
\end{equation}
\vskip 0.3cm\noindent
According to Theorem \ref{AlmSureConv} and the $\mathbb{P}$-almost sure convergence result  (\ref{AlmSureP}),  the area distance fluctuations $\{Y_{(\,\widehat{z}\,(i))}(\widehat\theta,\,\widehat\phi)\}$ are such that
\begin{equation}
\lim_{n\rightarrow\,\infty}\,\frac{1}{n}\,
\sum_{i=1}^n\,Y_{(\,\widehat{z}\,(i))}\,=\, \delta^{\;(1)}_{\widehat{\mathbb{CS}}}(\,\widehat{z}\,(c)) 
\end{equation} 
$\mathbb{P}$-almost surely. Thus, as $m\rightarrow\infty$,  (\ref{CentLimThm2}) implies that the normalized  fluctuations $\breve{S}_m$ around $\delta^{\;(1)}_{\widehat{\mathbb{CS}}}(\,\widehat{z}\,(c))$ are normally distributed (with mean $0$ and variance $1$). Moreover, these fluctuations are independent from $\delta^{\;(1)}_{\widehat{\mathbb{CS}}}(\,\widehat{z}\,(c))$, and, in turn, $\delta^{\;(1)}_{\widehat{\mathbb{CS}}}(\,\widehat{z}\,(c))$ is independent from the $Y_{(\,\widehat{z}\,(c))}(\widehat\theta,\,\widehat\phi)$ sample variance (\ref{samplevar0i}),  
\begin{align}
\label{samplevar0i1}
\mathrm{Var}_{\widehat{\mathbb{CS}}}(Y_{(\,\widehat{z}\,(c))})\,&:=\,
\frac{1}{4\pi}
\int_{\widehat{\mathbb{CS}}}
\left(Y_{(\,\widehat{z}\,(c))}\,-\,\delta^{\;(1)}_{\widehat{\mathbb{CS}}}(\,\widehat{z}\,(c))\right)^2
	d\mu_{\widehat{\mathbb{S}}^2}\nonumber\\
\\
&=\,\delta^{\;(2)}_{\widehat{\mathbb{CS}}}(\,\widehat{z}\,(c))\,-\,\left(\delta^{\;(1)}_{\widehat{\mathbb{CS}}}(\,\widehat{z}\,(c)) \right)^2\,.\nonumber
\end{align}
In particular, $\delta^{\;(2)}_{\widehat{\mathbb{CS}}}(\,\widehat{z}\,(c))$ is indipendent from $\delta^{\;(1)}_{\widehat{\mathbb{CS}}}(\,\widehat{z}\,(c))$. 
\end{theorem}
\begin{proof}
The first part of the theorem is a straightforward application of the Central Limit Theorem \cite{Stroock} to the sequence of mutually independent, identically distributed, square integrable random variables $\widehat{\mathbb{CS}}\,\ni\,
\widehat{n}(\widehat\theta,\,\widehat\phi)\,\rightarrow\,
\breve{S}_m(\,\widehat\theta,\,\widehat\phi\,)$, with mean value $0$ and variance $1$, defined by (\ref{partsum}). For the second part, note that from $\breve{S}_m\,\sim\, N_{(0, 1)}$ and (\ref{AlmSureP}) it follows that, for large $m$, the empirical mean of the normalized fluctuations of $Y_{(\,\widehat{z}\,(i))}$  around $\delta^{\;(1)}_{\widehat{\mathbb{CS}}}(\,\widehat{z}\,(c))$ are $\sim\,N_{(0, 1)}$ distributed. The independence of $\delta^{\;(1)}_{\widehat{\mathbb{CS}}}(\,\widehat{z}\,(c))$ from these fluctuations, and the independence of the mean value from the variance are characterizing properties of the normal distribution.  
\end{proof}
 
The results of Theorems \ref{AlmSureConv}\, and \ref{GaussCentLim} provide a deep interpretation of the asymptotic and limit relations (\ref{AsyFluctAreaDist4}), (\ref{FluctAreaDist4}), and (\ref{AsyFluctAreaDist4CRIT}) connecting  
$\delta^{\;(2)}_{\widehat{\mathbb{CS}}}(\,\widehat{z}\,(i))$  and $\delta^{\;(1)}_{\widehat{\mathbb{CS}}}(\,\widehat{z}\,(i))$  to the spacetime scalar curvature. To begin with, let us observe that in terms of the sequence of relative fluctuations 
\begin{equation}
\label{random1bisse}
\left\{\left.Y_{(\,\widehat{z}\,(i))}\,:=\,\frac{D_{\widehat{z}\,(i)}(\zeta_{(\,\widehat{z}\,(i))},\,v_{\;\widehat{z}\,(i)})\,-\,\widehat{D}_{\widehat{z}\,(i)}}{\widehat{D}_{\widehat{z}\,(i)}}\;\right|\;i\,\geq\,1\;\right\}\,,
\end{equation}
(see (\ref{random1}))
, we have the following result.
 
\begin{proposition}
\label{PropA}
With the same setting of Theorem \ref{AlmSureConv}, we have the $\mathbb{P}$-almost sure convergence 
\begin{align}
\label{AlmSureDeltas}
\lim_{n\rightarrow\,\infty}&\,\frac{1}{n}\,
\sum_{i=1}^n\,\left[Y^2_{(\,\widehat{z}\,(i))}\,+\,2 Y_{(\,\widehat{z}\,(i))}  \right]\,=\, \delta^{\;(2)}_{\widehat{\mathbb{CS}}}(\,\widehat{z}\,(c))\,+\,2\delta^{\;(1)}_{\widehat{\mathbb{CS}}}(\,\widehat{z}\,(c))\\
&=\,\mathbb{E}^{\mathbb{P}}\left[\frac{D^2_{\,\widehat{z}\,(c)}\,(\zeta_{(\widehat{z}\,(c))},\,v_{\;\widehat{z}\,(c)})\,-\,
\widehat{D}^2_{\,\widehat{z}\,(c)}}{\widehat{D}^2_{\,\widehat{z}\,(c)}}\right]\nonumber\\
&=\,
\frac{\left(1\,+\,\widehat{z}\,(c)\right)^2}{4\pi\,f^2(\widehat{r}\,(\widehat{z}\,(c)))}
\left(A\left(\Sigma_{\;\widehat{z}\,(c)}\right)\,-\, \widehat{A}\left(\widehat{\Sigma}_{\;\widehat{z}\,(c)}\right)\right)\nonumber\,,
\end{align}
where we exploited the relations (\ref{PhysAreazeta(ii)}), (\ref{FLRWPhysAreazeta(ii)}),
(\ref{FLRWDiRefer}), and (\ref{formulDiffD2}). Let  
\begin{equation}
\Delta\,A_{(i)}\,:=\,
\frac{D^2_{\,\widehat{z}\,(i)}\,(\zeta_{(\widehat{z}\,(c))},\,v_{\;\widehat{z}\,(i)})\,-\,
\widehat{D}^2_{\,\widehat{z}\,(i)}}{\widehat{D}^2_{\,\widehat{z}\,(i)}}\,
\end{equation}
denote the area measure relative fluctuations, and let 
\begin{equation}
\label{normArFluct}
\overline{\Delta\,A}_{(i)}\,:=\,\frac{\Delta\,A_{(i)}\,-\,\mathbb{E}^{\mathbb{P}}[\Delta\,A_{(i)}]}{\sqrt{\mathrm{Var}_{\widehat{\mathbb{CS}}}\,(\Delta\,A_{(i)})}}
\end{equation}
the associated random variable with mean $0$ and variance $1$. Let us consider the  associated empirical averages 
\begin{equation}
\label{empiricalA}
\mathbb{A}_n\,:=\,\frac{1}{n}\,\sum_{i=1}^{n}\,\overline{\Delta\,A}_{(i)}\,,
\end{equation}
and let us assume that they take values in some interval $I\subset\,\mathbb{R}$. Then, the distribution of $\mathbb{A}_n$ with respect to the push forward measure $\mathbb{P}_{\widehat{\mathbb{CS}}}(\mathbb{A}^{- 1}_n(I))$ converges for $n\,\rightarrow\,\infty$\;to\, a Gaussian distribution $N_{(0,1)}$ of $\mathbb{A}_n$, with mean $0$ and variance $1$, \emph{i.e.}
\begin{equation}
\label{CentLimThm2A}
\lim_{n\,\rightarrow\,\infty}\,\mathbb{E}^{\mathbb{P}}\left[\mathbb{A}_n,\,I \right]\,=\,
\frac{1}{\sqrt{2\pi}}\,\int_I\,\exp\left[-\,\frac{u^2}{2}\right]\,du\,.
\end{equation}
Thus, for $i\,\rightarrow\,\infty$,  the normalized area measure fluctuations (\ref{normArFluct}) are normally distributed around the expected mean value (\ref{AlmSureDeltas}). 
\end{proposition} 

\begin{proof}
The random variable $Y^2_{(\,\widehat{z}\,(i))}\,+\,2 Y_{(\,\widehat{z}\,(i))}$ in (\ref{AlmSureDeltas}) is a quadratic polynomial and the $\mathbb{P}$-almost sure convergence   
\begin{equation}
\label{AlmSureDeltasproof}
\mathbb{P}_{\widehat{\mathbb{CS}}}\,\left(\lim_{n\rightarrow\,\infty}\,\frac{1}{n}\,
\sum_{i=1}^n\,\left[Y^2_{(\,\widehat{z}\,(i))}\,+\,2 Y_{(\,\widehat{z}\,(i))}  \right]\,=\, \delta^{\;(2)}_{\widehat{\mathbb{CS}}}(\,\widehat{z}\,(c))\,+\,2\delta^{\;(1)}_{\widehat{\mathbb{CS}}}(\,\widehat{z}\,(c))  \right)\,=\,1\,,
\end{equation}
immediately follows from (\ref{AlmSureP}). Similarly, the central limit result (\ref{CentLimThm2A}) is a direct consequence of Theorem \ref{GaussCentLim}  and of the polynomial nature of the random variables involved as expressed in terms of the $Y_{(\,\widehat{z}\,(i))}$\,.
\end{proof}

\subsection{Area distance fluctuations and curvature factorization}
\label{ArDisCurvFact}

The asymptotics (\ref{FluctAreaDistasympt2}), associated with the spacetime scalar curvature, provides   
\begin{equation}
\frac{
A\left(\Sigma_{\;\widehat{z}\,(c)}\right)\,-\, \widehat{A}\left(\widehat{\Sigma}_{\;\widehat{z}\,(c)}\right)}{4\pi\,\widehat{r}^{\,2}(\,\widehat{z}\,(c))}\,=\,
\frac{\widehat{r}^{\,2}(\,\widehat{z}\,(c))}{72}\,\left(\widehat{\mathrm{R}}(\widehat{p})\,- \tfrac{1\,+\,v_{\widehat{z}\,(c)}}{1\,-\,v_{\widehat{z}\,(c)}}\,{\mathrm{R}}(p)\,+\,\ldots\,\right)\,,
\end{equation}
so that (\ref{AlmSureDeltas}) can be rephrased as the $\mathbb{P}$-almost sure convergence of the area measure fluctuations to the expectation value of the scalar curvature deviation at the event $p$. We have
 
\begin{lemma}
\label{noRest}
\begin{equation}
\label{noRest1}
\mathbb{P}_{\widehat{\mathbb{CS}}}\,\left(\lim_{n\rightarrow\,\infty}\,\tfrac{1}{n}\,
\sum_{i=1}^n\,\tfrac{D^2_{\,\widehat{z}\,(i)}\,(\zeta_{(\widehat{z}\,(i))},\,v_{\;\widehat{z}\,(i)})\,-\,
\widehat{D}^2_{\,\widehat{z}\,(i)}}{\widehat{D}^2_{\,\widehat{z}\,(i)}}\,=\,
\tfrac{\widehat{r}^{\,2}(\,\widehat{z}\,(c))}{72}\,\left(\widehat{\mathrm{R}}({p})\,- \tfrac{1\,+\,v_{\widehat{z}\,(c)}}{1\,-\,v_{\widehat{z}\,(c)}}\,{\mathrm{R}}(p)\right)   \right)\,=\,1\,.  
\end{equation}
\end{lemma}

\begin{proof}
The lemma is a direct consequence of the $\mathbb{P}$-almost sure convergence result
(\ref{AlmSureDeltas}) discussed in Proposition \ref{PropA} and of the asymptotics (\ref{FluctAreaDistasympt2}). The delicate point concerns the removal, in 
(\ref{noRest1}), of the higher 
order, $o(\widehat{r}^{\,2}(\,\widehat{z}\,(c))$,\, curvature contributions present in (\ref{FluctAreaDistasympt2}), and disappearing in the usual  $\widehat{r}(\widehat{z})\rightarrow\,0$ limit. This removal is justified by observing that  the empirical means
\begin{equation}
\label{empirD}
\frac{1}{n}\,
\sum_{i=1}^n\,\frac{D^2_{\,\widehat{z}\,(i)}\,(\zeta_{(\widehat{z}\,(i))},\,v_{\;\widehat{z}\,(i)})\,-\,
\widehat{D}^2_{\,\widehat{z}\,(i)}}{\widehat{D}^2_{\,\widehat{z}\,(i)}}
\end{equation} 
explore, as $n\rightarrow\,\infty$,  the $\widehat{r}(\widehat{z})\rightarrow\,0$ limit by computing it in the $\mathbb{P}$-almost sure convergence  sense. Roughly speaking, by $\mathbb{P}$-sampling the cluster with the sequence of sky sections $\{\Sigma_{\widehat{z}(i)}\}$, \,and\, $\{\widehat{\Sigma}_{\widehat{z}(i)}\}$ (and associated celestial spheres),  we are identifying the whole cluster region with the "observer" at $p$ (see (\ref{random1}) ff.).
\end{proof}

As a consequence of this lemma and of the relation (\ref{AlmSureDeltas}), we have also  
\begin{equation}
\label{AlmSureDeltasA}
\frac{\delta^{\;(2)}_{\widehat{\mathbb{CS}}}(\,\widehat{z}\,(c))\,+\,2\,\delta^{\;(1)}_{\widehat{\mathbb{CS}}}(\,\widehat{z}\,(c))}{\left(1\,+\,\widehat{z}\,(c)\right)^2}\,=\,\frac{\widehat{r}^{\,2}(\,\widehat{z}\,(c))}{72}\,\left(\widehat{\mathrm{R}}({p})\,- \tfrac{1\,+\,v_{\widehat{z}\,(c)}}{1\,-\,v_{\widehat{z}\,(c)}}\,{\mathrm{R}}(p)\right)\,.
\end{equation}
\vskip 0.2cm
\noindent
Since the (normalized) empirical averages $\mathbb{A}_n$ (see (\ref{empiricalA})) of the area measures are in the $n\rightarrow\infty$ limit $\sim\,N_{(0,1)}$ distributed, we can promote (\ref{AlmSureDeltas}) to a relation characterizing the scalar curvature deviation in a probabilistic sense in terms of the statistical independent quantities   $\delta^{\;(1)}_{\widehat{\mathbb{CS}}}(\,\widehat{z}\,(c))$ and $\delta^{\;(2)}_{\widehat{\mathbb{CS}}}(\,\widehat{z}\,(c))$. To this end, by tracing the Einstein equations (\ref{FLRWeinstein}) and  (\ref{einsteinEFE}), respectively associated with the reference FLRW spacetime $(M, \widehat{g}, \widehat{\gamma}_s(\widehat\tau))$ and the physical spacetime $(M, {g}, {\gamma}_s(\tau))$, we get
\begin{equation}
\label{RmattCosmCon}
\,\widehat{\mathrm{R}}(p)\,-\,\frac{1\,+\,v_{\widehat{z}\,(c)}}{1\,-\,v_{\widehat{z}\,(c)}}\,\mathrm{R}(p)\,=\,4\,\left(\frac{\widehat{\Lambda}-\Lambda^{(v)}}{\widehat{\Lambda}}\right)\,\widehat{\Lambda}\,
+\,8\pi\,\left(\frac{\widehat{\rho}(p)\,-\,\rho^{(v)}(p)}{\widehat{\rho}(p)}\right) \widehat{\rho}(p)\,,
\end{equation}  
where for notational ease we dropped the FLRW superscript from the FLRW cosmological constant term $\widehat{\Lambda}^{(FLRW)}$ and matter density $\widehat{\rho}^{\,(FLRW)}$ \,(just mantaining the carat $\widehat{}\,$), and  normalized the physical cosmological constant $\Lambda^{(phys)}$ and matter density according to
\begin{align}
\Lambda^{(v)}\,&:=\,\frac{1\,+\,v_{\widehat{z}\,(c)}}{1\,-\,v_{\widehat{z}\,(c)}}\,\Lambda^{(phys)}\\
\rho^{(v)}(p)\,&:=\,\frac{1\,+\,v_{\widehat{z}\,(c)}}{1\,-\,v_{\widehat{z}\,(c)}}\,\rho(p)\,,
\end{align}
where $\widehat{\rho}(p)$ and ${\rho}(p)$ respectively denote the FLRW and the physical matter density at $p$. We have assumed a pure matter scheme for both the reference FLRW as well as for the physical spacetime. The characterization of  ${\rho}(p)$ in the pre-homogeneity region should take care of the fact that the physical matter density data are gathered from the past light cone (directly for the baryonic component, indirectly via the cold dark matter assumption). Thus, we define ${\rho}(p)$ in terms of the physical matter density ${\rho}_{\,\Sigma_{\,\widehat{z}\,(i)}}$ as described on  the sequence of sky sections $\{\Sigma_{\widehat{z}(i)}\}$ foliating the physical past lightcone region  $\mathbb{W}(\gamma(\tau))\,\cap\,{\mathcal{C}}^-(p, {g})$, and converging to the observation event $(p, \dot\gamma(p))$ (see Theorem \ref{Theasympt1}). To this end,
let us consider the map  
\begin{align}
\label{psimaprho}
\widehat{\mathbb{CS}}_{\,\widehat{z}\,(i)}
\,&\longrightarrow \,\Sigma_{\,\widehat{z}\,(i)}\\
(\widehat\theta, \widehat\phi\,)\,&\longmapsto \,{q}\,:=\,\exp_p\,\left(\zeta_{(\,{\widehat{z}})}\,(\widehat\theta, \widehat\phi\,)\right)\,,\nonumber
\end{align}
in terms of which we can pull back ${\rho}^{(v)}_{\,\Sigma_{\,\widehat{z}\,(i)}}$ on the corresponding  celestial spheres $\mathbb{CS}_{\,\widehat{z}\,(i)}$. We define
\begin{equation}
\rho^{(v)}_{(\,\widehat{z}\,(i))}(\widehat\theta, \widehat\phi\,)\,:=\,
\left(\exp_p\,\circ\zeta_{(\,\widehat{z}\,(i))}\right)^*{\rho}^{(v)}_{\,\Sigma_{\,\widehat{z}\,(i)}}(q)\,,
\end{equation}
and denote by
\begin{equation}
\rho^{(v)}\left(\widehat{\mathbb{CS}}_{\widehat{z}(i)}\right)\,=\,
\frac{1}{4\pi}\,\int_{\widehat{\mathbb{CS}}}\,
\rho^{(v)}_{\widehat{z}\,(i)}(\widehat\theta, \widehat\phi\,)\,
d\mu_{\widehat{\mathbb{S}}^2}\,,
\end{equation} 
the average density over $\widehat{\mathbb{CS}}_{\widehat{z}(i)}$, and by $\widehat{\rho}(\widehat{\mathbb{CS}}_{\widehat{z}(i)})$ the FLRW reference density. It is natural to assume that the relative matter density fluctuation
\begin{equation}
\Gamma_{\widehat{z}\,(i)}\left(\widehat\theta, \widehat\phi\, \right)\,:=\,  
\frac{\widehat{\rho}(\widehat{\mathbb{CS}}_{\widehat{z}(i)})\,-\,
\rho^{(v)}_{\widehat{z}\,(i)}(\widehat\theta, \widehat\phi\,)}{\widehat{\rho}(\widehat{\mathbb{CS}}_{\widehat{z}(i)})}
\end{equation} 
characterizes a sequence $\{\Gamma_{\widehat{z}\,(i)}\,:\;i\,\in\,\mathbb{Z}^+\}$ of independent, identically distributed random variables with respect to $\mathbb{P}_{\widehat{\mathbb{CS}}}$. As in the previous case, the i.i.d. assumption implies the $\mathbb{P}$-almost sure convergence result
\begin{equation}
\lim_{n\rightarrow\,\infty}\,\frac{1}{n}\,\sum_{i=1}^{n}\Gamma_{\widehat{z}\,(i)}\,=\,
\frac{\widehat{\rho}(\widehat{\mathbb{CS}}_{\widehat{z}(i)})\,-\,
\mathbb{E}^{\mathbb{P}}[\rho^{(v)}_{\widehat{z}\,(c)}]}{\widehat{\rho}(\widehat{\mathbb{CS}}_{\widehat{z}(i)})}\,,
\end{equation}  
where $\widehat{z}(c)$ is the FLRW reference redshift marking the cosmological uncoupling sky section $\Sigma_{\widehat{z}(1)}\,:=\,\partial V_{(c)}(p)$. We identify this expression with the matter density fluctuation at the observation point $p$, \emph{i. e.} 
\begin{equation}
\label{mattMeanValue}
\frac{\widehat{\rho}(p)-\rho^{(v)}(p)}{\widehat{\rho}(p)}\,:=\,
\frac{\widehat{\rho}(\widehat{\mathbb{CS}}_{\widehat{z}(i)})\,-\,
\mathbb{E}^{\mathbb{P}}[\rho^{(v)}_{\widehat{z}\,(c)}]}{\widehat{\rho}(\widehat{\mathbb{CS}}_{\widehat{z}\,(c)})}\,=:\,
\frac{\Delta_{\widehat{z}\,(c)}\,\rho}{\widehat{\rho}_{\widehat{z}\,(c)}}\,,
\end{equation} 
and use this characterization in the expression for the scalar curvature deviation (\ref{RmattCosmCon}). If we shift and normalize $\Gamma_{\widehat{z}\,(i)}$ in such a way to have mean value $0$ and variance $1$,
\begin{equation}
\label{stablematter}
\overline{\Gamma}_{\widehat{z}\,(i)}\,:=\,\frac{\mathbb{E}^{\mathbb{P}}[\Gamma_{\widehat{z}\,(i)}]\,-\,\Gamma_{\widehat{z}\,(i)}}{\sqrt{\mathrm{Var}(\Gamma_{\widehat{z}\,(i)})}}\,,
\end{equation}
then, as in the case of  Theorem \ref{GaussCentLim} and Proposition \ref{PropA}, the distribution of
\begin{equation}
\label{normMattFluct}
\breve{\Gamma}_n\,:=\,\frac{1}{n}\,\sum_{i=1}^n\,\overline{\Gamma}_{\widehat{z}\,(i)}\,,
\end{equation}
with respect to  the push forward measure $\mathbb{P}_{\widehat{\mathbb{CS}}}(\breve{\Gamma}^{- 1}_n(I))$\,,\,with $I\,\subset\,\mathbb{R}$,\, converges for $n\,\rightarrow\,\infty$\;to\, a Gaussian distribution $N_{(0,1)}$ of $\breve{\Gamma}_n$, with mean $0$ and variance $1$, \emph{i.e.}
\begin{equation}
\label{CentLimThm2Gamma}
\lim_{n\,\rightarrow\,\infty}\,\mathbb{E}^{\mathbb{P}}\left[\breve{\Gamma}_n,\,I \right]\,=\,
\frac{1}{\sqrt{2\pi}}\,\int_I\,\exp\left[-\,\frac{u^2}{2}\right]\,du\,.
\end{equation}
Thus, for $i\,\rightarrow\,\infty$,  the matter density fluctuations (\ref{normMattFluct}) are normally distributed around the expected mean value (\ref{mattMeanValue}). If we note that the cosmological term in (\ref{RmattCosmCon}), being a constant,   is statistically independent from the fluctuating term (\ref{mattMeanValue}), then by gathering all the results above,  we have the following theorem.

\begin{theorem}
\label{MattCentLim}
In the same probabilistic setting of Theorem \ref{AlmSureConv} and Proposition \ref{PropA}, and in the $N_{(0,1)}$-distribution sense, we can factorize the relation (see (\ref{AlmSureDeltasA}) and (\ref{AsyFluctAreaDist4CRIT}))  
\begin{align}
&\,\frac{72\left(\delta^{\;(2)}_{\widehat{\mathbb{CS}}}(\,\widehat{z}\,(c))\,+\,2\,\delta^{\;(1)}_{\widehat{\mathbb{CS}}}(\,\widehat{z}\,(c))\right)}{\left(1\,+\,\widehat{z}\,(c)\right)^2\,\widehat{r}^{\,2}(\,\widehat{z}\,(c))}\,=\,\left(\widehat{\mathrm{R}}({p})\,- \tfrac{1\,+\,v_{\widehat{z}\,(c)}}{1\,-\,v_{\widehat{z}\,(c)}}\,{\mathrm{R}}(p)\right)\nonumber\\
\\
&= 4\,\left(\frac{\widehat{\Lambda}-\Lambda^{(v)}}{\widehat{\Lambda}}\right)\,\widehat{\Lambda}\,
+\,8\pi\,\left(\frac{\Delta_{\widehat{z}\,(c)}\,\rho}{\widehat{\rho}_{\widehat{z}\,(c)}}\right) \widehat{\rho}(p)\,,
\nonumber
\end{align}
according to  
\begin{align}
\label{FactorizationN}
\delta^{\;(1)}_{\widehat{\mathbb{CS}}}(\,\widehat{z}\,(c))\,&=\,\frac{\widehat{\Omega}_m}{48}\,\left(\frac{\Delta_{\widehat{z}\,(c)}\,\rho}{\widehat{\rho}_{\widehat{z}\,(c)}}\right)\,\left(1\,+\,\widehat{z}\,(c)\right)^2\,\widehat{z}^{\,2}(c)\,\nonumber\\
\\
\delta^{\;(2)}_{\widehat{\mathbb{CS}}}(\,\widehat{z}\,(c))\,&=\,\frac{\widehat{\Omega}_{\widehat\Lambda}}{6}\,\left(\frac{\widehat{\Lambda}-\Lambda^{(v)}}{\widehat{\Lambda}}\right)\,\left(1\,+\,\widehat{z}\,(c)\right)^2\,\widehat{z}^{\,2}(c)\,,
\end{align}  
where\, $\widehat\Omega_\Lambda:=\widehat{\Lambda}/{3\,H_0^2}$,\,and\,  $\widehat\Omega_m:=8\pi\widehat{\rho}/{3\,H_0^2}$ are the present values of the FLRW density cosmological parameters (see Section \ref{capsule}).
\end{theorem}

\begin{proof}
The factorization (\ref{FactorizationN}) is a direct consequence of 
Theorem \ref{AlmSureConv}, according to which the area distance fluctuations interpreted as a sequence of i.i.d. random variables $\{Y_{(\widehat{z}\,(i))}\,:\;i\,\in\,\mathbb{Z}^+\}$ converge, in the $\mathbb{P}$-almost sure sense, to the mean value $\delta^{\;(1)}_{\widehat{\mathbb{CS}}}(\,\widehat{z}\,(c))$. In full generality, the sequence $\{Y_{(\widehat{z}\,(i))}\,:\;i\,\in\,\mathbb{Z}^+\}$ is not normally distributed. However, according to Theorem \ref{GaussCentLim}, based on direct use of a Central Limit Theorem, the empirical means, $\breve{S}_n$, associated with the stabilized random variables  $X_{(\,\widehat{z}\,(i))}\,:=\,\left(
Y_{(\,\widehat{z}\,(i))}\,-\,\delta^{\;(1)}_{\widehat{\mathbb{CS}}}(\,\widehat{z}\,(i)\,)\right)\,\mathrm{Var}^{-\,1/2}_{\widehat{\mathbb{CS}}}(Y_{(\,\widehat{z}\,(i))})$
 are $N_{(0, 1)}$ distributed in the large $n$ limit. Mutual independence of the (normalized) fluctuations from $\delta^{\;(1)}_{\widehat{\mathbb{CS}}}(\,\widehat{z}\,(c))$  and $\delta^{\;(2)}_{\widehat{\mathbb{CS}}}(\,\widehat{z}\,(c))$ easily follows.        Proposition \ref{PropA} extends this $N_{(0, 1)}$ control to the distribution of the area fluctuations  (\ref{normArFluct}), and by exploiting the $\mathbb{P}$-almost sure convergence we got a random variable version of the Lorentzian Bertrand-Puiseaux formula  
 \begin{align}
&\,\lim_{n\rightarrow\,\infty}\,\frac{1}{n}\,
\sum_{i=1}^n\,\left[\frac{D^2_{\,\widehat{z}\,(c)}\,(\zeta_{(\widehat{z}\,(c))},\,v_{\;\widehat{z}\,(c)})\,-\,
\widehat{D}^2_{\,\widehat{z}\,(c)}}{\widehat{D}^2_{\,\widehat{z}\,(c)}} \right]=\frac{\delta^{\;(2)}_{\widehat{\mathbb{CS}}}(\,\widehat{z}\,(c))\,+\,2\,\delta^{\;(1)}_{\widehat{\mathbb{CS}}}(\,\widehat{z}\,(c))}{\left(1\,+\,\widehat{z}\,(c)\right)^2}\,\nonumber\\
\\ 
&=\,\frac{\widehat{r}^{\,2}(\,\widehat{z}\,(c))}{72}\,\left(\widehat{\mathrm{R}}({p})\,- \tfrac{1\,+\,v_{\widehat{z}\,(c)}}{1\,-\,v_{\widehat{z}\,(c)}}\,{\mathrm{R}}(p)\right)\nonumber\,.
\end{align}
By exploiting this random variable rendering and the trace of the Einstein field equations, we can connect the i.i.d. area distance fluctuations with the corresponding i.i.d matter density fluctuations  $\frac{\Delta_{\widehat{z}\,(c)}\,\rho}{\widehat{\rho}_{\widehat{z}\,(c)}}$ whose sample mean $\breve{\Gamma}_n$  is, for $n\rightarrow \infty$, normally distributed   $N_{(0,1)}$ like the $X_{(\,\widehat{z}\,(i))}$  (see (\ref{CentLimThm2Gamma})). This, together with the independence between $\delta^{\;(1)}_{\widehat{\mathbb{CS}}}(\,\widehat{z}\,(c))$ and $\delta^{\;(2)}_{\widehat{\mathbb{CS}}}(\,\widehat{z}\,(c))$ on one side, and the independence between $\frac{\Delta_{\widehat{z}\,(c)}\,\rho}{\widehat{\rho}_{\widehat{z}\,(c)}}$ and the constant term $(\widehat{\Lambda}-\Lambda^{(v)})/{\widehat{\Lambda}}$ on the other side,  immediately implies the statistical factorization (\ref{FactorizationN}).    
\end{proof}

\section{A distance functional between lightcones}
\label{TSKSM}

We conclude this appendix by characterizing a functional distance between the physical and the reference FLRW lightcones. This distance provides a profound geometrical meaning to the factorization (\ref{FactorizationN}) and its connection with the cosmological constant problem. 
   
We start by recalling that, as long as the null exponential map $\exp_p$ is a diffeomorphism,  		
the  map	
\begin{align}
\label{Finterpolation2}	
F_{\widehat{z}_1,\,\widehat{z}_2}(\epsilon)\,&:\,[0,\,1]\times\left(\Sigma_{\widehat{z}_1},\,g^{(2)}_{\widehat{z}_1}\right)\,\longrightarrow\,\left(\Sigma_{\widehat{z}_2},\,g^{(2)}_{\widehat{z}_2}\right)\,,\\
&(\epsilon,\,y)\,\longmapsto\,F_{\widehat{z}_1,\,\widehat{z}_2}(\epsilon,\,y)\,:=\,\exp_p\left(r(\widehat{z}_\varepsilon)\,\ell(n(y))\right)\,,\nonumber
\end{align}
defined by  (\ref{Finterpolation}) (see Proposition \ref{lemmaDiffF}),
 is a diffeomorphism 
interpolating between the sky sections $\Sigma_{\widehat{z}_1}=F_{\widehat{z}_1,\,\widehat{z}_2}(0)$ and $\Sigma_{\widehat{z}_2}=F_{\widehat{z}_1,\,\widehat{z}_2}(1)$. Moreover, 
\begin{equation}
\label{FpullbackDopo}
g^{(2)}_{\widehat{z}_1}\,=\, \left.F^*_{\widehat{z}_1\,\widehat{z}_2}\right|_{\epsilon=1}\,g^{(2)}_{\widehat{z}_2}\,,
\end{equation}
where $F^*_{\widehat{z}_1,\,\widehat{z}_2}$ denotes the pull-back action associated with (\ref{Finterpolation2}). This interpolation can be extended to the case when $\exp_p$  is only bi-Lipschitz since Rademacher's theorem implies that the pullback action
\begin{equation}
\label{LipF}
\left(F_{\widehat{z}_1,\,\widehat{z}_2}^*g^{(2)}_{\,\widehat{z}_{\;2}}\right)_{ab}\,=\,\frac{\partial\,F_{\widehat{z}_1,\,\widehat{z}_2}^i(y)}{\partial y^a}
\frac{\partial\,F_{\widehat{z}_1,\,\widehat{z}_2}^k(y)}{\partial y^b}\,(g^{(2)}_{\,\widehat{z}_{\;2}})_{ik}\,,
\end{equation}
is defined\footnote{The $\{y^i\}$ denote local coordinates on $\Sigma_{\widehat{z}_1}$.} almost everywhere on $\Sigma_{\widehat{z}_1}$. More generally, if we denote by 
\begin{align}
\zeta_{(\,\widehat{z})}\,:\,\left(\widehat{\mathbb{C\,S}}_{(\,\widehat{z})}(p),\,h_{\mathbb{S}^2(\,\widehat{z})}\right)\,&\longrightarrow\,\left(\mathbb{C\,S}_{(\,\widehat{z})}(p),\,h_{\widehat{\mathbb{S}}^2(\,\widehat{z})}\right)\\
\widehat{w}\,&\longmapsto\,\zeta_{(\,\widehat{z})}(\,\widehat{w})\,=\,\sqrt{\frac{1\,+\,v_{\;\widehat{z}}\,(p)}{1\,-\,v_{\;\widehat{z}}\,(p)}}\,e^{i\,\alpha_{\;\widehat{z}}}\,\widehat{w}\,.
\nonumber
\end{align}
the extension of the $\mathrm{PSL}(2,\,\mathbb{C})$ action to the celestial spheres $\widehat{\mathbb{C\,S}}_{(\,\widehat{z})}(p)$ and ${\mathbb{C\,S}}_{(\,\widehat{z})}(p)$ (see Proposition \ref{PSL2mappingFLRW}), then we have the following result that provides a basic set-up for the comparison between the physical sky section $(\Sigma_{\;\widehat{z}},\,g^{(2)}_{\,\widehat{z}}\,)$ and the reference FLRW sky section 
$(\widehat\Sigma_{\,\widehat{z}},\,\widehat{g}^{\;(2)}_{\,\widehat{z}}\,)$.

\begin{proposition}
\label{comp1}
The map 
\begin{align}
\label{psimap}
\psi_{{\;\widehat{z}}}\,:\,\widehat{\Sigma}_{\,\widehat{z}}\,&\longrightarrow\,\Sigma_{\,\widehat{z}}\\
\widehat{q}\,&\longmapsto \,\psi_{\hat{z}}(\,\widehat{q})\,:=\,\exp_p\,\circ\,\zeta_{(\,{\widehat{z}})}\,\circ\,\widehat{\exp}_p^{\,-1}(\,\widehat{q}\,)\,,\nonumber
\end{align}
defined by the commutative diagram 
\begin{align}
\begin{CD}
\widehat{\mathbb	{CS}}_{\,\widehat{z}}(p)    @>\zeta_{(\,\widehat{z}\,)}>>{\mathbb{CS}}_{\,\widehat{z}}(p)\\
@VV\widehat{\exp}_pV     @VV\exp_pV\\
\left(\widehat{\Sigma}_{\;\widehat{z}},\,\widehat{g}^{\;(2)}_{\;\widehat{z}}\right)   @>\psi_{\;\widehat{z}}>> \left({\Sigma}_{\;\widehat{z}},\,{g}^{\;(2)}_{\;\widehat{z}}\right)
\end{CD}
\end{align}
is a bi-Lipschitz map between the sky sections $(\widehat\Sigma_{\,\widehat{z}},\,d_{\,\widehat\Sigma_{\,\widehat{z}}})$ and $(\Sigma_{\;\widehat{z}},\,d_{\Sigma_{\;\widehat{z}}})$, with isometric distortion 
$\mathrm{Lip}(\psi_{\,\widehat{z}})$  defined as the smallest redshift-dependent constant such that 
\begin{align}
\frac{d_{\,\widehat\Sigma_{\,\widehat{z}}}(\widehat{q}_{\;1}, \widehat{q}_{\;2})\,\sqrt{\frac{1\,+\,v_{\;\widehat{z}}\,(p)}{1\,-\,v_{\;\widehat{z}}\,(p)}}}{\mathrm{Lip}(\psi_{\,\widehat{z}})}
\,&\leq\,
d_{\Sigma_{\;\widehat{z}}}\left(\psi_{\,\widehat{z}}(\widehat{q}_{\;1}), \psi_{\,\widehat{z}}(\widehat{q}_{\;2})\right)\\
&\leq\,\mathrm{Lip}(\psi_{\,\widehat{z}})\;
d_{\,\widehat\Sigma_{\,\widehat{z}}}(\widehat{q}_{\;1}, \widehat{q}_{\;2})\,\sqrt{\frac{1\,+\,v_{\;\widehat{z}}\,(p)}{1\,-\,v_{\;\widehat{z}}\,(p)}}\,,\nonumber
\end{align}
for all $\widehat{q}_{\;1},\,\widehat{q}_{\;2}\,\in\,(\widehat{\Sigma}_{\;\widehat{z}},\,\widehat{g}^{\;(2)}_{\;\widehat{z}})$, and where we have inserted the aberration factor that takes into account the relative velocity of the physical and FRW observers. In particular, we have the uniform bound for the Hausdorff measure $\mathcal{H}^2(\Sigma_{\;\widehat{z}})$ of $({\Sigma}_{\;\widehat{z}},\,{g}^{\;(2)}_{\;\widehat{z}})$

\begin{equation}
4\pi \mathrm{Lip}^{-2}_{\,\widehat{z}\,}(\psi_{\,\widehat{z}})\,\left(\tfrac{1\,+\,v_{\,\widehat{z}}\,(p)}{1\,-\,v_{\,\widehat{z}}\,(p)}\right)\,\widehat{r}^{\,2}(\,\widehat{z}\,)\,\,
\,\leq\,\mathcal{H}^2\left(\Sigma_{\;\widehat{z}}\right)\,\leq\,
4\pi\,\mathrm{Lip}^2_{\,\widehat{z}\,}(\psi_{\,\widehat{z}})\,\left(\tfrac{1\,+\,v_{\,\widehat{z}}\,(p)}{1\,-\,v_{\,\widehat{z}}\,(p)}\right)\,
\widehat{r}^{\,2}(\,\widehat{z}\,)\,,
\end{equation}
\vskip 0.3cm\noindent
Since the physical spacetime $(M,\,g,\,\gamma_s(\tau))$ is statistically Friedmannian for large redshifts,  we can assume that there exists $\epsilon_{\,\widehat{z}}\,>\,0$, small enough,   such that the isometric distortion is given by
\begin{equation}	
\label{smallLip}
\mathrm{Lip}(\psi_{\,\widehat{z}})\,=\,1\,+\,\epsilon_{\,\widehat{z}}\,,
\end{equation}
for all $\widehat{z}\, >\,\widehat{z}_0$, where $\widehat{z}_0$ marks the assumed transition to statistical isotropy and homogeneity\footnote{As already stressed, the actual value of $\widehat{z}_0$ is much debated. Mainstream analysis of cosmological data put this value in the range \; $\widehat{z}_0\,\gtrapprox\,100\div200\,h^{\,- 1}\,\mathrm{Mpc}$)}. The parameter $\epsilon_{\,\widehat{z}}$ characterizes, in the large ${\,\widehat{z}}$ regime,  the (perturbative) isometric distortion from the underlying FLRW geometry\footnote{In the smooth case, this distortion is related to perturbation theory around the FLRW background spacetime geometry. In this appendix, we do not belabor on this well-known subject.}. It follows that for ${\,\widehat{z}}$ large enough, the map $\psi_{\,\widehat{z}}$ is $(1\,+\,\epsilon_{\,\widehat{z}})$-biLipschitz and, up to the Lorentz aberration factor, a quasi-isometry. In particular, $\psi_{(\,{\;\widehat{z}}\,)}$ can be deformed to a (locally) bi-Lipschitz conformal map, \emph{i.e.} for all $\widehat{z}_c\,>>\,\widehat{z}_0$, there exists a positive Lipschitz function $\Phi_{\;\widehat{z}_c}(\zeta_{\,(\widehat{z})\,})$, defined up to $\mathrm{PSL}(2,\mathbb{C})$ transformations, such that\footnote{The conformal factor $\Phi_{\;\widehat{z}_c}(\zeta_{\,(\widehat{z})\,})$ functionally depends on $\psi_{(\,{\;\widehat{z}}\,)}$, \,namely from the exponential maps $\exp_p$, \;$\widehat{\exp}_p$, and from the $\mathrm{PSL}(2,\mathbb{C})$ map $\zeta_{\,(\widehat{z})\,}$. In the adopted notation for $\Phi_{\;\widehat{z}_c}$ we stressed the role of $\zeta_{\,(\widehat{z})\,}$ since whereas $\exp_p$, \;$\widehat{\exp}_p$ are fixed, the M\"obius transformation $\zeta_{\,(\widehat{z})\,}$ parametrizes $\Phi_{\;\widehat{z}_c}$ and plays a role in controlling the optimal distortion among the sky sections and the corresponding celestial spheres.}  
\begin{align}
\label{psiconfbisse}
\left(\psi_{{\;\widehat{z}_c}}^*g^{(2)}_{\widehat{z}_c}\right)_{ab}\,&=\,\frac{\partial\psi_{\,{\widehat{z}_c}}^i\,(y)}{\partial y^a}
\frac{\partial\psi_{\,{\widehat{z}_c}}^k\,(y)}{\partial y^b}
\,\left(g^{(2)}_{\widehat{z}_c}\right)_{ik}\nonumber\\
\\
&=\,\left(\frac{1\,+\,v_{\,\widehat{z}}\,(p)}{1\,-\,v_{\,\widehat{z}}\,(p)}\right)\,\Phi^2_{\;\widehat{z}_c}(\zeta_{\,(\widehat{z})\,})\,
\left(\widehat{g}^{\,(2)}_{\,\widehat{z}_c}\right)_{ab}\,,\nonumber
\end{align} 
where we have explicitly introduced the aberration factor $\left(\tfrac{1\,+\,v_{\,\widehat{z}}\,(p)}{1\,-\,v_{\,\widehat{z}}\,(p)}\right)$ for later convenience.\,
The metrics ${g}^{\,(2)}_{\,\widehat{z}}$ of all sky sections $\Sigma_{\,\widehat{z}}$ can be obtained from $g^{(2)}_{\,\widehat{z}_c}$ under the action of the bi-Lipschitz map $F_{\widehat{z},\,\widehat{z}_c}(\epsilon)$, \, $0\,\leq\,\epsilon\,\leq\,1$,\; interpolating between the sky sections $\Sigma_{\hat{z}}=F_{\widehat{z},\,\widehat{z}_c}(0)$ and $\Sigma_{\widehat{z}_c}=F_{\widehat{z},\,\widehat{z}_c}(1)$, and  defined by (see 
(\ref{Finterpolation2}))  

\begin{align}
\label{FinterpolationBisse}
F_{\widehat{z},\,\widehat{z}_c}(\epsilon)\,&:\,[0,\,1]\times\left(\Sigma_{\hat{z}},\,g^{(2)}_{\hat{z}}\right)\,\longrightarrow\,\left(\Sigma_{\widehat{z}_c},\,g^{(2)}_{\widehat{z}_c}\right)\,,\\
&(\epsilon,\,y)\,\longmapsto\,F_{\widehat{z},\,\widehat{z}_c}(\epsilon,\,y)\,:=\,\exp_p\left(r(\widehat{z}_\varepsilon)\,\ell(n(y))\right)\,,\nonumber 
\end{align}
with
\begin{equation}
\label{Fpullbackbiss}
g^{(2)}_{\hat{z}}\,=\, \left.F^*_{\widehat{z}\,\widehat{z}_c}\right|_{\epsilon=1}\,g^{(2)}_{\widehat{z}_c}\,=\,
\frac{\partial\,F_{\widehat{z},\,\widehat{z}_c}^i(y)}{\partial y^a}
\frac{\partial\,F_{\widehat{z},\,\widehat{z}_c}^k(y)}{\partial y^b}\,(g^{(2)}_{\,\widehat{z}_{\;c}})_{ik}\,.
\end{equation}
It follows that for every $\widehat{z}\,\geq\,0$, the sky section metric 
$g^{(2)}_{\,\widehat{z}}$ is in the same conformal class of $g^{(2)}_{\,\widehat{z}_{\;c}}$ and we may assume that the map  $\psi_{(\,{\;\widehat{z}}\,)}$ can be deformed to a (locally) bi-Lipschitz conformal map for all $\widehat{z}$. Explicitly, for each each sky section $(\Sigma_{\;\widehat{z}},\,g^{(2)}_{\,\widehat{z}})$   there exists a corresponding positive Lipschitz function $\Phi_{\;\widehat{z}}$, defined up to $\mathrm{PSL}(2,\mathbb{C})$ transformations, such that
\begin{equation}
\label{psiconftrisse}
\psi_{\,\widehat{z}}^*g^{(2)}_{\,\widehat{z}}\,=\,\left(\frac{1\,+\,v_{\,\widehat{z}}\,(p)}{1\,-\,v_{\,\widehat{z}}\,(p)}\right)\,\Phi^2_{\,\widehat{z}}(\zeta_{\,(\widehat{z})\,})\,
\widehat{g}^{\,(2)}_{\,\widehat{z}}\,.
\end{equation}
Finally, the conformal factor $\Phi^2_{\,\widehat{z}}(\zeta_{\,(\widehat{z})\,})$ can be locally computed in terms of  the area distances
  $D^2_{\,\widehat{z}}\,(\zeta_{(\widehat{z})},\,v_{\;\widehat{z}})$\,  and 
$\widehat{D}^2_{\,\widehat{z}}$ according to
\begin{equation}
\label{ConfArDist}
\Phi^2_{\,\widehat{z}}(\zeta_{\,(\widehat{z})\,})\,=\,\frac{D^2_{\,\widehat{z}}\,(\zeta_{(\widehat{z})},\,v_{\;\widehat{z}})}{\widehat{D}^2_{\,\widehat{z}}}\,.
\end{equation}   
\end{proposition}

\begin{proof}
The first part of the theorem is a direct consequence of the definition of the map $\psi_{\,\widehat{z}}$ and of the assumed bi-Lipschitz nature of $\exp_p$. The existence of the conformal factor $\Phi^2_{\,\widehat{z}_c}(\zeta_{\,(\widehat{z})\,})$ for large redshifts 
$\widehat{z}_c\,>>\,z_0$ \,follows by observing that the physical spacetime $(M,\,g,\,\gamma_s(\tau))$ is (statistically) Friedmannian for large redshifts, and the Lipschitz distortion (\ref{smallLip}) of the sky section $(\Sigma_{\,\widehat{z}_c},\,g^{(2)}_{\,\widehat{z}_c})$   is a small perturbation of an isometry. Under these conditions, if we pull back the
 metric  $g^{(2)}_{\,\widehat{z}_c}$ from $\Sigma_{\,\widehat{z}_c}$ to $\widehat{\Sigma}_{\,\widehat{z}_c}$ under the action of the map $\psi_{\,\widehat{z}_c}$, we get a Riemann-Lipschitz surface $(\widehat{\Sigma}_{\,\widehat{z}_c},\,\psi_{\,\widehat{z}_c}^*\,g^{(2)}_{\,\widehat{z}_c})$ which is an almost round $2$-sphere. We can apply the Poincar\'e--Koebe uniformization theorem\footnote{The Poincar\'e--Koebe uniformization theorem states that every simply connected Riemann surface is conformally
equivalent to the open unit disk, the complex plane, or, as in our case, the Riemann sphere.} (see \emph{e.g.} \cite{Abikoff}) to compare, by a conformal transformation, $\psi_{\,\widehat{z}_c}^*\,g^{(2)}_{\,\widehat{z}_c}$ to the (rescaled) round metric $\widehat{g}^{(2)}_{\,\widehat{z}_c}$ defined of the FLRW sky section $\widehat{\Sigma}_{\,\widehat{z}_c}\simeq\,\mathbb{S}^2$. If we define
\begin{equation}
\label{conftrisse}
\Phi^2_{\,\widehat{z}}\left(\zeta_{(\,\widehat{z}\,)},\,   v_{\,\widehat{z}}\,\right)\,:=\,\left(\frac{1\,+\,v_{\,\widehat{z}}\,(p)}{1\,-\,v_{\,\widehat{z}}\,(p)}\right)\,\Phi^2_{\,\widehat{z}}\left(\zeta_{(\,\widehat{z}\,)} \right)\,,
\end{equation}
then the (normalized) conformal factor $\Phi^2_{\,\widehat{z}}\left(\zeta_{(\,\widehat{z}\,)},\,   v_{\,\widehat{z}}\,\right)$
is the solution, (unique up to the $\mathrm{PSL}(2, \mathbb{C})$ automorphisms of $(\widehat\Sigma_{\widehat{z}_c},\,\widehat{g}^{(2)}_{{\widehat{z}_c}})$), of the elliptic partial differential equation  on $(\widehat\Sigma_{\widehat{z}_c}, \widehat{g}^{(2)}_{{\widehat{z}_c}})$ defined by \cite{Berger}  
\begin{equation}
\label{lapconf0}
-\,\Delta _{\widehat{g}^{(2)}_{{\widehat{z}_c}}}\ln({\Phi}_{\widehat{z}_c}
^2(\zeta_{(\,\widehat{z}\,)},\,   v_{\,\widehat{z}}\,))\,+\,R(\widehat{g}^{(2)}_{{\widehat{z}_c}})\,=\, R(\psi_{{\widehat{z}_c}}^*g^{(2)}_{\widehat{z}_c})\,{\Phi}_{\widehat{z}_c}
^2(\zeta_{(\,\widehat{z}\,)},\,   v_{\,\widehat{z}}\,)\;,
\end{equation}
where $\Delta _{\widehat{g}^{(2)}_{{\widehat{z}_c}}}$ is the Laplace-Beltrami operator on $(\widehat\Sigma_{\widehat{z}_c}, \widehat{g}^{\,(2)}_{{\widehat{z}_c}})$, and where 
we respectively denoted by $R(\widehat{g}^{(2)}_{{\widehat{z}_c}})$ and  $R(\psi_{{\widehat{z}_c}}^*g^{(2)}_{\widehat{z}_c})$  the scalar curvature of the metrics $\widehat{g}^{(2)}_{{\widehat{z}_c}}$ and $\psi_{{\widehat{z}_c}}^*g^{(2)}_{\widehat{z}_c}$.
The quasi-isometry control on the Lipschitz distortion (\ref{smallLip}) of the sky section $(\Sigma_{\,\widehat{z}_c},\,g^{(2)}_{\,\widehat{z}_c})$ implies also a corresponding control on the conformal factor ${\Phi}_{\widehat{z}_c}^2(\zeta_{(\,\widehat{z}\,)},\,   v_{\,\widehat{z}}\,)$, but here we do not belabor on this point.

The extension of the map $\psi_{\hat{z}}$ as a conformal diffeomorphism for all $\widehat{z}$ is explained in the statement of the theorem, and it is a direct consequence of the pullback action of the bi-Lipschitz map (\ref{FinterpolationBisse}). In particular, we have that for every $\widehat{z}\,\geq\,0$, the sky section metric 
$g^{(2)}_{\,\widehat{z}}$ is in the same conformal class of $g^{(2)}_{\,\widehat{z}_{\;c}}$. Thus, the map  $\psi_{(\,{\;\widehat{z}}\,)}$ is a bi-Lipschitz conformal map described by a $\widehat{z}$-dependent conformal factor ${\Phi}_{\hat{z}}^2(\zeta_{(\,\widehat{z}\,)})$. The geometrical content of ${\Phi}_{\hat{z}}^2(\zeta_{(\,\widehat{z}\,)})$ varies significantly with the redshift since it depends on the vagaries of the physical exponential map (\ref{FinterpolationBisse}) which interpolates between the sky sections $\{\Sigma_{\hat{z}} \}$, as $\widehat{z}$ varies.

As already stressed, the solution of (\ref{lapconf0}) is not unique since we have a non-trivial action of the $\mathrm{PSL}(2, \mathbb{C})$ automorphisms group of $(\widehat\Sigma_{\widehat{z}_c},\,\widehat{g}^{(2)}_{{\widehat{z}_c}})$. In our case, this action can be explicitly described in terms of the  $\mathrm{PSL}(2, \mathbb{C})$ map $\zeta_{\,\widehat{z}}$ acting between the celestial spheres $\mathbb{CS}_{\hat{z}}(p)$ and $\mathbb{CS}_{\hat{z}}(p)$. This remark follows by observing that we have the equivalent characterization of the conformal factor ${\Phi}_{\hat{z}}^2(\zeta_{(\,\widehat{z}\,)})$ in terms of the metric celestial spheres $({\mathbb{C\,S}}_z(p),\;h_z)$ and  $(\widehat{\mathbb{C\,S}}_{\hat{z}}(p),\;\widehat{h}_{\hat{z}})$. By massaging pullbacks, we have 
\begin{align}
\label{Nphi10}
\left(\tfrac{1\,+\,v_{\,\widehat{z}}\,(p)}{1\,-\,v_{\,\widehat{z}}\,(p)}\right)\,
{\Phi}_{\hat{z}}^2(\zeta_{(\,\widehat{z}\,)})\,\widehat{g}^{(2)}_{{\widehat{z}}}\,&=\,\left(
\psi_{{\widehat{z}}}^*{g^{(2)}_{(z)}}\right)\\
&=\left(\exp_p\,\circ\,\zeta_{(\,{\widehat{z}}\,)}\,\circ\,\widehat{\exp}_p^{\,-1}\right)^*{g^{(2)}_{\hat{z}}}\nonumber\\
&= (\widehat{\exp}_p^{\,-1})^*(\zeta_{({\,\widehat{z}\,})}^*\,(\exp_p^*\,g^{(2)}_{\hat{z}}))\nonumber\\
&= (\widehat{\exp}_p^{\,-1})^*(\zeta_{(\,{\widehat{z}}\,)}^*h_{\hat{z}})\nonumber\,,
\end{align}
which, by taking the $\widehat\exp_p$ pullback of both members, provides
\begin{equation}
\left(\tfrac{1\,+\,v_{\,\widehat{z}}\,(p)}{1\,-\,v_{\,\widehat{z}}\,(p)}\right)\,
\widehat{\exp}_p^*\left(
{\Phi}_{\hat{z}}^2(\zeta_{(\,\widehat{z}\,)})\,{\widehat{g}^{(2)}_{{\widehat{z}}}}\right)\,
=\,\zeta_{(\,{\widehat{z}}\,)}^*h_{\hat{z}}.
\end{equation}
Since 
\begin{equation}
\widehat{\exp}_p^*\left(
{\Phi}_{\hat{z}}^2(\zeta_{(\,\widehat{z}\,)})\,{\widehat{g}^{(2)}_{{\widehat{z}}}}\right)\,=\,
\left({\Phi}_{\hat{z}}^2(\zeta_{(\,\widehat{z}\,)})\,\circ\,\widehat{\exp}_p\right)\,\widehat{\exp}_p^*\,{\widehat{g}^{(2)}_{{\widehat{z}}}}\,=\,
{\Phi}_{\hat{z}}^2(\zeta_{(\,\widehat{z}\,)})\,{\widehat{h}}_{\hat{z}}\circ\,\widehat{\exp}_p\,,
\end{equation}
we eventually get
\begin{equation}
\label{fiacca0}
\left(\tfrac{1\,+\,v_{\,\widehat{z}}\,(p)}{1\,-\,v_{\,\widehat{z}}\,(p)}\right)\,
{\Phi}_{\hat{z}}^2(\zeta_{(\,\widehat{z}\,)})\,{\widehat{h}}_{\hat{z}}\circ\,\widehat{\exp}_p\,=\,\zeta_{(\,{\widehat{z}}\,)}^*h_{\hat{z}}\,,
\end{equation}
\vskip 0.3cm\noindent     
showing that the celestial spheres $({\mathbb{C\,S}}_z(p),\;h_z)$ and  $(\widehat{\mathbb{C\,S}}_{\hat{z}}(p),\;\widehat{h}_{\hat{z}})$ are conformally related via the composition of the maps $\widehat{\exp}_p$ and $\zeta_{(\,{\widehat{z}}\,)}$.

To prove the relation (\ref{ConfArDist}) connecting the area distances with the conformal factor ${\Phi}_{\hat{z}}^2(\zeta_{(\,\widehat{z}\,)})$\,, we use (\ref{psiconftrisse}) to compute  

\begin{align}
\label{Nphi12}
\left(\tfrac{1\,+\,v_{\,\widehat{z}}\,(p)}{1\,-\,v_{\,\widehat{z}}\,(p)}\right)\,
{\Phi}_{\hat{z}}^2(\zeta_{(\,\widehat{z}\,)})\,d\mu_{\widehat{g}^{(2)}_{{\widehat{z}}}}\,&=\,
{\psi}^*_{\hat{z}}\,d\mu_{{g}^{(2)}_{{\widehat{z}}}}\\
=\,
\left(\exp_p\,\circ\,\zeta_{(\,{\widehat{z}}\,)}\,\circ\,\widehat{\exp}_p^{\,-1}\right)^*\,d\mu_{{g}^{(2)}_{{\widehat{z}}}}\,
&=\, \left(\zeta_{(\,{\widehat{z}}\,)}\,\circ\,\widehat{\exp}_p^{\,-1}\right)^*\,(\exp_p^*\,d\mu_{{g}^{(2)}_{{\widehat{z}}}})\nonumber\\
= \left(\zeta_{(\,{\widehat{z}}\,)}\,\circ\,\widehat{\exp}_p^{\,-1}\right)^*\,d\mu_{h_{\hat{z}}}\,
&=\, \left(\zeta_{(\,{\widehat{z}}\,)}\,\circ\,\widehat{\exp}_p^{\,-1}\right)^*\,D^2_{\hat{z}}\,d\mu_{\mathbb{S}^2}\nonumber\\
= (\widehat{\exp}_p^{\,-1})^*\left(D^2_{\hat{z}}\,d\mu_{\widehat{\mathbb{S}}^2}\,\circ\,\zeta_{(\,{\widehat{z}}\,)}\right)\nonumber\,,
\end{align}
from which, by taking the $\widehat{\exp}_p$ pull back of both members, we get 

\begin{align}
\label{phiareadist}
&\widehat{\exp}_p^*\,\left(\left(\tfrac{1\,+\,v_{\,\widehat{z}}\,(p)}{1\,-\,v_{\,\widehat{z}}\,(p)}\right)\,
{\Phi}_{\hat{z}}^2(\zeta_{(\,\widehat{z}\,)})\,d\mu_{\widehat{g}^{(2)}_{{\widehat{z}}}}\,  \right)\,=\,
\left(\tfrac{1\,+\,v_{\,\widehat{z}}\,(p)}{1\,-\,v_{\,\widehat{z}}\,(p)}\right)\,
{\Phi}_{\hat{z}}^2(\zeta_{(\,\widehat{z}\,)})\,\widehat{D}^2_{\hat{z}}\,d\mu_{\widehat{\mathbb{S}}^2}\,
\circ\,\widehat{\exp}_p\nonumber\\
\\
&=\,D^2_{\hat{z}}\,d\mu_{\widehat{\mathbb{S}}^2}\,\circ\,\zeta_{(\,{\widehat{z}}\,)}\nonumber\\
&
\Longrightarrow\, {\Phi}_{\hat{z}}^2(\zeta_{(\,\widehat{z}\,)})\,\circ\,\widehat{\exp}_p\,=\,
\left(\tfrac{1\,-\,v_{\,\widehat{z}}\,(p)}{1\,+\,v_{\,\widehat{z}}\,(p)}\right)\,
\frac{D^2_{\hat{z}}\,\circ\,\zeta_{(\,{\widehat{z}}\,)}}
{\widehat{D}^2_{\hat{z}}\,\circ\,\widehat{\exp}_p}\,=\,
\frac{D^2_{\hat{z}}(\zeta_{{\widehat{z}}},\,v_{\,\widehat{z}})}{\widehat{D}^2_{\hat{z}}\,\circ\,\widehat{\exp}_p}\,,
\nonumber
\end{align}
where we exploited the relation (\ref{Dnotation}) defining the (aberration) normalized area distance $D^2_{\hat{z}}(\zeta_{{\widehat{z}}},\,v_{\,\widehat{z}})$. 
\end{proof}
Proposition \ref{comp1} shows that the conformal factor ${\Phi}_{\hat{z}}(\zeta_{(\,\widehat{z}\,)})$, evaluated for a given source $\widehat{q}$ at FLRW celestial sphere coordinates $y\,\in\,\widehat{\mathbb{CS}}_{\hat{z}}(p)$, can be expressed as the ratio between the (normalized) physical area distance of the source, $D(\zeta_{\hat{z}}(y),\,v_{\hat{z}})$,\; and the corresponding reference area distance $\widehat{D}(\widehat{\exp}_p(y))$. It follows that  ${\Phi}_{\hat{z}}(\zeta_{(\,\widehat{z}\,)})$ is an explicit and, at least in principle, measurable quantity, also when the exponential mapping is not smooth.

\section{The comparison functional at a given FLRW redshift}
\label{CompFunctRedZ}
We can think of  the  conformal factor ${\Phi}_{\hat{z}}(\zeta_{(\,\widehat{z}\,)})$ as  stretching the elastic surface $({\mathbb{C\,S}}_{\hat{z}}(p),\,{h}_{\hat{z}})$ on the rigid surface $(\widehat{\mathbb{C\,S}}_{\hat{z}}(p),\widehat{h}_{\hat{z}})$. In \cite{CarFam, CarFam2},  we associated with this stretching an \emph{elastic energy functional},   whose structure was suggested to us by the rich repertoire of functionals used 
 in the problem of comparing shapes of surfaces in computer graphics and visualization problems (see \textit{e.g.} \cite{JinYauGu} and \cite{YauGu}, to quote two relevant papers in a vast literature). In particular, we were inspired by an energy functional introduced, under the name of \textit{elastic energy}, in a remarkable paper by J. Hass and P. Koehl  \cite{HassKoehl}, who use it as a powerful means of comparing the shapes of genus-zero surfaces in problems relevant to surface visualization. In the more complex cosmological setting, we introduce the following comparison functional between the celestial spheres $({\mathbb{C\,S}}_{\hat{z}}(p),\,{h}_{\hat{z}})$ and $(\widehat{\mathbb{C\,S}}_{\hat{z}}(p),\widehat{h}_{\hat{z}})$
 
 \begin{equation}
\label{Efunctn1}
E_{\widehat{\mathbb{C\,S}}_{\hat{z}},\;{\mathbb{C\,S}}_{\hat{z}}}[\zeta_{\hat{z}}]\,:=\,\int_{\widehat{\mathbb{CS}}_{\hat{z}}}\left({\Phi}_{\hat{z}}(\zeta_{(\,\widehat{z}\,)})\,-\,1 \right)^2\,d\mu_{\hat{h}_{\hat{z}}}\,=\,
\int_{\widehat{\mathbb{CS}}_{\hat{z}}}\,
\left[\frac{D_{{z}}(\zeta_{\hat{z}},\,v_{\hat{z}})\,-\,\widehat{D}_{\hat{z}}}{\widehat{D}_{\hat{z}}}\right]^2
\,d\mu_{\widehat{h}_{\hat{z}}}\,,
\end{equation} 
where we exploited (\ref{phiareadist}). As the notation suggests, $E_{\widehat{\mathbb{C\,S}}_{\hat{z}},\;{\mathbb{C\,S}}_{\hat{z}}}[\zeta_{\hat{z}}]$, is a functional of the M\"obius map $\zeta_{\hat{z}}$ controlling, on the FLRW celestial sphere $\widehat{\mathbb{CS}}_{\hat{z}}$,  the conformal distortion between the FLRW metric $\widehat{h}_{\hat{z}}$ and the pulled back physical metric 
\begin{equation}
\zeta_{(\,{\widehat{z}}\,)}^*h_{\hat{z}}\,=\,
\label{fiacca0bisse}
\left(\tfrac{1\,+\,v_{\,\widehat{z}}\,(p)}{1\,-\,v_{\,\widehat{z}}\,(p)}\right)\,
{\Phi}_{\hat{z}}^2(\zeta_{(\,\widehat{z}\,)})\,{\widehat{h}}_{\hat{z}}\circ\,\widehat{\exp}_p\,,
\end{equation}
defined by (\ref{fiacca0}).
Note that one can easily rephrase the characterization of the comparison functional $E_{\widehat{\mathbb{C\,S}}_{\hat{z}},\;{\mathbb{C\,S}}_{\hat{z}}}[\zeta_{\hat{z}}]$ in terms of the map $\psi_{\hat{z}}$ defined by (\ref{psimap}), acting between the sky sections $\widehat{\Sigma}_{\hat{z}}$ and ${\Sigma}_{\hat{z}}$.  

\begin{lemma}
\label{EnMapping}
As an immediate consequence of the preceding remarks, we have
\begin{equation}
E_{\widehat{\mathbb{C\,S}}_{\hat{z}},\;{\mathbb{C\,S}}_{\hat{z}}}[\zeta_{\hat{z}}]\,=\,
E_{\widehat\Sigma,\,\Sigma}[\psi_{(\widehat{z})}]\,.
\end{equation}
\end{lemma}

\begin{proof}
By exploiting the pull back action associated with $\widehat{\exp}_p$ we can write 
\begin{align}
\label{enfunct1new}
E_{\widehat{\mathbb{C\,S}}_{\hat{z}},\;{\mathbb{C\,S}}_{\hat{z}}}[\zeta_{\hat{z}}]\,&:=\,\int_{\widehat{\mathbb{CS}}_{\hat{z}}}\left({\Phi}\,-\,1 \right)^2\,d\mu_{\widehat{h}_{(\widehat{z})}}\\
&=\,
\int_{\widehat{\mathbb{CS}}_{\hat{z}}}\left({\Phi}(\widehat{\exp}_p(y))\,-\,1 \right)^2\,(\widehat{\exp}_p)^*d\mu_{\hat{g}_{(\widehat{z})}}\nonumber\\
&=
\int_{\widehat{\exp}_p\left(\widehat{\mathbb{CS}}_{\hat{z}}\right)}\left({\Phi}(\widehat{q})\,-\,1 \right)^2\,d\mu_{\hat{g}_{(\widehat{z})}}(\widehat{q})
\nonumber\\
&=\,\int_{\widehat{\Sigma}_{\hat{z}}}\left({\Phi}\,-\,1 \right)^2\,d\mu_{\hat{g}_{(\widehat{z})}}\,=\,E_{\widehat\Sigma,\,\Sigma}[\psi_{(\widehat{z})}]\,.
\end{align}
\end{proof}
 
Lemma \ref{EnMapping} shows that $E_{\widehat{\mathbb{C\,S}}_{\hat{z}},\;{\mathbb{C\,S}}_{\hat{z}}}[\zeta_{(\widehat{z})}]$ can be used to compare sky sections at given redshift, and eventually, as the redshift varies, it compares the physical and the reference FLRW past light cones in the pre-homogeneity region\cite{CarFam, CarFam2}.

From a geometric analysis point of view, the (bi)-Lipschitz continuous exponential map $\exp_p$ admits a Sobolev $(1,2)$-norm, requiring  that $\exp_p$ and its first differential be square integrable. This property implies that the functional $E_{\widehat{\mathbb{C\,S}}_{\hat{z}},\;{\mathbb{C\,S}}_{\hat{z}}}[\zeta_{(\widehat{z})}]$  is well-defined in the assumed bi-Lipschitzian setting here adopted, and that it 
can be minimized over  $\zeta_{(\widehat{z})}\in\mathrm{PSL}(2, \mathbb{C})$ providing a redshift-dependent distance functional between the reference FLRW celestial sphere  $(\widehat{\mathbb{C\,S}}_{\hat{z}},\,\widehat{h}_{\hat{z}})$ and the physical celestial sphere $({\mathbb{C\,S}}_{\hat{z}},\, h_{\hat{z}})$, or, equivalently, between the corresponding sky sections $\widehat{\Sigma}_{\hat{z}}$ and $\Sigma_{\hat{z}}$.  We have the following result.

\begin{theorem} (The distance functional at redshift $\widehat{z}$ \cite{CarFam, CarFam2}).
\label{distheorem} 

For any given FLRW reference redshift $\widehat{z}$, the functional $E_{\widehat{\mathbb{C\,S}}_{\hat{z}},\;{\mathbb{C\,S}}_{\hat{z}}}[\zeta_{(\widehat{z})}]$  can be minimized over $\mathrm{PSL}(2, \mathbb{C})$, and 
\begin{equation}
\label{distCelSph}
d_{\hat{z}}\left[\widehat{\mathbb{C\,S}}_{\hat{z}},\;{\mathbb{C\,S}}_{\hat{z}}\right]\,:=\,\inf_{\zeta_{(\widehat{z})}\,\in \,\mathrm{PSL}(2, \mathbb{C})}\;\,E_{\widehat{\mathbb{C\,S}}_{\hat{z}},\;{\mathbb{C\,S}}_{\hat{z}}}[\zeta_{(\widehat{z})}] 
\end{equation}
defines a redshift-dependent distance functional between the reference FLRW celestial sphere  $(\widehat{\mathbb{C\,S}}_{\hat{z}},\,\widehat{h}_{\hat{z}})$ and the physical celestial sphere $({\mathbb{C\,S}}_{\hat{z}},\, h_{\hat{z}})$.\,\; Since
\begin{equation}
E_{\widehat{\mathbb{C\,S}}_{\hat{z}},\;{\mathbb{C\,S}}_{\hat{z}}}[\zeta_{\hat{z}}]\,=\,
E_{\widehat\Sigma,\,\Sigma}[\psi_{(\widehat{z})}]\,.
\end{equation}
 the distance functional $d_{\hat{z}}[\widehat{\mathbb{C\,S}}_{\hat{z}},\;{\mathbb{C\,S}}_{\hat{z}}]$ extends to a corresponding functional between the sky sections $(\widehat\Sigma_{\hat{z}},\,\widehat{g}^{\,(2)}_{\hat{z}})$ and $(\Sigma_{\hat{z}},\,{g}^{\,(2)}_{\hat{z}})$.
\end{theorem}

\begin{proof}
Without loss in generality, we can safely assume that, for the given $\widehat{z}$,
\begin{equation}
\label{assumedPHI}
\int_{\widehat{\mathbb{CS}}_{\hat{z}}}\,
\frac{D_{{z}}(\zeta_{\hat{z}},\,v_{\hat{z}})}{\widehat{D}_{\hat{z}}}
\,d\mu_{\widehat{h}_{\hat{z}}}\,>\,0\,.
\end{equation}
Since we can write (see (\ref{ConfArDist}))
\begin{equation}
\label{ConfArDistNew}
\Phi_{\,\widehat{z}}(\zeta_{\,(\widehat{z})\,})\,=\,\frac{D_{\,\widehat{z}}\,(\zeta_{(\widehat{z})},\,v_{\;\widehat{z}})}{\widehat{D}_{\,\widehat{z}}}\,,
\end{equation} 
the condition (\ref{assumedPHI}) rules out an almost everywhere collapsing behavior of the conformal factor $\Phi_{\,\widehat{z}}(\zeta_{\,(\widehat{z})\,})$ induced by the $\mathrm{PSL}(2, \mathbb{C})$ map 
$\zeta_{\hat{z}}\,:\,\widehat{\mathbb{C\,S}}_{\hat{z}}\,\longrightarrow\,{\mathbb{C\,S}}_{\hat{z}}$. In particular,  let $\widehat{q}_{(j)}\,=\,\widehat{\exp}_p(y_{(j)})$,\; $j=1,2,3$,\; denote the three selected sources on the reference FLRW celestial sphere $\widehat{\mathbb{C\,S}}_{\hat{z}}$, and let
$\zeta_{(\widehat{z})}^{(k)}(\widehat{q}_{(j)})\,=\,{q}_{(j)}(k)\,=\,{\exp}_p(x_{(j)}(k))$,\; be the corresponding images in ${\mathbb{C\,S}}_{\hat{z}}$ . Then, if the sequences $\{ {x}_{(j)}(k)\}$ \;  converge to three distinct points $\{\overline{x}_{(j)}\}$, \; $j=1,2,3$,\; and (\ref{ConfArDistNew}) holds,  the sequence $\{\zeta_{(\widehat{z})}^{(k)}\}_{k\in\mathbb{N}}$ is equicontinuous and converges to the unique $\mathrm{PSL}(2, \mathbb{C})$ map $\overline{\zeta}_{(\widehat{z})}$ that takes ${y}_{(j)}$ to $\{\overline{x}_{(j)}\}$, \; $j=1,2,3$.\;\, This is the configuration of interest in mapping our celestial spheres.

Under these conditions a minimizing sequence  $\{\zeta_{(\widehat{z})}^{(k)}\}_{k\in\mathbb{N}}\,\in\,\mathrm{PSL}(2, \mathbb{C})$\;
for the functional $E_{\widehat{\mathbb{C\,S}}_{\hat{z}},\;{\mathbb{C\,S}}_{\hat{z}}}[\zeta_{\hat{z}}]$ is equicontinuous, and  by selecting a subsequence, we may assume that 
$\{\zeta_{(\widehat{z})}^{(k)}\}_{k\in\mathbb{N}}$ converges to a smooth map
$\overline{\zeta}_{(\widehat{z})}$. Thus, given $\delta\,>\,0$,\, there exists $k_0$ such that for all $k\,\geq\,k_0$, we have
 
\begin{equation}
\label{plsemi}
E_{\widehat{\mathbb{C\,S}}_{\hat{z}},\;{\mathbb{C\,S}}_{\hat{z}}}[\overline{\zeta}_{(\widehat{z})}]\,\leq\,E_{\widehat{\mathbb{C\,S}}_{\hat{z}},\;{\mathbb{C\,S}}_{\hat{z}}}[\zeta_{(\widehat{z})}^{(k)}]\,+\,\delta\;,
\end{equation}
\vskip 0.2cm\noindent
for all $k\,\geq\,k_0$.  Since the choice of $\delta\,>\,0$ is arbitrary, (\ref{plsemi}) implies  that the functional $E_{\widehat{\mathbb{C\,S}}_{\hat{z}},\;{\mathbb{C\,S}}_{\hat{z}}}[\zeta_{(\widehat{z})}^{(k)}]$ is lower semicontinuous,  \textit{i.e.},
\begin{equation}
E_{\widehat{\mathbb{C\,S}}_{\hat{z}},\;{\mathbb{C\,S}}_{\hat{z}}}[\overline\zeta_{(\widehat{z})}]\,\leq\,\lim_k\inf\,E_{\widehat{\mathbb{C\,S}}_{\hat{z}},\;{\mathbb{C\,S}}_{\hat{z}}}[\zeta_{(\widehat{z})}^{(k)}]\,.  
\end{equation}
Hence,  $\{\zeta_{(\widehat{z})}^{(k)}\}\,\longrightarrow\,\overline\zeta_{(\widehat{z})}$ minimizes  $E_{\widehat{\mathbb{C\,S}}_{\hat{z}},\;{\mathbb{C\,S}}_{\hat{z}}}[\zeta_{(\widehat{z})}]$ over $\mathrm{PSL}(2, \mathbb{C})$, as stated.

If, according to (\ref{distCelSph}) we set
\begin{equation}
\label{distCelSph2}
d_{\hat{z}}\left[\widehat{\mathbb{C\,S}}_{\hat{z}},\;{\mathbb{C\,S}}_{\hat{z}}\right]\,:=\,\inf_{\zeta_{(\widehat{z})}\,\in \,\mathrm{PSL}(2, \mathbb{C})}\;\,E_{\widehat{\mathbb{C\,S}}_{\hat{z}},\;{\mathbb{C\,S}}_{\hat{z}}}[\zeta_{(\widehat{z})}]\,=\,
E_{\widehat{\mathbb{C\,S}}_{\hat{z}},\;{\mathbb{C\,S}}_{\hat{z}}}[\overline\zeta_{(\widehat{z})}]\,,
\end{equation}
then, by exploiting the rewriting of the functional 
$E_{\widehat{\mathbb{C\,S}}_{\hat{z}},\;{\mathbb{C\,S}}_{\hat{z}}}[\overline\zeta_{(\widehat{z})}]$ in  terms of the area distances, is not difficult to prove that $d_{\hat{z}}\left[\widehat{\mathbb{C\,S}}_{\hat{z}},\;{\mathbb{C\,S}}_{\hat{z}}\right]$ provides 
a redshift-dependent distance functional 
between the celestial spheres  $(\widehat{\mathbb{C\,S}}_{\hat{z}},\,\widehat{h}_{\hat{z}})$ and $({\mathbb{C\,S}}_{\hat{z}},\,{h}_{\hat{z}})$. Explicitly, we have
\begin{itemize}
\item{(Non-negativity)

\begin{equation} 
d_{\hat{z}}\left[\widehat{\mathbb{C\,S}}_{\hat{z}},\;{\mathbb{C\,S}}_{\hat{z}}\right]\,\geq\,0\,;
\end{equation}}
\item{(Isometry)

\begin{equation}
d_{\hat{z}}\left[\widehat{\mathbb{C\,S}}_{\hat{z}},\;{\mathbb{C\,S}}_{\hat{z}}\right]\,=\,0\,,
\end{equation}
iff $(\widehat{\mathbb{C\,S}}_{\hat{z}},\,\widehat{h}_{\hat{z}})$ and $(\mathbb{C\,S}_{\hat{z}},\,\zeta_{(\widehat{z})}^*{h}_{\hat{z}})$ are isometric\,; }
\item{(Symmetry)

\begin{equation}
d_{\hat{z}}\left[\widehat{\mathbb{C\,S}}_{\hat{z}},\;{\mathbb{C\,S}}_{\hat{z}}\right]\,=\,
d_{\hat{z}}\left[{\mathbb{C\,S}}_{\hat{z}},\;\widehat{\mathbb{C\,S}}_{\hat{z}}\right]\,;
\end{equation}}
\item{(Triangular inequality)

\begin{equation} 
d_{\hat{z}}\left[\widehat{\mathbb{C\,S}}_{\hat{z}},\;\widetilde{\mathbb{C\,S}}_{\hat{z}}\right]\,\leq\,
d_{\hat{z}}\left[\widehat{\mathbb{C\,S}}_{\hat{z}},\;{\mathbb{C\,S}}_{\hat{z}}\right]\,
+\,d_{\hat{z}}\left[{\mathbb{C\,S}}_{\hat{z}},\;\widetilde{\mathbb{C\,S}}_{\hat{z}}\right]
\,,
\end{equation}
where $\widetilde{\,}\,$ refers to a third candidate spacetime (either another FLRW selected for enlarging the possible choices among potential FLRW best-fitting candidates or another physical spacetime among those candidates that the astrophysical observations may suggest). 
}
\end{itemize}
\end{proof}

According to  (\ref{Farea}), the area element associated with the metric $\widehat{h}_{\hat{z}}$ on $\widehat{\mathbb{CS}}_{\hat{z}}(\,p)$ is given by 
$d\mu_{\widehat{h}_{\hat{z}}}\,=\,(1\,+\,\widehat{z})^{-2}\,f^2(\widehat{r}(\,\widehat{z}))\,d\mu_{\widehat{\mathbb{S}}^2}$, where $d\mu_{\widehat{\mathbb{S}}^2}$ denotes the solid angle measure on the FLRW observer's celestial sphere $\widehat{\mathbb{CS}}_{\hat{z}}$, where the radial density is constant. It follows that if we consider a model FLRW spacetime with flat spatial sections, then $f^2(\widehat{r}(\,\widehat{z}))\,=\,\widehat{r}^2(\,\widehat{z})$ and we can profitably rewrite the functional $E_{\widehat{\mathbb{C\,S}}_{\hat{z}},\;{\mathbb{C\,S}}_{\hat{z}}}[\zeta_{(\widehat{z})}]$ in terms of the mean square relative fluctuation $\delta^{\;(2)}_{\widehat{\mathbb{CS}}}(\,\widehat{z}\,)$  of the area distance (see (\ref{meansquareFunc0})) according to
\begin{equation}
\label{rewd2}
E_{\widehat{\mathbb{C\,S}}_{\hat{z}},\;{\mathbb{C\,S}}_{\hat{z}}}\,=\,
\frac{4\pi\,\widehat{r}^{\;2}_{(\widehat{z})}}{\left(1\,+\,\widehat{z}\right)^{\;2}}
\,\delta^{\;(2)}_{\widehat{\mathbb{CS}}}(\,\widehat{z}\,)\,.
\end{equation}
 From this rewriting and (\ref{FactorizationN}), we easily obtain the following physical characterization of the comparison functional $E_{\widehat{\mathbb{C\,S}}_{\hat{z}},\;{\mathbb{C\,S}}_{\hat{z}}}$.

\begin{theorem}
\label{FactorEnergyD}
The functional $E_{\widehat{\mathbb{C\,S}}_{\hat{z}},\;{\mathbb{C\,S}}_{\hat{z}}}[\zeta_{(\widehat{z})}$ minimized on the celestial sphere $\widehat{\mathbb{C\,S}}_{\widehat{z}\,(c)}$ associated with the cosmological uncoupling surface $\partial\,V_{(c)} (p)$ is given by 
\begin{equation}
\label{disFunctEst}
d_{\widehat{z}\,(c)}\left[\widehat{\mathbb{C\,S}}_{\widehat{z}\,(c)},\;{\mathbb{C\,S}}_{\widehat{z}\,(c)}\right]\,=\,
E_{\widehat{\mathbb{C\,S}}_{\widehat{z}\,(c)},\;{\mathbb{C\,S}}_{\widehat{z}\,(c)}}[\overline\zeta_{(\widehat{z}\,(c))}]\,=\,
\frac{2\pi}{3}\,
\frac{\widehat{\Omega}_{\widehat\Lambda}}{H_0^2}\,\left(\frac{\widehat{\Lambda}-\Lambda^{(v)}}{\widehat{\Lambda}}\right)\,\widehat{z}^{\,4}(c)\,.
\end{equation}
In particular, if in the pre-homogeneity region  $(\frac{\Delta_{\widehat{z}\,(c)}\,\rho}{\widehat{\rho}_{\widehat{z}\,(c)}})^2\,\not=\,0$,\, then we have the lower bound
\begin{equation}
\label{Dbound}
d_{\widehat{z}\,(c)}\left[\widehat{\mathbb{C\,S}}_{\widehat{z}\,(c)},\;{\mathbb{C\,S}}_{\widehat{z}\,(c)}\right]\,>\,
\pi\,\left(\frac{\widehat{\Omega}_m}{24 H_0}\right)^2\,\left(\frac{\Delta_{\widehat{z}\,(c)}\,\rho}{\widehat{\rho}_{\widehat{z}\,(c)}}\right)^2\,\left(1\,+\,\widehat{z}\,(c)\right)^2\,\widehat{z}^{\,6}(c)\,.
\end{equation}
\end{theorem}
\noindent
This result implies that as long as there is a non-trivial matter density contrast in the pre-homogeneity region  (and for a cluster, this contrast can be very large), the distance functional $d_{\widehat{z}\,(c)}\left[\widehat{\mathbb{C\,S}}_{\widehat{z}\,(c)},\;{\mathbb{C\,S}}_{\widehat{z}\,(c)}\right]$ is different from zero.  
    
 \begin{proof}
 The relation (\ref{disFunctEst}) follows from a straightforward substitution of the factorization  (\ref{FactorizationN}) of $\delta^{\;(1)}_{\widehat{\mathbb{CS}}}(\,\widehat{z}\,(c))$ into (\ref{rewd2}), and of the standard low redshift Hubble formula $\widehat{z}=H_0\,\widehat{r}(\,\widehat{z})$. The lower bound is a consequence of the variance inequality (see (\ref{samplevar0}) and (\ref{samplevar0i}) ) 
\begin{align}
\label{samplevar0iii}
\mathrm{Var}_{\widehat{\mathbb{CS}}}(Y_{(\,\widehat{z}\,(i))})\,&:=\,
\frac{1}{4\pi}
\int_{\widehat{\mathbb{CS}}}
\left(Y_{(\,\widehat{z}\,(i))}\,-\,\delta^{\;(1)}_{\widehat{\mathbb{CS}}}(\,\widehat{z}\,(i))\right)^2
	d\mu_{\widehat{\mathbb{S}}^2}\nonumber\\
\\
&=\,\delta^{\;(2)}_{\widehat{\mathbb{CS}}}(\,\widehat{z}\,(i))\,-\,\left(\delta^{\;(1)}_{\widehat{\mathbb{CS}}}(\,\widehat{z}\,(i)) \right)^2\,>\,0.\nonumber
\end{align}
The bound (\ref{Dbound}) follows from a direct computation exploiting the factorization 
\begin{equation}
\delta^{\;(1)}_{\widehat{\mathbb{CS}}}(\,\widehat{z}\,(c))\,=\,\frac{\widehat{\Omega}_m}{48}\,\left(\frac{\Delta_{\widehat{z}\,(c)}\,\rho}{\widehat{\rho}_{\widehat{z}\,(c)}}\right)\,\left(1\,+\,\widehat{z}\,(c)\right)^2\,\widehat{z}^{\,2}(c)
\end{equation}
provided by (\ref{FactorizationN}).     
 \end{proof}

\end{document}